\documentclass[nocopyrightspace, preprint, 9pt]{sigplanconf}

\usepackage{amsmath, amssymb}
\usepackage{amsfonts}
\usepackage{subfigure}
\usepackage{wrapfig}
\usepackage{amsthm}
\usepackage{multirow}
\usepackage{mathtools}
\usepackage{mathrsfs}
\usepackage{stmaryrd}
\usepackage{bm}
\usepackage[ligature,inference]{semantic}
\usepackage[backref]{hyperref}

\usepackage{caption}

\newtheorem{definition}{\bf Definition} 
\newtheorem{theorem}{\bf Theorem}       
\newtheorem{corollary}{\bf Corollary}   [section]
\newtheorem{lemma}{\bf Lemma}           [section]

\def\dups [#1] [#2]                          {\text{{#1}$^{\mathit{#2}}$}}

\def\kStepArrow [#1]                         {\text{$\overset{{#1}}{\rightarrow}$}}

\newcommand{\myNewLine}     {\vspace{3 mm}}

\mathlig{->}{\rightarrow}
\mathlig{->*}{\overset{*}{\rightarrow}}
\mathlig{=>}{\Rightarrow}

\def \state {m}
\def \vals  {\sigma}

\begin{document}

\setlength{\pdfpageheight}{\paperheight}
\setlength{\pdfpagewidth}{\paperwidth}


\conferenceinfo{}{}
\copyrightyear{}
\copyrightdata{}
\doi{}





\title{A Formal Study on Backward Compatible Dynamic Software Updates}
\subtitle{}

\authorinfo{Jun Shen}
           {Arizona State University}
           {jun.shen.1@asu.edu}
\authorinfo{Rida A. Bazzi}
           {Arizona State University}
           {bazzi@asu.edu}

\maketitle

\begin{abstract}
We study the dynamic software update problem for programs interacting with an environment that is not necessarily updated.
We argue that such updates should be backward compatible.
We propose a general definition of backward compatibility and cases of backward compatible program update.
Based on our detailed study of real world program evolution, we propose classes of backward compatible update for interactive programs, which are included at an average of 32\% of all studied program changes.
The definitions of update classes are parameterized by our novel framework of program equivalence, which generalizes existing results on program equivalence to non-terminating executions.
Our study of backward compatible updates is based on a typed extension of {\em W} language.
\end{abstract}

\category{D.3.1}{Formal Definitions and Theory}{Semantics, Syntax}
\category{D.2.4}{Software/Program Verification}{Correctness Proof, Formal Methods}
\category{F.3.2}{Semantics of Programming Languages} {Operational Semantics, Program Analysis}
\category{D.3.3}{Language Constructs and Features}{Input/output, Procedures, functions, and subroutines}

\terms{Theory}

\keywords{dynamic software update, backward compatibility, program equivalence, proof rule, operational semantic}

\tableofcontents

\section{Introduction}

Dynamic software update (DSU) allows programs to be updated in the middle of their execution by mapping a state of an old
version of the program to that of a newer version. The ability to update programs without having to restart them
is useful for high-availability applications  that cannot afford the downtime  incurred by offline updates~\cite{hicks_thesis}. DSU has been an active area of research\cite{hicks_thesis,upstare_dissertation_09,linux_hotswapping,ksplice_eurosys_09}
with much of the published work emphasizing the {\em update mechanism} that implements a {\em state mapping}
which maps the execution state of an old version of the program to that of a new version.
DSU {\em safety} has not yet been successfully studied. Existing studies on DSU safety are lacking in one way
or another:
high-level studies are concerned with change management for system components~\cite{KramerConsistentDSU, BloomCorrectDSU} and 
lower-level studies typically require significant programmer annotations~\cite{hayden12dsucorrect,Magill_automap,ZhangCorrectness}
 or have a restricted class of applications to which they apply (e.g., controller systems~\cite{PanzicaLaManna_criteria}).

In this paper, we consider the safety of DSU when applied to possibly non-terminating programs interacting with an environment that is not necessarily updated. For such updates, the new version of the program must be able to interact with the old environment, which means that it should be, in some sense, {\em backward compatible} with the old version. A strict definition of backward compatibility would require the new version to exhibit the same I/O behavior as
the old version; in other words the two programs are observationally equivalent. It should be immediately clear that
a more nuanced definition is needed because observational-equivalence does not allow changes which one would want to allow as backward compatible such as bug fixes, new functionalities, or usability improvement (e.g., improved user messages). Allowing for such differences would be needed in any practical definition of backward compatibility.
One contribution of this work is a general definition of backward compatibility, a classification of common backward compatible program behavior changes, as well as classes of program change from  real world program evolution.

Determining backward compatibility, which allows for differences between two program versions, is requiring one to solve the {\em semantic equivalence} problem which has been extensively studied~\cite{Horwitz88,Godlin10,Benton04,Kundu09,staticaffineprog-equivcheck,arrayLoopTransformation13,optimizationCorrectness02,lucanu13}. Unfortunately, existing results turned out to be lacking in one or more aspects
which rules out retrofitting them for our setting.
In fact, existing work on program equivalence typically guarantees equivalence at the end of an execution.
Such equivalence is not adequate for our purposes because it does not allow us to express that a point in the middle
of a loop execution of one program {\em corresponds} (in a well defined sense) to a point in the middle of a loop execution of another program. The ability to express such correspondences is desirable for dynamic software update.
Besides, existing formulations of the
program equivalence problem either do not use formal semantics~\cite{Horwitz88,Binkley89,arrayLoopTransformation13}, only apply to terminating programs~\cite{Benton04,Kundu09}, severely restrict the programming model~\cite{Godlin10,staticaffineprog-equivcheck,arrayLoopTransformation13}, or rely on some form of model checking~\cite{optimizationCorrectness02,lucanu13} (which is not appropriate
for non-terminating programs with infinite states). Our goal for program equivalence is to establish compile-time conditions ensuring that two programs have the same I/O behavior in {\em all} executions. In particular, if one program enters an infinite loop and does not produce a certain output, the other program should not produce that output either. This is different from much of the literature on program equivalence which only guarantees same behavior in terminating executions.

The closest work that aims to establish program equivalence for nonterminating programs is that of Godlin and Strichman~\cite{Godlin10} who give sufficient conditions for semantic equivalence for a language that includes recursive functions, but does not allow loops (loops are extracted as recursive functions).
That and the fact that equivalence is enforced on corresponding functions severely limits the applicability of the work to general transformations affecting loops such as loop-invariant code motion, loop fission/fusion. So, as a major component of our formal treatment of backward compatible updates, we set out to develop sufficient conditions for semantic equivalence for programs in a typed extension of the {\emph W} languages~\cite{Cartwright:semanticOfDep} with small-step operational semantics. The syntax of language is extended with arrays and enumeration types and the semantics take into consideration the execution environment to allow various classes of updates.

In summary, the paper makes the following contributions:
\begin{enumerate}
\item We formally define backward compatibility and identify cases of backward compatible program behavior for typical program update motivation.
\item We identify and formally define classes of program changes that result in backward compatible program update based on empirical study of real world program evolution.
\item We give a formal treatment of the semantic equivalence for nonterminating imperative programs.
\end{enumerate}

The rest of the paper is organized as follows.
Section~\ref{sec:backwardComptibility} proposes the general backward compatibility and cases of backward compatible new program behavior.
Then we describe real world update classes that result in backward compatible update in Section~\ref{sec:briefUpdateClasses}.
Section~\ref{sec:formallanguage} formally defines our extension of the {\bf W} language to study backward compatible updates.
Section~\ref{sectionDefs} shows terms, notations and definitions (e.g., execution) heavily used in the technical result.
The technical results on semantic equivalence are presented in Section~\ref{sec:equivalence}.
We propose our formal treatment of real world update classes in Section~\ref{sec:realupdateclasses}.
A more detailed comparison to related work is given in Section~\ref{sec:relatedwork} .
Section~\ref{sec:conclusion} concludes the paper. 

\section{Backward compatibility}\label{sec:backwardComptibility}

\subsection{Programs and Specifications}

Programs are designed to satisfy specifications.
Specification can be explicitly provided or implicitly defined by the behavior of a program.
Programs interact with their environment by receiving inputs and producing outputs.
In this section we introduce enough of a computing model to describe the input/output behavior of programs;
 In the next section we introduce a  specific programming language to reason about specific software updates.

An execution of a program consists of a sequence of steps from a finite set of steps, $\mathcal{S} = \mathcal{S}_{in} \cup  \mathcal{S}_{internal} \cup \mathcal{S}_{out} \cup \{halt\}$.
A step of a program can either be an input step in which input is received, an internal step in which the state of the program
is modified, an output step in which output is produced, or a halt.

We make a distinction between internal state of a program and external state (e.g., application settings) of the local environment in which the program
executes. Such external state can include the state of a file system that program can access; we include both as part of
the program state.
The state of a program is an element of a set $\mathcal{M} \times \mathcal{I}$,
where the set $\mathcal{M} = \mathcal{M}_{int} \times \mathcal{M}_{ext}$,
$\mathcal{M}_{int} = \prod_{k=0}^{n_{int}} V_k$ is a cartesian product of $n_{int}$ sets of values, one for each internal memory location, and $\mathcal{M}_{ext} = \prod_{k=0}^{n_{ext}} V_k$ is a cartesian product of $n_{ext}$ sets of values, one for each external location.
The input value last received is an element of the set $\mathcal{I}$ of input values.

A program executes in an \emph{execution environment}.
An execution environment $(M_{{ext}_0}, I)$ specifies an initial value for the external program state $M_{{ext}_0}$ and a
 possibly infinite sequence of input values $I$. The input sequence is assumed to be produced by \emph{users} that we do not model explicitly.

A step of a program $P$ is a mapping that specifies the next program state and the next step to execute.
For an internal step $s_{internal} \in \mathcal{S}_{internal}$, the mapping is
$s_{internal} : \mathcal{M} \times \mathcal{I} \mapsto \mathcal{S} \times \mathcal{M} \times \{\bot\}$,
which specifies the next step and how the state is modified. The internal steps clear input in the state if any.
For an output step $s_{out} \in \mathcal{S}_{out}$,
the mapping $s_{out} : \mathcal{M} \mapsto  \mathcal{S} \times \mathcal{O}$
which specifies the next step to execute and the output value produced.
$\mathcal{O}$ is the set of output values produced by the program.
An input step $s_{in} \in \mathcal{S}_{in}$ is simply an element of  $\mathcal{S} \times \mathcal{I}$
and specifies the next step to execute and the input obtained from the environment. (We simply write $s_{in}()$ to denote the next step and the input received.)
Because the input value is received by the program, we do not restrict the next step to execute. We allow the input value to be ignored by the program by two consecutive input steps.
When the step is $halt$, there is no further action as if $halt$ were mapped to itself.

\begin{definition}\label{def:generalProg}
{\bf (Program)} A program $P$ is a tuple $(\mathcal{S}, \mathcal{M}, M_{{int}_0}, s_0,\allowbreak \mathcal{I}, \mathcal{O})$,
where $\mathcal{S}$ is the set of steps as defined above, $\mathcal{M}$ is the set of program states,
$M_{{int}_0}$ is the initial internal state, $s_0$ is the initial step, and $\mathcal{I}$ and $\mathcal{O}$ are disjoint sets of input and output values.
\end{definition}
We do not include the initial external state $M_{{ext}_0}$ in the program definition; we include it in the execution environment of $P$.

\begin{definition}\label{def:generalExec}
{\bf(Execution)}
An execution of a program $P = (\mathcal{S}, \mathcal{M}, M_{{int}_0}, s_0, \mathcal{I}, \mathcal{O})$ in execution environment $(M_{{ext}_0}, I)$, where $I$ is a possibly infinite sequence of input values from $\mathcal{I}$,
is a sequence of configurations $C$ from the infinite set $\{(M, s, i, I_r, IO)\}$.
A configuration $c$ has the form $c = (M, s, i, I_r,\allowbreak IO)$,
where $M$ is a state, $s$ is a step, $i$ is the last input received, $I_r$ is a sequence of remaining input values and $IO$ is the input/output sequence produced so far. The $k$th configuration $c_k$ in an execution is obtained from the $(k-1)$th configuration $c_{k-1} = (M, s, i, I_r, IO)$ where $s \neq halt$ in one of the following cases:
\begin{enumerate}
\item The first configuration $c_0$ is of the form $(M_0, s_0, \bot, I, \varnothing)$, where $M_0 =
(M_{{int}_0}, M_{{ext}_0})$;

\item $s \in \mathcal{S}_{internal}: c_k = (M', s', \bot, I_r, IO)$, where $(s', M', \bot) = s(M, i)$;

\item $s \in \mathcal{S}_{in}$ and the remaining inputs $I_r$ is not empty: $c_k = (M, s', head(I_r), tail(I_r), IO\cdot head(I_r))$ where $(s', head(I_r))\allowbreak = s(I_r)$;

\item $s \in \mathcal{S}_{in}$ and the remaining inputs $I_r$ is empty: $c_k = c_{k-1}$;

\item $s \in \mathcal{S}_{out}: c_k = (M, s', i, I_r, IO\cdot o')$, where $(s', o') = s(M)$;

\end{enumerate}
\end{definition}
In the definition, $head(I)$  denotes the head (leftmost) element in the sequence $I$ and $tail(I)$ denotes the remaining sequence without the head. The input value in $i$ is either consumed by the next internal step or updated by another input from the next input step. Execution is stuck if an input step is attempted in state in which there are no remaining inputs.
In what follows, we include the execution environment in the execution and we abuse notation to say $(M_{{ext}_0}, I, C)$ is
an execution of a program $P$.

\paragraph{Specifications}

We consider specifications that define the input/output behavior of programs.
Specifications are not concerned with how fast an output is produced or about the internal state of the program.

\begin{definition}\label{def:generalSpec} {\bf (Specification)}
Given a set $\mathcal{M}_{ext}$ of external states, a set $seq(\mathcal{I})$ of input sequences, and a set $seq(\mathcal{I} \cup \mathcal{O})$ of I/O sequences, specification $\Sigma$ is a predicate:
$\mathcal{M}_{ext} \times seq(\mathcal{I}) \times seq(\mathcal{I} \cup \mathcal{O}) \times \mapsto \{true, false\}$.
\end{definition}

We define the I/O sequence of a sequence of configurations $C$ to be a sequence $IO(C)$ of values from $\mathcal{I} \cup \mathcal{O}$ such that every finite prefix of $IO(C)$ is the IO sequence of some configuration $c \in C$ and every I/O sequence of a configuration $c \in C$ is a finite prefix of $IO(C)$.

An execution $(M_{{ext}_0}, I, C)$ of program $P$ satisfies a specification $\Sigma$ if $\Sigma(M_{{ext}_0}, I, IO(C)) = {true}$.
A specification distinguishes executions into those that satisfy the specification and those that do not.

A specification defines the external behavior of a program that is observed by a user.
The input sequence and I/O sequence are obviously part of external behavior.
We also include $\mathcal{M}_{ext}$ in specification domain because a user can have information about the external state.
For example, a user who has data stored in the file system considers the program's refusal to access the stored data a violation of the service specification; this is not the case if the user has no stored data.

\subsection{Hybrid executions and state mapping}

DSU is a process of updating software while it is running. This results in a hybrid execution in which part of the execution is that of the old program and part of the execution is for the new program. 

State mapping is a function $\delta$ mapping an internal state and a non-halt step of one program $P$ to an internal state and a step of another program $P'$, $\delta: \mathcal{M}_{int}^P \times (\mathcal{S}^P \setminus \{halt\}) \mapsto \mathcal{M}_{int}^{P'} \times \mathcal{S}^{P'}$.
The external state is not mapped because the environment is not necessarily updated. In addition, we cannot change input and output that
already occurred and that I/O must be part of the hybrid execution.

\begin{definition}\label{def:hybridExec} {\bf (Hybrid execution)}
A hybrid execution $(M_{{ext}_0}, I, \allowbreak C_P;C_{P'})$, produced by DSU using state mapping $\delta$ from program $P$
 to program $P'$,
is an execution $(M_{{ext}_0}, I, C_P)$ of $P$ concatenated with an execution $(M_{ext}', I_r', C_{P'})$ of $P'$
where the first configuration $c_{P'} = ((M_{int}', M_{ext}'), s', i', I_r', IO')$ in $C_{P'}$ is obtained by applying the state mapping to the last configuration $c_{P} = ((M_{int}, M_{ext}), s(\neq halt), i, I_r, IO)$ in $C_{P}$ as follows:
\begin{itemize}
\item $(M_{int}', s') = \delta(M_{int}, s)$;

\item $(i' = i) \wedge (I_r' = I_r) \wedge (IO' = IO) \wedge (M_{ext} \subseteq M_{ext}')$.
\end{itemize}
\end{definition}

\subsection{Backward compatibility}

In this paper, we consider updates in which the environment is not necessarily updated.
It follows that in order for the hybrid execution to be meaningful, the new program should provide functionality expected by both old and new users of the system. 

In practice, specifications are not explicitly available. Instead, the program is its own specification. This means that the specification that the program satisfies can only be inferred by the external behavior of the program.
Bug fixes create a dilemma for dynamic software updates. When a program has a bug, its external behavior does not captures its
 {\em implicit specification} and the update will change the behavior of the program. In what follows, we first discuss what flexibility we can be afforded for a backward compatible update and then we give formal definitions of backward compatibility and state our assumptions for allowing bug fixes.

We consider a hybrid execution starting from a program
$P = (\mathcal{S}, \mathcal{M}_{int} \times \mathcal{M}_{ext}, M_{{int}_0}, s_0, \mathcal{I}, \mathcal{O})$
and being updated to a program
$P' = (\mathcal{S}', \mathcal{M}_{int}' \times \mathcal{M}_{ext}', M_{{int}_0}', s_0', \mathcal{I}', \mathcal{O}')$.
We examine how the two programs should be related for a meaningful hybrid execution.
\begin{enumerate}
\item (Inputs)
Input set $\mathcal{I}'$ of $P'$ should be a superset of that $\mathcal{I}$ of $P$ to allow for {\em old users} to interact with $P'$
	after the update.
It is possible to allow for new input values in $\mathcal{I}'$ to accommodate new functionality under the assumption
	that old users do not generate new input values. Such new input values should be expected to produce erroneous output by old users
	as they are not part of $P$'s specification.

\item (Outputs)
Output produced by $P'$ should be identical to output produced by $P$ if all the input in an execution comes from the input set of $P$.
This is needed to ensure that interactions between old users and the program $P'$ can make sense from the perspective of old users.
This is true in the case that the update does not involve a bug fix, but what should be done if the update indeed involves a bug fix and
	the output produced by the old program was not correct to start with?
As far as syntax, a bug fix should not introduce new output values. As far as semantics, we should allow the bug fix to change what
	output is produced for a given input. We discuss this further under the bug fix heading.
In summary, if we ignore bug fixes, the new program should behave as the old program when provided with input meant for the old program.

\item (Bugfix)
Handling bug fixes is problematic. If the produced output already violates the fix, then there is no way for the hybrid execution
	to satisfy the {\em implicit semantics} of the program or the semantics of the new program.
Some bug fixes can be handled. For example, a bug that causes a program to crash for some input can be fixed to allow the program
	to continue executing. Applying the fix to a program that has not encountered the bug should not be problematic.
Another case is when the program should terminate for some input sequence, but the old program does not terminate.
A bug fix that allows the program to terminate should not present a semantic difficulty for old users.

In general, we assume that there are  {\em valid} executions and {\em invalid} executions of the old program.
I/O sequences produced in {invalid} executions are not in specification of the program.
We assume that an invalid execution will lead to an {\em error} configuration not explicitly handled by the program developers.
We do not expect the state mapping to change an error configuration into an non-error configuration just as static updating does not fix occurred errors. Besides, we do not attempt to determine if a particular configuration is an error configuration.
Such determination is not possible in general and very hard in practice.
We simply assume that the configuration at the time of the update is not an error configuration.
(which is equivalent to assuming the existence of an oracle $\mathcal{J}_P$ to determine if  a particular configuration is erroneous,
$\mathcal{J}_P(C_P)$ = true if the configuration $C_P$ is not erroneous).

\item (New functionality) New functionality is usually accompanied by new inputs/outputs and the expansion of external state. We assume that new functionality is independent of existing functionality in the sense that programs $P$ and $P'$ produce the same I/O sequence when receiving inputs in $\mathcal{I}$ only.
    We therefore assume all new inputs $\mathcal{I}'\setminus\mathcal{I}$ are introduced by new functionality.

Every external state of $P$ is part of some external state of program $P'$ because of the definition of the specification of $P$. We only consider expansion of the external state of $P$ for new functionalities in $P'$ where the expansion of external state is independent of values in existing external state. One of the motivating examples is to add application settings for new program feature.
\end{enumerate}

In light of the discussion above we give the following definition of backward compatibility in the absence of bug fixes.

\begin{definition}\label{def:bwdcomptGeneral} {\bf (Backward compatible hybrid executions)}
Let $P = (\mathcal{S}, \mathcal{M}_{int} \times \mathcal{M}_{ext}, M_{{int}_0}, s_0, \mathcal{I}, \mathcal{O})$ be a program satisfying a specification $\Sigma$.
We say that a hybrid execution $(M_{ext}, I, C_P;C_P')$ from $P$ to a program
$P' = (\mathcal{S}', \mathcal{M}_{int}' \times \mathcal{M}_{ext}', M_{{int}_0}', s_0', \mathcal{I}', \mathcal{O}')$ is backward compatible with implicit specification of $P$ if all of the following hold:
\begin{itemize}
\item The last configuration in $C_P$ is not an error configuration, $C_P = ``C';(M, s', i, I_r, IO)" : \mathcal{J}_P(C_P)$=true.

\item The hybrid execution satisfies the specification $\Sigma$ of $P$,

\noindent$\Sigma(M_{ext}, I, IO(C_P;C_P')) = true$;

\item Inputs/outputs/external states of $P$ are a subset of those of $P': \mathcal{I} \subseteq \mathcal{I}', \mathcal{O} \subseteq \mathcal{O}' \text{ and } \mathcal{M}_{ext} \subseteq \mathcal{M}_{ext}'$;
\end{itemize}
\end{definition}

If there is bug fix between programs $P'$ and $P$, we need to adapt Definition~\ref{def:bwdcomptGeneral} to allow for
some executions on input sequences from $\mathcal{I}$ to violate the specification of $P$. Above we identified two cases in
which bug fixes are safe (replacing a response with no response or replacing a no response with a correct response without
introducing new output values). We omit the definition.

We have the backward compatible updates by extending the definition of a backward compatible hybrid execution to all possible hybrid executions.
\begin{definition}\label{def:bcupdates} {\bf (Backward compatible updates)}
We say an updated program $P'$ is backward compatible with a program $P$ in configuration $C$ if there is hybrid execution, from configuration $C$ of $P$ to $P'$ that is backward compatible with specification of $P$.
\end{definition}

\subsection{Backward compatible program behavior changes}
With the formal definition of backward compatibility,
it is desirable to check what behavior changes of an updated program help ensure a safe update.
Backward compatibility is essentially a relation between I/O sequences produced by an old program and those produced by an updated program. We summarized typical possibilities of the relation into six cases in Figure~\ref{NewProgBehaviorFig} by considering consequence of major update motivation (i.e., new functionality, bug fix and program perfective/preventive needs~\cite{iso14764}).
According to David Parnas~\cite{Parnas94}, a program is updated to adapt to changing needs.
In other words, program changes are to produce more or less or different output according to changing needs. These changes are captured by case 2, 3, 4, 5 and 6 in Figure~\ref{NewProgBehaviorFig}.
We also capture output-preserving changes which are most likely motivated by the program developer's own needs (e.g., software maintainability),
which is case 1 in Figure~\ref{NewProgBehaviorFig}.

\begin{figure}
		\begin{center}
			\begin{tabular}{| c | l |}
                \hline
				Case          & Formal new program behavior                                         \\ \hline
                 1            & the old behavior including external state extension: \\
				              & $\Sigma_{P} \subseteq \Sigma_{P'}$,
                                 or $\Sigma_{P'}$ = \{$(M_{ext'}, \text{oneseq}(\mathcal{I}), \text{oneseq}(\mathcal{I}\cup\mathcal{O})) \rightarrow \text{val}$       \\
                              & $| \exists (M_{ext}, \text{oneseq}(\mathcal{I}), \text{oneseq}(\mathcal{I}\cup\mathcal{O})) \rightarrow \text{val}$
                                 in $\Sigma_{P}$ \\
                              &  and $M_{ext} \subseteq M_{ext'}$ \} where $\mathcal{I} = \mathcal{I'}$, $\mathcal{O} = \mathcal{O'}$
                                  and $\mathcal{M}_{ext} \subseteq \mathcal{M}_{ext'}$    \\ \hline

                 2            & the old behavior for old input and consuming inputs \\
                              & that are only from new clients:                               \\
				              & $\Sigma_{P} \subseteq \Sigma_{P'} \land$ \\
                              &   $\Sigma_{P'} \setminus \Sigma_{P}$ =
                                     \{$({M}_{ext}, \text{oneseq}(\mathcal{I'}), \text{oneseq}(\mathcal{I'} \cup \mathcal{O'})) \rightarrow \text{true}$ \\
                              &   $| \, \text{oneseq}(\mathcal{I'} \cup \mathcal{O'}) \text{ includes at least one input in } (\mathcal{I}' \setminus \mathcal{I})$\} $\neq \emptyset$ \\
                              & where $\mathcal{I} \subset \mathcal{I'}$, $\mathcal{O} \subseteq \mathcal{O'}$
                                  and $\mathcal{M}_{ext} = \mathcal{M}_{ext'}$    \\ \hline

                 3            & producing more output while the old program terminates:      \\
                              & $\Sigma_{P'} \setminus \Sigma_{P}$ =
                                   $\{(M_{ext}, \text{oneseq}(\mathcal{I}), \text{oneseq}(\mathcal{I}\cup\mathcal{O})) \mapsto \text{false}$  \\
                              & $| (M_{ext}, \text{oneseq}(\mathcal{I}), \text{oneseq}(\mathcal{I}\cup\mathcal{O})) \in \Delta_{f} \neq \emptyset\}$ \\
                              & $\cup$ $\{(M_{ext}, \text{oneseq}(\mathcal{I}), \text{oneseq}(\mathcal{I}\cup\mathcal{O})\centerdot\text{oneseq'}(\mathcal{I}\cup\mathcal{O})) \mapsto \text{true}$ \\
                              & $| (M_{ext}, \text{oneseq}(\mathcal{I}), \text{oneseq}(\mathcal{I}\cup\mathcal{O})\centerdot\text{oneseq'}(\mathcal{I}\cup\mathcal{O})) \in \Delta_{t} \neq \emptyset\}$ \\
                              & $\Sigma_{P} \setminus \Sigma_{P'}$ =
                                   $\{(M_{ext}, \text{oneseq}(\mathcal{I}), \text{oneseq}(\mathcal{I}\cup\mathcal{O})) \mapsto \text{false}$  \\
                              & $| (M_{ext}, \text{oneseq}(\mathcal{I}), \text{oneseq}(\mathcal{I}\cup\mathcal{O})) \in \Delta_{t} \neq \emptyset\}$ \\
                              & $\cup$ $\{(M_{ext}, \text{oneseq}(\mathcal{I}), \text{oneseq}(\mathcal{I}\cup\mathcal{O})\centerdot\text{oneseq'}(\mathcal{I}\cup\mathcal{O})) \mapsto \text{true}$ \\
                              & $| (M_{ext}, \text{oneseq}(\mathcal{I}), \text{oneseq}(\mathcal{I}\cup\mathcal{O})\centerdot\text{oneseq'}(\mathcal{I}\cup\mathcal{O})) \in \Delta_{f} \neq \emptyset\}$ \\
                              & where $\mathcal{I} = \mathcal{I'}$, $\mathcal{O} = \mathcal{O'}$
                                  and $\mathcal{M}_{ext} = \mathcal{M}_{ext'}$    \\ \hline

                 4            & termination while the old program produces erroneous output: \\
				              & $\Sigma_{P'} \setminus \Sigma_{P}$ =
                                   $\{(M_{ext}, \text{oneseq}(\mathcal{I}), \text{oneseq}(\mathcal{I}\cup\mathcal{O})) \mapsto \text{true}$  \\
                              & $| (M_{ext}, \text{oneseq}(\mathcal{I}), \text{oneseq}(\mathcal{I}\cup\mathcal{O})) \in \Delta_{t} \neq \emptyset\}$ \\
                              & $\cup$ $\{(M_{ext}, \text{oneseq}(\mathcal{I}), \text{oneseq}(\mathcal{I}\cup\mathcal{O})\centerdot\text{oneseq'}(\mathcal{I}\cup\mathcal{O})) \mapsto \text{false}$ \\
                              & $| (M_{ext}, \text{oneseq}(\mathcal{I}), \text{oneseq}(\mathcal{I}\cup\mathcal{O})\centerdot\text{oneseq'}(\mathcal{I}\cup\mathcal{O})) \in \Delta_{f} \neq \emptyset\}$ \\

                              & $\Sigma_{P} \setminus \Sigma_{P'}$ =
                                   $\{(M_{ext}, \text{oneseq}(\mathcal{I}), \text{oneseq}(\mathcal{I}\cup\mathcal{O})) \mapsto \text{true}$  \\
                              & $| (M_{ext}, \text{oneseq}(\mathcal{I}), \text{oneseq}(\mathcal{I}\cup\mathcal{O})) \in \Delta_{f} \neq \emptyset\}$ \\
                              & $\cup$ $\{(M_{ext}, \text{oneseq}(\mathcal{I}), \text{oneseq}(\mathcal{I}\cup\mathcal{O})\centerdot\text{oneseq'}(\mathcal{I}\cup\mathcal{O})) \mapsto \text{false}$ \\
                              & $| (M_{ext}, \text{oneseq}(\mathcal{I}), \text{oneseq}(\mathcal{I}\cup\mathcal{O})\centerdot\text{oneseq'}(\mathcal{I}\cup\mathcal{O})) \in \Delta_{t} \neq \emptyset\}$ \\
                              & where $\mathcal{I} = \mathcal{I'}$, $\mathcal{O} = \mathcal{O'}$
                                  and $\mathcal{M}_{ext} = \mathcal{M}_{ext'}$    \\ \hline

                5             & different output that is functionally equivalent to old output: \\
				              & $(\Sigma_{P} \neq \Sigma_{P'}) \land (\Sigma_{P} \equiv \Sigma_{P'})$     \\
                              &   where $\mathcal{I} = \mathcal{I'}$,
                                        $(\mathcal{O} \neq \mathcal{O'}) \land (\mathcal{O} \equiv \mathcal{O'})$
                                  and $\mathcal{M}_{ext} = \mathcal{M}_{ext'}$    \\ \hline
                                  
                6             & enforcing restrictions on program state:    \\
				              & $\Sigma_{P'} \setminus \Sigma_{P}$ =
                                   $\{(M_{ext}, \text{oneseq}(\mathcal{I}), \text{oneseq}(\mathcal{I}\cup\mathcal{O})) \mapsto \text{false}$  \\
                              & $| (M_{ext}, \text{oneseq}(\mathcal{I}), \text{oneseq}(\mathcal{I}\cup\mathcal{O})) \in \Delta_{arbi} \neq \emptyset\}$ \\
                              & $\Sigma_{P} \setminus \Sigma_{P'}$ =
                                   $\{(M_{ext}, \text{oneseq}(\mathcal{I}), \text{oneseq}(\mathcal{I}\cup\mathcal{O})) \mapsto \text{true}$  \\
                              & $| (M_{ext}, \text{oneseq}(\mathcal{I}), \text{oneseq}(\mathcal{I}\cup\mathcal{O})) \in \Delta_{arbi} \neq \emptyset\}$ \\
                              & where $\mathcal{I} = \mathcal{I'}$, $\mathcal{O} = \mathcal{O'}$
                                  and $\mathcal{M}_{ext} = \mathcal{M}_{ext'}$    \\ \hline                                  
			\end{tabular}
		\end{center}
		\caption{Six cases of formalized general new program behavior}\label{NewProgBehaviorFig}
\end{figure}

Furthermore, we find that an update is {\em backward compatible} if in every execution the new program behavior is one of the six cases in Figure~\ref{NewProgBehaviorFig}.
Cases 1 and 2 are obviously backward compatible because an old client is guaranteed to get old responses.
Cases 3, 4, 5, and 6 are not obviously backward compatible.
Unlike case 1 and 2, case 3, 4 and 5 are backward compatible under specific assumptions on program semantics while case 6 is different.
Case 3 is backward compatible because we assume the change is either adding new functionality, or fixing a bug in which the old program hanged or crashed. Similarly, case 4 is backward compatible.
Case 5 is backward compatible because different I/O interaction could express the same application semantics. For example, a greeting message could be changed from ``hi" to ``hello".
Case 6 is backward compatible in that the new program makes implicit specification of the program explicit by enforcing restrictions on program state and therefore eliminating undesired I/O sequence.  

The six cases in Fig.~\ref{NewProgBehaviorFig} have covered the changes of output, including more or less or different output.
There exists more specific cases of backward compatible program behavior changes under various specific assumptions.
However, these more specific cases could be attributed to one of the six cases as far as the changes of output are concerned.
In conclusion, it is not possible to go much beyond the six cases of backward compatibility in Fig.~\ref{NewProgBehaviorFig}.

\section{Real world backward compatible update classes: brief description}\label{sec:briefUpdateClasses}
We have studied evolution of three real world programs (i.e., vsftpd, sshd and icecast)
 to identify real world changes that are backward compatible.
 We chose these three programs because the programs are widely used in practice~\cite{vsftpdweb,sshusers} and are widely studied in the DSU community~\cite{upstare_usenix_09,iuliandissertation}.
We have studied several years of releases of vsftpd and consecutive updates of sshd and icecast. This is because vsftpd is more widely studied by the DSU community~\cite{upstare_usenix_09,iuliandissertation, upstare_dissertation_09}.

Our study of real world program evolution is carried out as follows.
We examined every changed function
manually to classify updates. For every individual change, we first identified the motivation of the change, then the assumptions under which the change could be considered backward compatible. If the assumption under which the change is considered backward compatible is reasonable, we recorded the change into one particular update class. Finally we summarized common update classes observed in the evolution of studied programs.

\begin{figure}
		\begin{center}
			\begin{tabular}{l | c | r | r}
				Software version                 & Update date          & Total         & Class       \\ \hline
				vsftpd 1.1.0 \textendash 1.1.1   & 2002-10-07           & 16            & 8           \\
				vsftpd 1.1.1 \textendash 1.1.2   & 2002-10-16           & 8             & 1           \\
				vsftpd 1.1.2 \textendash 1.1.3   & 2002-11-09           & 8             & 4           \\
				vsftpd 1.1.3 \textendash 1.2.0   & 2003-05-29           & 61            & 9           \\
				vsftpd 1.2.0 \textendash 1.2.1   & 2003-11-13           & 33            & 11          \\
				vsftpd 1.2.1 \textendash 1.2.2   & 2004-04-26           & 10            & 6           \\
				vsftpd 1.2.2 \textendash 2.0.0   & 2004-07-01           & 52            & 13          \\
				vsftpd 2.0.0 \textendash 2.0.1   & 2004-07-02           & 7             & 4           \\
				vsftpd 2.0.1 \textendash 2.0.2   & 2005-03-03           & 23            & 4           \\
				vsftpd 2.0.2 \textendash 2.0.3   & 2005-03-19           & 18            & 8           \\
				vsftpd 2.0.3 \textendash 2.0.4   & 2006-01-09           & 14            & 9           \\
				vsftpd 2.0.4 \textendash 2.0.5   & 2006-07-03           & 21            & 15  \\
				vsftpd 2.0.5 \textendash 2.0.6   & 2008-02-13           & 20            & 9   \\
                vsftpd 2.0.6 \textendash 2.0.7   & 2008-07-30           & 16            & 8   \\
                vsftpd 2.0.7 \textendash 2.1.0   & 2009-02-19           & 53            & 11   \\
                vsftpd 2.1.0 \textendash 2.1.2   & 2009-05-29           & 21            & 9   \\
                vsftpd 2.1.2 \textendash 2.2.0   & 2009-08-13           & 34            & 14   \\
			\end{tabular}
			\begin{tabular}{l | c | r | r}
				Software version                 & Update date          & Total         & Class \\ \hline
                vsftpd 2.2.0 \textendash 2.2.2   & 2009-10-19           & 21            & 5   \\
                vsftpd 2.2.2 \textendash 2.3.0   & 2010-08-06           & 13            & 3  \\
                vsftpd 2.3.0 \textendash 2.3.2   & 2010-08-19           & 5             & 0  \\
                vsftpd 2.3.2 \textendash 2.3.4   & 2011-03-12           & 7             & 0  \\
                vsftpd 2.3.4 \textendash 2.3.5   & 2011-12-19           & 14            & 6  \\
                vsftpd 2.3.5 \textendash 3.0.0   & 2012-04-10           & 23            & 4  \\
                vsftpd 3.0.0 \textendash 3.0.2   & 2012-09-19           & 40            & 2   \\ \hline
				sshd 3.5p1 \textendash 3.6p1     & 2003-03-31           & 95            & 34  \\
				sshd 3.6p1 \textendash 3.6.1p1   & 2003-04-01           & 13            & 12  \\
				sshd 3.6.1p1 \textendash 3.6.1p2 & 2003-04-29           & 16            & 12  \\
				sshd 4.5p1 \textendash 4.6p1     & 2007-03-07           & 48            & 13  \\
                sshd 6.6p1 \textendash 6.7p1     & 2014-10-06           & 283           & 51 \\ \hline
				icecast 0.8.0 \textendash 0.8.1  & 2004-08-04           & 4             & 3   \\
				icecast 0.8.1 \textendash 0.8.2  & 2004-08-04           & 2             & 0   \\
				icecast 2.3.0 \textendash 2.3.1  & 2005-11-30           & 47            & 10  \\
				icecast 2.3.1 \textendash 2.3.2  & 2008-06-02           & 250           & 28 \\
                icecast 2.4.0 \textendash 2.4.1  & 2014-11-19           & 178           & 154 \\
			\end{tabular}
		\end{center}
		\caption{Statistics of classified real world software update}\label{CategoryPercentFig}
\end{figure}

Fig.~\ref{CategoryPercentFig} shows the statistics from our study of real world program evolution where ``total" refers to the number of all updated functions, ``class" refers to the number of updated functions with at least one classified update pattern.
In summary, 32\% of all updated functions include at least one classified program update; the unclassified updates are mostly bug fix that are related to specific program logic.
We summarized seven most common real world update classes from all the studied updates in Fig.~\ref{UpdateClassRequiredAssumptions} and we believe that these update classes are also widespread in other program evolution.
Each of the six real world update classes falls in one of the five cases of backward compatibility in Fig.~\ref{NewProgBehaviorFig}.
We present informal descriptions of all update classes including required assumptions for the two programs to produce same or equivalent output sequence which guarantees backward compatible DSU.

\begin{figure}
		\begin{center}
			\begin{tabular}{ | c | l |} \hline
			Update class (Case)              & Required assumptions for backward          \\
                                             & compatible update                          \\ \hline
			program equivalence (1)          & none                                       \\ \hline
			new config. variables (1)        & no redefinitions of new config             \\
                                             &  variables after initialization            \\ \hline
			enum type extension (2)          & no inputs from old clients match           \\
                                             & the extended enum labels                   \\ \hline
			var. type weakening (3)          & no intentional use of value type           \\
                                             & mismatch and array out of bound            \\ \hline
			exit on error (4)                & correct error check before exit            \\ \hline
            improved prompt msgs (5)         & changing prompt messages for               \\
                                             & more effective communication               \\ \hline
			missing var. init. (6)           & no intentional use of undefined            \\
                                             & variables                                  \\ \hline
			\end{tabular}
		\end{center}
		\caption{Required assumptions for real world backward compatible update classes}\label{UpdateClassRequiredAssumptions}
\end{figure}

\subsection{Observational equivalence: the old behavior}
In case 1 in Fig.~\ref{NewProgBehaviorFig}, two programs are backward compatible because the new program keeps all old behaviors (``observational equivalence"). In our study, we differentiate two types of ``observational equivalence" based on if assumptions are required.

\noindent{\bf Program equivalence}
We consider several types of program changes that are allowed by ``observational equivalence" without user assumptions. These changes include: loop fission or fusion, statement reordering or duplication, and extra statements unrelated to output(e.g., logging related changes).
We incorporate these changes in our framework of program equivalence which ensures two programs produce the same output regardless of whether the programs terminate or not.
The details of the formal treatment is in Section~\ref{sec:equivalence}.

\noindent{\bf Specializing new configuration variables}
Another update class of ``observational equivalence" is ``specializing new configuration variables", which is backward compatible under user assumptions. In this update class, new configuration variables are introduced to generalize functionality.
\begin{figure}
\begin{tabbing}
xxxxxx\=xxx\=xxx\=xxx\=xxxxxxxxxxxxxxxx\=xxxx\=xxx\=xxxxx\= \kill
\>1: \>\>\>                                           \>1': \> {\bf If} $(b)$  {\bf then}       \> \\
\>2: \>\>\>                                           \>2': \> \> {\bf output} $a * 2$                   \> \\
\>3: \>\>\>                                           \>3': \> {\bf else}                        \> \\
\>4: \> {\bf output} $a + 2$ \>\>                     \>4': \> \> {\bf output} $a + 2$                   \> \\

\> \> old\>\>                                         \>  \> new \>
\end{tabbing}
\vspace{-3ex}
\caption{Specializing new configuration variables}\label{fig:newParamExample}
\end{figure}
For example, in Fig.~\ref{fig:newParamExample}, a new configuration variable $b$ is used to introduce new code.
The two statement sequences in Fig.~\ref{fig:newParamExample} are equivalent when the new variable $b$ is specialized to 0.
In general, if all new code is introduced in a way that is similar to that  in Fig.~\ref{fig:newParamExample} where there is a valuation of new configuration variables under which new code is not executed, and new configuration variables are not redefined after initialization, then the new program and the old program produce the same output sequence.
The point is that new functionality is not introduced abruptly in interaction with an old client.
Instead new functionality could be enabled for a new client when old clients are not a concern.

\subsection{Enum. type extension: old behavior for old input and allowing new input}

Enumeration types allow developers to list similar items. New code is usually accompanied with the introduction of new enumeration labels. Fig.~\ref{fig:enumTypeExtExample} shows an example of the update. The new enum label $o_2$ gives a new option for matching the value of the variable $a$, which introduces the new code ``{\bf output} $3 + c$".
\begin{figure}
\begin{center}
\begin{tabbing}
xxxxxx\=xxx\=xxx\=xxx\=xxxxxxxxxxxxxxxx\=xxxx\=xxx\=xxxxx\= \kill
\>1: \>{\bf enum} ${id}$ \{$o_1$\}\>\>          \>1': \> {\bf enum} ${id}$ \{$o_1, o_2$\} \> \\
\>2: \>a : enum $id$\>\>                        \>2': \> a : enum $id$                    \> \\
\>3: \>{\bf If} $(a == o_1)$ {\bf then}\>\>     \>3': \> {\bf If} $(a == o_1)$ {\bf then}\> \\
\>4: \> \>$\text{{\bf output}} \, 2 + c$ \>           \>4': \> \> $\text{{\bf output}} \, 2 + c$                  \> \\
\>5: \>\>\>                                     \>5': \> {\bf If} $(a == o_2)$ {\bf then}\> \\
\>6: \>\>\>                                     \>6': \> \> $\text{{\bf output}} \, 3 + c$                  \> \\

\> \> old\>\>                                         \>  \> new \>
\end{tabbing}
\vspace{-3ex}
\end{center}
\caption{Enumeration type extension}\label{fig:enumTypeExtExample}
\end{figure}
To show enumeration type extensions to be backward compatible,
we assume that values of enum variables, used in the If-predicate introducing the new code, are only from inputs that cannot be translated to new enum labels. This is case 2 of the backward compatibility.

\subsection{Variable type weakening: more output when the old program terminates}
In program updates, variable types are changed either to allow for larger ranges (weakening)
or smaller ranges to save space (strengthening). For example,
an integer variable might be changed to become a long variable
to avoid integer overflow or a long variable might be changed to an
integer variable because the larger range of long is not needed.
Type weakening also includes adding a new enumeration value and increasing array size.
The kinds of strengthening or weakening that should be allowed are application dependent and would need to be defined by the user in general.
The type weakening considered is either changes from type int to long or increase of array size.
These updates fix integer overflow or array index out of bound respectively, the case 3 of backward compatibility.
Implicitly, we assume that there is no intentional use of integer overflow and array out of bound as program semantics.

\subsection{Exit on errors: stopping execution while the old program produces more output}
One kind of bug fix, which we call {\em exit on error}, causes a program to exit in observation of errors that depend on application semantic.
\begin{figure}
\begin{tabbing}
xxxxxx\=xxx\=xxx\=xxx\=xxxxxxxxxxxxxxxx\=xxxx\=xxx\=xxxxx\= \kill
\>1: \>\>\>                                    \>1': \> {\bf If} $(1/(a-5))$ {\bf then}\> \\
\>2: \>\>\>                                    \>2': \> \> {\bf skip} \> \\
\>3: \> {\bf output} $a$\>\>                   \>3': \> {\bf output} $a$  \> \\
\>\\
\> \> old\>\>                                         \>  \> new \>
\end{tabbing}
\vspace{-3ex}
\caption{Exit-on-error}\label{fig:exitOnErrExample}
\end{figure}
Fig.~\ref{fig:exitOnErrExample} shows an example of exit-on-error update.
In the example, the fixed bugs refer to the program semantic error that $a = 5$.
Instead of using an ``exit" statement, we rely on the crash from expression evaluations to model the ``exit".
When errors do not occur, the two programs in Fig.~\ref{fig:exitOnErrExample} produce the same output sequence.
This is case 4 of backward compatibility.
Naturally, we assume that all error checks are correct.

\subsection{Improved prompt messages: functionally equivalent outputs}
In practice, outputs could be classified into prompt outputs and actual outputs. Prompt outputs are those asking clients for inputs, which are constants hardcoded in output statements. Actual outputs are dynamic messages produced by evaluation of non-constant expressions in execution.
If the differences between two programs are only the prompt messages that a client receives, we consider that the two programs are equivalent.
The prompt messages are the replaceable part of program semantics.
We observe cases of improving prompt messages in program evolution for effective communication.
The changes of prompt outputs do not matter only for human clients.
This is case 5 of backward compatibility.

\subsection{Missing variable initialization: enforcing restrictions on program states}
Another kind of bug fix, which we call {\em missing variable initialization}, includes initializations for variables whose arbitrary initial values can affect the output sequence in the old program.
\begin{figure}
\begin{tabbing}
xxxxxx\=xxx\=xxx\=xxx\=xxxxxxxxxxxxxxxx\=xxxx\=xxx\=xxxxx\= \kill
\>1: \>\>\>                                    \>1': \> $b: = 2$ \> \\
\>2: \>{\bf If} $(a > 0)$ {\bf then}\>\>       \>2': \> {\bf If} $(a > 0)$ {\bf then} \> \\
\>3: \> \> $b := c + 1$\>                      \>3': \> \> $b := c + 1$               \> \\
\>4: \> {\bf output} $b + c$ \>\>              \>4': \> {\bf output} $b + c$          \> \\

\> \> old\>\>                                         \>  \> new \>
\end{tabbing}
\vspace{-3ex}
\caption{Missing initialization}\label{fig:missInitExample}
\end{figure}
Fig.~\ref{fig:missInitExample} shows an example of missing variable initialization. The initialization $b: = 2$ ensures the value used in ``{\bf output} $b + c$" not to be undefined. Despite of initialization statements, the two programs are same.
In general, initializations of variables only affect rare buggy executions of the old program,
 where undefined variables affect the output sequence.
This update class is case 6 of backward compatibility and we assume that there is no intentional use of undefined variable in the program.
When there are no uses of variables with undefined variables in executions of the old program,
the two programs produce the same output sequence.

\section{Formal programming language}\label{sec:formallanguage}
We present the formal programming language based on which we prove our semantic equivalence results and describe categories of backward compatible changes. We first explain the language syntax, then the language semantics.

\subsection{Syntax of the formal language}

\begin{figure}
\begin{scriptsize}

\begin{center}
    \begin{tabular}{p{1cm} p{0.3cm} p{0.9cm} p{0.3cm} p{.9cm} p{0.3cm} p{.9cm} p{0.3cm}}
        {Identifier}       & {$id$}      & {Constant}    & {{$n$}}
        & {Label}          & {{$l$}}     &   & \\
        \hline
    \end{tabular}

    \begin{tabular}{  p{1.7cm}  p{0.3cm}  p{0.2cm}  l }
    {Enum Items}             & {$el$}                & ::=             &  {$l \; | \; {el}_1,{el}_2$} \\
    {Enumeration}            & {$EN$}                & ::=             &  $\varnothing \; | \; \text{enum} \, id \, \{el\} \, | \, {EN}_1,{EN}_2$ \\

    {Prompt Msg}             & {$msg$}               & ::=             &  $l : n \; | \; {msg}_1, {msg}_2$ \\
    {Prompts}                & {$Pmpt$}              & ::=             &  $\varnothing \; | \; \{{msg}\}$  \\

    {Base type}              & {$\tau$}              & ::=             &  {$\text{Int} \, | \, \text{Long} \, | \, \text{pmpt} \; | \; \text{enum}\, id$}          \\

    {Variables}              & {$V$}                 & ::=             & {$\varnothing \, | \, \tau \, id \, | \, \tau \, id[n] \, | \, V_1,V_2$} \\

    {Left value}              & {$lval$}               & ::=             & {$id \, | \, {id_1}[{id}_2] \, | \, id[n]$} \\

    {Expression}             & {$e$}                 & ::=             & {$id == l \, | \, lval \, | \, \text{other}$} \\

    {Statement}              & {$s$}                 & ::=            & {$lval := e \; | \; \text{input} \; id \; | \; \text{output} \; e \; | \; \text{skip}$}\\
                             &                       &                & $|$ {while ($e$) \{$S$\} }  \\

    {Stmt Seq.}             & {$S$}                 & ::=            & {$s_1;...;s_k \; \text{for} \; k\geq1$}\\

    {Program}                & {$P$}                 & ::=             & {${Pmpt;EN;V;S_{entry}}$} \\
\\
\hline
    \end{tabular}
\end{center}
\caption{Abstract syntax}\label{fig:syntaxfig}
\end{scriptsize}
\end{figure}

The language syntax is in Figure.~\ref{fig:syntaxfig}.
We use $id$ to range over the set of identifiers, $n$ to range over integers, $l$ to range over labels.
We assume unique identifiers across all syntactic categories, unique labels across all enumeration types and the prompt type.
We have base type Int and Long for integer values.
The integers defined in type Int are also defined in type Long. Every label defined in the prompt type is related with an integer constant as the actual value used in output statement.
We differentiate type Long and Int to define the bug fix of type relaxation from Int to Long to prevent overflow in calculation (e.g., a + b can cause an error with Int but not with Long).
The type Int is necessary reflecting the concern of space and time efficiency in practical computation.
We also have user-defined enumeration type, prompt type and array type.

We explicitly have ``$id == l$" and $lval$ as expressions for convenience of the definition of specific updates.
To make our programming language general and to separate the concern of expression evaluation, we parameterize the language by ``other" expressions which are unspecified.

We have explicit input and output statement because we model the program behavior as the I/O sequence which is the observational behavior of a program. The I/O statement makes it convenient for the argument of program behavior correspondence. In this paper, every I/O value is an integer value which is a common I/O representation~\cite{gordon1994functional}.
A Statement sequence is defined as $s_1;...;s_k$ where $k>0$ for the convenience of syntax-direct definition from both ends of the sequence.

A program is composed of a possibly empty prompt type $Pmpt$, a possibly empty sequence of enumeration types $EN$, a possibly empty sequence of global variables $V$ and a sequence of entry statements $S_{entry}$.
Finally, we have a standard type system based on our syntax.

\subsection{Small-step operational semantics of the formal language}

\begin{figure}
\begin{center}
\begin{scriptsize}
\begin{tabular} {p{1.6cm}p{0.4cm}p{0.1cm}p{1.6cm}p{3cm}}%
{Values}       & $v$      &$\in$   &{$\mathbb{Z}_L \, \cup \, \mathbb{L}$}     & {integer values in type long and enum/prompt labels}\\

{I/O values}   & $v_{io}$ &$\in$   &\multicolumn{2}{l}{$\mathbb{Z}_L$}                \\

{Inputs}       & $v_i$       & ::=    &{$\underline{v}_{io}$}                    &  {tagged input values} \\
{Eval. values} & $v_{\text{err}}$ & ::=    &{$v \, | \, \text{error}$}                &  {values and the runtime error}\\
\\

{Param. types} & ${\tau}_{\top}$ & ::=  & \multicolumn{2}{l}{$\tau \, | \, \text{array}(\tau, n)$} \\
\\

{Loop Labels}   & ${loop}_{lbl}$   & $\in$  &$\mathbb{N}$                &  \\  
\\
\hline
\end{tabular}

\end{scriptsize}
\end{center}

\vspace{-3ex}
\caption{Values, types and domains}\label{fig:domains}
\end{figure}

\begin{figure}[t!]
\begin{center}
\begin{scriptsize}
\begin{tabular} {p{1.5cm}p{.2cm}p{.1cm}p{2.6cm}p{2.5cm}}

{Crash flag}           & $\mathfrak{f}$  & ::=     & 0 $|$ 1                                     & \\

{Overflow flag}        & $\mathfrak{of}$ & ::=     & 0 $|$ 1                                        & \\

{Type Env.}            & $\Gamma$        & ::=     & \multicolumn{2}{l}{$\varnothing \, | \, id : {\tau}_{\top} \, | \, id : \{l_1,...,l_k\} \, | \, {\Gamma}_1, {\Gamma}_2$}                               \\

{Loop counter}         & ${loop}_c$            & ::=     &\multicolumn{2}{l}{$({loop}_{lbl} \mapsto (n \, | \, \bot))$}                     \\

{Value store}          & $\vals$         & ::=     & $\mit{id} \mapsto (\mit{v} \, | \, \bot)$                     & values of scalar variables   \\
                       &                 &$\;|\;$  &${id} \mapsto (n \mapsto ({v} \, | \, \bot))$         & values of array elements \\
                       &                 &$\;|\;$  &${id}_{I}  \mapsto  {v}_{io}^{*}$                                    & input sequence \\
                       &                 &$\;|\;$  &${id}_{IO} \mapsto  {(v_i \; | \; v_o)}^* $                          & I/O sequence \\
\\
{State}                & $m$             &::=    &\multicolumn{2}{l}{$(\mathfrak{f}, \mathfrak{of}, \Gamma, {loop}_c, \vals)$} \\
\\
\hline
\end{tabular}
\end{scriptsize}
\end{center}
\caption{Elements of an execution state }\label{fig:semelems}
\end{figure}

%
Figure~\ref{fig:domains} shows semantic categories of our language.
We consider values to be either labels $\mathbb{L}$ or integer numbers $\mathbb{Z}_L$ defined in type Long.
The integer numbers defined $\mathbb{Z}_I$ of type Int are a proper subset of those in type Long,
$\mathbb{Z}_I \subset \mathbb{Z}_L$. We use the notation $\mathbb{Z}_{L+}$ for the positive integers defined in type Long.
We use the notation $\text{udf}\llbracket \tau \rrbracket$ for an undefined value of type $\tau$. Unlike the ``undef" in Clight~\cite{Clight}, we need to parameterize the undefined value with a type $\tau$ because we do not have an underlining memory model that can interpret any block content according to a type.
An individual value in I/O sequence is an integer number with tag differentiating inputs and outputs, our tags for inputs and outputs are standard notations~\cite{gordon1994functional}.
The value from expression evaluation is a pair. One of the pair is  either a value $v$ or ``error" for runtime errors(e.g., division by zero); the other is the overflow flag (i.e., 0 for no overflow).

We use notation ${\tau}_{\top}$ for all types  that are defined in syntax, including array types.

Every loop statement in a program is with a unique label ${loop}_{lbl}$ of a natural number in order to differentiate their executions.

The composition of an execution state is in Figure~\ref{fig:semelems}.
\begin{enumerate}
\item
The crash flag $\mathfrak{f}$ is initially zero and is set to one whenever an exception occurs. Once the crash flag is set, it is not cleared.
We only consider unrecoverable crashes.
The crash flag is used to make sure that updates do not occur in error states.

\item
The overflow flag $\mathfrak{of}$ is initially zero and is set to one whenever an integer overflow in expression evaluation occurs. Overflow flag is sticky in the sense that once it is set, the flag is not cleared.
According to~\cite{intoverflow}, integer overflows are common in mature programs.

\item
$\Gamma$ is the type environment mapping enumeration type identifers and variable identifiers to their types.
Type environment is necessary for checking array index out of bound or checking value mismatch in execution of input/assignment statement.

\item
Loop counters ${loop}_c$ are to record the number of iterations  for one instance of a loop statement. The loop counters ${loop}_c$ is not necessary for program executions but are needed for our reasoning of the execution of loops.
When a counter entry for loop label $n$ is not defined in loop counters ${loop}_c$, we write ${loop}_c(n) = \bot$.
Otherwise, we write ${loop}_c(n) \neq \bot$.

\item
The value store $\vals$ is a valuation for scalar variables, array elements, the input sequence variable,  and the I/O sequence variable.
\end{enumerate}

Execution state $m$ is a composition of elements discussed above.
In our SOS rules, we only show components of a state $m$ when necessary (e.g., $m(\Gamma, \vals)$).

\begin{figure}[t!]
\begin{small}
\begin{center}
\begin{tabular}{ll}
{\boxed{\mbox{$(S, {\state}) -> (S', {\state}')$}}}
&
\inference[]
{(r, \state) -> (r', \state')}
{(\mathbb{E}[r], \state) -> (\mathbb{E}[r'], \state')}
\end{tabular}

\begin{tabular}{l}
{\mbox{{\bf Eval. Context}
$\mathbb{E} \, ::= \, \_ \,
                 | \, id[\mathbb{E}] \,
                 | \, \mathbb{E} == l $}} \\
{{\mbox{$| \, id \, := \, \mathbb{E} \,
         | \, id[\mathbb{E}] := e \,
         | \, id[v] := \mathbb{E}  \,
         | \, \text{output } \mathbb{E}
         $}}} \\
{\mbox{$| \, \text{while }(\mathbb{E}) \{S\} \,
        | \, \text{If }(\mathbb{E}) \text{ then }\{S_t\} \text{ else }\{S_f\} \,
        | \, \mathbb{E}; S$}}
\\
\\
\hline
\end{tabular}

\caption{Contextual semantic rule}\label{fig:semctxt}
\end{center}
\end{small}
\end{figure}


\begin{figure}[t!]

\begin{tabular} {l}
{\boxed{\mbox{$(r, {\state}) -> (r', {\state}')$}}}
\end{tabular}

\begin{center}
\scriptsize{
\begin{tabular}{l}
$\mathcal{E}: \text{other} -> \vals -> (v_{\text{err}} \times \{0, 1\})$ \\
$\text{Err}:  \text{other} -> \{id\} \;\;\;\; \text{(unspecified)}$ \\
\\

\inference[Var]
{\mathfrak{f} = 0  &    \vals(id) = v
}
{(id, \state(\mathfrak{f}, \vals)) -> (v, \state)}\\
\\

\inference[Arr-1]
{\mathfrak{f} = 0  &    \vals(id, v_1) = v_2
}
{(id[v_1], \state(\mathfrak{f}, \vals)) -> (v_2, \state)}\\
\\

\inference[Arr-2]
{\mathfrak{f} = 0 & (\Gamma \vdash id : \text{array}(\tau, n)) \wedge \neg(1 \leq v_1 \leq n)} 
{(id[v_1], \state(\mathfrak{f}, \Gamma)) -> (id[v_1], \state(1/\mathfrak{f}))}\\
\\

\inference[Eq-T]
{\mathfrak{f} = 0 }
{(l == l, \state(\mathfrak{f})) -> (1, \state)}\\
\\
\\

\inference[Eq-F]
{\mathfrak{f} = 0 & l_1 \neq l_2}
{(l_1 == l_2, \state(\mathfrak{f})) -> (0, \state)}\\
\\
\\

\inference[EEval]
{\mathfrak{f} = 0 & e = \text{other}}
{(e, \state(\mathfrak{f}, \vals)) -> (\mathcal{E}\llbracket e\rrbracket\vals, \state)}
\\
\\

\inference[ECrash]
{\mathfrak{f} = 0}
{((\text{error}, {v}_{\mathfrak{of}}), m(\mathfrak{f})) -> (0, m(1 / \mathfrak{f}))}\\
\\

\inference[EOflow-1]
{\mathfrak{f} = 0 & \mathfrak{of} = 0}
{((v, v_{\mathfrak{of}}), \state(\mathfrak{f}, \mathfrak{of})) -> (v, m( v_{\mathfrak{of}} / \mathfrak{of}))}\\
\\

\inference[EOflow-2]
{\mathfrak{f} = 0 & \mathfrak{of} = 1}
{((v, v_{\mathfrak{of}}), \state(\mathfrak{f}, \mathfrak{of})) -> (v, \state)}\\
\\
\hline
\end{tabular}
}
\end{center}
\caption{SOS rules for expressions}\label{fig:sosrulesExpr}
\end{figure}

\begin{figure}[t!]

\begin{tabular} {l}
{\boxed{\mbox{$(r, {\state}) \rightarrow (r', {\state}')$}}}
\end{tabular}

\begin{center}
\scriptsize{
\begin{tabular}{l}
\\

\inference[As-Scl]
{\mathfrak{f} = 0  & \vals(id) \neq \bot}
{\begin{array}{l}
(id := v, m(\mathfrak{f}, \vals)) ->  (\text{skip}, m(\vals[v / id]))
 \end{array}
}
\\
\\

\inference[As-Arr]
{\mathfrak{f} = 0    &      \vals(id, v_1) \neq \bot}
{\begin{array}{l}
(id[v_1] := v_2, m(\mathfrak{f}, \vals)) ->
(\text{skip}, m(\vals[v_2 / (id, v_1)]))
 \end{array}
}
\\
\\

\inference[As-Err1]
{\mathfrak{f} = 0  & (\Gamma \vdash id : \text{array}(\tau, n)) \wedge \neg(1 \leq v_1 \leq n)
}
{\begin{array}{l}
\begin{array}{l}
(id[v_1] := v_2, m(\mathfrak{f}, \Gamma)) ->
(id[v_1] := v_2, m(1/\mathfrak{f}))
\end{array}
 \end{array}
}
\\
\\

\inference[As-Err2]
{\mathfrak{f} = 0    &    \vals(id) \neq \bot  \\
  (\Gamma \vdash id : \text{Int}) \wedge (v \in (\mathbb{Z}_{L}\setminus \mathbb{Z}_{I}))}
{\begin{array}{l}
(id := v, m(\mathfrak{f}, \Gamma, \vals)) ->  (id := v, m(1/\mathfrak{f}))
 \end{array}
}
\\
\\

\inference[As-Err3]
{\mathfrak{f} = 0  &     \vals(id, v_1) \neq \bot  \\
  (\Gamma \vdash id : \text{array}(\text{Int}, n)) \wedge (v_2 \in (\mathbb{Z}_{L}\setminus \mathbb{Z}_{I}))}
{\begin{array}{l}
(id[v_1] := v_2, m(\mathfrak{f}, \Gamma, \vals)) ->  (id[v_1] := v_2, m(1/\mathfrak{f}))
 \end{array}
}
\\
\\

\inference[If-T]
{\mathfrak{f} = 0 & (v \in \mathbb{Z}_{L}) \wedge (v \neq 0)}
{(\text{If } (v) \text{ then } \{S_t\} \text{ else } \{S_f\}, \state(\mathfrak{f})) -> (S_t, \state)}
\\
\\

\inference[If-F]
{\mathfrak{f} = 0}
{(\text{If } (0) \text{ then } \{S_t\} \text{ else } \{S_f\}, \state(\mathfrak{f})) -> (S_f, \state)}
\\
\\

\inference[Wh-T]
{\mathfrak{f} = 0 & (v \in \mathbb{Z}_{L}) \wedge (v \neq 0)  &  {loop}_c(n) = k }
{\begin{array}{l}
(\text{while}_{\langle n\rangle} \; (v) \; \{S\}, \state(\mathfrak{f}, {loop}_c)) -> \\
(S; \text{while}_{\langle n\rangle} \; (e) \; \{S\}, \state({loop}_c[(k + 1) / n])
 \end{array}
}
\\
\\

\inference[Wh-F]
{\mathfrak{f} = 0 & {loop}_c(n) \neq \bot }
{
\begin{array}{l}
(\text{while}_{\langle n\rangle} \; (0) \; \{S\}, m(\mathfrak{f}, {loop}_c)) ->
(\text{skip}, m({loop}_c[0/n]))
\end{array}
}
\\
\\
\end{tabular}

\begin{tabular}{ll}

\inference[Seq]
{\mathfrak{f} = 0}
{(\text{skip};S, m(\mathfrak{f})) -> (S, m)}
&
\inference[Crash]
{\mathfrak{f} = 1}
{(s, \state(\mathfrak{f})) -> (s, \state)}
\\
\\
\hline
\end{tabular}
}
\end{center}
\caption{SOS rules for Assignment, If, and While statements}\label{fig:basicrules}
\end{figure}

\begin{figure}[t!]

\begin{tabular} {l}
{\boxed{\mbox{$(r, m) -> (r', m')$}}}
\end{tabular}

\begin{scriptsize}
\begin{center}
\begin{tabular} {l}

\inference[In-1]
{\mathfrak{f} = 0  &  \vals(id) \neq \bot & \text{hd}(\vals({id}_I)) = v_{io}    &     \Gamma \vdash id : \text{Long}
}
{\begin{array}{l}
 (\text{input }id, m(\mathfrak{f}, \Gamma, \vals)) -> \\
 (\text{skip}, m(\vals[v_{io} / id] [\text{tl}(\vals({id}_I))/{id}_I] [``\vals({id}_{IO}) \cdot \underline{v}_{io}"/{id}_{IO}])
 \end{array}
}
\\
\\

\inference[In-2]
{\mathfrak{f} = 0    &  \vals(id) \neq \bot \\
 \text{hd}(\vals({id}_I)) = v_{io}    &      (\Gamma \vdash id : \text{Int}) \wedge (v_{io} \in \mathbb{Z}_{I})
}
{\begin{array}{l}
 (\text{input }id, m(\mathfrak{f}, \Gamma, \vals)) -> (\text{skip},\\
  m(\vals[{v_{io}} / id] [\text{tl}(\vals({id}_I))/{id}_I] [``\vals({id}_{IO}) \cdot \underline{v}_{io}"/{id}_{IO}])
 \end{array}
}
\\
\\

\inference[In-3]
{\mathfrak{f} = 0     &  \vals(id) \neq \bot \\
  \text{hd}(\vals({id}_I)) = v_{io}    &      (\Gamma \vdash id : \text{Int}) \wedge (v_{io} \notin \mathbb{Z}_{I})
}
{
 (\text{input }id, m(\mathfrak{f}, \Gamma, \vals)) -> (\text{input }id, m(1 / \mathfrak{f}))
}
\\
\\

\inference[In-4]
{\mathfrak{f} = 0 &  \vals(id) \neq \bot & \text{hd}(\vals({id}_I)) = v_{io} \\
 (\Gamma \vdash id : \text{enum }{id}') \wedge (\Gamma \vdash id' : \{l_1,...,l_k\}) \wedge (1 \leq v_{io} \leq k)
}
{\begin{array}{l}
 (\text{input }id, \state(\mathfrak{f}, \Gamma,\vals)) -> (\text{skip},\\
  \state(\vals[l_{v_{io}} / id, \text{tl}(\vals({id}_I))/{id}_I] [``\vals({id}_{IO}) \cdot \underline{v}_{io}"/{id}_{IO}])
 \end{array}
}
\\
\\

\inference[In-5]
{\mathfrak{f} = 0  & \vals(id) \neq \bot & \text{hd}(\vals({id}_I)) = v_{io} \\
 (\Gamma \vdash id : \text{enum }{id}') \wedge (\Gamma \vdash id' : \{l_1,...,l_k\}) \wedge \neg(1 \leq v_{io} \leq k)
}
{
 (\text{input }id, m(\mathfrak{f}, \Gamma, \vals)) -> (\text{input }id, m(1 / \mathfrak{f}))
}
\\
\\

\inference[In-6]
{\mathfrak{f} = 0    &         \vals({id}_I) = \varnothing}
{\begin{array}{l}
 (\text{input } id, m(\mathfrak{f}, \vals)) ->
 (\text{input } id, m(1/\mathfrak{f}))
 \end{array}
}
\\
\\

\inference[Out-1]
{\mathfrak{f} = 0 & v \in \mathbb{Z}_{L}}
{
\begin{array}{l}
(\text{output } v, m(\mathfrak{f},\vals)) -> (\text{skip},
m(\vals[``\vals({id}_{IO}) \cdot \overline{v}"/{id}_{IO}]))
\end{array}
}
\\
\\

\inference[Out-2]
{\mathfrak{f} = 0 & \Gamma \vdash id : \{l_1,..., l_k\} \wedge v = l_i \in \{l_1,..., l_k\}}
{
\begin{array}{l}
(\text{output }v, m(\mathfrak{f}, \Gamma, \vals)) -> (\text{skip},
m(\vals[``\vals({id}_{IO}) \cdot \overline{i}"/{id}_{IO}]))
\end{array}
}
\\
\\

\inference[Out-3]
{\mathfrak{f} = 0 & \Gamma \vdash pmpt : \{l_1 : n_1,..., l_k : n_k\} \\
 ``l : n" \in \{l_1 : n_1,..., l_k : n_k\}}
{
\begin{array}{l}
(\text{output }l, m(\mathfrak{f}, \Gamma)) -> (\text{output }n, m)
\end{array}
}
\\
\\
\hline
\end{tabular}
\end{center}
\end{scriptsize}
\caption{SOS rules for input/output statements}\label{fig:iorules}
\end{figure}

Figure~\ref{fig:semctxt} shows typical contextual rule and Figure~\ref{fig:sosrulesExpr}, ~\ref{fig:basicrules} and ~\ref{fig:iorules} show all SOS rules.

Figure~\ref{fig:sosrulesExpr} shows rules for expression evaluation.
We use the expression meaning function $\mathcal{E}: \text{other} -> \vals -> (v_{\text{err}} \times \{0, 1\})$ to evaluate ``other" expressions.
In evaluation of expression ``other" against a value store $\vals$, the expression meaning function $\mathcal{E}$ returns a pair $(v_{\text{err}}, \mathfrak{of})$ where the value $v_{\text{err}}$  is either a value $v$ or an ``$\text{error}$",  $\mathfrak{of}$ is a flag indicating if there is integer overflow in the evaluation (e.g., 1 if there is overflow).
The meaning function $\mathcal{E}$ interprets ``other" expressions deterministically.
In addition, there is a function $\text{Use}: \text{other} -> \{id\}$ maps an ``other" expression to a set of variables used in the expression;
there is a function $\text{Err}: \text{other} -> \{id\}$ maps an expression to a set of variables whose values decide if the evaluation of expression leads to crash.
We assume function $\text{Use}$ and $\text{Err}$ available.
The value returned by the expression meaning function only depends on the values of variables in the use set of the expression and the error evaluation only depends on the variables in the error set.

As to integer overflow, there are two ways of handling overflow in practice
one is to wrap around overflow using twos-complement representation (e.g., the gcc option -fwrapv); the other is to generates traps for overflow (e.g., the gcc option -ftrapv). We adopt a combination of the two handling of overflow: the meaning function $\mathcal{E}$ wraps the overflow in some representation (e.g., two-complement) and notifies the overflow in return value.
Rule EOflow-1 and EOflow-2 update the sticky overflow flag.
The evaluation of $lval$ or $id == l$ is shown by respective rules in Figure~\ref{fig:sosrulesExpr}.

Figure~\ref{fig:basicrules} shows SOS rules for assignment, If, while statements, statement sequence, and crash, which are almost standard.
There are four particular crash in execution of assignment statements.
One is array out of bound for array access for l-value (e.g., rule As-Err1);
the second is assigning a value defined in type Long but not type Int to an Int-typed variable (e.g., rule As-Err2);
the third is value mismatch in input statement; the last is expression evaluation exception.
As to loop statement, if the predicate expression evaluates to a nonzero integer, corresponding loop counter value increments by one;
otherwise, the loop counter value is reset to zero.
We use rule Crash to treat crash as non-terminating execution, telling apart normally terminating executions and others.

Figure~\ref{fig:iorules} shows rules for the execution of input/output statements.
As to input, there are conversion from values of type Long to those of Int or enumeration types but not the prompt type.
For an enumeration type, the Long-typed value is transformed to the label with index of that value if possible.
There is crash when value conversion is impossible.
Besides, there is crash when executing input statement with empty input sequence.
We use standard list operation hd and tl for fetching the list head(leftmost element) or the list tail(the list by removing its head) respectively~\cite{pierce2002types}.

Last, we construct initial state in following steps:
First, crash flag $\mathfrak{f}$, overflow flag $\mathfrak{of}$ are zero.
Second, type environment is obtained after parsing of the program.
Third, every loop counter value in $\text{loop}_c$ is initially zero.
Fourth, every scalar variable or array element has an entry in value store with some initial value if specified.
Last, there is initial input sequence and empty I/O sequence. 

\subsection{Preliminary terms and notations}\label{sectionDefs}
We present terms, notations and definitions for program equivalence and backward compatible update classes.

We use $\text{Use}(e)$ or $\text{Use}(S)$ to denote used variables in an expression $e$ or a statement sequence $S$; $\text{Def}(S)$ denotes the set of defined variables in a statement sequence $S$.
The full definitions of Use and Def are in appendix~\ref{appendix:definitions}.

%

We use symbol $\in$ for two different purposes: $x \in X$ denotes one variable to be in a set of variables, $s \in S$ denotes a statement to be in a statement sequence.
We use the symbol $\subset$ to refer to proper subset relation.



We call an ``If" statement or a ``while" statement as a compound statement;
all other statements are simple statements.
We introduce terms referring to a part of a compound statement.
Let $s = ``\text{If}(e) \, \text{then}\{S_t\} \, \text{else}\{S_f\}"$ be an ``If" statement, we call $e$ in $s$ the predicate expression, $S_t/S_f$ the true/false branch of $s$.

\section{Program equivalence}\label{sec:equivalence}

We consider several types of program changes that are allowed by ``observational equivalence" without user assumptions. These changes include: statement reordering or duplication, extra statements unrelated to output(e.g., logging related changes), loop fission or fusion, and extra statements unrelated to output.
Our program equivalence ensures two programs produce the same output, which means two programs produce same I/O sequence till any output. The program equivalence is established upon two other kinds of equivalence, namely equivalent terminating computation of a variable and equivalent termination behavior.

We first define terminating and nonterminating execution.
Then we present the framework of program equivalence in three steps in which every later step relies on prior ones.
We first propose a proof rule ensuring two programs to compute a variable in the same way.
We then suggest a condition ensuring two programs to either both terminate or both do not terminate.
Finally we describe a condition ensuring two programs to produce the same output sequence.
Our proof rule of program equivalence gives program point mapping as well as program state mapping.
Though we express the program equivalence as a whole program relation, it is easy to apply the equivalence check for local changes using our framework under user's various assumptions for equivalence.

\subsection{Definitions of execution}
We define an execution to be a sequence of configurations which are pairs $(S, m)$ where $S$ is a statement sequence and $m$ is a execution state shown in Figure~\ref{fig:semelems}.
Let $(S_1, m_1)$, $(S_2, m_2)$ be two consecutive configurations in an execution, the later configuration $(S_2, m_2)$ is obtained by applying one semantic rule w.r.t to the configuration $(S_1, m_1)$, denoted $(S_1, m_1) -> (S_2, m_2)$, called one step (of execution).
For our convenience, we use the notation
 $(S, m)$ {\kStepArrow[k] } $(S', m')$ for $k$ steps execution where $k > 0$.
When we do not care the exact (finite) number of steps, we write the execution as
$(S, m) ->* (S', m')$.
We express terminating executions, nonterminating executions including crash in Definition~\ref{def:normalExec} and~\ref{def:nontermExec}.
\begin{definition}\label{def:normalExec}
{\bf (Termination)} A statement sequence $S$ normally terminates when started in a state $\state$ iff $(S, \state) ->* (\text{skip}, \state'(\mathfrak{f}))$ where $\mathfrak{f} = 0$.
\end{definition}

\begin{definition}\label{def:nontermExec}
{\bf (Nontermination)} A statement sequence $S$ does not terminate when started in a state $\state$ iff,
$\forall k > 0\,:\, (S, \state)$ {\kStepArrow [k] } $(S_k, \state_k)$ where $S_k \neq \text{skip}$.
\end{definition}

\subsection{Equivalent computation for terminating programs}\label{sec_equiv_comp}
We propose a proof rule under which two terminating programs are computing a variable in the same way. We start by giving the definition of equivalent computation for terminating programs right after this paragraph. Then we present the proof rule of equivalent computation in the same way.
We prove that the proof rule ensures equivalent computation for terminating programs by induction on the program size of the two programs in the proof rule.
We also list auxiliary lemmas required by the soundness proof for the proof rule for equivalent computation for terminating programs.

\begin{definition}\label{def:equiv_term_comp}
{\bf{(Equivalent computation for terminating programs)}}
Two statement sequences $S_1$ and $S_2$  compute a variable $x$ equivalently when started in states $m_1$ and $m_2$ respectively, written $(S_1, m_1) \equiv_{x} (S_2, m_2)$, iff
$({S_1}, m_1) ->* ({\text{skip}}, m_1'(\vals_{1'}))$ and
$({S_2}, m_2) ->* ({\text{skip}}, m_2'(\vals_{2'}))$ imply $\vals_{1'}(x) = \vals_{2'}(x)$.
\end{definition}

\subsubsection{Proof rule for equivalent computation for terminating programs}

We define a proof rule under which $(S_1, m_1) \equiv_{x} (S_2, m_2)$ holds for generally constructed initial states $m_1$ and $m_2$, written $S_1 \equiv_{x}^{S} S_2$.
Our proof rule for equivalent computation for terminating programs allows updates including statement reordering or duplication, loop fission or fusion, additional statements unrelated to the computation and statements movement across if-branch.

Definition~\ref{def:syntactic-equivalence} includes the recursive proof rule of equivalent computing for terminating programs. The base case is the condition for two simple statements in Definition~\ref{def:localcondForEquivTermCompSimpleStmt}. Definition~\ref{def:impvars} of imported variables captures the variable def-use chain which is the essence of our equivalence.
In Definition~\ref{def:impvars}, the Def and Use refer to variables defined or used in a statement (sequence) or an expression similar to those in the optimization chapter in the dragon book~\cite{Aho86}; $S^{i}$ refers to $i$ consecutive copies of a statement sequence $S$.

\begin{definition}\label{def:impvars}
{\bf (Imported variables)}
The imported variables in a sequence of statements $S$ relative to variables $X$, written Imp$(S, X)$, are defined in one of the following cases:
\begin{enumerate}
\item
Def $(S) \; \cap \; X \; = \; \emptyset$:
Imp $(S, X)$ = $X$;

\item
$S$ = ``$id$ := $e$" or ``input $id$" or ``output $e$" and
 $\text{Def}(S) \; \cap \; X \; \neq \; \emptyset$:

\noindent$\text{Imp}(S, X) = \text{Use}(S) \; \cup \; (X \setminus \; \text{Def}(S))$;

\item
$S$ = ``If $(e)$ then \{$S_t$\} else \{$S_f$\}" and $\text{Def}(S) \; \cap \; X \; \neq \emptyset$:

\noindent$\text{Imp}(S, X) = \text{Use}(e) \cup \bigcup_{y \in X}$
$\big(\text{Imp}(S_t, \{y\}) \cup \text{Imp}(S_f, \{y\})\big)$;

\item
$S$ = ``while($e$) \{$S'$\}" where (Def$(S') \; \cap \; X) \neq \emptyset$):
Imp $(S, X)$ = $\bigcup_{i \geq 0}$
Imp $({S'}^i, \text{Use}(e) \; \cup \;X)$;

\item
For $k>0$, $S = s_1;...;s_{k+1}$:

\noindent$\text{Imp}(S, X)  = \text{Imp}(s_1;...;s_k, \text{Imp}(s_{k+1}, X))$
\end{enumerate}
\end{definition}

\begin{definition}\label{def:localcondForEquivTermCompSimpleStmt}
{\bf{(Base cases of the proof rule for equivalent computation for terminating programs)}}
Two simple statements $s_1$ and $s_2$ satisfy the proof rule of equivalent computation of a variable $x$, written $s_1 \equiv_x^S s_2$, iff one of the following holds:
\begin{enumerate}
\item $s_1 = s_2$;

\item $s_1 \neq s_2$ and one of the following holds:
      \begin{enumerate}
      \item $s_1 = ``\text{input} \; {id}_1", s_2 = ``\text{input} \; {id}_2", x \notin \{{id}_1, {id}_2\}$;

      \item Case {\em a)} does not hold and $x \notin \text{Def}(s_1) \cup \text{Def}(s_2)$;
      \end{enumerate}
\end{enumerate}
\end{definition}

\begin{definition}
{\bf{(Proof rule of equivalent computation for terminating programs)}}\label{def:syntactic-equivalence}
Two statement sequences $S_1$ and $S_2$ satisfy the proof rule of equivalent computation of a variable $x$, written $S_1 \equiv_{x}^S S_2$, iff one of the following holds:
\begin{enumerate}

\item $S_1$ and $S_2$ are one statement and one of the following holds:

\begin{enumerate}

\item $S_1$ and $S_2$  are simple statement:
$s_1 \equiv_x^S s_2$;

\item $S_1 = ``\text{If }(e) \text{ then }\{S_1^t\} \text{ else }\{S_1^f\}"$,
      $S_2 = ``\text{If }(e) \text{ then }\{S_2^t\} \text{ else } \allowbreak\{S_2^f\}"$ such that all of the following hold:

     \begin{itemize}
     \item $x \in \text{Def}(S_1) \cap \text{Def}(S_2)$;

     \item $(S_1^t \equiv_{x}^S S_2^t) \wedge (S_1^f \equiv_{x}^S S_2^f)$;
     \end{itemize}

\item $S_1 = ``\text{while}_{\langle n_1\rangle} (e) \; \{S_1''\}", S_2 = ``\text{while}_{\langle n_2\rangle} (e) \; \{S_2''\}"$ such that both of the following hold:
   \begin{itemize}
   \item $x \in \text{Def}(S_1) \cap \text{Def}(S_2)$;

   \item $\forall y \in
        \text{Imp}(S_1, \{x\})
        \cup
        \text{Imp}(S_2, \{x\})\,:\,$
        $S_1'' \equiv_{y}^S S_2''$;
   \end{itemize}

\item $S_1$ and $S_2$ do not define the variable $x$: $x \notin \text{Def}(S_1) \cup \text{Def}(S_2)$.

\end{enumerate}

\item $S_1$ and $S_2$ are not both one statement and one of the following holds:
\begin{enumerate}

\item $S_1=S_1';s_1, S_2=S_2';s_2$ and last statements both define the variable $x$ such that both of the following hold:
\begin{itemize}
\item  $\forall y \in \text{Imp}(s_1, \{x\}) \cup \text{Imp}(s_2, \{x\})\,:\,
S_1' \equiv_{y}^S S_2'$;

\item  $s_1 \equiv_x^S s_2$ where $x \in \text{Def}(s_1) \cap \text{Def}(s_2)$;
\end{itemize}

\item Last statement in $S_1$ or $S_2$ does not define the variable $x$:

$\big(x \notin \text{Def}(s_2) \wedge (S_1 \equiv_{x}^S S_2')\big) \vee \big(x \notin \text{Def}(s_1) \wedge (S_1' \equiv_{x}^S S_2)\big)$;

\item $S_1=S_1';s_1, S_2=S_2';s_2$ and there are statements moving in/out of If statement:
      $s_1 = ``\text{If} \, (e) \, \text{then} \, \{S_1^t\} \, \text{else} \, \{S_1^f\}"$,
      $s_2 = ``\text{If} \, (e) \, \text{then} \, \{S_2^t\} \, \text{else} \, \allowbreak\{S_2^f\}"$ such that none of the above cases hold and all of the following hold:
     \begin{itemize}
     \item $\forall y \in \text{ Use}(e)\,:\, S_1' \equiv_{y}^S S_2'$;

     \item $(S_1';S_1^t \equiv_{x}^S S_2';S_2^t) \wedge (S_1';S_1^f \equiv_{x}^S S_2';S_2^f)$;
     \end{itemize}
\end{enumerate}

\end{enumerate}
\end{definition}

The generalization of definition $S_1 \equiv_{x}^S S_2$ to a set of variables is as follows.
\begin{definition}\label{condSameCompXVar}
Two statement sequences $S_1$ and $S_2$ have equivalent computation of variables $X$, written $S_1 \equiv_{X}^S S_2$, iff
$\forall x \in X\,:\, S_1 \equiv_{x}^S S_2$.
\end{definition}

%
\subsubsection{Soundness of the proof rule for equivalent computation for terminating programs}

We show that if two programs satisfy the proof rule of equivalent computation of a variable $x$ (Definition~\ref{def:syntactic-equivalence}) and their value stores in initial states agree on values of the imported variables relative to $x$, then the two programs compute the same value of $x$ if they terminate.
We start by proving the theorem for the base cases of terminating computation equivalently.
\begin{theorem}\label{thm:equivTermCompOfSimpleStmt}
If $s_1$ and $s_2$ are simple statements  that satisfy the proof rule for equivalent computation of $x$, $s_1 \equiv_x^S s_2$, and their initial states $m_1(\vals_{1})$ and $m_2(\vals_{2})$ agree on the values of the imported variables relative to $x$, $\forall y \in\, \text{Imp}(s_1, \{x\}) \cup \text{ Imp}(s_2, \{x\})\,:\, \vals_{s_1}(y) = \vals_{s_2}(y)$, then $s_1$ and $s_2$ equivalently compute $x$ when started in states $m_1$ and $m_2$ respectively, $(s_1,m_1) \equiv_{x} (s_2,m_2)$.
\end{theorem}

\begin{proof}
The proof is a case analysis according to the cases in the definition of the proof rule for equivalent computation (i.e., Definition~\ref{def:localcondForEquivTermCompSimpleStmt}).

\begin{enumerate}
\item{$s_1 = s_2$}

Since the two statements are identical, they have the same imported variables. By assumption, the imported variables of $s_1$ and $s_2$ have the same initial values, so it is enough to show that the value of $x$ at the end of the computation only depends on the initial values of the imported variables.
\begin{enumerate}
\item{$s_1 = s_2 = ``\text{skip}"$}.
In this case, the states before and after the execution of \text{skip} are the same and $\text{Imp}(\text{skip}, \{x\}) = \{x\}$.

\item{$s_1 = s_2 = ``lval := e"$}.
\begin{enumerate}
\item{$lval = x$}.

$s_1 = s_2 = ``x := e"$. By the definition of imported variables, $\text{Imp}(x := e, \{x\}) = \text{Use}(e)$.
The execution of $s_1$ proceeds as follows.

\begin{tabbing}
xx\=xx\=\kill
\>\>        $(x := e, m(\vals))$\\
\>$->$\> $(x := \mathcal{E}'\llbracket e\rrbracket\vals, m(\vals))$ by the EEval' rule\\
\>$->$\> $(\text{skip}, m(\vals[\mathcal{E}'\llbracket e\rrbracket\vals/x]))$ by the Assign rule.
\end{tabbing}
The value of $x$ after the full execution is  $\vals[(\mathcal{E}\llbracket e\rrbracket\vals)/x](x)$ which only depend on the
initial values of the imported variables by the property of the expression meaning function.

\item{$lval \neq id$}.

By the definition of imported variables, $\text{Imp}(s_1, \{x\}) = \text{Imp}(s_2, \{x\}) = \{x\}$. It follows, by assumption, that $\vals_1(x) = \vals_2(x)$ and also $s_1$ terminate, $(s_1, m_1(\vals_1)) ->* (\text{skip}, m_1'(\vals_1'))$. Hence, $\vals_1'(x) = \vals_1(x)$ by Corollary~\ref{coro:defExclusion}. Similarly, $s_2$ terminates, $(s_2, m_2(\vals_2)) ->* (\text{skip}, m_2'(\vals_2'))$ and $\vals_2'(x) = \vals_2(x)$.
Therefore, $\vals_2'(x) = \vals_2(x) = \vals_1(x) = \vals_1'(x)$ and the theorem holds.
\end{enumerate}

\item{$s_1 = s_2 = ``\text{input }id"$}.
\begin{enumerate}
\item{$x \in \text{Def}(\text{input }id) = \{ id , {id}_I , id_{IO}\}$}.

By the In rule, the execution of $\text{input }id$ is the following.

\begin{tabbing}
xx\=xx\=\kill
\>\>        $(\text{input }id, m(\vals))$ \\
\>$->$\> $(\text{skip}, m(\vals[\text{tl}(\vals({id}_I))/{id}_I]$\\
\>\>     $[``\vals({id}_{IO})\cdot\underline{\text{hd}(\vals({id}_I))}"/{id}_{IO}] [\text{hd}(\vals({id}_I))/id]))$.
\end{tabbing}

The value of $x$ after the execution of ``$\text{input }id$" is one of the following:

\begin{enumerate}
\item $\text{tl}(\vals({id}_I))$ if  $x = {id}_I$.

\item $\vals_1({id}_{IO})\cdot\underline{\text{hd}(\vals({id}_I))}$ if  $x = {id}_{IO}$.

\item $ \text{hd}(\vals({id}_I))$ if  $x = id$.
\end{enumerate}

By the definition of imported variables, $\text{Imp}(\text{input }id, \{x\})\allowbreak = \{{id}_{IO}, {id}_I\}$. So, in all cases, the value of $x$ only depends on the initial values of the imported variables ${id}_I$ and ${id}_{IO}$.

\item{$x \notin \text{Def}(\text{input }id) = \{ id , {id}_I , id_{IO}\}$}.

By same argument in the subcase $id \neq x$ of case $s_1 = s_2 = ``id := e"$, the theorem holds.
\end{enumerate}

\item{$s_1 = s_2 = ``\text{output }e"$}.

\begin{enumerate}
\item{$x = {id}_{IO}$}

By the definition of imported variables, $\text{Imp}(\text{output }e, \{x\}) = \{{id}_{IO}\} \cup \text{Use}(e)$.
The execution of $s_1$ proceeds as follows.

\begin{tabbing}
xx\=xx\=\kill
\>\>     $(\text{output }e, m(\vals))$\\
\>$->$\> $(\text{output }\mathcal{E}\llbracket e\rrbracket\vals, m(\vals))$ \\
\>$->$\> $(\text{skip}, m(\vals[``\vals({id}_{IO})\cdot\bar{\mathcal{E}\llbracket e\rrbracket\vals}"/{id}_{IO}]))$.
\end{tabbing}

The value of $x$ after the execution is  $``\vals({id}_{IO})\cdot\bar{\mathcal{E}\llbracket e\rrbracket\vals}"$, which only depends on the initial value of the imported variables of the statement $``\text{output }e"$ by the expression meaning function.

\item{$x \neq {id}_{IO}$}

By same argument in the subcase $id \neq x$ of case $s_1 = s_2 = ``id := e"$, the theorem holds.
\end{enumerate}
\end{enumerate}

\item{$s_1 \neq s_2$}

\begin{enumerate}
\item{$s_1 = ``\text{input }{id}_1", s_2 = ``\text{input }{id}_2", x \notin \{{id}_1, {id}_2\}$}.

\begin{enumerate}
\item{$x \in \{{id}_I, {id}_{IO}\}$}.

By the definition of imported variables, $\text{Imp}(s_1, \{x\}) = \text{Imp}(s_2, \{x\}) = \{{id}_{IO}, {id}_{I}\}$. It follows, by assumption, that $\vals_1(y) = \vals_2(y), \forall y \in \{{id}_{IO}, {id}_{I}\}$. The execution of $s_1$ proceeds as follows.

\begin{tabbing}
xx\=xx\=\kill
\>\>     $(s_1,m_1)$\\
\>= \>   $(\text{input }{id}_1, m_1(\vals_1))$ \\
\>$->$\> $(\text{skip}, m_1(\vals_1[\text{tl}(\vals_1({id}_I))/{id}_I]$\\
\>\>     $[``\vals_1({id}_{IO})\cdot\underline{\text{hd}(\vals_1({id}_I))}"/{id}_{IO}] [\text{hd}(\vals_1({id}_I))/{id}_1]))$
\end{tabbing}

Let $\vals_1' = \vals_1[\text{tl}(\vec{v})/{id}_I, ``\vals_1({id}_{IO})\cdot\underline{\text{hd}(\vec{v})}"/{id}_{IO}, \text{hd}(\vec{v})/{id}_1]$. The value of $x$ after the execution of $s_1$ is one of the following:

\begin{enumerate}
\item{$\vals_1'(x) = \text{tl}(\vals_1({id}_I))$ if $x = {id}_I$}.

\item{$\vals_1'(x) = \vals_1({id}_{IO})\cdot\underline{\text{hd}(\vals_1({id}_I))}$ if $x = {id}_{IO}$}.
\end{enumerate}
Similarly, $(s_2,m_2) -> (\text{skip}, m_2(\vals_2[\text{tl}(\vals_2({id}_I))/{id}_2]\allowbreak [``\vals_2({id}_{IO})\cdot\underline{\text{hd}(\vals_2({id}_I))}"/{id}_{IO}] [\text{hd}(\vals_2({id}_I))/{id}_2]))$.
Let $\vals_2' = \vals_2[\text{tl}(\vals_2({id}_I))/{id}_I]\allowbreak [``\vals_2({id}_{IO})\cdot\underline{\text{hd}(\vals_2({id}_I))}"/{id}_{IO}]\allowbreak [\text{hd}(\vals_2({id}_I))/{id}_2]$. Then the value of $x$ after the execution of $s_2$ is one of the following:
\begin{enumerate}
\item{$\vals_2'(x) = \text{tl}(\vals_2({id}_I))$ if $x = {id}_I$}

\item{$\vals_2'(x) = \vals_2({id}_{IO})\cdot\underline{\text{hd}(\vals_2({id}_I))}$ if $x = {id}_{IO}$}
\end{enumerate}
Repeatedly, $\vals_2({id}_I) = \vals_1({id}_I)$ and  $\vals_2({id}_{IO}) = \vals_1({id}_{IO})$. Therefore, the theorem holds.

\item{$x \notin \{{id}_I, {id}_{IO}\}$}

Repeatedly, $x \notin \{{id}_1, {id}_2\}$.
By same argument in the subcase $id \neq x$ of case $s_1 = s_2 = ``id := e"$, the theorem holds.

\end{enumerate}

\item{all the above cases do not hold and $ x \notin \text{ Def}(s_1) \cup \text{ Def}(s_2)$}

By same argument in the subcase $id \neq x$ of case $s_1 = s_2 = ``id := e"$, the theorem holds.
\end{enumerate}
\end{enumerate}
\end{proof}

\begin{theorem}\label{thm:equivCompMain}
If statement sequence $S_1$ and $S_2$ satisfy the proof rule of equivalent computation of a variable $x$, $S_1 \equiv_x^S S_2$, and their initial states $m_1(\vals_1)$ and $m_2(\vals_2)$ agree on the initial values of the imported variables relative to $x$, $\forall y \in \text{Imp}(S_1, \{x\}) \cup \text{Imp}(S_2, \{x\})\,:\, \vals_1(y) = \vals_2(y)$, then $S_1$ and $S_2$ equivalently compute the variable $x$ when started in state $m_1$ and $m_2$ respectively,  $(S_1, m_1) \equiv_x (S_2, m_2)$.
\end{theorem}

\begin{proof}
By induction on $\text{size}(S_1)+\text{size}(S_2)$, the sum of the program size of $S_1$ and $S_2$.

\noindent{\bf Base case}.

\noindent$S_1 \equiv_x^S S_2$ where $S_1$ and $S_2$ are two simple statements.
This theorem holds by theorem~\ref{thm:equivTermCompOfSimpleStmt}.

\myNewLine

\noindent{\bf Induction step}

\noindent The hypothesis IH is that Theorem~\ref{thm:equivCompMain} holds when  $\text{size}(S_1) + \text{size}(S_2) = k \geq 2$.

\noindent Then we show that the Theorem holds when $\text{size}(S_1) + \text{size}(S_2) = k + 1$.
The proof is a case analysis according to the cases in the definition of the proof rule of terminating computation of statement sequence.
the two big categories enum
\begin{enumerate}

\item $S_1$ and $S_2$ are one statement such that one of the following holds:

\begin{enumerate}

\item $S_1$ and $S_2$  are If statement that define the variable $x$:

      $S_1 = ``\text{If }(e) \text{ then }\{S_1^t\} \text{ else }\{S_1^f\}"$,
      $S_2 = ``\text{If }(e) \text{ then }\{S_2^t\}\allowbreak \text{ else }\{S_2^f\}"$ such that all of the following hold:

     \begin{itemize}
     \item $x \in \text{Def}(S_1) \cap \text{Def}(S_2)$;

     \item $S_1^t \equiv_{x}^S S_2^t$;

     \item $S_1^f \equiv_{x}^S S_2^f$;
     \end{itemize}

We first show that the evaluations of the predicate expression of $S_1$ and $S_2$ produce the same value when started from state $m_1(\vals_1)$ and $m_2(\vals_2)$, w.l.o.g. say zero.
Next, we show that $S_1^f$ started in the state $m_1$ and $S_2^f$ in the state $m_2$ equivalently compute the variable $x$.

In order to show that the evaluations of predicate expression of $S_1$ and $S_2$ produce same value when started from state $m_1(\vals_1)$ and $m_2(\vals_2)$, we show that the variables used in predicate expression of $S_1$ and $S_2$ are a subset of imported variables in $S_1$ and $S_2$ relative to $x$. This is true by the definition of imported variables, $\text{Use}(e) \subseteq \text{Imp}(S_1, \{x\}), \text{Use}(e) \subseteq \text{Imp}(S_2, \{x\})$.
By assumption, the value stores $\vals_1$ and $\vals_2$ agree on the values of the variables used in predicate expression of $S_1$ and $S_2$, $\vals_1(y) = \vals_2(y), \forall y \in \text{Use}(e)$.
By the property of expression meaning function $\mathcal{E}$,
 the predicate expression of $S_1$ and $S_2$ evaluate to the same value when started in states $m_1(\vals_1)$ and $m_2(\vals_2)$, $\mathcal{E}\llbracket e\rrbracket\vals_1 = \mathcal{E}\llbracket e\rrbracket\vals_2$, w.l.o.g, $\mathcal{E}\llbracket e\rrbracket\vals_1 = \mathcal{E}\llbracket e\rrbracket\vals_2=(0,v_\mathfrak{of})$. Then the execution of $S_1$ proceeds as follows.

\begin{tabbing}
xx\=xx\=\kill
\>\>     $(S_1,m_1(\vals_1))$\\
\>= \>   $(\text{If }(e) \text{ then }\{S_1^t\} \text{ else }\{S_1^f\}, m_1(\vals_1))$ \\
\>$->$\> $(\text{If }((0,v_\mathfrak{of})) \text{ then }\{S_1^t\} \text{ else }\{S_1^f\}, m_1(\vals_1))$\\
\>\>      by the EEval' rule  \\
\>$->$\> $(\text{If }(0) \text{ then }\{S_1^t\} \text{ else }\{S_1^f\}, m_1(\vals_1))$\\
\>\>      by the E-Oflow1 or E-Oflow2 rule  \\
\>$->$\> $(S_1^f, m_1(\vals_1))$ by the If-F rule.
\end{tabbing}

Similarly, the execution from $(s_2, m_2(\vals_2))$ gets to $(S_2^f, m_2(\vals_2))$.

By the hypothesis IH, we show that $S_1^f$ and  $S_2^f$ compute the variable $x$ equivalently when started in state $m_1(\vals_1)$ and $m_2(\vals_2)$ respectively. To do that, we show that all required conditions are satisfied for the application of hypothesis IH.
\begin{itemize}
\item $\text{size}(S_1^f) + \text{size}(S_2^f) < k$.

Because $\text{size}(S_1) = 1 + \text{size}(S_1^t) + \text{size}(S_1^f)$, $\text{size}(S_2) = 1 + \text{size}(S_2^t) + \text{size}(S_2^f)$.

\item the value stores $\vals_1$ and $\vals_2$ agree on the values of the imported variables in $S_1^f$ and $S_2^f$ relative to $x$, $\vals_1(y) = \vals_2(y), \forall y \in \text{Imp}(S_1^f, \{x\}) \cup \text{Imp}(S_2^f, \{x\})$.

By the definition of imported variables,

$\text{Imp}(S_1^f, \{x\}) \subseteq \text{Imp}(S_1, \{x\}), \text{Imp}(S_2^f, \{x\}) \subseteq \text{Imp}(S_2, \{x\})$.
\end{itemize}
By the hypothesis IH, $S_1^f$ and  $S_2^f$ compute the variable $x$ equivalently when started in state $m_1(\vals_1)$ and $m_2(\vals_2)$ respectively.
Therefore, the theorem holds.

\item $S_1$ and $S_2$  are while statement that define the variable $x$:

   $S_1 = ``\text{while}_{\langle n_1\rangle} (e) \; \{S_1''\}", S_2 = ``\text{while}_{\langle n_2\rangle} (e) \; \{S_2''\}"$ such that both of the following hold:
   \begin{itemize}
   \item $x \in \text{Def}(S_1) \cap \text{Def}(S_2)$;

   \item $S_1'' \equiv_{y}^S S_2''$ for $\forall y \in
        \text{Imp}(S_1, \{x\})
        \cup
        \text{Imp}(S_2, \{x\})$;
   \end{itemize}

By Lemma~\ref{lmm:sameFinalValueX}, we show $S_1$ and $S_2$ compute the variable $x$ equivalently when started from state $m_1(m_c^1, \vals_1)$ and $m_2(m_c^2, \vals_2)$ respectively. The point is to show that all required conditions are satisfied for the application of lemma~\ref{lmm:sameFinalValueX}.
\begin{itemize}
\item loop counter value of $S_1$ and $S_2$ are zero.

By our assumption, the loop counter value of $S_1$ and $S_2$ are initially zero.

\item $S_1$ and $S_2$ have same imported variables relative to $x$,
$\text{Imp}(S_1, \{x\}) = \text{Imp }(S_2, \{x\}) = \text{Imp}(\Delta)$.

This is obtained by Lemma~\ref{lmm:sameImpfromEquivCompCond}.

\item the initial value store $\vals_1$ and $\vals_2$ agree on the values of the imported variables in $S_1$ and $S_2$ relative to $x$, $\vals_1(y) = \vals_2(y), \forall y \in \text{Imp}(S_1, \{x\}) \cup \text{Imp}(S_2, \{x\})$.

  By assumption, this holds.

\item $S_1''$ and $S_2''$ compute the imported variables in $S_1$ and $S_2$ relative to $x$ equivalently,
$(S_1'', m_{S_1''}(\vals_{S_1''})) \equiv_{y} (S_2'', m_{S_2''}(\vals_{S_2''})), \forall y \in \text{Imp}(\Delta)$ with value stores $\vals_{S_1''}$ and $\vals_{S_2''}$ agreeing on the values of the imported variables in $S_1''$ and $S_2''$ relative to $\text{Imp}(\Delta)$,
$\vals_{S_1''}(z) = \vals_{S_2''}(z), \forall z \in \text{Imp}(S_1'', \text{Imp}(\Delta)) \cup \text{Imp}(S_2'', \text{Imp}(\Delta))$.

By the definition of program size, the sum of the program size of $S_1'$ and $S_2'$ is less than $k$, $\text{size}(S_1'') + \text{size}(S_2'') < k$. By the hypothesis IH, $S_1''$ and $S_2''$ compute the imported variables in $S_1$ and $S_2$ relative to $x$ equivalently when started in states $m_{S_1''}(\vals_{S_1''})$ and $m_{S_2''}(\vals_{S_2''})$ with value store $\vals_{S_1''}$ and $\vals_{S_2''}$ agreeing on the values of the imported variables in $S_1''$ and $S_2''$ relative to the variables $\text{Imp}(\Delta)$.
\end{itemize}

By Lemma~\ref{lmm:sameFinalValueX}, we show $S_1$ and $S_2$ compute the variable $x$ equivalently when started from state $m_1(m_c^1, \vals_1)$ and $m_2(m_c^2, \vals_2)$ respectively. The theorem holds.

\item $S_1$ and $S_2$ do not define the variable $x$: $x \notin \text{Def}(S_1) \cup \text{Def}(S_2)$.

By the definition of imported variable, the imported variables in $S_1$ and $S_2$ relative to $x$ are both $x$, $\text{Imp}(S_1, \{x\}) = \text{Imp}(S_2, \{x\}) = \{x\}$.
By assumption, the initial values $\vals_1$ and $\vals_2$ agree on the value of the variable $x$, $\vals_1(x) = \vals_2(x)$.
In addition, by assumption, execution of $S_1$ and $S_2$ when started in state $m_1(\vals_1)$ and $m_2(\vals_2)$ terminate, $(S_1, m_1(\vals_1)) ->* (\text{skip}, m_1'(\vals_1')), (S_2, m_2(\vals_2)) ->* (\text{skip}, m_2'(\vals_2'))$.
Finally, by Corollary~\ref{coro:defExclusion}, the value of $x$ is not changed in execution of $S_1$ and $S_2$, $\vals_1'(x) = \vals_1(x) = \vals_2(x) = \vals_2'(x)$. The theorem holds.

\end{enumerate}

\item $S_1$ and $S_2$ are not both one statement such that one of the following holds:

\begin{enumerate}

\item  Last statements both define the variable $x$ such that all of the following hold:
\begin{itemize}
\item  $S_1' \equiv_{y}^S S_2', \forall y \in \text{Imp}(s_1, \{x\}) \cup \text{Imp}(s_2, \{x\})$;

\item  $x \in \text{Def}(s_1) \cap \text{Def}(s_2)$;

\item  $s_1 \equiv_x^S s_2$;
\end{itemize}

We show that $S_1'$ and $S_2'$ compute the imported variables in $s_1$ and $s_2$ relative to the variable $x$ equivalently when started in state $m_1(\vals_1)$ and $m_2(\vals_2)$ respectively by the hypothesis IH. To do that, we show the required conditions are satisfied for applying the hypothesis IH.
\begin{itemize}

\item $\text{size}(S_1') + \text{size}(S_2') < k$.

By the definition of program size, $\text{size}(s_1) \geq 1$, $\text{size}(s_2) \geq 1$.
Hence, $\text{size}(S_1') + \text{size}(S_2') < k$.

\item the executions from $({S_1}, m_1(\vals_1))$ and $({S_2}, m_2(\vals_2))$ terminate respectively,

$({S_1'}, m_1(\vals_1)) ->* (\text{skip}, m_1''(\vals_1''))$,
           $({S_2'}, m_2(\vals_2)) ->* (\text{skip}, m_2''(\vals_2''))$.

By assumption, the execution from $({S_1}, m_1(\vals_1))$ and $({S_2}, m_2(\vals_2))$ terminate,
then the execution of $S_1'$ and $S_2'$ from state $m_1(\vals_1)$ and $m_2(\vals_2)$ terminate,
 $({S_1'}, m_1(\vals_1)) ->* (\text{skip}, m_1''(\vals_1''))$,
           $({S_2'}, m_2(\vals_2)) ->* (\text{skip}, m_2''(\vals_2''))$.

\item the initial value stores agree on the values of the variables:

$\allowbreak\text{Imp}(S_1', \text{Imp}(s_1, \{x\})) \cup \allowbreak \text{Imp}(S_2', \text{Imp}(s_2, \{x\}))$.


By Lemma~\ref{lmm:sameImpfromEquivCompCond}, $s_1$ and $s_2$ have the same imported variables relative to $x$,
$\text{Imp}(s_1, \{x\}) = \text{Imp}(s_2, \{x\})\allowbreak = \text{Imp}(x)$.
By the definition of imported variables, imported variables in $S_1'$ relative to $\text{Imp}(x)$ are same as the imported variables in $S_1$ relative to $x$, $\text{Imp}(S_1', \text{Imp}(s_1, \{x\})) = \text{Imp}(S_1, \{x\})$. Similarly,
$\text{Imp}(S_2', \text{Imp}(s_2, \{x\})) = \text{Imp}(S_2, \{x\})$.
Then, by assumption, the initial value stores agree on the values of the variables
$\allowbreak\text{Imp}(S_1', \text{Imp}(s_1, \{x\}))$ and
$\forall y \in \text{Imp} (S_1', \text{Imp} (s_1, \{x\}))$ $\cup$ $\text{Imp}(S_2', \text{Imp}(s_2, \{x\}))$,
$\text{Imp}(S_2', \text{Imp}(s_2, \{x\}))$, $\vals_1(y) = \vals_2(y)$.

\end{itemize}
By the hypothesis IH, after the full execution of $S_1'$ from state $m_1(\vals_1)$ and the execution of $S_2'$ from state $m_2(\vals_2)$, the value stores agree on the values of the imported variables in $s_1$ and $s_2$ relative to $x$, $\vals_1''(y) = \vals_2''(y), \forall y \in \text{Imp }(x) = \text{Imp }(s_1, \{x\}) = \text{Imp }(s_2, \{x\})$.

Then, we show $s_1$ and $s_2$ compute $x$ equivalently. By Corollary~\ref{coro:termSeq}, $s_1$ and $s_2$ continue execution after the full execution of $S_1'$ and $S_2'$ respectively,
$({S_1'};s_1, m_1(\vals_1)) ->* ({s_1}, m_1''(\vals_1''))$,
$({S_2'};s_2, m_2(\vals_2)) ->* ({s_2}, m_2''(\vals_2''))$.
When $s_1$ and $s_2$ are while statements, by our assumption of unique loop labels, $s_1$ is not in $S_1'$. By Corollary~\ref{coro:loopCntRemainsSame}, the loop counter value of $s_1$ is not redefined in the execution of $S_1'$.
Similarly, the loop counter value of $s_2$ is not redefined in the execution of $S_2'$.
By the hypothesis IH again, after the full execution of $s_1$ and $s_2$, the value stores agree on the value of $x$,
$({s_1}, m_1''(\vals_1'')) ->* (\text{skip}, m_1'(\vals_1')), \allowbreak ({s_2}, m_2''(\vals_2'')) ->* (\text{skip}, m_2'(\vals_2'))$ such that $\vals_1'(x) = \vals_2'(x)$.
The theorem holds.

\item One last statement does not define the variable $x$:
W.l.o.g., $(x \notin \text{ Def}(s_2)) \wedge (S_1 \equiv_{x}^S S_2')$.

We show that $S_1$ and $S_2'$ compute the variable $x$ equivalently when started from state $m_1(\vals_1)$ and $m_2(\vals_2)$ by the hypothesis IH.
First, by the definition of program size, $\text{size}(s_2) \geq 1$. Hence, $\text{size}(S_1) + \text{size}(S_2') \leq k$ .
Next,
by the definition of imported variables, $\text{ Imp }(S_2', \{x\}) \subseteq \text{ Imp }(S_2, \{x\})$.
By assumption, $\vals_1(y) = \vals_2(y)$ for $\forall y \in
\text{ Imp }(S_2', \{x\})
\cup
\text{ Imp }(S_1, \{x\})$.
By the hypothesis IH, $S_1$ and $S_2'$ compute the variable $x$ equivalently when started in state $m_1(\vals_1)$ and $m_2(\vals_2)$ respectively,
$(S_2', m_2(\vals_2)) ->* (\text{skip}, m_2''(\vals_2''))$,
     $({S_1}, m_1(\vals_1)) ->* (\text{skip}, m_1'(\vals_1'))$ such that
$\vals_1'(x) = \vals_2''(x)$.

Then, we show that $S_1$ and $S_2$ compute the variable $x$ equivalently after the full execution of $s_2$. By Corollary~\ref{coro:termSeq}, $s_2$ continues execution immediately after the full execution of $S_2'$,
$({S_2';s_2}, m_2) ->* ({s_2}, m_2'')$.
By assumption, the execution from $({S_2';s_2}, m_2)$ terminates,
$({s_2}, m_2''(\vals_2'')) ->* (\text{skip}, m_2'(\vals_2'))$.
By Corollary~\ref{coro:defExclusion}, the value of $x$ is not changed in the execution of $s_2$, $\vals_2'(x) = \vals_2''(x)$.
Hence, $\vals_1'(x) = \vals_2'(x)$. The theorem holds.

\item There are statements moving in/out of If statement:

      $s_1 = ``\text{If }(e) \text{ then }\{S_1^t\} \text{ else }\{S_1^f\}"$,
      $s_2 = ``\text{If }(e) \text{ then }\{S_2^t\}\allowbreak \text{ else }\{S_2^f\}"$ such that none of the above cases hold and all of the following hold:
     \begin{itemize}
     \item $S_1' \equiv_{y}^S S_2'$ for $\forall y \in \text{ Use}(e)$;

     \item $S_1';S_1^t \equiv_{x}^S S_2';S_2^t$;

     \item $S_1';S_1^f \equiv_{x}^S S_2';S_2^f$;

     \item $x \in \text{Def}(s_1) \cap \text{Def}(s_2)$;
     \end{itemize}

 Repeatedly $S_1 = S_1';s_1, S_2 = S_2';s_2$. We first show that, after the full execution of $S_1'$ and $S_2'$ started in state $m_1$ and $m_2$, the predicate expression of $s_1$ and $s_2$ evaluate to the same value, w.l.o.g, zero.
Next we show that $S_1$ and $S_1';S_1^f$ compute the variable $x$ equivalently when
(1) both started in state $m_1$ and
(2) the predicate expression of $s_1$ evaluates to zero after the full execution of $S_1'$ started in state $m_1$, similarly $S_2$ and $S_2';S_2^f$ compute the variable $x$ equivalently when
(1) both started in the state $m_2$ and
(2) the predicate expression of $s_2$ evaluates to zero after the full execution of $S_1'$ when started in state $m_2$. Last we prove the theorem by showing that $S_1';S_1^f$ started in state $m_1$ and $S_2';S_2^f$ started in state $m_2$ compute the variable $x$ equivalently.

\myNewLine

In order to show that $S_1'$ and $S_2'$ compute the variables used in predicate expression of $s_1$ and $s_2$ equivalently by the hypothesis IH, we show that all required conditions are satisfied for the application of hypothesis IH.
 \begin{itemize}
  \item $\text{size}(S_1')+\text{size}(S_2')<k$.

The sum of program size of $S_1'$ and $S_2'$ are less than $k$ by the definition of program size for $s_1$ and $s_2$, $\text{size}(S_1')+\text{size}(S_2')<k$.

 \item the execution of $S_1'$ and $S_2'$ terminate, $({S_1'}, m_1) ->* (\text{skip}, m_1''(\vals_1''))$, and
$({S_2'}, m_2) ->* (\text{skip}, m_2''(\vals_2''))$.

By assumption, the execution of $S_1$ and $S_2$ from the state $m_1$ and $m_2$ respectively terminate, then the execution of $S_1'$ and $S_2'$ terminate when started in state $m_1$ and $m_2$ respectively.

 \item the initial value stores $\vals_1$ and $\vals_2$ agree on the values of the imported variables in $S_1'$ and $S_2'$ relative to the variables used in the predicate expression of $s_1$ and $s_2$.

By Lemma~\ref{lmm:sameImpfromEquivCompCond}, the imported variables in $S_1'$ and $S_2'$ relative to the variables used in predicate expression of $s_1$ and $s_2$ are same,
 $\text{Imp}(S_1', \text{Use}(e)) = \text{Imp}(S_2', \text{Use}(e)) = \text{Imp}(e)$.
By the definition of imported variable, the imported variables in $S_1'$ relative to the variables used in predicate expression of $s_1$ are a subset of the imported variables in $S_1$ relative to $x$ respectively, $\text{Imp}(S_1', \text{Use}(e)) \subseteq \text{Imp}(S_1', \text{Imp}(s_1, \{x\})) = \text{Imp}(S_1, \{x\})$. Similarly $\text{Imp}(S_2', \text{Use}(e)) \subseteq \text{Imp}(S_2, \{x\})$.
Then, by assumption, the initial value stores agree on the values of the imported variables in $S_1'$ and $S_2'$ relative to the variables used in the predicate expression of $s_1$ and $s_2$, $\vals_1(y) = \vals_2(y), \forall y \in \text{Imp}(e) = \text{Imp}(S_1', \text{Use}(e)) = \text{Imp}(S_2', \text{Use}(e))$.

\end{itemize}
By the hypothesis IH, after the full execution of $S_1'$ and $S_2'$, the value stores agree on the values of the variables used in the predicate expression of $s_1$ and $s_2$, $\vals_1''(y) = \vals_2''(y), \forall y \in \text{Use}(e)$. By Corollary~\ref{coro:termSeq},
$s_1$ and $s_2$ continue execution after the full execution of $S_1'$ and $S_2'$ respectively, $({S_1';s_1}, m_1) ->* (s_1, m_1''(\vals_1''))$, and
$({S_2'};s_2, m_2) ->* (s_2, m_2''(\vals_2''))$.

 By the property of expression meaning function $\mathcal{E}$, expression $e$ evaluates to the same value w.r.t value stores $\vals_1''$ and $\vals_2''$, w.l.o.g., zero, $\mathcal{E}\llbracket e\rrbracket\vals_1'' = \mathcal{E}\llbracket e\rrbracket\vals_2'' = 0$. Then the execution of $s_1$ proceeds as follows.

\begin{tabbing}
xx\=xx\=\kill
\>\>     $(s_1,m_1''(\vals_1''))$\\
\>= \>   $(\text{If }(e) \text{ then }\{S_1^t\} \text{ else }\{S_1^f\}, m_1''(\vals_1''))$ \\
\>$->$\> $(\text{If }(0) \text{ then }\{S_1^t\} \text{ else }\{S_1^f\}, m_1''(\vals_1''))$ by the EEval rule.  \\
\>$->$\> $(S_1^f, m_1''(\vals_1''))$ by the If-F rule.
\end{tabbing}

Similarly, the execution from $(s_2, m_2''(\vals_2''))$ gets to $(S_2^f, m_2''(\vals_2''))$.

\myNewLine

Then, we show that $S_1$ and $S_1';S_1^f$ compute the variable $x$ equivalently when both started from state $m_1(\vals_1)$.
 The execution of ${S_1'};S_1^f$ started from state $m_1$ also gets to configuration $(S_1^f, m_1''(\vals_1''))$ because execution of $S_1=S_1';s_1$ and ${S_1'};S_1^f$ share the common execution $({S_1'}, m_1) ->* (\text{skip}, m_1''(\vals_1''))$. By Corollary~\ref{coro:termSeq}, $S_1^f$ continues execution after the full execution of $S_1'$, $({S_1'};S_1^f, m_1) ->* ({S_1^f}, m_1'')$.
Therefore, the execution of $S_1$ and ${S_1'};S_1^f$ from state $m_1$ compute the variable $x$ equivalently because both executions get to same intermediate configuration. Similarly, $S_2$ and ${S_2'};S_2^f$  compute the variable $x$ equivalently when both started from state $m_2(\vals_2)$.

Lastly, we show that $S_1';S_1^f$ and $S_2';S_2^f$ compute the variable $x$ equivalently when started in states $m_1(\vals_1)$ and $m_2(\vals_2)$ respectively by the hypothesis IH. To do that, we show that all required conditions are satisfied for the application of hypothesis IH.
\begin{itemize}
\item $\text{size}(S_1';S_1^f) + \text{size}(S_2';S_2^f) < k$.

This is obtained by the definition of program size.

\item execution of $S_1';S_1^f$ and $S_2';S_2^f$ terminate when started in state $m_1(\vals_1)$ and $m_2(\vals_2)$ respectively.

This is obtained by above argument.

\item $\vals_1(y) = \vals_2(y),  \forall y \in \text{Imp}(S_1';S_1^f, \{x\}) \allowbreak \cup \text{Imp}(S_2';S_2^f, \{x\})$.

We show that $\text{Imp}(S_1';S_1^f, \{x\}) \subseteq \text{Imp}(S_1, \{x\})$ as follows.
\begin{tabbing}
xx\=xx\=\kill
\>\>             $\text{Imp}(S_1^f, \{x\})$\\
\>$\subseteq$\>  $\text{Imp}(s_1, \{x\})$  (1) by the definition of imported variables.\\
\end{tabbing}

\begin{tabbing}
xx\=xx\=\kill
\>\>             $\text{Imp}(S_1';S_1^f, \{x\})$\\
\>=\>            $\text{Imp}(S_1', \text{Imp}(S_1^f, \{x\}))$ by Lemma~\ref{lmm:ImpPrefixLemma}\\
\>$\subseteq$\>  $\text{Imp}(S_1', \text{Imp}(s_1, \{x\}))$ by (1)\\
\>=\>            $\text{Imp}(S_1, \{x\})$ by the definition of imported variables.
\end{tabbing}
Similarly, $\text{Imp}(S_2';S_2^f, \{x\})  \subseteq \text{Imp}(S_2, \{x\})$.
Then, by assumption, the initial value stores agree on the values of the imported variables in $S_1';S_1^f$ and $S_2';S_2^f$ relative to $x$.
\end{itemize}
Then, by the hypothesis IH, after the full execution of $S_1';S_1^f$ and $S_2';S_2^f$, the value stores agree on the value of $x$,
$({S_1'};S_1^f, m_1) ->* (\text{skip}, m_1'(\vals_1'))$,
$({S_2'};S_2^f, m_2) ->* (\text{skip}, m_2'(\vals_2'))$ such that $\vals_1'(x) = \vals_2'(x)$.

In conclusion, after execution of $S_1$ and $S_2$, the value stores agree on the value of $x$.
Therefore, the theorem holds.

\end{enumerate}
\end{enumerate}
\end{proof}

\subsubsection{Supporting lemmas for the soundness proof of equivalent computation for terminating programs}
The lemmas include the proof of two while statements computing a variable equivalently used in the proof of Theorem~\ref{thm:equivCompMain} and
the property that two programs have same imported variables relative to a variable $x$ if the two programs satisfy the proof rule of equivalent computation of the variable $x$.
From the proof rule of terminating computation of a variable $x$ equivalently,
we have the two programs either both define $x$ or both do not.

\begin{lemma}\label{lmm:equivTermCompSameLoopIteration}
Let $s_1$ = ``$\text{while}_{\langle n_1\rangle} (e) \; \{S_1\}$" and $s_2$ = ``$\text{while}_{\langle n_2\rangle} (e) \;\allowbreak \{S_2\}$" be two while statements with the same set of imported variables relative to a variable $x$ (defined in $s_1$ and $s_2$), $\text{Imp}(x)$, and whose loop bodies $S_1$ and $S_2$ terminatingly compute the variables in $\text{Imp}(x)$ equivalently when started in states that agree on the values of the variables imported by $S_1$ or $S_2$ relative to $\text{Imp}(x)$:
\begin{itemize}
\item $x \in \text{Def}(s_1) \cap \text{Def}(s_2)$;

\item $\text{Imp}(s_1, \{x\}) = \text{Imp}(s_2, \{x\}) = \text{Imp}(x)$;

\item $\forall y \in \text{Imp}(x), \forall m_{S_1}(\vals_{S_1}), m_{S_2}(\vals_{S_2}):$

$((\forall z\in \text{Imp}(S_1, \text{Imp}(x)) \cup \text{Imp}(S_2, \text{Imp}(x))\,:\, \vals_{S_1}(z) = \vals_{S_2}(z))\allowbreak
=> (S_1, m_{S_1}(\vals_{S_1})) \equiv_{y} (S_2, m_{S_2}(\vals_{S_2})))$.
\end{itemize}

If the executions of $s_1$ and $s_2$ terminate when started in states $m_1(\text{loop}_c^1, \vals_1)$ and $m_2(\text{loop}_c^2, \vals_2)$ in which
$s_1$ and $s_2$ have not already executed (loop counter initially 0:
$\text{loop}_c^1(n_{1}) = \text{loop}_c^2(n_{2}) = 0$),
and whose value stores $\vals_1$ and $\vals_2$ agree on the values of the variables in $\text{Imp}(x)$,
$ \forall y \in \text{Imp}(x),\; \vals_1(y) = \vals_2(y)$, then, for any positive integer $i$, one of the following holds:
\begin{enumerate}
\item{The loop counters for $s_1$ and $s_2$ are always less than $i$:}

$\forall m_{1}', m_{2}'$ such that
$({s_1}, m_1) ->* ({S_1'}, m_{1}'(\text{loop}_c^{1'})$ and
$({s_2}, m_2) \allowbreak->* ({S_2'}, m_{2}'(\text{loop}_c^{2'}))$,

$\text{loop}_c^{1'}(n_1) < i$ and $\text{loop}_c^{2'}(n_2) < i$;

\item{There are two configurations $(s_1, m_{1_i})$ and $(s_2, m_{2_i})$ reachable from $(s_1, m_1)$ and $(s_2, m_2)$, respectively, in which the loop counters of $s_1$ and $s_2$ are equal to $i$ and value stores agree on the values of imported variables relative to $x$ and, for every state in execution, $(s_1, m_1) ->* (s_1, m_{1_i})$ or $(s_2, m_2) ->* (s_2, m_{2_i})$ the loop counters for $s_1$ and $s_2$ are less than or equal to $i$ respectively:}

 $\exists (s_1, m_{1_i}), (s_2, m_{2_i})\,:\, ({s_1}, m_1) ->* ({s_1}, m_{1_i}(\text{loop}_c^{1_i}, \vals_{1_i})) \wedge
 ({s_2}, m_2) ->* ({s_2}, m_{2_i}(\text{loop}_c^{2_i}, \vals_{2_i}))$ where
\begin{itemize}
 \item $\text{loop}_c^{1_i}(n_1) = \text{loop}_c^{2_i}(n_2) = i$; and

 \item $\forall y \in \text{Imp}(x)\,:\, \vals_{1_i}(y) = \vals_{2_i}(y)$ and

 \item $\forall m_1'\,:\, ({s_1}, m_1) ->* (S_1', m_1'(\text{loop}_c^{1'})) ->*$

 \noindent$({s_1}, m_{1_i}(\text{loop}_c^{1_i}, \vals_{1_i})), \; \text{loop}_c^{1'}(n_1) \leq i$; and

 \item $\forall m_2'\,:\, ({s_2}, m_2) ->* (S_2', m_2'(\text{loop}_c^{2'})) ->*$

 \noindent$({s_2}, m_{2_i}(\text{loop}_c^{2_i}, \vals_{2_i})), \; \text{loop}_c^{2'}(n_2) \leq i$;
\end{itemize}
\end{enumerate}
\end{lemma}

\begin{proof}
By induction on $i$.

\noindent{\bf Base case}. $i = 1$.

By assumption, initial loop counters of $s_1$ and $s_2$ are of value zero. Initial value stores $\vals_1$ and $\vals_2$ agree on the values of the variables in $\text{Imp}(x)$. Then we show one of the following cases hold:
\begin{enumerate}
\item{The loop counters for $s_1$ and $s_2$ are always less than 1:}

$\forall m_{1}', m_{2}'$ such that
$({s_1}, m_1) ->* ({S_1'}, m_1'(\text{loop}_c^{1'}))$ and
$({s_2}, m_2) \allowbreak->* ({S_2'}, m_2'(\text{loop}_c^{2'}))$,
$\text{loop}_c^{1'}(n_1) < 1$ and $\text{loop}_c^{2'}(n_2) \allowbreak< 1$;

\item{There are two configurations $(s_1, m_{1_1})$ and $(s_2, m_{2_1})$ reachable from $(s_1, m_1)$ and $(s_2, m_2)$, respectively, in which the loop counters of  $s_1$ and $s_2$ are equal to 1 and value stores agree on the values of imported variables relative to $x$ and, for every state in execution, $(s_1, m_1) ->* (s_1, m_{1_1})$ or $(s_2, m_2) ->* (s_2, m_{2_1})$ the loop counters for $s_1$ and $s_2$ are less than or equal to one respectively:}

 $\exists (s_1, m_{1_1}), (s_2, m_{2_1})\,:\,
 ({s_1}, m_1) ->* ({s_1}, m_{1_1}(\text{loop}_c^{1_1}, \vals_{1_1})) \wedge
 ({s_2}, m_2) ->* ({s_2}, m_{2_1}(\text{loop}_c^{2_1}, \vals_{2_1}))$ where
\begin{itemize}
 \item $\text{loop}_c^{1_1}(n_1) = \text{loop}_c^{2_1}(n_2) = 1$; and

 \item $\forall y \in \text{Imp}(x)\,:\,
   \vals_{1_1}(y) = \vals_{2_1}(y)$; and

 \item $\forall m_1'\,:\,
 ({s_1}, m_1) ->* (S_1', m_1'(\text{loop}_c^{1'})) ->* ({s_1}, m_{1_1}(\text{loop}_c^{1_1}, \vals_{1_1})), \; \allowbreak \text{loop}_c^{1'}(n_1) \leq 1$; and

 \item $\forall m_2'\,:\, ({s_2}, m_2) ->* (S_2', m_2'(\text{loop}_c^{2'})) ->* ({s_2}, m_{2_1}(\text{loop}_c^{2_1}, \vals_{2_1})), \; \allowbreak \text{loop}_c^{2'}(n_2) \leq 1$.
\end{itemize}
\end{enumerate}

We show evaluations of the predicate expression of $s_1$ and $s_2$ w.r.t value stores $\vals_1$ and $\vals_2$ produce same value.
By the definition of imported variables,
$\text{Imp}(s_1, \{x\}) = \bigcup_{j\geq0}\text{Imp}(S_1^j, \{x\} \cup \text{Use}(e))$.
By our notation of $S^0$, $S_1^0 = \text{skip}$.
By the definition of imported variables, $\text{Imp}(S_1^0, \{x\} \cup \text{Use}(e)) = \{x\} \cup \text{Use}(e)$.
Then $\text{Use}(e) \subseteq \text{Imp}(x)$. By assumption, value stores $\vals_1$ and $\vals_2$ agree on the values of the variables in $\text{Use}(e)$.
By Lemma~\ref{lmm:expEvalSameVal}, the predicate expression $e$ of $s_1$ and $s_2$ evaluates to same value $v$ w.r.t value stores $\vals_1, \vals_2$,
$\mathcal{E}'\llbracket e\rrbracket\vals_1 = \mathcal{E}'\llbracket e\rrbracket\vals_2 = v$.
Then there are two possibilities to consider.
\begin{enumerate}
\item{$\mathcal{E}'\llbracket e\rrbracket\vals_1 =
       \mathcal{E}'\llbracket e\rrbracket\vals_2 = v = 0$}

The execution from $({s_1}, m_1(\text{loop}_c^{1}, \vals_1))$ proceeds as follows.

\begin{tabbing}
xx\=xx\=\kill
\>\>     $({s_1}, m_{1}(\text{loop}_c^1,\vals_1))$\\
\>= \>   $(\text{while}_{\langle n_1\rangle} (e) \; \{S_1\}, m_{1}(\text{loop}_c^{1}))$\\
\>$->$\> $(\text{while}_{\langle n_1\rangle} (0) \; \{S_1\}, m_{1}(\text{loop}_c^{1}))$ by the EEval' rule\\
\>$->$\> $({\text{skip}}, m_1(\text{loop}_c^{1}[0/n_1]))$ by the Wh-F rule.
\end{tabbing}

Similarly, $({s_2}, m_{2}(\text{loop}_c^{2},\vals_{2}))$ {\kStepArrow [2] }
           $({\text{skip}}, m_2(\text{loop}_c^{2}[0/n_2]))$.

In conclusion, the loop counters of $s_1$ and $s_2$ in any states of the execution from $({s_1}, m_1)$ and $({s_2}, m_2)$ respectively are less than 1,
$\forall m_1', m_2'$ such that
$({s_1}, m_1) ->* ({S_1}', m_{1}'(\text{loop}_c^{1'})$ and
$({s_2}, m_2) ->* ({S_2}', m_{2}'(\text{loop}_c^{2'}))$,
$\text{loop}_c^{1'}(n_1) < 1$, $\text{loop}_c^{2'}(n_2) < 1$.

\item{$\mathcal{E}'\llbracket e\rrbracket\vals_{1} =
       \mathcal{E}'\llbracket e\rrbracket\vals_{2} = v \neq 0$}

The execution from $({s_1}, m_{1}(\text{loop}_c^{1}, \vals_{1}))$ proceeds as follows.

\begin{tabbing}
xx\=xx\=\kill
\>\>     $({s_1}, m_{1}(\text{loop}_c^{1}, \vals_{1}))$\\
\>= \>   $(\text{while}_{\langle n_1\rangle} (e) \; \{S_1\}, {m_{1}(\text{loop}_c^{1}, \vals_{1})})$\\
\>$->$\> $(\text{while}_{\langle n_1\rangle} (v) \; \{S_1\}, {m_{1}(\text{loop}_c^{1}, \vals_{1})})$ by the EEval' rule\\
\>$->$\> $(S_1;\text{while}_{\langle n_1\rangle} (e) \; \{S_1\}, m_{1}(\text{loop}_c^{1}[1/(n_1)], \vals_{1}))$\\
\>\>     by the Wh-T rule.
\end{tabbing}

Similarly,
$({s_2}, m_{2}(\text{loop}_c^{2}, \vals_{2}))$ {\kStepArrow [2] }
$(S_2;\text{while}_{\langle n_2\rangle} (e) \{S_2\}, \allowbreak m_{2}\allowbreak(\text{loop}_c^{2}[1/(n_2)], \vals_{2}))$.
Then, the loop counters of $s_1$ and $s_2$ are 1, value stores $\vals_{1_{1}}$ and $\vals_{2_{1}}$
agree on values of variables in $\text{Imp}(x)$:
$\text{loop}_c^{1}[1/n_1](n_1) =$

\noindent$\text{loop}_c^{2}[1/(n_2)](n_2) = 1$; and
$\forall y \in \text{Imp}(x), \vals_{1_{1}}(y)$ = $\vals_{2_{1}}(y)$.
By assumption, the execution of $s_1$ terminates when started in the state $m_1(, \vals_1)$, then the execution of $S_1$ terminates when started in the state
$m_{1}(\text{loop}_c^{1}[1/(n_1)], \vals_{1})$,
$({s_1}, m_1) ->*$

\noindent$(S_1;\text{while}_{\langle n_1\rangle} (e) \{S_1\}, m_{1}(\text{loop}_c^{1}[1/(n_1)], \vals_{1})) ->*
 (\text{skip}, m_{1}') =>$

$(S_1, m_{1}(\text{loop}_c^{1}[1/(n_1)], \vals_{1})) ->*
(\text{skip}, m_{1_{1}}(\text{loop}_c^{1_{1}}, \vals_{1_{1}}))$.
Similarly, the execution of $S_2$ terminates when started in the state
$m_{2}(\text{loop}_c^{2}[1/(n_2)], \vals_{2})$,

\noindent$(S_2, m_{2}(\text{loop}_c^{2}[1/(n_2)], \vals_{2})) ->* (\text{skip}, m_{2_{1}}(\text{loop}_c^{2_{1}}, \vals_{2_{1}}))$.

We show that, after the full execution of $S_1$ and $S_2$, the following four properties hold.
\begin{itemize}
\item The loop counters of $s_1$ and $s_2$ are of value $1$, $\text{loop}_c^{1_{1}}(n_1) = \text{loop}_c^{2_{1}}(n_2) = 1$.

By assumption of unique loop labels, $s_1 \notin S_1$. Then, the loop counter value of $n_1$ is not redefined in the execution of $S_1$ by corollary~\ref{coro:loopCntRemainsSame}, $\text{loop}_c^{1}[1/(n_1)](n_1)$ = $\text{loop}_c^{1_{1}}(n_1) = 1$.
Similarly, the loop counter value of $n_2$ is not redefined in the execution of $S_2$, $\text{loop}_c^{2}[1/(n_2)](n_2)$ = $\text{loop}_c^{2_{1}}(n_2) = 1$.

\item In any state in the execution $(s_1, m_1) ->* (s_1, m_{1_1}(\text{loop}_c^{1_1}, \vals_{1_1}))$, the loop counter of $s_1$ is less than or equal to 1.

    The loop counter of $s_1$ is zero in any of the two states in the one step execution
    $({s_1}, m_1) -> (\text{while}_{\langle n_1\rangle} (v) \; \{S_1\}, \allowbreak{m_{1}(\text{loop}_c^{1},  , \vals_{1})})$, and
    the loop counter of $s_1$ is 1 in any states in the execution

    \noindent$(S_1;\text{while}_{\langle n_1\rangle} (e) \; \{S_1\}, m_{1}(\text{loop}_c^{1}[i/(n_1)], \vals_{1})) ->*$

    \noindent$(s_1, m_{1_1}(\text{loop}_c^{1_1}, \vals_{1_1}))$.

\item In any state in the executions $(s_2, m_2) ->* (s_2, m_{2_1}(\text{loop}_c^{2_1}, \vals_{2_1}))$, the loop counter of $s_2$ is less than or equal to 1.

    By similar argument for the loop counter of $s_1$.

\item The value stores $\vals_{1_{1}}$ and $\vals_{2_{1}}$ agree on the values of the imported variables in $s_1$ and $s_2$ relative to the variable $x$: $\forall y \in \text{Imp}(x), \vals_{1_{1}}(y) = \vals_{2_{1}}(y)$.


We show that the imported variables in $S_1$ relative to those in $\text{Imp}(x)$ are a subset of $\text{Imp}(x)$.


\begin{tabbing}
xx\=xx\=\kill
\>\>              $\text{Imp}(S_1, \text{Imp}(x))$ \\
\>=\>             $\text{Imp}(S_1, \text{Imp}(s_1, \{x\} \cup \text{Use}(e)))$ by the definition of $\text{Imp}(x)$\\
\>=\>             $\text{Imp}(S_1, \bigcup_{j\geq0} \text{Imp}(S_1^j, \{x\} \cup \text{Use}(e)))$\\
\>\>               by the definition of imported variables\\
\>=\>             $\bigcup_{j\geq0}\text{Imp}(S_1, \text{Imp}(S_1^j, \{x\} \cup \text{Use}(e)))$ by Lemma~\ref{lmm:impVarUnionLemma}\\
\>=\>             $\bigcup_{j>0}\text{Imp}(S_1^j, \{x\} \cup \text{Use}(e))$ by Lemma~\ref{lmm:ImpPrefixLemma}\\
\>$\subseteq$\>   $\bigcup_{j\geq0} \text{Imp}(S_1^j, \{x\} \cup \text{Use}(e))$ \\
\>=\>             $\text{Imp}(s_1, \{x\} \cup \text{Use}(e)) = \text{Imp}(x)$.
\end{tabbing}

Similarly, $\text{Imp}(S_2, \text{Imp}(x)) \subseteq \; \text{Imp}(x)$.
Consequently, the value stores $\vals_{1_1}$ and $\vals_{2_1}$ agree on the values of the imported variables in $S_1$ and $S_2$ relative to those in $\text{Imp}(x)$, $\forall y \in \text{Imp}(S_1, \text{Imp}(x))\allowbreak \cup \text{Imp}(S_2, \text{Imp}(x)), \vals_{1}(y)$ = $\vals_{2}(y)$.
Because $S_1$ and $S_2$ have computation of every variable in $\text{Imp}(x)$ equivalently when started in states agreeing on the values of the imported variables relative to $\text{Imp}(x)$, then value store $\vals_{1_{1}}$ and $\vals_{2_{1}}$ agree on the values of the variables $\text{Imp}(x)$,
$\forall y \in \text{Imp}(x), \vals_{1_{1}}(y)$ = $\vals_{2_{1}}(y)$.
 \end{itemize}

It follows that, by corollary~\ref{coro:termSeq},

$(S_1;s_1, m_{1}(\text{loop}_c^{1}[1/n_1], \vals_{1})) ->*$
$(s_1, m_{1_{1}}(\text{loop}_c^{1_{1}}, \vals_{1_{1}}))$ and
$(S_2;s_2, m_{2}(\text{loop}_c^{2}[1/n_2], \vals_{2})) ->*$
$(s_2, m_{2_{1}}(\text{loop}_c^{2_{1}}, \vals_{2_{1}}))$.
\end{enumerate}


\noindent{\bf Induction Step}.

\noindent The induction hypothesis IH is that, for a positive integer $i$, one of the following holds:
\begin{enumerate}
\item{The loop counters for $s_1$ and $s_2$ are always less than $i$:}

$\forall m_{1}', m_{2}'$ such that
$({s_1}, m_1) ->* ({S_1'}, m_1'(\text{loop}_c^{1'})$ and
$({s_2}, m_2) ->* ({S_2'}, m_2'(\text{loop}_c^{2'}))$,

$\text{loop}_c^{1'}(n_1) < i$ and $\text{loop}_c^{2'}(n_2) < i$;

\item{There are two configurations $(s_1, m_{1_i})$ and $(s_2, m_{2_i})$ reachable from $(s_1, m_1)$ and $(s_2, m_2)$, respectively, in which the loop counters of $s_1$ and $s_2$ are equal to $i$ and value stores agree on the values of imported variables relative to $x$ and, for every state in execution, $(s_1, m_1) ->* (s_1, m_{1_i})$ and $(s_2, m_2) ->* (s_2, m_{2_i})$ the loop counters for $s_1$ and $s_2$ are less than or equal to $i$ respectively:}

 $\exists (s_1, m_{1_i}), (s_2, m_{2_i})\,:\,
 (s_1, m_1) ->* (s_1, m_{1_i}(\text{loop}_c^{1_i}, \vals_{1_i})) \wedge
 (s_2, m_2) ->* (s_2, m_{2_i}(\text{loop}_c^{2_i}, \vals_{2_i}))$ where
\begin{itemize}
 \item $\text{loop}_c^{1_i}(n_1) = \text{loop}_c^{2_i}(n_2) = i$; and

 \item $\forall y \in \text{Imp}(x), \vals_{1_i}(y) = \vals_{2_i}(y)$; and

 \item $\forall m_1'\,:\, ({s_1}, m_1) ->* (S_1', m_1'(\text{loop}_c^{1'}) ->*$

 \noindent$({s_1}, m_{1_i}(\text{loop}_c^{1_i}, \vals_{1_i})), \; \text{loop}_c^{1'}(n_1) \leq i$; and

 \item $\forall m_2'\,:\, ({s_2}, m_2) ->* (S_2', m_2'(\text{loop}_c^{2'})) ->*$

 \noindent$({s_2}, m_{2_i}(\text{loop}_c^{2_i}, \vals_{2_i})), \; \text{loop}_c^{2'}(n_2) \leq i$.
\end{itemize}
\end{enumerate}

\noindent Then we show that, for the positive integer $i+1$, one of the following holds:
\begin{enumerate}
\item{The loop counters for $s_1$ and $s_2$ are always less than $i+1$:}

$\forall m_{1}', m_{2}'$ such that
$({s_1}, m_1) ->* ({S_1'}, m_1'(\text{loop}_c^{1'}))$ and
$({s_2}, m_2) \allowbreak->* ({S_2'}, m_2'(\text{loop}_c^{2'}))$,

$\text{loop}_c^{1'}(n_1) < i+1$ and $\text{loop}_c^{2'}(n_2) < i+1$;

\item{There are two configurations $(s_1, m_{1_{i+1}})$ and $(s_2, m_{2_{i+1}})$ reachable from $(s_1, m_1)$ and $(s_2, m_2)$, respectively, in which the loop counters of $s_1$ and $s_2$ are equal to $i+1$ and value stores agree on the values of imported variables relative to $x$ and, for every state in executions $(s_1, m_1) ->* (s_1, m_{1_{i+1}})$ and $(s_2, m_2) ->* (s_2, m_{2_{i+1}})$ the loop counters for $s_1$ and $s_2$ are less than or equal to $i+1$ respectively:}

 $\exists (s_1, m_{1_{i+1}}), (s_2, m_{2_{i+1}})\,:\,
 (s_1, m_1) ->* (s_1, m_{1_{i+1}}(\text{loop}_c^{1_{i+1}},\allowbreak \vals_{1_{i+1}})) \wedge
 (s_2, m_2) ->* (s_2, m_{2_{i+1}}(\text{loop}_c^{2_{i+1}}, \vals_{2_{i+1}}))$ where
\begin{itemize}
 \item $\text{loop}_c^{1_{i+1}}(n_1) = \text{loop}_c^{2_{i+1}}(n_2) = i+1$; and

 \item $\forall y \in \text{Imp}(x), \vals_{1_{i+1}}(y) = \vals_{2_{i+1}}(y)$; and

 \item $\forall m_1'\,:\, ({s_1}, m_1) ->* (S_1', m_1'(\text{loop}_c^{1'})) ->*$

 \noindent$({s_1}, m_{1_{i+1}}(\text{loop}_c^{1_{i+1}}, \vals_{1_{i+1}})), \; \text{loop}_c^{1'}(n_1) \leq i+1$; and

 \item $\forall m_2'\,:\, ({s_2}, m_2) ->* (S_2', m_2'(\text{loop}_c^{2'})) ->*$

  \noindent$({s_2}, m_{2_{i+1}}(\text{loop}_c^{2_{i+1}}, \vals_{2_{i+1}})), \; \text{loop}_c^{2'}(n_2) \leq i+1$.
\end{itemize}
\end{enumerate}

\noindent By the hypothesis IH, one of the following holds:
\begin{enumerate}
\item{The loop counters for $s_1$ and $s_2$ are always less than $i$:}

$\forall m_{1}', m_{2}'$ such that
$({s_1}, m_1) ->* ({S_1'}, m_1'(\text{loop}_c^{1'})$ and
$({s_2}, m_2) \allowbreak ->* ({S_2'}, m_2'(\text{loop}_c^{2'}))$,

$\text{loop}_c^{1'}(n_1) < i$ and $\text{loop}_c^{2'}(n_2) < i$;

When this case holds, then we have the loop counters for $s_1$ and $s_2$ are always less than $i+1$:

$\forall m_{1}', m_{2}'$ such that
$({s_1}, m_1) ->* ({S_1'}, m_1'(\text{loop}_c^{1'})$ and
$({s_2}, m_2) \allowbreak ->* ({S_2'}, m_2'(\text{loop}_c^{2'}))$,

$\text{loop}_c^{1'}(n_1) < i+1$ and $\text{loop}_c^{2'}(n_2) < i+1$.

\item{There are two configurations $(s_1, m_{1_i})$ and $(s_2, m_{2_i})$ reachable from $(s_1, m_1)$ and $(s_2, m_2)$, respectively, in which the loop counters of $s_1$ and $s_2$ are equal to $i$ and value stores agree on the values of imported variables relative to $x$ and, for every state in executions $(s_1, m_1) ->* (s_1, m_{1_i})$ and $(s_2, m_2) ->* (s_2, m_{2_i})$ the loop counters for $s_1$ and $s_2$ are less than or equal to $i$ respectively:}

 $\exists (s_1, m_{1_i}), (s_2, m_{2_i})\,:\,
 (s_1, m_1) ->* (s_1, m_{1_i}(\text{loop}_c^{1_i}, \vals_{1_i})) \wedge
 (s_2, m_2) ->* (s_2, m_{2_i}(\text{loop}_c^{2_i}, \vals_{2_i}))$ where
\begin{itemize}
 \item $\text{loop}_c^{1_i}(n_1) = \text{loop}_c^{2_i}(n_2) = i$; and

 \item $\forall y \in \text{Imp}(x), \vals_{1_i}(y) = \vals_{2_i}(y)$; and

 \item $\forall m_1'\,:\, ({s_1}, m_1) ->* (S_1', m_1'(\text{loop}_c^{1'}) ->* ({s_1}, m_{1_i}(\text{loop}_c^{1_i}, \allowbreak , \vals_{1_i})), \; \text{loop}_c^{1'}(n_1) \leq i$; and

 \item $\forall m_2'\,:\, ({s_2}, m_2) ->* (S_2', m_2'(\text{loop}_c^{2'})) ->* ({s_2}, m_{2_i}(\text{loop}_c^{2_i}, \allowbreak ), \vals_{2_i})), \; \text{loop}_c^{2'}(n_2) \leq i$.
\end{itemize}
By similar argument in base case, evaluations of the predicate expression of $s_1$ and $s_2$ w.r.t value stores $\vals_{1_i}$ and $\vals_{2_i}$ produce same value. Then there are two possibilities:
\begin{enumerate}
\item{$\mathcal{E}'\llbracket e\rrbracket\vals_{1_i} = \mathcal{E}'\llbracket e\rrbracket\vals_{2_i} = (0, v_\mathfrak{of})$}

The execution from $({s_1}, m_1(\text{loop}_c^{1}, \vals_1))$ proceeds as follows.

\begin{tabbing}
xx\=xx\=\kill
\>\>     $(s_1, m_{1_i}(\text{loop}_c^{1_i}, \vals_{1_i}))$\\
\>=\>    $(\text{while}_{\langle n_1\rangle} (e) \; \{S_1\}, {m_{1_i}(\text{loop}_c^{1_i}, \vals_{1_i})})$\\
\>$->$\> $(\text{while}_{\langle n_1\rangle} ((0, v_\mathfrak{of})) \; \{S_1\}, {m_{1_i}(\text{loop}_c^{1_i}, \vals_{1_i})})$ by the EEval' rule\\
\>$->$\> $(\text{while}_{\langle n_1\rangle} (0) \; \{S_1\}, {m_{1_i}(\text{loop}_c^{1_i}, \vals_{1_i})})$\\
\>\>     by the E-Oflow1 and E-Oflow2 rule\\
\>$->$\> $({\text{skip}}, m_{1_i}(\text{loop}_c^{1_i}[0/n_1], \vals_{1_i}))$ by the Wh-F rule.
\end{tabbing}

By the hypothesis IH, the loop counter of $s_1$ and $s_2$ in any configuration in executions
$(s_1, m_1) ->* (s_1, m_{1_i}(\text{loop}_c^{1_i}, \vals_{1_i}))$ and
$(s_2, m_2) ->* (s_2, m_{2_i}(\text{loop}_c^{2_i}, \vals_{2_i}))$
respectively are less than or equal to $i$,

 $\forall m_1'\,:\, ({s_1}, m_1) ->* (S_1', m_1'(\text{loop}_c^{1'}) ->* ({s_1}, m_{1_i}(\text{loop}_c^{1_i}, \allowbreak , \vals_{1_i})), \; \text{loop}_c^{1'}(n_1) \leq i$; and

 $\forall m_2'\,:\, ({s_2}, m_2) ->* (S_2', m_2'(\text{loop}_c^{2'})) ->* ({s_2}, m_{2_i}(\text{loop}_c^{2_i}, \allowbreak ), \vals_{2_i})), \; \text{loop}_c^{2'}(n_2) \leq i$.

Therefore, the loop counter of $s_1$ and $s_2$ in any configuration in executions

$(s_1, m_1) ->* (\text{skip}, m_{1_i}(\text{loop}_c^{1_i}\setminus\{(n_1)\}, \vals_{1_i}))$ and

$(s_2, m_2) ->* (\text{skip}, m_{2_i}(\text{loop}_c^{2_i}[0/n_2], \vals_{2_i}))$
respectively are less than $i+1$.

\item{$\mathcal{E}'\llbracket e\rrbracket\vals_{1_i} = \mathcal{E}'\llbracket e\rrbracket\vals_{2_i} = v \neq 0$}

The execution from $({s_1}, m_{1_i}(\text{loop}_c^{1_i}, \vals_{1_i}))$ proceeds as follows.

\begin{tabbing}
xx\=xx\=\kill
\>\>     $({s_1}, m_{1_i}(\text{loop}_c^{1_i}, \vals_{1_i}))$\\
\>= \>   $(\text{while}_{\langle n_1\rangle} (e) \; \{S_1\}, {m_{1_i}(\text{loop}_c^{1_i}, \vals_{1_i})})$\\
\>$->$\> $(\text{while}_{\langle n_1\rangle} (v) \; \{S_1\}, {m_{1_i}(\text{loop}_c^{1_i}, \vals_{1_i})})$ by the EEval' rule\\
\>$->$\> $(S_1;\text{while}_{\langle n_1\rangle} (e) \; \{S_1\}, m_{1_i}(\text{loop}_c^{1_i}[i+1/(n_1)]$\\
\>\> $, \vals_{1_i}))$ by the Wh-T rule.
\end{tabbing}

Similarly, $({s_2}, m_{2_i}(\text{loop}_c^{2_i}, \vals_{2_i}))$ {\kStepArrow [2] } $(S_2;\text{while}_{\langle n_2\rangle} (e) \{S_2\}, \allowbreak  m_{2_{i}}(\text{loop}_c^{2_{i}}[i+1/(n_2)], \vals_{2_{i}}))$.

By similar argument in base case, the executions of $S_1$ and $S_2$ terminate when started in states $m_{1_i}(\text{loop}_c^{1_i}[i+1/(n_1)], \allowbreak , \vals_{1_i})$ and
$m_{2_i}(\text{loop}_c^{2_i}[i+1/(n_2)], \vals_{2_i})$ respectively,
$(S_1;s_1, m_{1_i}(\text{loop}_c^{1_i}[i+1/(n_1)], \vals_{1_i})) ->* $

\noindent$(s_1, m_{1_{i+1}}(\text{loop}_c^{1_{i+1}}, \vals_{1_{i+1}}))$ and
$(S_2;s_1, m_{2_i}(\text{loop}_c^{2_i}[i+1/(n_2)], \vals_{2_i})) ->*
(s_2, m_{2_{i+1}}(\text{loop}_c^{2_{i+1}}, \vals_{2_{i+1}}))$ such that all of the following holds:
\begin{itemize}
\item
$\text{loop}_c^{1_{i+1}}(n_1) = \text{loop}_c^{2_{i+1}}(n_2) = i+1$; and

\item $\forall y\in\text{Imp}(x), \vals_{1_{i+1}}(y) = \vals_{2_{i+1}}(y)$, and

\item in any state in the execution

\noindent$(s_{1}, m_{1_i}) ->* (s_1, m_{1_{i+1}}(\text{loop}_c^{1_{i+1}}, \vals_{1_{i+1}}))$, the loop counter of $s_1$ is less than or equal to $i+1$.


\item in any state in the executions

\noindent$(s_2, m_{2_i}) ->* (s_2, m_{2_{i+1}}(\text{loop}_c^{2_{i+1}}, \vals_{2_{i+1}}))$, the loop counter of $s_2$ is less than or equal to $i+1$.

\end{itemize}
With the hypothesis IH, there are two configurations $(s_1, m_{1_{i+1}})$ and $(s_2, m_{2_{i+1}})$ reachable from $(s_1, m_1)$ and $(s_2, m_2)$, respectively, in which the loop counters of $s_1$ and $s_2$ are equal to $i+1$ and value stores agree on the values of imported variables relative to $x$ and, for every state in executions $(s_1, m_1) ->* (s_1, m_{1_{i+1}})$ and $(s_2, m_2) ->* (s_2, m_{2_{i+1}})$ the loop counters for $s_1$ and $s_2$ are less than or equal to $i+1$ respectively:

 $\exists (s_1, m_{1_{i+1}}), (s_2, m_{2_{i+1}})\,:\,$

 \noindent$(s_1, m_1) ->* (s_1, m_{1_{i+1}}(\text{loop}_c^{1_{i+1}}, \vals_{1_{i+1}})) \wedge$

 \noindent$(s_2, m_2) ->* (s_2, m_{2_{i+1}}(\text{loop}_c^{2_{i+1}}, \vals_{2_{i+1}}))$ where
\begin{itemize}
 \item $\text{loop}_c^{1_{i+1}}(n_1) = \text{loop}_c^{2_{i+1}}(n_2) = i+1$; and

 \item $\forall y \in \text{Imp}(x), \vals_{1_{i+1}}(y) = \vals_{2_{i+1}}(y)$; and

 \item $\forall m_1'\,:\, ({s_1}, m_1) ->* (S_1', m_1'(\text{loop}_c^{1'})) ->*$

 \noindent$({s_1}, m_{1_{i+1}}(\text{loop}_c^{1_{i+1}}, \vals_{1_{i+1}})), \; \text{loop}_c^{1'}(n_1) \leq i+1$; and

 \item $\forall m_2'\,:\, ({s_2}, m_2) ->* (S_2', m_2'(\text{loop}_c^{2'})) ->*$

  \noindent$({s_2}, m_{2_{i+1}}(\text{loop}_c^{2_{i+1}}, \vals_{2_{i+1}})), \; \text{loop}_c^{2'}(n_2) \leq i+1$.
\end{itemize}
\end{enumerate}
\end{enumerate}
\end{proof}

\begin{lemma}\label{lmm:sameFinalValueX}
Let $s_1$ = ``$\text{while}_{\langle n_1\rangle} (e) \; \{S_1\}$" and $s_2$ = ``$\text{while}_{\langle n_2\rangle} (e) \allowbreak \{S_2\}$" be two while statements with the same set of imported variables relative to a variable $x$ (defined in $s_1$ and $s_2$), and whose loop bodies $S_1$ and $S_2$ terminatingly compute the variables in $\text{Imp}(x)$ equivalently when started in states that agree on the values of the variables imported by $S_1$ or $S_2$ relative to $\text{Imp}(x)$:
\begin{itemize}
\item $x \in \text{Def}(s_1) \cap \text{Def}(s_2)$;
\item $\text{Imp}(s_1, \{x\}) = \text{Imp}(s_2, \{x\}) = \text{Imp}(x)$;
\item $\forall y \in \text{Imp}(x)\; \forall m_{S_1}(\vals_{S_1})\,m_{S_2}(\vals_{S_2}):$

$((\forall z\in \text{Imp}(S_1, \text{Imp}(x)) \cup \text{Imp}(S_2, \text{Imp}(x)), \vals_{S_1}(z) = \vals_{S_2}(z))
=> ((S_1, m_{S_1}(\vals_{S_1})) \equiv_{y} (S_2, m_{S_2}(\vals_{S_2})))$.
\end{itemize}

If the executions of $s_1$ and $s_2$ terminate when started in states
$m_1(\text{loop}_c^1, \vals_1)$ and $m_2(\text{loop}_c^2, \vals_2)$ in which
$s_1$ and $s_2$ have not already executed (loop counter initially 0:
$\text{loop}_c^1(n_{1}) = \text{loop}_c^2(n_{2}) = 0$),
and whose value stores $\vals_1$ and $\vals_2$ agree on the values of the variables in $\text{Imp}(x)$,
$ \forall y \in \text{Imp}(x)\;\vals_1(y)$ = $\vals_2(y)$,
when $s_1$ and $s_2$ terminate,
$({s_1}, {m_1}) ->* (\text{skip}, m_{1_i}(\vals_1'))$ and
$({s_2}, {m_2}) ->* (\text{skip}, m_{2_i}(\vals_2'))$, value stores $\vals_1'$ and $\vals_2'$ agree on the value of $x$,
$\vals_1'(x) = \vals_2'(x)$.
\end{lemma}
\begin{proof}
We show that there must exist a finite integer $k$ such that the loop counters of $s_1$ and $s_2$ in executions started in states $m_1$ and $m_2$ is always less than $k$.
By the definition of terminating execution, there are only finite number of steps in executions of $s_1$ and $s_2$ started in states $m_1$ and $m_2$ respectively. Then, by Lemma~\ref{lmm:loopCntStepwiseInc}, there must be a finite integer $k$ such that the loop counter of $s_1$ and $s_2$ is always less than $k$. In the following, we consider $k$ be the smallest positive integer such that the loop counter of $s_1$ and $s_2$ in executions started in states $m_1$ and $m_2$ is always less than $k$.

By Lemma~\ref{lmm:equivTermCompSameLoopIteration}, there are two possibilities:
\begin{enumerate}
\item{The loop counters for $s_1$ and $s_2$ are always less than 1 $(k=1)$:}

$\forall m_{1}', m_{2}'$ such that
$({s_1}, m_1) ->* ({S_1'}, m_{1}'(\text{loop}_c^{1'}))$ and
$({s_2}, m_2) \allowbreak ->* ({S_2'}, m_{2}'(\text{loop}_c^{2'}))$,

$\text{loop}_c^{1'}(n_1) < 1$ and $\text{loop}_c^{2'}(n_2) < 1$;

By the proof in base case of Lemma~\ref{lmm:equivTermCompSameLoopIteration},
the execution of $s_1$ proceeds as follows:

\begin{tabbing}
xx\=xx\=\kill
\>\>     $(s_1, m_1)$ \\
\>=\>    $(\text{while}_{\langle n_1\rangle} (e) \; \{S_1\}, {m_1(\text{loop}_c^1, \vals_1)})$\\
\>$->$\> $(\text{while}_{\langle n_1\rangle} (0) \; \{S_1\}, {m_1(\text{loop}_c^1, \vals_1)})$ by the EEval' rule\\
\>$->$\> $(\text{skip}, m_1(\text{loop}_c^1[0/(n_1)], \vals_1))$ by the Wh-F rule.
\end{tabbing}

Similarly, the execution of $s_2$ proceeds to

\noindent$(\text{skip}, m_2(\text{loop}_c^2[0/n_2], \vals_2))$. Therefore, $\vals_1' = \vals_1$ and $\vals_2' = \vals_2$.

By the definition of imported variables, $x \in \text{Imp}(s_1, \{x\})$. By assumption, value stores $\vals_1$ and $\vals_2$ agree on the value of $x$, $\vals_1'(x) = \vals_1(x) = \vals_2(x) = \vals_2'(x)$. The lemma holds.

\item For some finite positive $k(>1)$, both of the following hold:
\begin{itemize}
\item{The loop counters for $s_1$ and $s_2$ are always less than $k$:}

$\forall m_{1}', m_{2}'$ such that
$({s_1}, m_1) ->* ({S_1'}, m_{1}'(\text{loop}_c^{1'})$ and
$({s_2}, m_2) ->* ({S_2'}, m_{2}'(\text{loop}_c^{2'}))$,

$\text{loop}_c^{1'}(n_1) < k$ and $\text{loop}_c^{2'}(n_2) < k$;

\item{There are two configuration $(s_1, m_{1_{k-1}})$ and $(s_2, m_{2_{k-1}})$ reachable from $(s_1, m_1)$ and $(s_2, m_2)$, respectively, in which the loop counters of $s_1$ and $s_2$ are equal to $k-1$ and value stores agree on the values of imported variables relative to $x$ and, for every state in execution, $(s_1, m_1) ->* (s_1, m_{1_{k-1}})$ or $(s_2, m_2) ->* (s_2, m_{2_{k-1}})$ the loop counters for $s_1$ and $s_2$ are less than or equal to $k-1$ respectively:}

 $\exists (s_1, m_{1_{k-1}}), (s_2, m_{2_{k-1}})\,:\,$

 \noindent $({s_1}, m_1) ->* ({s_1}, m_{1_{k-1}}(\text{loop}_c^{1_{k-1}}, \vals_{1_{k-1}})) \wedge$

 \noindent $({s_2}, m_2) ->* ({s_2}, m_{2_{k-1}}(\text{loop}_c^{2_{k-1}}, \vals_{2_{k-1}}))$ where
\begin{itemize}
 \item $\text{loop}_c^{1_{k-1}}(n_1) = \text{loop}_c^{2_{k-1}}(n_2) = k-1$; and

 \item $\forall y \in \text{Imp}(x)\,:\, \vals_{1_{k-1}}(y) = \vals_{2_{k-1}}(y)$; and

 \item $\forall m_1'\,:\, ({s_1}, m_1) ->* (S_1', m_1'(\text{loop}_c^{1'}) ->* $

 \noindent$({s_1}, m_{1_{k-1}}(\text{loop}_c^{1_{k-1}}, \vals_{1_{k-1}})), \; \text{loop}_c^{1'}(n_1) \leq k-1$; and

 \item $\forall m_2'\,:\, ({s_2}, m_2) ->* (S_2', m_2'(\text{loop}_c^{2'})) ->*$

  \noindent$({s_2}, m_{2_{k-1}}(\text{loop}_c^{2_{k-1}}, \vals_{2_{k-1}})), \; \text{loop}_c^{2'}(n_2) \leq k-1$;
\end{itemize}
\end{itemize}

 By proof of Lemma~\ref{lmm:equivTermCompSameLoopIteration}, value stores $\vals_{1_{k-1}}$  and $\vals_{2_{k-1}}$ agree on the values of the variables in $\text{Use}(e)$. By Lemma~\ref{lmm:expEvalSameVal},
 $\mathcal{E}'\llbracket e\rrbracket\vals_{1_{k-1}} =
  \mathcal{E}'\llbracket e\rrbracket\vals_{2_{k-1}} = v$.
 Because the loop counter of $s_1$ and $s_2$ is less than $k$ in executions of $s_1$ and $s_2$ when started in states $m_1$ and $m_2$, then by our semantic rules, the predicate expression of $s_1$ and $s_2$ must evaluate to zero w.r.t value stores $\vals_{1_{k-1}}$ and $\vals_{2_{k-1}}$,
 $\mathcal{E}'\llbracket e\rrbracket\vals_{1_{k-1}} =
  \mathcal{E}'\llbracket e\rrbracket\vals_{2_{k-1}} = (0, v_\mathfrak{of})$.
 Then the execution of $s_1$ proceeds as follows.

\begin{tabbing}
xx\=xx\=\kill
\>\>     $(\text{while}_{\langle n_1\rangle} (e) \; \{S_1\}, {m_{1_{k-1}}(\text{loop}_c^{1_{k-1}}, \vals_{1_{k-1}})})$\\
\>$->$\> $(\text{while}_{\langle n_1\rangle} ((0, v_\mathfrak{of})) \; \{S_1\}, {m_{1_{k-1}}(\text{loop}_c^{1_{k-1}}, \vals_{1_{k-1}})})$\\
\>\>      by the EEval' rule\\
\>$->$\> $(\text{while}_{\langle n_1\rangle} (0) \; \{S_1\}, {m_{1_{k-1}}(\text{loop}_c^{1_{k-1}}, \vals_{1_{k-1}})})$\\
\>\>      by the E-Oflow1 or E-Oflow2 rule\\
\>$->$\> $(\text{skip}, m_{1_{k-1}}(\text{loop}_c^{1_{k-1}}[0/n_1], \vals_{1_{k-1}}))$\\
\>\>       by the Wh-F rule.
\end{tabbing}

Similarly, the execution of $s_2$ proceeds to

\noindent$(\text{skip}, m_{2_{k-1}}(\text{loop}_c^{2_{k-1}}[0/n_2], \vals_{2_{k-1}}))$.
Therefore, $\vals_2' = \vals_{2_{k-1}}$, $\vals_1' = \vals_{1_{k-1}}$.
By the definition of imported variables, $x \in \text{Imp}(x)$. In conclusion, $\vals_2'(x) = \vals_{2_{k-1}}(x) = \vals_{1_{k-1}}(x) = \vals_1'(x)$.
\end{enumerate}
\end{proof}

\begin{lemma}\label{lmm:sameImpfromEquivCompCond}
If two statement sequences $S_1$ and $S_2$ satisfy the proof rule of terminating computation of a variable $x$ equivalently, then $S_1$ and $S_2$ have same imported variables relative to $x$:
$(S_1 \equiv_{x}^S S_2) => (\text{Imp}(S_1, \{x\}) = \text{Imp}(S_2, \{x\}))$.
\end{lemma}
\begin{proof}
By induction on $\text{size}(S_1) + \text{size}(S_2)$, the sum of the program size of $S_1$ and $S_2$.

\noindent{\bf Base case}.

\noindent$S_1$ and $S_2$ are simple statement. Then the proof is a case analysis according to the cases in the definition of the proof rule of computation equivalently for simple statements.

\begin{enumerate}
\item{$S_1 = S_2$}

By the definition of imported variables, same statement have same imported variables relative to same $x$.

\item{$S_1 \neq S_2$}

There are two further cases:
\begin{itemize}
\item{$S_1 = ``\text{input }{id}_1", S_2 = ``\text{input }{id}_2"$ and $x \notin \{{id}_1, {id}_2\}$.}

When $x = {id}_I$, by the definition of imported variables, $\text{Imp}(S_1, \{{id}_I\}) = \text{Imp}(S_2, \{{id}_I\}) = \{{id}_I\}$.
When $x = {id}_{IO}$, by the definition of imported variables, $\text{Imp}(S_1, \{{id}_{IO}\}) = \text{Imp}(S_2, \{{id}_{IO}\}) = \{{id}_I, {id}_{IO}\}$.

When $x \notin \{{id}_I, {id}_{IO}\}$, by the definition of imported variables, $\text{Imp}(S_1, \{x\}) = \text{Imp}(S_2, \{x\}) = \{x\}$.

\item{the above cases do not hold and $x \notin \text{Def}(S_1)\cup\text{Def}(S_2)$.}

By the definition of imported variables, $\text{Imp}(S_1, \{x\}) = \text{Imp}(S_2, \{x\}) = \{x\}$.
\end{itemize}
\end{enumerate}

\noindent{\bf Induction Step}.

\noindent The hypothesis IH is that the lemma holds when size$(S_1) + \text{size}(S_2) = k \geq 2$.

\noindent Then we show the lemma holds when $\text{size}(S_1) + \text{size}(S_2) = k + 1$.
The proof is a case analysis based on the cases in the definition of the proof rule of computation equivalently for statement sequence, $S_1 \equiv_{x}^S S_2$:
\begin{enumerate}
\item $S_1$ and $S_2$ are one statement such that one of the following holds:

\begin{enumerate}

\item $S_1$ and $S_2$  are If statement that define the variable $x$:

      $S_1 = ``\text{If }(e) \, \text{then} \, \{S_1^t\} \text{else} \, \{S_1^f\}"$,
      $S_2 = ``\text{If }(e) \, \text{then} \, \{S_2^t\} \allowbreak\text{else} \, \{S_2^f\}"$ such that all of the following hold:

     \begin{itemize}
     \item $x \in \text{Def}(S_1) \cap \text{Def}(S_2)$;

     \item $S_1^t \equiv_{x}^S S_2^t$;

     \item $S_1^f \equiv_{x}^S S_2^f$;
     \end{itemize}

By the hypothesis IH, the imported variables in $S_1^t$ and $S_2^t$ relative to $x$ are same, $\text{Imp}(S_1^t, \{x\}) = \text{Imp}(S_2^t, \{x\})$. Similarly, $\text{Imp}(S_1^f, \{x\}) = \text{Imp}(S_2^f, \{x\})$.
By the definition of imported variables, $\text{Imp}(S_1, \{x\}) = \text{Use}(e) \cup \text{Imp}(S_1^t, \{x\}) \cup \text{Imp}(S_1^f, \{x\})$.
Similarly, $\text{Imp}(S_2, \{x\})$

$= \text{Use}(e) \cup \text{Imp}(S_2^t, \{x\}) \cup \text{Imp}(S_2^f, \{x\})$.
Then, the lemma holds.

\item $S_1$ and $S_2$ are while statement that define the variable $x$:

   $S_1 = ``\text{while}_{\langle n_1\rangle} (e) \; \{S_1''\}", S_2 = ``\text{while}_{\langle n_2\rangle} (e) \; \{S_2''\}"$ such that both of the following hold:
   \begin{itemize}
   \item $x \in \text{Def}(S_1) \cap \text{Def}(S_2)$;

   \item $\forall y \in
        \text{Imp}(S_1, \{x\})
        \cup
        \text{Imp}(S_2, \{x\}), S_1'' \equiv_{y}^S S_2''$;
   \end{itemize}

By the definition of imported variables,
$\text{Imp}(S_1, \{x\}) =$
$ \bigcup_{i\geq0} \text{Imp}(S_{1}^{''\mathit{i}}, \{x\}\allowbreak \cup \text{Use}(e))$.
Similarly, $\text{Imp}(S_2, \{x\}) = \bigcup_{i\geq0} \text{Imp}(S_{2}^{''\mathit{i}}, \{x\} \cup \text{Use}(e))$.
Then, we show that
$\text{Imp}(S_{1}^{''\mathit{i}}, \{x\} \; \cup \; \text{Use}(e)) =
                    \text{Imp}(S_{2}^{''\mathit{i}}, \{x\} \; \cup \; \text{Use}(e))$ by induction on $i$.


\noindent{Base case}.

By our assumption of the notation $S^0$, $S_{1}^{''\mathit{0}}$ = skip, $S_{2}^{''\mathit{0}}$ = skip.

Then,
$\text{Imp}(S_{1}^{''\mathit{0}},\{x\} \cup \text{Use}(e))$ = $\{x\} \cup \text{Use}(e)$,
$\text{Imp}(S_{2}^{''\mathit{0}},\{x\} \cup \text{Use}(e))$ = $\{x\} \cup \text{Use}(e)$.

Hence, $\text{Imp}(S_{1}^{''\mathit{0}}, \{x\} \cup \text{Use}(e))$ =
       $\text{Imp}(S_{2}^{''\mathit{0}}, \{x\} \cup \text{Use}(e))$.


\noindent{Induction step.}

The hypothesis IH2 is that
$\text{Imp}(S_{1}^{''\mathit{i}}, \{x\} \cup \text{Use}(e))$ =
$\text{Imp}(S_{2}^{''\mathit{i}}, \{x\} \cup \text{Use}(e))$ for $i\geq 0$.

Then we show that
$\text{Imp}(S_{1}^{''\mathit{i+1}}, \{x\} \cup \text{Use}(e))$ =
$\text{Imp}(S_{2}^{''\mathit{i+1}}, \{x\} \cup \text{Use}(e))$.

\begin{tabbing}
xx\=xx\=\kill
\>\>     $\text{Imp}(S_{1}^{''\mathit{i+1}}, \{x\} \cup \text{Use}(e))$\\
\>= \>   $\text{Imp}(S_{1}^{''}, \text{Imp}(S_{1}^{''\mathit{i}}, \{x\} \cup \text{Use}(e)))$ (1) by corollary~\ref{lmm:dupStmtPrefixLemma}
\end{tabbing}

\begin{tabbing}
xx\=xx\=\kill
\>\>     $\text{Imp}(S_{2}^{''\mathit{i+1}}, \{x\} \cup \text{Use}(e))$\\
\>=\>    $\text{Imp}(S_{2}^{''}, \text{Imp}(S_{2}^{''\mathit{i}}, \{x\} \cup \text{Use}(e)))$ (2) by corollary~\ref{lmm:dupStmtPrefixLemma}
\end{tabbing}

\begin{tabular}{l}
    $\text{Imp}(S_{1}^{''\mathit{i}}, \{x\} \cup \text{Use}(e))$
    =  $\text{Imp}(S_{2}^{''\mathit{i}}, \{x\} \cup \text{Use}(e))$
    \\
     {\text{by the hypothesis IH2}}
     \\
     $\text{Imp}(S_{1}^{''}, \text{Imp}(S_{1}^{''\mathit{i}}, \{x\} \cup \text{Use}(e)))$ \\
     =
     $\text{Imp}(S_{2}^{''}, \text{Imp}(S_{2}^{''\mathit{i}}, \{x\} \cup \text{Use}(e)))$
     \\
     {\text{by the hypothesis IH}}
     \\
     $\text{Imp}(S_{1}^{''\mathit{i+1}}, \{x\} \cup \text{Use}(e))$  \\
     =  $\text{Imp}(S_{2}^{''\mathit{i+1}}, \{x\} \cup \text{Use}(e))$ {\text{by (1),(2)}}\\

\end{tabular}

Therefore,
$\text{Imp}(S_{1}^{''\mathit{i+1}}, \{x\} \cup \text{Use}(e))$ =
$\text{Imp}(S_{2}^{''\mathit{i+1}}, \{x\} \cup \text{Use}(e))$.

In conclusion, $\text{Imp}(S_1, \{x\}) = \text{Imp}(S_2, \{x\})$. The lemma holds.

\item $S_1$ and $S_2$ do not define the variable $x$: $x \notin \text{Def}(S_1) \cup \text{Def}(S_2)$.

By the definition of imported variable, the imported variables in $S_1$ and $S_2$ relative to $x$ is $x$, $\text{Imp}(S_1, \{x\}) = \text{Imp}(S_2, \{x\}) = \{x\}$. The lemma holds.

\end{enumerate}

\item $S_1$ and $S_2$ are not both one statement such that one of the following holds:
\begin{enumerate}

\item Last statements both define the variable $x$ such that all of the following hold:
\begin{itemize}
\item  $\forall y \in \text{Imp}(s_1, \{x\}) \cup \text{Imp}(s_2, \{x\}), S_1' \equiv_{y}^S S_2'$;

\item  $x \in \text{Def}(s_1) \cap \text{Def}(s_2)$;

\item  $s_1 \equiv_x^S s_2$;
\end{itemize}

By the hypothesis IH, we have  $\text{Imp}(s_1, \{x\}) = \text{Imp}(s_2, \{x\}) = \text{Imp}(\Delta)$.
Then, by the hypothesis IH again, we have that
$\forall y \in \text{Imp}(\Delta) = \text{Imp}(s_1, \{x\}) = \text{Imp}(s_2, \{x\}),\text{Imp}(S_1', \{y\})\allowbreak = \text{Imp}(S_2', \{y\})$. By taking the union of all $\forall y \in \text{Imp}(\Delta),\allowbreak \text{Imp}(S_1', \{y\})$ and $\text{Imp}(S_2', \{y\})$, by the Lemma~\ref{lmm:impVarUnionLemma}, $\text{Imp}(S_1',\allowbreak \text{Imp}(\Delta)) = \text{Imp}(S_2', \text{Imp}(\Delta))$.
By the definition of imported variables,

$\text{Imp}(S_1, \{x\}) = \text{Imp}(S_1', \text{Imp}(s_1, \{x\}))$, $\text{Imp}(S_2, \{x\}) = \text{Imp}(S_2', \text{Imp}(s_2, \{x\}))$. Therefore, the lemma holds.

\item One last statement does not define the variable $x$:
w.l.o.g., $\big(x \notin \text{ Def}(s_1) \wedge (S_1' \equiv_{x}^S S_2)\big)$;

By the definition of imported variables, we have  $\text{Imp}(s_1, \{x\})$ = \{$x$\}.
By the hypothesis IH,  $\text{Imp}(S_1', \{x\}) = \text{Imp}(S_2, \{x\})$
Therefore, by the definition of imported variables, $\text{Imp}(S_1, \{x\})$ = $\text{Imp}(S_1', \text{Imp}(s_1, \{x\})) = \text{Imp}(S_1', \{x\}) = \text{Imp}(S_2, \{x\})$.

\item There are statements moving in/out of If statement:

      $s_1 = ``\text{If }(e) \text{ then }\{S_1^t\} \text{ else }\{S_1^f\}"$,
      $s_2 = ``\text{If }(e) \text{ then }\{S_2^t\}\allowbreak \text{ else }\{S_2^f\}"$ such that none of the above cases hold and all of the following hold:
     \begin{itemize}
     \item $\forall y \in \text{ Use}(e), S_1' \equiv_{y}^S S_2'$;

     \item $S_1';S_1^t \equiv_{x}^S S_2';S_2^t$;

     \item $S_1';S_1^f \equiv_{x}^S S_2';S_2^f$;

     \item $x \in \text{Def}(s_1) \cap \text{Def}(s_2)$;
     \end{itemize}

We show all of the following hold.
\begin{enumerate}
\item $\text{Imp}(S_1', \text{Use}(e)) = \text{Imp}(S_2', \text{Use}(e))$.

By the hypothesis IH and the assumption that $\forall y \in \text{Use}(e), S_1' \equiv_{y}^S S_2'$.
Then, by Lemma~\ref{lmm:impVarUnionLemma},

\noindent$\text{Imp}(S_1', \text{Use}(e)) = \text{Imp}(S_2', \text{Use}(e))$.

\item $\text{Imp}(S_1', \text{Imp}(S_1^t, \{x\})) = \text{Imp}(S_2', \text{Imp}(S_2^t, \{x\}))$.

Because $\text{size}(``\text{If}(e) \text{ then } \{S_t\} \text{ else }\{S_f\}") = 1+ \text{size}(S_t)$ + $\text{size}(S_f)$,
then

$\text{size}(S_1';S_1^t) + \text{size}(S_2';S_2^t)\allowbreak < k$.
By the hypothesis IH, $\text{Imp}(S_1';S_1^t, \{x\}) = \text{Imp}(S_2';S_2^t, \{x\})$.

Besides, by Lemma~\ref{lmm:ImpPrefixLemma},

$\text{Imp}(S_1', \text{Imp}(S_1^t, \{x\}))\allowbreak = \text{Imp}(S_1';S_1^t, \{x\})$

\noindent$= \text{Imp}(S_2';S_2^t, \{x\}) = \text{Imp}(S_2', \text{Imp}(S_2^t, \{x\}))$.

\item $\text{Imp}(S_1', \text{Imp}(S_1^f, \{x\})) = \text{Imp}(S_2', \text{Imp}(S_2^f, \{x\}))$.

By similar argument in the case that

\noindent$\text{Imp}(S_1', \text{Imp}(S_1^t, \{x\})) = \text{Imp}(S_2', \text{Imp}(S_2^t, \{x\}))$.
\end{enumerate}

Then, by combining things together,

\begin{tabbing}
xx\=xx\=\kill
\>\>     $\text{Imp}(S_1, \{x\})$\\
\>=\>    $\text{Imp}(S_1', \text{Imp}(s_1, \{x\}))$\\
\>\>      by the definition of imported variables\\
\>=\>    $\text{Imp}(S_1', \text{Imp}(S_1^t, \{x\}) \cup \text{Imp}(S_1^f, \{x\}) \cup \text{Use}(e))$\\
\>\>     by the definition of imported variables\\
\>=\>    $\text{Imp}(S_1', \text{Imp}(S_1^t, \{x\})) \cup \text{Imp}(S_1', \text{Imp}(S_1^f, \{x\}))$\\
\>\>     $\cup \text{Imp}(S_1', \text{Use}(e))$ by Lemma~\ref{lmm:impVarUnionLemma}\\
\>=\>    $\text{Imp}(S_2', \text{Imp}(S_2^t, \{x\})) \cup \text{Imp}(S_2', \text{Imp}(S_2^f, \{x\}))$\\
\>\>     $\cup \text{Imp}(S_2', \text{Use}(e))$ by i, ii, iii\\
\>=\>    $\text{Imp}(S_2', \text{Imp}(S_2^t, \{x\}) \cup \text{Imp}(S_2^f, \{x\}) \cup \text{Use}(e))$\\
\>\>     by Lemma~\ref{lmm:impVarUnionLemma}\\
\>=\>    $\text{Imp}(S_2', \text{Imp}(S_2^t, \{x\}) \cup \text{Imp}(S_2^f, \{x\}) \cup \text{Use}(e))$\\
\>\>     by the definition of imported variables\\
\>=\>    $\text{Imp}(S_2', \text{Imp}(s_2, \{x\}))$ \\
\>\>     by the definition of imported variables\\
\>=\>     $\text{Imp}(S_2, \{x\})$.
\end{tabbing}
Hence, the lemma holds.
\end{enumerate}
\end{enumerate}
\end{proof}
\subsection{Termination in the same way}\label{sec_equiv_term}
We proceed to propose a proof rule under which two statement sequences either both terminate or both do not terminate. We start by giving the definition of termination in the same way. Then we present the proof rule of termination in the same way. Our proof rule of termination in the same way allows updates such as statement duplication or reordering, loop fission or fusion and additional terminating statements.
We prove that the proof rule ensures terminating in the same way by induction on the program size of the two programs in the proof rule. We also list auxiliary lemmas required by the proof of termination in the same way.

\begin{definition}\label{def:termInSameWay}
{\bf{(Termination in the same way)}}
Two statement sequences $S_1$ and $S_2$ terminate in the same way when started in states $m_1$ and $m_2$ respectively, written $(S_1, m_1) \equiv_{H} (S_2, m_2)$, iff one of the following holds:
\begin{enumerate}
\item
$({S_1}, m_1) ->* (\text{skip}, m_1')$ and
$({S_2}, m_2) ->* (\text{skip}, m_2')$;

\item
$\forall i\geq 0$, $({S_1}, m_1)$ {\kStepArrow [i] } $({S_1^i}, m_1^i)$ and
                     $({S_2}, m_2)$ {\kStepArrow [i] } $({S_2^i}, m_2^i)$ where
                     $S_1^i \neq \text{skip}, S_2^i \neq \text{skip}$.
\end{enumerate}
\end{definition}

\subsubsection{Proof rule for termination in the same way}
We define proof rules under which two statement sequences $S_1$ and $S_2$ terminate in the same way.
We summarize the cause of non-terminating execution and then give the proof rule.

We consider two causes of nonterminating executions: crash and infinite iterations of loop statements.
As to crash~\cite{crashReason}, we consider four common causes based on our language: expression evaluation exceptions, the lack of input value, input/assignment value type mismatch and array index out of bound.
In essence, the causes of nontermination are partly due to the values of some particular variables during executions.
We capture variables affecting each source of nontermination;
loop deciding variables LVar$(S)$ are variables affecting the evaluation of a loop statements in the statement sequence $S$,
crash deciding variables CVar$(S)$ are variables whose values decide if a crash occurs in $S$.
We list the definitions of LVar$(S)$ and CVar$(S)$ in Definition~\ref{def:LVar} and~\ref{def:CVar}.
Definition~\ref{def:TermVars} summarizes the variables whose values decide if one program terminates.
%
%




\begin{definition}\label{def:LVar}
{\bf (Loop deciding variables)} The loop deciding variables of a statement sequence $S$, written LVar$(S)$, are defined as follows:
\begin{enumerate}
\item $\text{LVar}(S) = \emptyset$ if $\nexists s = ``\text{while} (e) \; \{S'\}$" and $s \in S$;

\item  $\text{LVar}(``\text{If }(e) \text{ then }\{S_t\} \text{ else }\{S_f\}") =
\text{Use}(e) \, \cup \,
\text{LVar}(S_t) \, \cup \,
\text{LVar}(S_f)$ if $``\text{while} (e) \{S'\}" \in S$;

\item $\text{LVar}(``\text{while} (e) \{S'\}") = \text{Imp}(S, \text{Use}(e) \, \cup \, \text{LVar}(S'))$;

\item For $k>0$, $\text{LVar}(s_1;...;s_k;s_{k+1}) = \text{LVar}(s_1;...;s_k) \, \cup \, \text{Imp}(s_1;...;s_k, \text{LVar}(s_{k+1}))$;
\end{enumerate}
\end{definition}

\begin{definition}\label{def:CVar}
{\bf{(Crash deciding variables)}} The crash deciding variables of a statement sequence $S$, written $\text{CVar}(S)$, are defined as follows:
\begin{enumerate}
\item  $\text{CVar}(\text{skip}) = \emptyset$;

\item  $\text{CVar}(lval := e) = \text{Idx}(lval) \cup \text{Use}(e)$ if $(\Gamma \vdash lval : \text{Int}) \wedge (\Gamma \vdash e : \text{Long})$;

\item  $\text{CVar}(lval := e) = \text{Idx}(lval) \cup \text{Err}(e)$ if $(\Gamma \vdash lval : \text{Int}) \wedge (\Gamma \vdash e : \text{Long})$ does not hold;

\item  $\text{CVar}(\text{input }id) = \{{id}_I\}$;

\item  $\text{CVar}(\text{output }e) = \text{Err}(e)$;

\item  $\text{CVar}(``\text{If} \; (e) \text{ then } \{S_t\} \text{ else }\{S_f\}") = \text{Err}(e)$,
  if $\text{CVar}(S_t) \, \allowbreak \cup \, \text{CVar}(S_f) = \emptyset$;

\item $\text{CVar}(``\text{If} \; (e) \text{ then } \{S_t\} \text{ else }\{S_f\}")$ =
  $\text{Use}(e) \,  \cup \, \text{CVar}(S_t) \, \allowbreak \cup \, \text{CVar}(S_f)$,
  if $\text{CVar}(S_t) \, \cup \, \text{CVar}(S_f)\allowbreak \neq \emptyset$;

\item $\text{CVar}(``\text{while}_{\langle n\rangle} (e) \{S'\}") =
\text{Imp}(``\text{while}_{\langle n\rangle} (e) \{S'\}", \text{Use}(e) \, \cup \, \text{CVar}(S'))$;

\item For $k>0$,  $\text{CVar}(s_1;...;s_{k+1}) = \text{CVar}(s_1;...;s_k)$

\noindent$ \, \cup \,
\text{Imp}(s_1;...;s_k, \text{CVar}(s_{k+1}))$;
\end{enumerate}
\end{definition}

\begin{definition}\label{def:TermVars}
{\bf (Termination deciding variables)}
The termination deciding variables of statement sequence $S$ are
 $\text{CVar}(S) \cup \text{LVar}(S)$, written $\text{TVar}(S)$.
\end{definition}

\begin{definition}\label{def:condTermInSameWaySimpleStmt}
{\bf (Base cases of the proof rule of termination in the same way)}
Two simple statements $s_1$ and $s_2$ satisfy the proof rule of termination in the same way, written $s_1 \equiv_{H}^S s_2$, iff one of the following holds:
\begin{enumerate}
\item $s_1$ and $s_2$ are same,
      $s_1 = s_2$;

\item $s_1$ and $s_2$ are input statement with same type variable: $s_1=``\text{input }{id}_1",\allowbreak s_2 = ``\text{input }{id}_2"$
  where $(\Gamma_{s_1} \vdash {id}_1 : \tau) \wedge (\Gamma_{s_2} \vdash {id}_2 : \tau)$;

\item  $s_1 = ``\text{output} \, e" \, \text{or} \, ``{id}_1 := e",
        s_2 = ``\text{output} \, e" \, \text{or} \, ``{id}_2 := e"$ where both of the following hold:

       \begin{itemize}
       \item There is no possible value mismatch in $``{id}_1 := e"$,
       $\neg(\Gamma_{s_1} \vdash {id}_1 : \text{Int}) \vee
        \neg(\Gamma_{s_1} \vdash e : \text{Long}) \vee
            (\Gamma_{s_1} \vdash e : \text{Int})$.

       \item There is no possible value mismatch in $``{id}_2 := e"$,
       $\neg(\Gamma_{s_2} \vdash {id}_2 : \text{Int}) \vee
        \neg(\Gamma_{s_2} \vdash e : \text{Long}) \vee
            (\Gamma_{s_2} \vdash e : \text{Int})$.
       \end{itemize}
\end{enumerate}
\end{definition}

\begin{definition}\label{def:lCondTermInSameWay}
{\bf (proof rule of termination in the same way)}
Two statement sequences $S_1$ and $S_2$ satisfy the proof rule of termination in the same way, written $S_1 \equiv_{H}^S S_2$, iff one of the following holds:
\begin{enumerate}
\item $S_1$ and $S_2$ are both one statement and one of the following holds.

\begin{enumerate}
\item $S_1$ and $S_2$ are simple statements:
$s_1 \equiv_{H}^S s_2$;

\item $S_1$ = ``If($e$) then \{$S_1^t$\} else \{$S_1^f$\}",
      $S_2$ = ``If($e$) then \{$S_2^t$\} else \{$S_2^f$\}" and one of the following holds:

\begin{enumerate}
\item $S_1^t, S_1^f, S_2^t, S_2^f$ are all sequences of ``skip";

\item At least one of $S_1^t, S_1^f, S_2^t, S_2^f$ is not a sequence of ``skip" such that:
      $(S_1^t \equiv_{H}^S S_2^t) \land (S_1^f \equiv_{H}^S S_2^f)$;
\end{enumerate}

\item $S_1$ = ``$\text{while}_{\langle n_1\rangle} (e) \{S_1''\}$",
      $S_2$ = ``$\text{while}_{\langle n_2\rangle} (e) \{S_2''\}$" and both of the following hold:

\begin{itemize}
\item $S_1'' \equiv_{H}^S S_2''$;

\item $S_1''$ and $S_2''$ have equivalent computation of $\text{TVar}(S_1) \cup \text{ TVar}(S_2)$;
\end{itemize}
\end{enumerate}

\item $S_1$ and $S_2$ are not both one statement and one of the following holds:

\begin{enumerate}
\item $S_1=S_1';s_1$ and $S_2 = S_2';s_2$ and all of the following hold:

\begin{itemize}
\item $S_1' \equiv_{H}^S S_2'$;

\item $S_1'$ and $S_2'$ have equivalent computation of $\text{TVar}(s_1) \cup \text{TVar}(s_2)$;

\item
$s_1 \equiv_{H}^S s_2$ where $s_1$ and $s_2$ are not ``skip";
\end{itemize}

\item One last statement is ``skip":

      $\big((S_1 = S_1';``\text{skip}") \wedge (S_1' \equiv_{H}^S S_2)\big)$
$\vee$
      $\big((S_2 = S_2';``\text{skip}") \wedge (S_1 \equiv_{H}^S S_2')\big)$.

\item  One last statement is a ``duplicate" statement and
 one of the following holds:

\begin{enumerate}
\item $S_1 = S_1';s_1';S_1'';s_1$ and all of the following hold:
\begin{itemize}
\item $S_1';s_1';S_1'' \equiv_{H}^S S_2$;

\item $(s_1' \equiv_{H}^S s_1) \wedge (s_1 \neq ``\text{skip}")$;

\item $\text{Def}(s_1';S_1'') \cap \text{TVar}(s_1) = \emptyset$;
\end{itemize}

\item $S_2 = S_2';s_2';S_2'';s_2$ and all of the following hold:
\begin{itemize}
\item $S_1 \equiv_{H}^S S_2';s_2';S_2''$;

\item $(s_2' \equiv_{H}^S s_2) \wedge (s_2 \neq ``\text{skip}")$;

\item $\text{Def}(s_2';S_2'') \cap \text{TVar}(s_2) = \emptyset$;
\end{itemize}
\end{enumerate}

\item $S_1 = S_1';s_1;s_1'$ and $S_2 = S_2';s_2;s_2'$ where $s_1$ and $s_2$ are reordered and all of the following hold:

\begin{itemize}
\item ${S_1'} \equiv_{H}^S {S_2'}$;

\item $S_1'$ and $S_2'$ have equivalent computation of $\text{TVar}(s_1;s_1') \cup \text{TVar}(s_2;s_2')$.

\item ${s_1} \equiv_{H}^S {s_2'}$;

\item ${s_1'} \equiv_{H}^S {s_2}$;

\item $\text{Def}(s_1) \cap \text{TVar}(s_1') = \emptyset$;

\item $\text{Def}(s_2) \cap \text{TVar}(s_2') = \emptyset$;
\end{itemize}
\end{enumerate}
\end{enumerate}
\end{definition}

\subsubsection{Soundness of the proof rule for termination in the same way}
We show that two statement sequences satisfy the proof rules of termination in the same way, and their initial states agree on the values of their termination deciding variables, then they either both terminate or both do not terminate.
%

\begin{theorem}\label{thm:mainTermSameWaySimpleStmt}
If two simple statements $s_1$ and $s_2$ satisfy the proof rule of termination in the same way,
$s_1 \equiv_{H}^s s_2$,
and their initial states
$m_1(\mathfrak{f}_1, \vals_1)$ and $m_2(\mathfrak{f}_2, \vals_2)$ with crash flags not set,
$\mathfrak{f}_1 = \mathfrak{f}_2 = 0$,
 and whose value stores agree on values of the termination deciding variables of $s_1$ and $s_2$,
$\forall x \in \text{TVar}(s_1) \cup \text{TVar}(s_2)\,:\,
  \vals_1(x) = \vals_2(x)$,
when executions of $s_1$ and $s_2$ start in states $m_1$ and $m_2$ respectively,
then $s_1$ and $s_2$ terminate in the same way when started in states $m_1$ and $m_2$ respectively:
\noindent$(s_1, m_1) \equiv_{H} (s_2, m_2)$.
\end{theorem}
\begin{proof}
The proof is a case analysis of those cases in the definition of $s_1 \equiv_{H}^s s_2$.
Because $s_1$ is a simple statement and $s_1$'s execution is without function call,
we only care the crash variables of $s_1$ in the termination deciding variables of $s_1$, $\text{CVar}(s_1)$.
Similarly, we only care $\text{CVar}(s_2)$.
\begin{description}
\item[First] $s_1$ and $s_2$ are same: $s_1 = s_2$;

We show the theorem by induction on abstract syntax of $s_1$ and $s_2$.
\begin{enumerate}
\item $s_1 = s_2 = \text{skip}$.

By definition of termination in the same way, both $s_1$ and $s_2$ terminate. The theorem holds.

\item $s_1 = s_2 = ``lval := e"$.

There are further cases regarding what $lval$ is.
\begin{enumerate}
\item $lval = id$.

By definition, $\text{CVar}(s_1) = \text{CVar}(s_2) = \text{Err}(e)$ or $\text{Use}(e)$
based on if there is possible value mismatch (e.g., assigning value defined only in type Long to a variable of type Int). There are two subcases.
\begin{itemize}
\item Left value $id$ is of type Int and expression $e$ is of type Long but not type Int,
       $(\Gamma \vdash id : \text{Int}) \wedge
       (\Gamma \vdash e : \text{Long}) \wedge
       \neg(\Gamma \vdash e : \text{Int})$.

By definition, $\text{CVar}(s_1) = \text{CVar}(s_2) = \text{Use}(e)$.
By assumption, $\forall x \in \text{Use}(e), \vals_1(x) = \vals_2(x)$.
By Lemma~\ref{lmm:expEvalSameVal}, the expression evaluates to the same value w.r.t two pairs of value stores $\vals_1$ and $\vals_2$ respectively,
\begin{itemize}
\item Both evaluations of expression lead to crash,

\noindent$\mathcal{E}\llbracket e\rrbracket\vals_1 =
 \mathcal{E}\llbracket e\rrbracket\vals_2 = (\text{error}, v_\mathfrak{of})$.

Then the execution of $s_1$ is as follows:
\begin{tabbing}
xx\=xx\=\kill
\>\>     $(s_1,m_1)$ \\
\>= \>   $(id := e, m_1(\vals_1))$ \\
\>$->$\> $(id := (\text{error}, *), m_1(\vals_1))$ by rule EEval' \\
\>$->$\> $(id := 0, m_1(1/\mathfrak{f}))$ by rule ECrash. \\
\>{\kStepArrow [i] }\> $(id := 0, m_1(1/\mathfrak{f}))$ for any $i>0$ by rule Crash.
\end{tabbing}

Similarly, $s_2$ does not terminate. The theorem holds.

\item Both evaluations of expression lead to no crash,
$\mathcal{E}\llbracket e\rrbracket\vals_1 =
 \mathcal{E}\llbracket e\rrbracket\vals_2 = (v, v_\mathfrak{of})$.

Then there are cases regarding if value mismatch occurs.

\noindent$\surd$ The value $v$ is only defined in type Long,
$(\Gamma \vdash v : \text{Long}) \wedge \neg(\Gamma \vdash v : \text{Int})$.

The execution of $s_1$ is as follows:
\begin{tabbing}
xx\=xx\=\kill
\>\>     $(s_1,m_1)$ \\
\>= \>   $(id := e, m_1(\vals_1))$ \\
\>$->$\> $(id := (v, v_\mathfrak{of}), m_1(\vals_1))$ by rule EEval \\
\>$->$\> $(id := v, m_1(\vals_1))$ by rule EOflow-1 or EOflow-2. \\
\>$->$\> $(id := v, m_1(1/\mathfrak{f}))$ by rule Assign-Err. \\
\>{\kStepArrow [i] }\> $(id := v, m_1(1/\mathfrak{f}))$ for any $i>0$ by rule Crash.
\end{tabbing}

Similarly, $s_2$ does not terminate. The theorem holds.

\noindent$\surd$  The value $v$ is defined in type Int,
$\Gamma \vdash v : \text{Int}$.

Assuming that the variable $id$ is a global one, the execution of $s_1$ is as follows:
\begin{tabbing}
xx\=xx\=\kill
\>\>     $(s_1,m_1)$ \\
\>= \>   $(id := e, m_1(\vals_1))$ \\
\>$->$\> $(id := (v, v_\mathfrak{of}), m_1(\vals_1))$ by rule EEval \\
\>$->$\> $(id := v, m_1(\vals_1))$ by rule EOflow-1 or EOflow-2. \\
\>$->$\> $(\text{skip}, m_1(\vals_1(\vals_1[v/id])))$ by rule Assign.
\end{tabbing}

Similarly, $s_2$ terminate. The theorem holds.

When the variable $id$ is a local variable, by similar argument for the global variable, we can show that $s_1$ and $s_2$ terminate. Then the theorem holds.
\end{itemize}

\item It is not the case that left value $id$ is of type Int and the expression $e$ is of type Long only,

\noindent$\neg\big((\Gamma \vdash id : \text{Int}) \wedge
                   (\Gamma \vdash e : \text{Long}) \wedge
                   \neg(\Gamma \vdash e : \text{Int})\big)$.

There are two cases based on if there is crash in evaluation of expression $e$.

\noindent$\surd$ Both evaluations of expression lead to crash,

\noindent$\mathcal{E}\llbracket e\rrbracket\vals_1 =
 \mathcal{E}\llbracket e\rrbracket\vals_2 = (\text{error}, v_\mathfrak{of})$.

By the same argument in case where left value $id$ is of type Int and the expression $e$ is of type Long only,
this theorem holds.


\noindent$\surd$ Both evaluations of expression lead to no crash,

\noindent$\mathcal{E}\llbracket e\rrbracket\vals_1 =
 \mathcal{E}\llbracket e\rrbracket\vals_2 = (v, v_\mathfrak{of})$.

By the same argument in subcase of no value mismatch in case where left value $id$ is of type Int and the expression $e$ is of type Long only, this theorem holds.
\end{itemize}

\item $lval = id[n]$.

There are two subcases based on if $n$ is within the array bound of $id$.
By our assumption, array variable $id$ is of the same bound in two programs.
W.l.o.g., we assume $id$ is local variable.
\begin{enumerate}
\item $n$ is out of bound of array variable $id$,
$((id, n) \mapsto v_1) \notin \vals_1$ and
$((id, n) \mapsto v_2) \notin \vals_2$;

Then the execution of $s_1$ continues as follows:
\begin{tabbing}
xx\=xx\=\kill
\>\>     $(s_1, m_1)$ \\
\>= \>   $(id[n] := e, m_1(\vals_1))$ \\
\>$->$\> $(id[n] := e, m_1(1/\mathfrak{f})$ by rule Arr-3 \\
\>${\kStepArrow [i] }$\> $(id[n] := e, m_1(1/\mathfrak{f}))$ by rule Crash.
\end{tabbing}

Similarly, $s_2$ does not terminate. The theorem holds.

\item $n$ is within the bound of array variable $id$,
$((id, n) \mapsto v_1) \in \vals_1$ and
$((id, n) \mapsto v_2) \in \vals_2$;

There are cases of $\text{CVar}(s_1)$ and $\text{CVar}(s_2)$ based on if there is possible value mismatch exception in $s_1$ and $s_2$.
\begin{itemize}
\item Left value $id[n]$ is of type Int and expression $e$ is of type Long but not type Int,
       $(\Gamma \vdash id[n] : \text{Int}) \wedge
       (\Gamma \vdash e : \text{Long}) \wedge
       \neg(\Gamma \vdash e : \text{Int})$.

By definition, $\text{CVar}(s_1) = \text{CVar}(s_2) = \text{Use}(e)$.
By assumption, $\forall x \in \text{Use}(e), \vals_1(x) = \vals_2(x)$.
By Lemma~\ref{lmm:expEvalSameVal}, the expression evaluates to the same value w.r.t two value stores $\vals_1$ and $\vals_2$ respectively,
\begin{itemize}
\item Both evaluations of expression lead to crash,

\noindent$\mathcal{E}\llbracket e\rrbracket\vals_1 =
 \mathcal{E}\llbracket e\rrbracket\vals_2 = (\text{error}, v_\mathfrak{of})$.

Then the execution of $s_1$ is as follows:
\begin{tabbing}
xx\=xx\=\kill
\>\>     $(s_1,m_1)$ \\
\>= \>   $(id[n] := e, m_1(\vals_1))$ \\
\>$->$\> $(id[n] := (\text{error}, *), m_1(\vals_1))$ by rule EEval' \\
\>$->$\> $(id[n] := 0, m_1(1/\mathfrak{f}))$ by rule ECrash. \\
\>{\kStepArrow [i] }\> $(id[n] := 0, m_1(1/\mathfrak{f}))$ for any $i>0$\\
\>\>                    by rule Crash.
\end{tabbing}

Similarly, $s_2$ does not terminate. The theorem holds.

\item Both evaluations of expression lead to no crash,
$\mathcal{E}\llbracket e\rrbracket\vals_1 =
 \mathcal{E}\llbracket e\rrbracket\vals_2 = (v, v_\mathfrak{of})$.

Then there are cases regarding if value mismatch occurs.

\noindent$\surd$ The value $v$ is only defined in type Long,
$(\Gamma \vdash v : \text{Long}) \wedge \neg(\Gamma \vdash v : \text{Int})$.

The execution of $s_1$ is as follows:
\begin{tabbing}
xx\=xx\=\kill
\>\>     $(s_1,m_1)$ \\
\>= \>   $(id[n] := e, m_1(\vals_1))$ \\
\>$->$\> $(id[n] := (v, v_\mathfrak{of}), m_1(\vals_1))$ by rule EEval' \\
\>$->$\> $(id[n] := v, m_1(\vals_1))$  \\
\>\>     by rule EOflow-1 or EOflow-2.\\
\>$->$\> $(id[n] := v, m_1(1/\mathfrak{f}))$ by rule Assign-Err. \\
\>{\kStepArrow [i] }\> $(id[n] := v, m_1(1/\mathfrak{f}))$ for any $i>0$\\
\>\>                    by rule Crash.
\end{tabbing}

Similarly, $s_2$ does not terminate. The theorem holds.

\noindent$\surd$ The value $v$ is defined in type Int,
$\Gamma \vdash v : \text{Int}$.

The execution of $s_1$ is as follows:
\begin{tabbing}
xx\=xx\=\kill
\>\>     $(s_1,m_1)$ \\
\>= \>   $(id[n] := e, m_1(\vals_1))$ \\
\>$->$\> $(id[n] := (v, v_\mathfrak{of}), m_1(\vals_1))$ by rule EEval' \\
\>$->$\> $(id[n] := v, m_1(\vals_1))$\\
\>\>      by rule EOflow-1 or EOflow-2. \\
\>$->$\> $(\text{skip}, m_1(\vals_1(\vals_1[v/(id, n)])))$ \\
\>\>      by rule Assign-A.
\end{tabbing}

Similarly, $s_2$ terminate. The theorem holds.

When the variable $id$ is a global variable, by similar argument for the global variable, we can show that $s_1$ and $s_2$ terminate. Then the theorem holds.
\end{itemize}

\item It is not the case that left value $id$ is of type Int and the expression $e$ is of type Long only,

\noindent$\neg\big((\Gamma \vdash id : \text{Int}) \wedge
                   (\Gamma \vdash e : \text{Long}) \wedge
                   \neg(\Gamma \vdash e : \text{Int})\big)$.

There are two cases based on if there is crash in evaluation of expression $e$.
\begin{itemize}
\item Both evaluations of expression lead to crash,

\noindent$\mathcal{E}\llbracket e\rrbracket\vals_1 =
          \mathcal{E}\llbracket e\rrbracket\vals_2 = (\text{error}, v_\mathfrak{of})$.

By the same argument in case where left value $id$ is of type Int and the expression $e$ is of type Long only,
this theorem holds.

\item Both evaluations of expression lead to no crash,

\noindent$\mathcal{E}\llbracket e\rrbracket\vals_1 =
 \mathcal{E}\llbracket e\rrbracket\vals_2 = (v, v_\mathfrak{of})$.

By the same argument in subcase of no value mismatch in case where left value $id$ is of type Int and the expression $e$ is of type Long only, this theorem holds.
\end{itemize}
\end{itemize}

\end{enumerate}
If array variable $id$ is a global variable, by similar argument above, the theorem holds.

\item $lval = {id}_1[{id}_2]$.

By definition, $\text{Idx}(s_1) = \text{Idx}(s_2) = \{{id}_2\} \subseteq \text{CVar}(s_1) = \text{CVar}(s_2)$.
By assumption, $\vals_1({id}_2) = \vals_2({id}_2) = n$.
By the same argument in the case where $lval = id[n]$, the theorem holds.
\end{enumerate}

\item $s_1 = s_2 = ``\text{input} \; id"$,

By definition, $\text{CVar}(s_1) = \text{CVar}(s_2) = \{{id}_I\}$.
By assumption $\vals_1({id}_I) = \vals_2({id}_I)$.
There are cases regarding if input sequence is empty or not.
\begin{enumerate}
\item There is empty input sequence, $\vals_1({id}_I) = \vals_2({id}_I) = \varnothing$.

Then the execution of $s_1$ continues as follows:
\begin{tabbing}
xx\=xx\=\kill
\>\>     $(s_1,m_1)$ \\
\>= \>   $(\text{input} \, id, m_1(\vals_1))$ \\
\>$->$\> $(\text{input} \, id, m_1(1/\mathfrak{f}))$ by rule In-7 \\
\>{\kStepArrow [i] }\> $(\text{input} \, id, m_1(1/\mathfrak{f}))$ by rule Crash.
\end{tabbing}

Similarly, $s_2$ does not terminate. The theorem holds.

\item There is nonempty input sequence, $\vals_1({id}_I) = \vals_2({id}_I) \neq \varnothing$.

There are cases regarding if type of the variable $id$ is Long or not.
\begin{enumerate}
\item $id$ is of type Long, $\Gamma \vdash id : \text{Long}$;

Assuming $id$ is a local variable, then the execution of $s_1$ continues as follows:
\begin{tabbing}
xx\=xx\=\kill
\>\>     $(s_1, m_1)$ \\
\>= \>   $(\text{input} \, id, m_1(\vals_1))$ \\
\>$->$\> $(\text{skip}, m_1(\vals_1[{v_{io}} / id, \text{tl}(\vals_1({id}_I))/{id}_I,$\\
\>\>     $``\vals_1({id}_{IO}) \cdot \underline{v}_{io}"/{id}_{IO}]))$ by rule In-3.
\end{tabbing}
Similarly, $s_2$ terminates. The theorem holds.

When the variable $id$ is a global variable, by similar argument, the theorem holds.

\item $id$ is of type Int or enumeration, $\Gamma \vdash id : \text{Int}$ or $\text{enum} \, id'$;

There are cases regarding if the head of input sequence can be transformed to type of $id$.
Let $v_{io} = \text{hd}(\vals_1({id}_I))$.
\begin{itemize}
\item $id$ is of type Int.

If $v_{io}$ is not of type Int, $\Gamma \vdash v_{io} : \text{Long}$ and $\neg(\Gamma \vdash v_{io} : \text{Int})$,
then the execution of $s_1$ continues as follows:
\begin{tabbing}
xx\=xx\=\kill
\>\>     $(s_1, m_1)$ \\
\>= \>   $(\text{input} \, id, m_1(\vals_1))$ \\
\>$->$\> $(\text{input} \, id, m_1(1/\mathfrak{f}))$ by Rule In-4.\\
\>{\kStepArrow [i] }\> $(\text{input} \, id, m_1(1/\mathfrak{f}))$ by Rule crash.
\end{tabbing}
Similarly, $s_2$ does not terminate. The theorem holds.

If $v_{io}$ is of type Int, $\Gamma \vdash v_{io} : \text{Long}$ and $\Gamma \vdash v_{io} : \text{Int}$,
assuming $id$ is a local variable,
then the execution of $s_1$ continues as follows:
\begin{tabbing}
xx\=xx\=\kill
\>\>     $(s_1, m_1)$ \\
\>= \>   $(\text{input} \, id, m_1(\vals_1))$ \\
\>$->$\> $(\text{skip}, m_1(\vals_1[{v_{io}} / id, \text{tl}(\vals_1({id}_I))/{id}_I,$\\
\>\>     $``\vals_1({id}_{IO}) \cdot \underline{v}_{io}"/{id}_{IO}]))$ by Rule In-8.
\end{tabbing}
Similarly, $s_2$ terminates. The theorem holds.

When $id$ is a global variable, by similar argument, the theorem holds.

\item If $id$ is of type enum $id' = \{l_1,...,l_k\}$.

If $(v_{io}<1) \vee (v_{io}>k)$,
then the execution of $s_1$ continues as follows:
\begin{tabbing}
xx\=xx\=\kill
\>\>     $(s_1, m_1)$ \\
\>= \>   $(\text{input} \, id, m_1(\vals_1))$ \\
\>$->$\> $(\text{input} \, id, m_1(1/\mathfrak{f}))$ by Rule In-6.\\
\>{\kStepArrow [i] }\> $(\text{input} \, id, m_1(1/\mathfrak{f}))$ by Rule crash.
\end{tabbing}
Similarly, $s_2$ does not terminate. The theorem holds.
When $id$ is a global variable, by similar argument, the theorem holds.

If $1 \leq v_{io} \leq k$,
assuming $id$ is a local variable,
then the execution of $s_1$ continues as follows:
\begin{tabbing}
xx\=xx\=\kill
\>\>     $(s_1, m_1)$ \\
\>= \>   $(\text{input} \, id, m_1(\vals_1))$ \\
\>$->$\> $(\text{skip}, m_1(\vals_1[l_{v_{io}} /id, \text{tl}(\vals_1({id}_I))/{id}_I,$\\
\>\>     $``\vals_1({id}_{IO}) \cdot \underline{v}_{io}"/{id}_{IO}]))$ by Rule In-5.
\end{tabbing}
Similarly, $s_2$ terminates. The theorem holds.
When $id$ is a global variable, by similar argument, the theorem holds.
\end{itemize}
\end{enumerate}

\item $s_1 = s_2 = ``\text{output} \; e"$;

There are two cases based on if evaluation of expression $e$ crashes.
By definition, $\text{CVar}(s_1) = \text{CVar}(s_2) = \text{Err}(e)$.
By assumption, $\forall x \in \text{Err}(e), \allowbreak  \vals_1(x) = \vals_2(x)$.
By Lemma~\ref{lmm:expEvalSameTerm}, evaluation of the expression $e$ w.r.t two value stores
$\vals_1$ and $\vals_2$ either both crash or both do not crash.
\begin{enumerate}
\item There is crash in evaluation of the expression $e$ w.r.t
two value stores
$\vals_1$ and $\vals_2$,
$\mathcal{E}\llbracket e\rrbracket\vals_1 = (\text{error}, v_\mathfrak{of}^1)$ and
$\mathcal{E}\llbracket e\rrbracket\vals_2 = (\text{error}, v_\mathfrak{of}^2)$.

The execution of $s_1$ continues as follows:
\begin{tabbing}
xx\=xx\=\kill
\>\>     $(s_1, m_1)$ \\
\>= \>   $(\text{output} \, e, m_1(\vals_1))$ \\
\>$->$\> $(\text{output} \, (\text{error}, v_\mathfrak{of}^1), m_1(1/\mathfrak{f}))$ by Rule EEval'\\
\>$->$\> $(\text{output} \, 0, m_1(1/\mathfrak{f}))$ by Rule ECrash.\\
\>{\kStepArrow [i] }\> $(\text{output} \, 0, m_1(1/\mathfrak{f}))$ by Rule crash.
\end{tabbing}

Similarly, $s_2$ does not terminate. The theorem holds.

\item There is no crash in evaluation of the expression $e$ w.r.t
two value stores
$\vals_1$ and $\vals_2$,
$\mathcal{E}\llbracket e\rrbracket\vals_1 = (v_1, v_\mathfrak{of}^1)$ and
$\mathcal{E}\llbracket e\rrbracket\vals_2 = (v_2, v_\mathfrak{of}^2)$.

According to rule Out-1 and Out-2, there is no exception in transformation of different typed output value.
We therefore only show the execution for output value of Int type.
The execution of $s_1$ continues as follows:
\begin{tabbing}
xx\=xx\=\kill
\>\>     $(s_1, m_1)$ \\
\>=\>    $(\text{output} \, e, m_1(\vals_1))$ \\
\>$->$\> $(\text{output} \, (v_1, v_\mathfrak{of}^1), m_1(1/\mathfrak{f}))$ by Rule EEval'\\
\>$->$\> $(\text{output} \, v_1, m_1(v_\mathfrak{of}^1/\mathfrak{of}))$ \\
\>\>     by Rule EOflow-1 or EOflow-2.\\
\>$->$\> $(\text{skip}, m_1(\vals_1[``\vals({id}_{IO}) \cdot \overline{v}_1"/{id}_{IO}]))$\\
\>\>      by Rule Out-1.
\end{tabbing}
Similarly, $s_2$ terminates. Theorem holds.
\end{enumerate}
\end{enumerate}
\end{enumerate}

\item[Second] $s_1$ and $s_2$ are input statement with same type variable: $s_1=``\text{input }{id}_1", s_2 = ``\text{input }{id}_2"$
  where $(\Gamma_{s_1} \vdash {id}_1 : t) \wedge (\Gamma_{s_2} \vdash {id}_2 : t)$;

The theorem holds by similar argument for the case $s_1 = s_2 = \text{input }{id}$.

\item[Third]  $s_1 = ``\text{output} \, e" \, \text{or} \, ``{id}_1 := e",
        s_2 = ``\text{output} \, e" \, \text{or} \, ``{id}_2 := e"$ where both of the following hold:
       \begin{itemize}
       \item There is no possible value mismatch in $``{id}_1 := e"$,
       $\neg(\Gamma_{s_1} \vdash {id}_1 : \text{Int}) \vee
        \neg(\Gamma_{s_1} \vdash e : \text{Long}) \vee
            (\Gamma_{s_1} \vdash e : \text{Int})$.

       \item There is no possible value mismatch in $``{id}_2 := e"$,
       $\neg(\Gamma_{s_2} \vdash {id}_2 : \text{Int}) \vee
        \neg(\Gamma_{s_2} \vdash e : \text{Long}) \vee
            (\Gamma_{s_2} \vdash e : \text{Int})$.
       \end{itemize}

We show that the evaluations of the expression $e$ w.r.t the value stores $\vals_1$ and $\vals_2$ either both raise an exception or both do not.
By the definition of crash variables, the crash variables of $s_1$ are those obtained by the function $\text{Err}(e)$, $\text{CVar}(s_1) = \text{Err}(e)$. Similarly, the termination deciding variables of $s_2$ are $\text{Err}(e)$.
By assumption, the initial value stores $\vals_1$ and $\vals_2$ agree on values of those in  $\text{CVar}(s_1)$ and $\text{CVar}(s_2)$, $\forall x \in \text{Err}(e) = (\text{CVar}(s_1) \cup \text{CVar}(s_2))\,:\,
  \vals_1(x) = \vals_2(x)$.
By Lemma~\ref{lmm:expEvalSameTerm}, the evaluations of expression $e$ w.r.t two value stores,
$\vals_1$ and $\vals_2$, either both raise an exception or both do not raise an exception.

\begin{enumerate}
\item The evaluations of the expression $e$ raise an exception w.r.t two value stores
$\vals_1$ and $\vals_2$,
$\mathcal{E}'\llbracket e\rrbracket\vals_1 = (\text{error}, v_\mathfrak{of}^1),
 \mathcal{E}'\llbracket e\rrbracket\vals_2 = (\text{error}, v_\mathfrak{of}^2)$:

We show the execution of $s_1$ proceeds to an configuration where the crash flag is set and then does not terminate.

When $s_1 = ``\text{output }e"$, the execution of $``\text{output }e"$ proceeds as follows.

\begin{tabbing}
xx\=xx\=\kill
\>\>                   $(\text{output }e, m_1(\vals_1))$ \\
\>$->$\>               $(\text{output }(\text{error}, v_\mathfrak{of}^1), m_1(\vals_1))$ by rule EEval'\\
\>$->$\>               $(\text{output }0, m_1(1/\mathfrak{f}))$ by rule ECrash\\
\>{\kStepArrow [i] }\> $(\text{output }0, m_1(1/\mathfrak{f}))$ for any $i\geq0$, by rule Crash.
\end{tabbing}

When $s_1 = ``{id}_1 := e"$, the execution of $``{id}_1 := e"$ proceeds as follows.

\begin{tabbing}
xx\=xx\=\kill
\>\>                   $({id}_1 := e, m_1(\vals_1))$ \\
\>$->$\>               $({id}_1 := (\text{error}, v_\mathfrak{of}^1), m_1(\vals_1))$ by rule EEval'\\
\>$->$\>               $({id}_1 := 0, m_1(1/\mathfrak{f}))$ by rule ECrash\\
\>{\kStepArrow [i] }\> $({id}_1 := 0, m_1(1/\mathfrak{f}))$ for any $i\geq0$, by rule Crash.
\end{tabbing}

Similarly, the execution of $s_2$ proceeds to a configuration where the crash flag is set.
Then, by the crash rule, the execution of $s_2$ does not terminate.
The theorem~\ref{thm:mainTermSameWaySimpleStmt} holds.

\item the evaluations of expression $e$ do not raise an exception w.r.t two value stores,
$\vals_1$ and $\vals_2$,
$\mathcal{E}'\llbracket e\rrbracket\vals_1 = (v_1, v_\mathfrak{of}^1),
 \mathcal{E}'\llbracket e\rrbracket\vals_2 = (v_2, v_\mathfrak{of}^2)$:

We show the execution of $s_1$ terminates.

When $s_1 = \text{output }(e)$, the execution of $\text{output }(e)$ proceeds as follows.
W.l.o.g, we assume expression $e$ is of type Int.
This is allowed by the condition that it does not hold that
   $(\Gamma_{s_1} \vdash e : \text{Long}) \wedge
\neg(\Gamma_{s_1} \vdash e : \text{Int})$.
\begin{tabbing}
xx\=xx\=\kill
\>\>     $(\text{output }e, m_1(\vals_1))$ \\
\>$->$\> $(\text{output }(v_1, v_\mathfrak{of}^1), m_1(\vals_1))$ by rule EEval'\\
\>$->$\> $(\text{output }v_1, m_1(v_\mathfrak{of}^1/\mathfrak{of}, \vals_1))$ \\
\>\>     by rule E-Oflow1 or E-Oflow2\\
\>$->$\> $(\text{skip}, m_1(\vals_1[``\vals_1({id}_{IO})\cdot\bar{v_1}"/{id}_{IO}]))$ by rule Out.
\end{tabbing}

When $s_1 = ``{id}_1 := e"$,
by assumption, the expression $e$ is of type Int, there is no possible value mismatch in execution of
$``{id}_1 := e"$ because the only possible value mismatch occurs when assigning a value of type Long but not Int to a variable of type Int.
By the condition $\neg(\Gamma_{s_1} \vdash {id}_1 : \text{Int}) \vee
        \neg(\Gamma_{s_1} \vdash e : \text{Long}) \vee
            (\Gamma_{s_2} \vdash e : \text{Int})$,
when expression $e$ is of type Long, then the variable ${id}_1$ is not of type Int.
In summary, there is no value mismatch.

The execution of $``{id}_1 := e"$ proceeds as follows.
\begin{tabbing}
xx\=xx\=\kill
\>\>                   $({id}_1 := e, m_1(\vals_1))$ \\
\>$->$\>               $({id}_1 := (v_1, v_\mathfrak{of}^1), m_1(\vals_1))$ by rule EEval'\\
\>$->$\>               $({id}_1 := v_1, m_1(v_\mathfrak{of}^1/\mathfrak{of}, \vals_1))$ by rule EEval'\\
\>$->$\>               $(\text{skip}, m_1(\vals_1[v_1/{id}_1]))$ by the rule Assign.\\
\end{tabbing}
When ${id}_1$ is a variable of enumeration or Long type, by similar argument, the theorem still holds.

Similarly, the execution of $s_2$ terminates when started in the state $m_2(\vals_2)$.
Theorem~\ref{thm:mainTermSameWaySimpleStmt} holds.
\end{enumerate}
\end{description}
\end{proof}

\begin{theorem}\label{thm:mainTermSameWayLocal}
If two statement sequences $S_1$ and $S_2$ satisfy the proof rule of termination in the same way,
$S_1 \equiv_{H}^S S_2$,
and their respective initial states $m_1(\mathfrak{f}_1, \vals_1)$ and $m_2(\mathfrak{f}_2, \vals_2)$ with crash flags not set,
 $\mathfrak{f}_1 = \mathfrak{f}_2 = 0$,
 and whose value stores agree on values of the termination deciding variables of $S_1$ and $S_2$,
 $\forall x \in \text{TVar}(S_1) \cup \text{TVar}(S_2)\,:\,
  \vals_1(x) = \vals_2(x)$,
 then $S_1$ and $S_2$ terminate in the same way when started in states $m_1$ and $m_2$ respectively:
$(S_1, m_1) \equiv_{H} (S_2, m_2)$.
\end{theorem}

\begin{proof}
The proof is by induction on $\text{size}(S_1) + \text{size}(S_2)$, the sum of program size of $S_1$ and $S_2$.

\noindent{\bf Base case.}
$S_1$ and $S_2$ are simple statement. By Theorem~\ref{thm:mainTermSameWaySimpleStmt}, Theorem~\ref{thm:mainTermSameWayLocal} holds.

\noindent{\bf Induction step}.

\noindent There are two hypotheses.
The hypothesis IH is that Theorem~\ref{thm:mainTermSameWayLocal} holds when $\text{size}(S_1)+\text{size}(S_2) = k\geq2$.

\noindent We show Theorem~\ref{thm:mainTermSameWayLocal} holds when $\text{size}(S_1) + \text{ size}(S_2) = k+1$.

The proof of Theorem~\ref{thm:mainTermSameWayLocal} is a case analysis  according to the cases in the definition of the proof rule of termination in the same way, $S_1 \equiv_{H}^S S_2$.
\begin{enumerate}
\item $S_1$ and $S_2$ are one statement and one of the following holds.

\begin{enumerate}
\item
$S_1$ = ``If($e$) then \{$S_1^t$\} else \{$S_1^f$\}",
$S_2$ = ``If($e$) then \{$S_2^t$\} else \{$S_2^f$\}" such that one of the following holds:

\begin{enumerate}
\item $S_1^t, S_1^f, S_2^t, S_2^f$ are all sequences of ``skip";

We show that the evaluation of expression $e$ w.r.t the value store $\vals_1$ and $\vals_2$ either both raise an exception or both do not.
By the definition of crash/loop variables,
$\text{CVar}(S_1^t) = \text{CVar}(S_1^f) = \emptyset$,
$\text{LVar}(S_1) = \emptyset$.
By the definition of termination deciding variables, the termination deciding variables of $S_1$ is the crash variables of $S_1$, $\text{TVar}(S_1) = \text{CVar}(S_1) = \text{Err}(e)$.
By assumption, the value stores $\vals_1$ and $\vals_2$ agree on the values of those in the crash variables of $S_1$ and $S_2$, $\forall x \in \text{Err}(e) = \text{TVar}(S_1) = \text{TVar}(S_2), \vals_1(x) = \vals_2(x)$.
By the property of the expression meaning function $\mathcal{E}$, the evaluation of predicate expression $e$ of $S_1$ and $S_2$ w.r.t value store $\vals_1$ and $\vals_2$ either both crash or both do not crash, $(\mathcal{E}\llbracket e\rrbracket\vals_1 = \mathcal{E}\llbracket e\rrbracket\vals_2 = \text{error}) \vee
\big((\mathcal{E}\llbracket e\rrbracket\vals_1 \neq \text{error}) \wedge (\mathcal{E}\llbracket e\rrbracket\vals_2 \neq \text{error})\big)$.
Then we show that Theorem~\ref{thm:mainTermSameWayLocal} holds in either of the two possibilities.

\begin{enumerate}
\item $\mathcal{E}\llbracket e\rrbracket\vals_1 = \mathcal{E}\llbracket e\rrbracket\vals_2 = \text{error}$.

The execution of $S_1$ proceeds as follows:

\begin{tabbing}
xx\=xx\=\kill
\>\>                   $(\text{If}(e) \text{ then }\{S_1^t\} \text{ else }\{S_1^f\}, m_1(\vals_1))$ \\
\>$->$\>               $(\text{If}(\text{error}) \text{ then }\{S_1^t\} \text{ else }\{S_1^f\}, m_1(\vals_1))$ by rule EEval\\
\>$->$\>               $(\text{If}(0) \text{ then }\{S_1^t\} \text{ else }\{S_1^f\}, m_1(1/\mathfrak{f}, \vals_1))$ by rule ECrash\\
\>{\kStepArrow [i] }\> $(\text{If}(0) \text{ then }\{S_1^t\} \text{ else }\{S_1^f\}, m_1(1/\mathfrak{f}, \vals_1))$ for any $i\geq0$,\\
\>\>                   by rule Crash.
\end{tabbing}

Similarly, the execution of $S_2$ started in the state $m_2(\vals_2)$ does not terminate. The theorem~\ref{thm:mainTermSameWayLocal} holds.

\item $(\mathcal{E}\llbracket e\rrbracket\vals_1 \neq \text{error}) \wedge (\mathcal{E}\llbracket e\rrbracket\vals_2 \neq \text{error})$.

W.l.o.g, $\mathcal{E}\llbracket e\rrbracket\vals_1 = v_1 \neq 0$, $\mathcal{E}\llbracket e\rrbracket\vals_2 = 0$.
Then the execution of $S_1$ proceeds as follows.

\begin{tabbing}
xx\=xx\=\kill
\>\>            $(\text{If}(e) \text{ then }\{S_1^t\} \text{ else }\{S_1^f\}, m_1(\vals_1))$ \\
\>$->$\>        $(\text{If}(v_1) \text{ then }\{S_1^t\} \text{ else }\{S_1^f\}, m_1(\vals_1))$ by rule EEval\\
\>$->$\>        $(S_1^t, m_1(\vals_1))$ by rule If-T\\
\>$->*$\>       $(\text{skip}, m_1')$ by rule Skip.
\end{tabbing}

Similarly, the execution of $S_2$ started in the state $m_2(\vals_2)$ terminates. The theorem~\ref{thm:mainTermSameWayLocal} holds.
\end{enumerate}

\item At least one of $S_1^t, S_1^f, S_2^t, S_2^f$ is not a sequence of ``skip" and
      $(S_1^t \equiv_{H}^S S_2^t) \wedge (S_1^f \equiv_{H}^S S_2^f)$;

W.l.o.g., $S_1^t$ is not of ``skip" only.
We show that the evaluation of the expression $e$ w.r.t the value stores $\vals_1$ and $\vals_2$ either both raise an exception or both produce the same integer value.
Then there is either some loop statement in $S_1^t$ or the crash variables of $S_1^t$ are not $\emptyset$ or both.
\begin{enumerate}
\item When there is some loop statement in $S_1^t$, then, by the definition of loop variables, the loop variables of $S_1$ include all variables used in the predicate expression of $S_1$,
$\text{LVar}(S_1) = \text{Use}(e) \cup \text{LVar}(S_1^t) \cup \text{LVar}(S_1^f)$.

\item When the crash variables of $S_1^t$ are not $\emptyset$, then, by the definition of crash variables, the crash variables of $S_1$ include all variables used in the predicate expression of $S_1$,
$\text{CVar}(S_1) = \text{Use}(e) \cup \text{CVar}(S_1^t) \cup \text{CVar}(S_1^f)$.
\end{enumerate}
In summary, all variables used in predicate expression of $S_1$ is a subset of termination deciding variables of $S_1$,
$\text{Use}(e) \subseteq \text{TVar}(S_1)$.
By assumption, the value store $\vals_1$ and $\vals_2$ agree on the values of those in the termination deciding variables of $S_1$ and $S_2$. It follows, by the property of expression meaning function $\mathcal{E}$, the evaluation of the predicate expression $e$ of $S_1$ and $S_2$ produce the same value w.r.t the value store $\vals_1$ and $\vals_2$,
$\mathcal{E}\llbracket e\rrbracket\vals_1 = \mathcal{E}\llbracket e\rrbracket\vals_2$.
Then either the evaluations of the predicate expression $e$ of $S_1$ and $S_2$ both crash w.r.t the value store $\vals_1$ and $\vals_2$, or both evaluations produce the same integer value,
$(\mathcal{E}\llbracket e\rrbracket\vals_1 = \mathcal{E}\llbracket e\rrbracket\vals_2 = \text{error}) \vee
 (\mathcal{E}\llbracket e\rrbracket\vals_1 = \mathcal{E}\llbracket e\rrbracket\vals_2 = v \neq \text{error})$. We show Theorem~\ref{thm:mainTermSameWayLocal} holds in either of the two possibilities.

\begin{enumerate}
\item $\mathcal{E}\llbracket e\rrbracket\vals_1 = \mathcal{E}\llbracket e\rrbracket\vals_2 = \text{error}$.

The execution of $S_1$ proceeds as follows:

\begin{tabbing}
xx\=xx\=\kill
\>\>                   $(\text{If}(e) \text{ then }\{S_1^t\} \text{ else }\{S_1^f\}, m_1(\vals_1))$ \\
\>$->$\>               $(\text{If}(\text{error}) \text{ then }\{S_1^t\} \text{ else }\{S_1^f\}, m_1(\vals_1))$ by rule EEval\\
\>$->$\>               $(\text{If}(0) \text{ then }\{S_1^t\} \text{ else }\{S_1^f\}, m_1(1/\mathfrak{f}, \vals_1))$ by rule ECrash\\
\>{\kStepArrow [i] }\> $(\text{If}(0) \text{ then }\{S_1^t\} \text{ else }\{S_1^f\}, m_1(1/\mathfrak{f}, \vals_1))$ for any $i\geq0$,\\
\>\>                   by rule Crash.
\end{tabbing}

Similarly, the execution of $S_2$ started from state $m_2(\vals_2)$ does not terminate. The theorem~\ref{thm:mainTermSameWayLocal} holds.

\item $\mathcal{E}\llbracket e\rrbracket\vals_1 = \mathcal{E}\llbracket e\rrbracket\vals_2 = v \neq \text{error}$, w.l.o.g., $v = 0$.

Then the execution of $S_1$ proceeds as follows:

\begin{tabbing}
xx\=xx\=\kill
\>\>                   $(\text{If}(e) \text{ then }\{S_1^t\} \text{ else }\{S_1^f\}, m_1(\vals_1))$ \\
\>$->$\>               $(\text{If}(0) \text{ then }\{S_1^t\} \text{ else }\{S_1^f\}, m_1(\vals_1))$ by rule EEval\\
\>$->$\>               $(S_1^f, m_1(\vals_1))$ by rule If-F.
\end{tabbing}

Similarly, after two steps of execution, $S_2$ gets to the configuration $(S_2^f, m_2(\vals_2))$.

\myNewLine

We show that $S_1^f$ and $S_2^f$ terminate in the same way when started in the state $m_1(\vals_1)$ and $m_2(\vals_2)$ respectively.
Because $S_1^f \equiv_{H}^S S_2^f$, by Corollary~\ref{coro:sameTermVarFromEquivTerm}, the termination deciding variables of $S_1^f$ and $S_2^f$ are same, $\text{TVar}(S_1^f) = \text{TVar}(S_2^f)$.
By the definition of crash/loop variables, $\text{CVar}(S_1^f) \subseteq \text{CVar}(S_1)$ and
                                           $\text{LVar}(S_1^f) \subseteq \text{LVar}(S_1)$.
Hence, the termination deciding variables of $S_1^f$ are a subset of the termination deciding variables of $S_1$, $\text{TVar}(S_1^f) \subseteq \text{TVar}(S_1)$. Similarly, $\text{TVar}(S_2^f) \subseteq \text{TVar}(S_2)$.
Therefore, the value store $\vals_1$ and $\vals_2$ agree on the values of those in the termination deciding variables of $S_1^f$ and $S_2^f$, $\forall y \in \text{TVar}(S_1^f) \cup \text{TVar}(S_2^f)\,:\, \vals_1(y) = \vals_2(y)$.
In addition, the sum of program size of $S_1^f$ and $S_2^f $ is less than $k$ because program size of each of $S_1^t$ and $S_2^t$ is great than or equal to one, $\text{size}(S_1^f) + \text{size}(S_2^f) < k$.
As is shown, crash flags are not set.
Therefore, by the hypothesis IH, $S_1^f$ and $S_2^f$ terminate in the same way when started in state $m_1(\mathfrak{f}_1, \vals_1)$ and $m_2(\mathfrak{f}_2, \vals_2)$, $(S_1^f, m_1(\mathfrak{f}_1, \vals_1)) \equiv_{H} (S_2^f, m_2(\mathfrak{f}_2, \vals_2))$. Hence, Theorem~\ref{thm:mainTermSameWayLocal} holds.
\end{enumerate}
\end{enumerate}

\item
$S_1 = ``\text{while}_{\langle n_1\rangle} (e) \; \{S_1''\}"$,
$S_2 = ``\text{while}_{\langle n_2\rangle} (e) \; \{S_2''\}"$ such that both of the following hold:
    \begin{itemize}
    \item $S_1'' \equiv_{H}^S S_2''$;

    \item $S_1''$ and $S_2''$ have equivalent computation of $\text{TVar}(S_1) \cup \text{ TVar}(S_2)$;
    \end{itemize}

By Corollary~\ref{coro:loopTermInSameWay}, we show $S_1$ and $S_2$ terminate in the same way when started from state $m_1(\mathfrak{f}_1, m_c^1,\vals_1)$ and $m_2(\mathfrak{f}_2, m_c^2,\vals_2)$ respectively.
We need to show that all required conditions are satisfied.
\begin{itemize}
\item The crash flags are not set, $\mathfrak{f}_1 = \mathfrak{f}_2 = 0$.

\item The loop counter value of $S_1$ and $S_2$ are zero: $m_c^1(n_1) = m_c^2(n_2) = 0$.

\item The value stores $\vals_1$ and $\vals_2$ agree on the values of those in the termination deciding variables of $S_1$ and $S_2$, $\forall x \in \text{TVar}(S_1) \cap \text{TVar}(S_2)\,:\, \vals_1(x) = \vals_2(x)$.

The three above conditions are from assumption.

\item $S_1$ and $S_2$ have same set of termination deciding variables,
$\text{TVar}(S_1) = \text{TVar}(S_2)$.

By Corollary~\ref{coro:sameTermVarFromEquivTerm}.

\item The loop body $S_1''$ of $S_1$ and $S_2''$ of $S_2$ terminate in the same way when started in state $m_{S_1}(\mathfrak{f}_{S_1}, \vals_{S_1})$ and $m_{S_2}(\mathfrak{f}_{S_2}, \vals_{S_2})$ with crash flags not set and in which value stores
    agree on the values of those in the termination deciding variables of $S_1''$ and $S_2''$:
    $((\forall x \in \text{TVar}(S_1'') \cup \text{TVar}(S_2'')\,:\, \vals_{S_1}(x) = \vals_{S_2}(x)) \wedge
                (\mathfrak{f}_{S_1} = \mathfrak{f}_{S_2} = 0)) =>
                (S_1'', m_{S_1}(\mathfrak{f}_{S_1}, \vals_{S_1})) \equiv_{H} (S_2'', m_{S_2}(\mathfrak{f}_{S_2}, \vals_{S_2}))$.

By the definition of program size, $\text{size}(S_1) = \text{size}(S_1'') + 1, \text{size}(S_2) = \text{size}(S_2'') + 1$. Then, $\text{size}(S_1'') + \text{size}(S_2'') < k$.
Then, by the hypothesis IH, the loop body $S_1''$ of $S_1$ and $S_2''$ of $S_2$ terminate in the same way when started in state $m_{S_1}(\vals_{S_1})$ and $m_{S_2}(\vals_{S_2})$ with crash flags not set and whose value stores agree on values of the termination deciding variables of $S_1''$ and $S_2''$.

\end{itemize}

Then, by Corollary~\ref{coro:loopTermInSameWay}, $S_1$ and $S_2$ terminate in the same way when started in the states $m_1(m_c^1,\vals_1)$ and $m_2(m_c^2,\vals_2)$ respectively. The theorem~\ref{thm:mainTermSameWayLocal} holds.
\end{enumerate}

\item $S_1$ and $S_2$ are not both one statement and one of the following holds:

\begin{enumerate}

\item $S_1 =S_1';s_1, S_2 = S_2';s_2$ and all of the following hold:

\begin{itemize}
\item $S_1' \equiv_{H}^S S_2'$;

\item $S_1'$ and $S_2'$ have equivalent computation of $\text{TVar}(s_1) \cup \text{TVar}(s_2)$;

\item
$s_1 \equiv_{H}^S s_2$ where $s_1$ and $s_2$ are not sequences of ``skip";
\end{itemize}

By the hypothesis IH, we show that $S_1'$ and $S_2'$ terminate in the same way when started in the states
$\state_1(\mathfrak{f}_1, \vals_1), \state_2(\mathfrak{f}_2, \vals_2)$ respectively,
$(S_1', \state_1(\mathfrak{f}_1, \vals_1)) \equiv_{H} (S_2', \state_2(\mathfrak{f}_2, \vals_2))$.
We need to show all required conditions are satisfied.
\begin{itemize}
\item Crash flags are not set, $\mathfrak{f}_1 = \mathfrak{f}_2 = 0$;

By assumption.

\item $\text{size}(S_1') + \text{size}(S_2') < k$.

By the definition, $\text{size}(s_1) \geq 1, \text{size}(s_2) \geq 1$.
Hence $\text{size}(S_1') + \text{size}(S_2') < k$.

\item Value stores $\vals_1$ and $\vals_2$ agree on values of the termination deciding variables of $S_1'$ and $S_2'$.

Besides, by the definition of loop/crash variables, $\text{LVar}(S_1') \subseteq \text{LVar}(S_1)$ and $\text{CVar}(S_1') \subseteq \text{CVar}(S_1)$. Hence, $\text{TVar}(S_1') \subseteq \text{TVar}(S_1)$. Similarly, $\text{TVar}(S_2') \subseteq \text{TVar}(S_2)$.
Then, value stores $\vals_1$ and $\vals_2$ agree on the values of those in the termination deciding variables of $S_1'$ and $S_2'$,
$\forall x \in \text{TVar}(S_1') \cup \text{TVar}(S_2')\,:\, \vals_1(x) = \vals_2(x)$.
\end{itemize}
Then, by the hypothesis IH, $S_1'$ and $S_2'$ terminate in the same way when started in the states $\state_1(\mathfrak{f}_1, \vals_1), \state_2(\mathfrak{f}_2, \vals_2)$ respectively, $(S_1', m_1(\mathfrak{f}_1, \vals_1)) \equiv_{H} (S_2', m_2(\mathfrak{f}_2, \vals_2))$.

\myNewLine

If the execution of $S_1'$ and $S_2'$ terminate when started in the states
$m_1(\mathfrak{f}_1, \vals_1)$ and $m_2(\mathfrak{f}_2, \vals_2)$ respectively, we show that $s_1$ and $s_2$ terminate in the same way.
We prove that $S_1'$ and $S_2'$ equivalently compute the termination deciding variables of $s_1$ and $s_2$ by Theorem~\ref{thm:equivCompMain}.
\begin{itemize}
\item Crash flags are not set, $\mathfrak{f}_1 = \mathfrak{f}_2 = 0$;

By definition of terminating execution of $S_1'$ and $S_2'$ when started in states $m_1$ and $m_2$ respectively.

\item The executions of $S_1'$ and $S_2'$ terminate when started in the states $m_1(\vals_1)$ and $m_2(\vals_2)$.

By assumption,
$(S_1', \state_1(\vals_1)) ->* (\text{skip}, \state_1'(\vals_1'))$ and
$(S_2', \state_2(\vals_2)) ->* (\text{skip}, \state_2'(\vals_2'))$.

\item $s_1$ and $s_2$ have same termination deciding variables.

By Corollary~\ref{coro:sameTermVarFromEquivTerm}, $s_1$ and $s_2$ have same termination deciding variables, $\text{TVar}(s_1)$ = $\text{TVar}(s_2) = \text{TVar}(s)$.

\item Value stores $\vals_1$ and $\vals_2$ agree on the values of variables in $\text{Imp}(S_1', \text{TVar}(s)) \cup  \allowbreak\text{Imp}(S_2', \text{TVar}(s))$.

By the definition of loop/crash variables,
$\text{Imp}(S_1', \text{LVar}(s_1)) \subseteq \text{LVar}(S_1)$ and
$\text{Imp}(S_1', \text{CVar}(s_1))\allowbreak \subseteq \text{CVar}(S_1)$.
Hence, by Lemma~\ref{lmm:impVarUnionLemma}, the imported variables in $S_1'$ relative to the termination deciding variables of $s_1$ is a subset of the termination deciding variables of $S_1$, $\text{Imp}(S_1', \text{TVar}(s))\allowbreak \subseteq \text{TVar}(S_1)$.
Similarly, $\text{Imp}(S_2', \text{TVar}(s))\allowbreak \subseteq \text{TVar}(S_2)$.
Thus, by assumption, the value stores $\vals_1$ and $\vals_2$ agree on the values of the variables in
$\text{Imp}(S_1', \text{TVar}(s)) \cup \text{Imp}(S_2', \text{TVar}(s))$.
\end{itemize}
By Theorem~\ref{thm:equivCompMain}, $\forall x \in \text{TVar}(s)\,:\, \vals_1'(x) = \vals_2'(x)$.

By Corollary~\ref{coro:termSeq},
$(S_1';s_1, \state_1(\vals_1)) ->* (s_1, \state_1'(\mathfrak{f}_1, \vals_1'))$ and
$(S_2';s_2, \state_2(\vals_2)) ->* (s_2, \state_2'(\mathfrak{f}_2, \vals_2'))$.
Then, by the hypothesis IH, we show that $s_1$ and $s_2$ terminate in the same way when started in the states $m_1'(\vals_1')$ and $m_2'(\vals_2')$. We show that all required conditions are satisfied.
$\text{size}(s_1) + \text{size}(s_2) < k$ because $\text{size}(S_1') \geq 1, \text{size}(S_2') \geq 1$ by the definition of program size.
If $s_1, s_2$ are loop statement, then, by the assumption of unique loop labels, $s_1 \notin S_1', s_2 \notin S_2'$.
Then, by Corollary~\ref{coro:loopCntRemainsSame}, the loop counter value of $s_1$ and $s_2$ is not redefined in the execution of $S_1'$ and $S_2'$ respectively.
By the hypothesis IH, $s_1$ and $s_2$ terminate in the same way when started in the states $\state_1'(\mathfrak{f}_1, \vals_1')$ and $\state_2'(\mathfrak{f}_2, \vals_2')$, $(s_1, \state_1'(\mathfrak{f}_1, \vals_1')) \equiv_{H} (s_2, \state_2'(\mathfrak{f}_2, \vals_2'))$. The theorem~\ref{thm:mainTermSameWayLocal} holds.

\item One last statement is ``skip":
w.l.o.g., $(s_2 =``\text{skip}") \wedge (S_1 \equiv_{H}^S S_2')$.

We show that $S_1$ and $S_2'$ terminate in the same way when started in the states $m_1(\vals_1)$ and $m_2(\vals_2)$ respectively by the hypothesis IH.
By the definition of crash/loop variables, $\text{CVar}(S_2') \subseteq \text{CVar}(S_2)$, $\text{LVar}(S_2') \subseteq \text{LVar}(S_2)$.
Then, by assumption, $\forall x \in \text{TVar}(S_2') \cup \text{TVar}(S_1)\,:\, \vals_1(x) = \vals_2(x)$.
Besides, size $(s_2) \geq 1$ by the definition of program size. Then size $(S_1) + \text{size }(S_2') \leq k$.
By the hypothesis IH, $S_1$ and $S_2'$ terminate in the same way when started in the states $\state_1(\mathfrak{f}_1, \vals_1), \state_2(\mathfrak{f}_2, \vals_2)$,
$(S_1, \state_1(\mathfrak{f}_1, \vals_1)) \equiv_{H} (S_2', \state_2(\mathfrak{f}_2, \vals_2))$.

When the execution of $S_1$ and $S_2'$ terminate when started in the states $m_1(\vals_1)$ and $m_2(\vals_2)$ respectively, $s_2$ terminates after the execution of $S_2'$ by the definition of terminating execution.

\item One last statement is a ``duplicate" statement such that one of the following holds:

W.l.o.g., $S_2 = S_2';s_2';S_2'';s_2$ and all of the following hold:
\begin{itemize}
\item $S_1 \equiv_{H}^S S_2';s_2';S_2''$;

\item $s_2' \equiv_{H}^S s_2$;

\item $\text{Def}(s_2';S_2'') \cap \text{TVar}(s_2) = \emptyset$;

\item $s_2 \neq ``\text{skip}"$;
\end{itemize}

We show that $S_1$ and $S_2';s_2';S_2''$ terminate in the same way when started in the states
$\state_1(\mathfrak{f}_1, \vals_1),$

\noindent$\state_2(\mathfrak{f}_2, \vals_2)$ respectively by the hypothesis IH. The proof is same as that in case a).

We show that $s_2$ terminates if the execution of $S_2';s_2';S_2''$ terminates.
We need to prove that $s_2'$ and $s_2$ start in the states agreeing on the values of variables in $\text{TVar}(s_2)$.
By assumption, $S_2';s_2';S_2''$ terminates,
$(S_2';s_2';S_2'', \state_2(\mathfrak{f}_2, \vals_2)) ->* (\text{skip}, \state_2'(\mathfrak{f}_2, \vals_2'))$.
Then, by Corollary~\ref{coro:termSeq},
$(S_2';s_2';S_2'';s_2, \state_2(\mathfrak{f}_2, \vals_2))\allowbreak ->* (s_2, \state_2'(\mathfrak{f}_2, \vals_2'))$.
In addition, the execution of $S_2'$ and $s_2'$ must terminate because the execution of $S_2';s_2';S_2''$ terminates,

\noindent$(S_2';s_2';S_2'';s_2, \state_2(\mathfrak{f}_2, \vals_2)) ->*
          (s_2';S_2'';s_2, \state_2''(\mathfrak{f}_2, \vals_2''))->*
          (s_2, \state_2'(\mathfrak{f}_2, \vals_2'))$.

By assumptin, $\text{Def}(s_2';S_2'') \cap \text{TVar}(s_2) = \emptyset$.
Then, by Corollary~\ref{coro:defExclusion}, the value store $\vals_2''$ and $\vals_2'$ agree on values of the termination deciding variables of $s_2$, $\forall x \in \text{TVar}(s_2)\,:\, \vals_2''(x) = \vals_2'(x)$.
By Corollary~\ref{coro:sameTermVarFromEquivTerm}, $\text{TVar}(s_2') = \text{TVar}(s_2)$.
Because the execution of $s_2'$ terminates, then the execution of $s_2$ terminates when started in the state
$\state_2'(\mathfrak{f}_2, \vals_2')$ by the hypothesis IH,
$(s_2, \state_2'(\mathfrak{f}_2, \vals_2')) ->* (\text{skip}, \state_2'')$.

In addition, we show that there is no input statement in $s_2$ by contradiction.
Suppose there is input statement in $s_2$. By Lemma~\ref{lmm:inputVarInStmtSeqWithInputStmt}, ${id}_I \in \text{CVar}(s_2)$.
Hence, the input sequence variable is in the termination deciding variables of $s_2$, ${id}_I \in \text{TVar}(s_2)$.
By Corollary~\ref{coro:sameTermVarFromEquivTerm}, $\text{TVar}(s_2) = \text{TVar}(s_2')$.
Then, there must be one input statement in $s_2'$. Otherwise, by Lemma~\ref{coro:inputVarNotInTVarWithoutInputStmt}, the input sequence variable is not in the termination deciding variables of $s_2'$. A contradiction against the result that
${id}_I \in \text{TVar}(s_2')$. Since there is one input statement in $s_2'$, by Lemma~\ref{lmm:inputVarInStmtSeqWithInputStmt}, ${id}_I \in \text{Def}(s_2')$. Thus, by defintion, ${id}_I \in \text{Def}(s_2';S_2'')$. Then, $\text{Def}(s_2';S_2'') \cap \text{TVar}(s_2) \neq \emptyset$. A contradiction.
Therefore, there is no input statement in $s_2$.

\item $S_1 = S_1';s_1;s_1';$ and $S_2 = S_2';s_2;s_2'$ where $s_1$ and $s_2$ are reordered and all of the following hold:
\begin{itemize}
\item ${S_1'} \equiv_{H}^S {S_2'}$;

\item $S_1'$ and $S_2'$ have equivalent computation of $\text{TVar}(s_1;s_1') \cup \text{TVar}(s_2;s_2')$.

\item ${s_1} \equiv_{H}^S {s_2'}$;

\item ${s_1'} \equiv_{H}^S {s_2}$;

\item $\text{Def}(s_1) \cap \text{TVar}(s_1') = \emptyset$;

\item $\text{Def}(s_2) \cap \text{TVar}(s_2') = \emptyset$;
\end{itemize}

The proof is to show that if $S_1$ terminates when started in the state $m_1$, the $S_2$ terminates when started in the state $m_2$, and vice versa.
Due to the symmetric conditions, it is suffice to show one direction that, w.l.o.g., $(S_1, \state_1) ->* (\text{skip}, \state_1') => (S_2, \state_2) ->* (\text{skip}, \state_2')$.


We show that the execution of $S_2'$ terminates by the hypothesis IH.
We need to show that all required conditions are satisfied.
\begin{itemize}
\item $\text{size}(S_1') + \text{size}(S_2') < k$.

This is because $\text{size}(s_1;s_1') >1, \text{size}(s_2;s_2') >1$.

\item
Initial value stores $\vals_1$ and $\vals_2$ agree on  values of the termination deciding variables of $S_1'$ and $S_2'$,
$\forall x \in \text{TVar}(S_1') \cup \text{TVar}(S_2')\,:\, \vals_1(x) = \vals_2(x)$.

We show that $\text{TVar}(S_1') \subseteq \text{TVar}(S_1)$.
In the following, we prove that $\text{CVar}(S_1') \subseteq \text{CVar}(S_1)$.
\begin{tabbing}
xxxxx\=xxxx\=\kill
\>\>                   $\text{CVar}(S_1')$\\
\>$\subseteq$\>        $\text{CVar}(S_1';s_1)$      by the defintion of $\text{CVar}(S_1';s_1)$ \\
\>$\subseteq$\>        $\text{CVar}(S_1';s_1;s_1')$ by the defintion of $\text{CVar}(S_1';s_1;s_1')$
\end{tabbing}
Similarly, $\text{LVar}(S_1') \subseteq \text{LVar}(S_1)$.
Hence, $\text{TVar}(S_1') \subseteq \text{TVar}(S_1)$.
Similarly, $\text{TVar}(S_2') \subseteq \text{TVar}(S_2)$.
By assumption, initial value stores $\vals_1$ and $\vals_2$ agree on  values of the termination deciding variables of $S_1'$ and $S_2'$.
\end{itemize}
By the hypothesis IH, $(S_1', m_1(\vals_1)) \equiv_{H} (S_2', m_2(\vals_2))$. Because the execution of $S_1$ terminates, then $S_1'$ terminates when started in the state $m_1(\vals_1)$,
$(S_1', m_1(\vals_1)) ->* (\text{skip}, m_1'(\vals_1'))$.
Therefore, $S_2'$ termiantes when started in the state $m_2(\vals_2)$, $(S_2', m_2(\vals_2)) ->* (\text{skip}, m_2'(\vals_2'))$.

We show that after the execution of $S_1'$ and $S_2'$, value stores agree on values of the termination deciding variables of $s_1;s_1'$ and $s_2;s_2'$, $\forall x \in \text{TVar}(s_1;s_1') \cup \text{TVar}(s_2;s_2'), \vals_1'(x) = \vals_2'(x)$.
We split the argument into two steps.
\begin{enumerate}
\item
We show that $\text{TVar}(s_1;s_1') = \text{TVar}(s_2;s_2')$.

By Corollary~\ref{coro:sameTermVarFromEquivTerm},
$\text{TVar}(s_1) = \text{TVar}(s_2')$ and $\text{TVar}(s_1') = \text{TVar}(s_2)$.
Then we show that $\text{TVar}(s_1;s_1') = \text{TVar}(s_1) \cup \text{TVar}(s_1')$.
To do that, we show that $\text{CVar}(s_1;s_1') = \text{CVar}(s_1) \cup \text{CVar}(s_1')$.
\begin{tabbing}
xx\=xx\=\kill
\>\>                   $\text{CVar}(s_1;s_1')$\\
\>=\>                  $\text{CVar}(s_1) \cup \text{Imp}(s_1, \text{CVar}(s_1'))$      by the defintion of $\text{CVar}(s_1;s_1')$ \\
\>=\>                  $\text{CVar}(s_1) \cup \text{CVar}(s_1')$ by $\text{Def}(s_1) \cap \text{TVar}(s_1') = \emptyset$ and\\
\>\>                   the defintion of $\text{Imp}(\cdot)$.
\end{tabbing}
Similarly, $\text{LVar}(s_1;s_1') = \text{LVar}(s_1) \cup \text{LVar}(s_1')$.
Thus, $\text{TVar}(s_1;s_1')\allowbreak = \text{TVar}(s_1) \cup \text{TVar}(s_1')$.
Similarly, $\text{TVar}(s_2;s_2') = \text{TVar}(s_2) \cup \text{TVar}(s_2')$.
In summary, $\text{TVar}(s_1;s_1') = \text{TVar}(s_2;s_2')$.

\item We show that $\text{Imp}(S_1', \text{TVar}(s_1;s_1'))\allowbreak \subseteq \text{TVar}(S_1)$ and

\noindent$\text{Imp}(S_2', \text{TVar}(s_2;s_2')) \subseteq \text{TVar}(S_2)$.

W.l.o.g, we show that $\text{Imp}(S_1', \text{TVar}(s_1;s_1')) \subseteq \text{TVar}(S_1)$.

Specifically, we show $\text{Imp}(S_1', \text{CVar}(s_1;s_1')) \subseteq \text{CVar}(S_1)$.
\begin{tabbing}
xx\=xx\=\kill
\>\>                   $\text{CVar}(s_1;s_1')$\\
\>=\>                  $\text{CVar}(s_1) \cup \text{Imp}(s_1, \text{CVar}(s_1'))$ (1)\\
\>\>                   by the defintion of $\text{CVar}(s_1;s_1')$ \\
\end{tabbing}

\begin{tabbing}
xx\=xx\=\kill
\>\>                   $\text{Imp}(S_1', \text{CVar}(s_1;s_1'))$\\
\>=\>                  $\text{Imp}(S_1', \text{CVar}(s_1) \cup \text{Imp}(s_1, \text{CVar}(s_1')))$ by (1)\\
\>=\>                  $\text{Imp}(S_1', \text{CVar}(s_1)) \cup \text{Imp}(S_1', \text{Imp}(s_1, \text{CVar}(s_1')))$ (2)\\
\>\>                   by Lemma~\ref{lmm:impVarUnionLemma}\\
\end{tabbing}

\begin{tabbing}
xx\=xx\=\kill
\>\>                   $\text{Imp}(S_1', \text{CVar}(s_1))$\\
\>$\subseteq$\>        $\text{CVar}(S_1';s_1)$  by the defintion of $\text{CVar}(\cdot)$\\
\>$\subseteq$\>        $\text{CVar}(S_1';s_1;s_1')$ by the defintion of $\text{CVar}(\cdot)$\\
\\
\>\>                  $\text{Imp}(S_1', \text{Imp}(s_1, \text{CVar}(s_1')))$ \\
\>=\>                 $\text{Imp}(S_1';s_1, \text{CVar}(s_1'))$ by Lemma~\ref{lmm:ImpPrefixLemma}\\
\>$\subseteq$\>       $\text{CVar}(S_1';s_1;s_1')$ by the defintion of $\text{CVar}(\cdot)$.\\
\end{tabbing}

\begin{tabbing}
xx\=xx\=\kill
\>\>                   $\text{Imp}(S_1', \text{CVar}(s_1)) \cup \text{Imp}(S_1', \text{Imp}(s_1, \text{CVar}(s_1')))$\\
\>$\subseteq$\>        $\text{CVar}(S_1';s_1;s_1')$ by (3) and (4).
\end{tabbing}
In conclusion, $\text{Imp}(S_1', \text{CVar}(s_1;s_1')) \subseteq \text{CVar}(S_1)$.
Similarly, $\text{Imp}(S_1', \text{LVar}(s_1;s_1')) \subseteq \text{LVar}(S_1)$.
Thus, $\text{Imp}(S_1', \text{TVar}(s_1;s_1')) \subseteq \text{TVar}(S_1)$.
Similarly,

\noindent$\text{Imp}(S_2', \text{TVar}(s_2;s_2')) \subseteq \text{TVar}(S_2)$.
\end{enumerate}
Then, by Theorem~\ref{thm:equivCompMain}, after terminating execution of $S_1'$ and $S_2'$, value stores $\vals_1'$ and $\vals_2'$ agree on values of the termination deciding variables of $s_1;s_1'$ and $s_2;s_2'$,

\noindent$\forall x \in \text{TVar}(s_1;s_1') \cup \text{TVar}(s_2;s_2')\,:\, \vals_1'(x) = \vals_2'(x)$.

We show that the execution of $s_2$ terminates by the hypothesis IH.
By Corollary~\ref{coro:termSeq},

\noindent$(S_1';s_1;s_1', m_1(\vals_1)) ->* (s_1;s_1', m_1'(\vals_1'))$ and
$(S_2';s_2;s_2', \break m_2(\vals_2)) ->* (s_2;s_2', m_2'(\vals_2'))$.
By assumption, $S_1$ terminates, then $s_1$ terminates,
$(s_1, m_1'(\vals_1')) ->* (\text{skip}, m_1''(\vals_1''))$.
Because $s_1' \equiv_{H}^S s_2$, to apply the induction hypothesis, we need to show that all required conditions hold.
\begin{itemize}
\item $\text{size}(s_2) + \text{size}(s_1') < k$.

By definition, $\text{size}(S_2') > 1,  \text{size}(S_1') > 1, \text{size}(s_1) > 1,  \text{size}(s_2') > 1$.

\item
Value stores $\vals_1''$ and $\vals_2'$ agree on values of the termination deciding variables of $s_1'$ and $s_2$.
$\forall x \in \text{TVar}(s_1') \cup \text{TVar}(s_2)\,:\, \vals_1''(x) = \vals_2'(x)$.

By Corollary~\ref{coro:sameTermVarFromEquivTerm}, $\text{TVar}(s_1') = \text{TVar}(s_2)$.
Because of the condition $\text{Def}(s_1) \cap \text{TVar}(s_1') = \emptyset$,
by Corollary~\ref{coro:defExclusion}, value stores $\vals_1''$ and $\vals_1'$ agree on values of the termination deciding variables of $s_1'$, $\forall x \in \text{TVar}(s_1')\,:\, \vals_1''(x) = \vals_1'(x)$.
By the argument above, $\forall x \in \text{TVar}(s_2)\,:\, \vals_1'(x) = \vals_2'(x)$.
Thus, the condition holds.
\end{itemize}
By the induction hypothesis IH, $(s_1', m_1''(\vals_1'')) \equiv_{H} (s_2, m_2'(\vals_2'))$.
Because the execution of $s_1'$ terminates, then the exeuction of $s_2$ terminates when started in the state $m_2'(\vals_2')$,
$(s_2, m_2'(\vals_2')) ->* (\text{skip}, m_2''(\vals_2''))$.

We show that the execution of $s_2'$ terminates.
This is by the similar argument that $s_2$ terminates.

In conclusion, $S_2$ terminates when started in the state $m_2(\vals_2)$. The theorem holds.

In addition, we show that it is impossible that $s_1$ and $s_1'$ both include input statements  by contradiction.
Suppose there are input statements in both $s_1$ and $s_1'$.
By Lemma~\ref{lmm:inputVarInStmtSeqWithInputStmt}, ${id}_I \in \text{Def}(s_1) \cap \text{TVar}(s_1')$. A contradiction against the condition that $\text{Def}(s_1) \cap \text{TVar}(s_1') = \emptyset$.
Similarly, there are no input statements in both $s_2$ and $s_2'$.
\end{enumerate} 
\end{enumerate}
\end{proof}

\subsubsection{Supporting lemmas for the soundness proof of termination in the same way}

The supporting lemmas include various properties of {TVar}$(S)$, two statement sequences satisfying the proof rule of termination in the same way consume the same number of input values when both terminate,  and the proof for the case of while statement of theorem~\ref{thm:mainTermSameWayLocal}.

\paragraph{The properties of the termination deciding variables}
\begin{lemma}\label{lmm:prefCVarSubsetOfCVar}
The crash variables of $S_1;S_1'$ is same as the union of the crash variables of $S_1$ and the imported variables in $S_1$  relative to the crash variables of $S_1'$,
$\text{CVar}(S_1;S_1')$ = $\text{CVar}(S_1) \cup \text{Imp}(S_1, \text{CVar }(S_1'))$.
\end{lemma}
\begin{proof}
Let $S_1' = s_1;...;s_k$ for $k>0$.
We show the lemma by induction on $k$.

\noindent{\bf Base case}.

By the definition of $\text{CVar}(S)$, the lemma holds.

\noindent{\bf Induction step}.

The hypothesis IH is that $\text{CVar}(S_1;s_1;...;s_k)$ = $\text{CVar}(S_1) \cup \text{Imp}(S_1, \text{CVar}(s_1;...;s_k))$ for $k\geq1$.

Then we show that the lemma holds when $S_1' = s_1;...;s_{k+1}$.

\begin{tabbing}
xx\=xx\=\kill
\>\>     $\text{CVar}(S_1;s_1;...;s_{k+1})$\\
\>=\>    $\text{CVar}(S_1;s_1) \cup \text{Imp}(S_1;s_1, \text{CVar}(s_2;...;s_{k+1}))$ by IH\\
\\
\>\>     $\text{CVar}(S_1;s_1) $\\
\>=\>    $\text{CVar}(S_1) \cup \text{Imp}(S_1, \text{CVar}(s_1))$ (1) \\
\\
\>\>     $\text{Imp}(S_1;s_1, \text{CVar }(s_2;...;s_{k+1}))$\\
\>=\>    $\text{Imp}(S_1, \text{Imp}(s_1; \text{CVar }(s_2;...;s_{k+1})))$ (2)
\end{tabbing}

Combining (1) and (2), we have
\begin{tabbing}
xx\=xx\=\kill
\>\>     $\text{CVar}(S_1;s_1) \cup \text{Imp}(S_1;s_1, \text{CVar }(s_2;...;s_{k+1}))$ \\
\>=\>    $\text{CVar}(S_1) \cup \text{Imp }(S_1, \text{CVar }(s_1))$ \\
\>\>        $\cup \text{ Imp }(S_1, \text{Imp}(s_1; \text{CVar }(s_2;...;s_{k+1})))$ \\

\>=\>    $\text{CVar}(S_1) \cup \text{Imp}(S_1, \text{CVar}(s_1) \cup \text{Imp}(s_1; \text{CVar}(s_2;...;s_{k+1})))$\\
\>\>     by Lemma~\ref{lmm:impVarUnionLemma}           \\
\>=\>    $\text{CVar}(S_1) \cup \text{Imp}(S_1, \text{CVar}(s_1;...;s_{k+1}))$.
\end{tabbing}
\end{proof}

\begin{lemma}\label{lmm:prefLVarSubsetOfLVar}
The loop deciding variables of $S_1;S_1'$  is same as the union of the loop deciding variables of $S_1$ and the imported variables in $S_1$  relative to the loop deciding variables of $S_1'$,
$\text{LVar}(S_1;S_1')$ = $\text{LVar}(S_1) \cup \text{Imp}(S_1, \text{LVar }(S_1'))$.
\end{lemma}

By proof of Lemma~\ref{lmm:prefLVarSubsetOfLVar} similar to that of lemma~\ref{lmm:prefCVarSubsetOfCVar} above.

\begin{lemma}\label{lmm:equivLVarFromEquivTerm}
If two statement sequences $S_1$ and $S_2$ satisfy the proof rule of termination in the same way, then $S_1$ and $S_2$ have same loop variables, $(S_1 \equiv_{H}^S S_2) => (\text{LVar}(S_1) = \text{LVar}(S_2))$.
\end{lemma}
\begin{proof}
By induction on $\text{size}(S_1) + \text{size}(S_2)$, the sum of the program size of $S_1$ and $S_2$.

\noindent{\bf Base case}.

$S_1$ and $S_2$ are simple statement. There are three base cases according to the definition of  $s_1 \equiv_{H}^S s_2$.

\begin{enumerate}
\item two same simple statements, $S_1 = S_2$;

\item $S_1$ and $S_2$ are input statement with same type variable: $S_1=``\text{input }{id}_1", S_2 = ``\text{input }{id}_2"$
  where $(\Gamma_{S_1} \vdash {id}_1 : t) \wedge (\Gamma_{S_2} \vdash {id}_2 : t)$;.

\item
$S_1 = ``\text{output} \, e" \, \text{or} \, ``{id}_1 := e",
 S_2 = ``\text{output} \, e" \, \text{or} \, ``{id}_2 := e"$ where both of the following hold:

       \begin{itemize}
       \item There is no possible value mismatch in $``{id}_1 := e"$,
       $\neg(\Gamma_{S_1} \vdash {id}_1 : \text{Int}) \vee
        \neg(\Gamma_{S_1} \vdash e : \text{Long}) \vee
            (\Gamma_{S_1} \vdash e : \text{Int})$.

       \item There is no possible value mismatch in $``{id}_2 := e"$,
       $\neg(\Gamma_{S_2} \vdash {id}_2 : \text{Int}) \vee
        \neg(\Gamma_{S_2} \vdash e : \text{Long}) \vee
            (\Gamma_{S_2} \vdash e : \text{Int})$.
       \end{itemize}
\end{enumerate}
By the definition of loop variables, $\text{LVar}(S_1) = \text{LVar}(S_2) = \emptyset$ in both base cases.
Therefore, Lemma~\ref{lmm:equivLVarFromEquivTerm} holds.

\noindent{\bf Induction Step}.

The hypothesis IH is that Lemma~\ref{lmm:equivLVarFromEquivTerm} holds when $\text{size}(S_1) + \text{ size}(S_2) = k \geq 2$.

We show that Lemma~\ref{lmm:equivLVarFromEquivTerm} holds when $\text{size}(S_1) + \text{ size}(S_2) = k + 1$.

The proof is a case analysis according to the cases in the definition of $(S_1 \equiv_{H}^S S_2)$:
\begin{enumerate}
\item $S_1$ and $S_2$ are one statement and one of the following holds.

\begin{enumerate}
\item $S_1$ = ``If($e$) then \{$S_1^t$\} else \{$S_1^f$\}",
      $S_2$ = ``If($e$) then \{$S_2^t$\} else \{$S_2^f$\}" such that one of the following holds:

\begin{enumerate}
\item $S_1^t, S_1^f, S_2^t, S_2^f$ are all sequences of ``skip";

By the definition of loop variables, $\text{LVar}(S_1^t) = \text{LVar}(S_1^f) = \text{LVar}(S_2^t) = \text{LVar}(S_2^f) = \emptyset$.
Therefore, by the definition of loop variables, $\text{LVar}(S_1) = \text{LVar}(S_2) = \emptyset$. The lemma~\ref{lmm:equivLVarFromEquivTerm} holds.

\item At least one of $S_1^t, S_1^f, S_2^t, S_2^f$ is not a sequence of ``skip" such that:
      $(S_1^t \equiv_{H}^S S_2^t) \wedge (S_1^f \equiv_{H}^S S_2^f)$;

$\text{size}(S_1^t) + \text{size}(S_2^t) < k, \text{size}(S_1^f) + \text{size}(S_2^f) < k$.

Then, by the hypothesis IH1, $\text{LVar}(S_1^t) = \text{LVar}(S_2^t), \text{LVar}(S_1^f) = \text{LVar}(S_2^f)$.
Consequently, $(\text{LVar}(S_1^t) \cup \text{LVar}(S_1^f)) = (\text{LVar}(S_2^t) \cup \text{LVar}(S_2^f)) = \text{LVar}(\Delta)$.
When $\text{LVar}(\Delta) = \emptyset$, then $\text{LVar}(S_1) = \text{LVar}(S_2) = \emptyset$ by the definition of loop variables.
When $\text{LVar}(\Delta) \neq \emptyset$, then $\text{LVar}(S_1) = \text{LVar}(S_2) = \text{LVar}(\Delta) \cup \text{Use}(e)$ by the definition of loop variables.      The lemma~\ref{lmm:equivLVarFromEquivTerm} holds.
\end{enumerate}

\item $S_1$ = ``$\text{while}_{\langle n_1\rangle} (e) \{S_1''\}$",
      $S_2$ = ``$\text{while}_{\langle n_2\rangle} (e) \{S_2''\}$" such that both of the following hold:

\begin{itemize}
\item $S_1'' \equiv_{H}^S S_2''$;

\item $S_1''$ and $S_2''$ have equivalent computation of $\text{TVar}(S_1) \cup \text{TVar}(S_2)$;
\end{itemize}

By the hypothesis IH1, $\text{LVar}(S_1'') = \text{LVar}(S_2'')$. Then $\text{Use}(e) \cup \text{LVar}(S_1'')  = \text{Use}(e) \cup \text{LVar}(S_2'')$.
Then, we show that:

$\forall i \geq0, \text{Imp}({S_1''}^i, \text{Use}(e) \cup \text{LVar}(S_1'')) = \text{Imp}({S_2''}^i, \text{Use}(e) \cup \text{LVar}(S_2''))$ by induction on $i$.

\noindent{\bf Base case}

By our notation $S^0$, ${S_1''}^0 = \text{skip}, {S_2''}^0 = \text{skip}$.
Then, by the definition of imported variables,

$\text{Imp}({S_1''}^0, \text{Use}(e) \cup \text{LVar}(S_1'')) = \text{Use}(e) \cup \text{LVar}(S_1'')$,

$\text{Imp}({S_2''}^0, \text{Use}(e) \cup \text{LVar}(S_2'')) = \text{Use}(e) \cup \text{LVar}(S_2'')$.

Then, $\text{Imp}({S_1''}^0, \text{Use}(e) \cup \text{LVar}(S_1'')) = \text{Imp}({S_2''}^0, \text{Use}(e) \cup \text{LVar}(S_2''))$.

\noindent{\bf Induction step}

The hypothesis IH3 is that, $\forall i \geq0, \text{Imp}({S_1''}^i, \text{Use}(e) \cup \text{LVar}(S_1'')) = \text{Imp}({S_2''}^i, \text{Use}(e) \cup \text{LVar}(S_2''))$.

Then we show that $\text{Imp}({S_1''}^{i+1}, \text{Use}(e) \cup \text{LVar}(S_1'')) = \text{Imp}({S_2''}^{i+1}, \text{Use}(e) \cup \text{LVar}(S_2''))$.

By Corollary~\ref{lmm:dupStmtPrefixLemma},

$\text{Imp}({S_1''}^{i+1}, \text{Use}(e) \cup \text{LVar}(S_1'')) = \text{Imp}(S_1'', \text{Imp}({S_1''}^{i}, \text{Use}(e) \cup \text{LVar}(S_1'')))$,

$\text{Imp}({S_2''}^{i+1}, \text{Use}(e) \cup \text{LVar}(S_2'')) = \allowbreak\text{Imp}(S_2'', \text{Imp}({S_2''}^{i}, \text{Use}(e) \cup \text{LVar}(S_2'')))$.

By the hypothesis IH3, $\text{Imp}({S_1''}^{i}, \text{Use}(e) \cup \text{LVar}(S_1'')) = \text{Imp}({S_2''}^{i}, \text{Use}(e) \cup \text{LVar}(S_2'')) = \text{LVar}(\Delta)$.

Besides, by the definition of loop variables, $\text{LVar}(\Delta) \subseteq \text{LVar}(S_1), \text{LVar}(\Delta) \subseteq \text{LVar}(S_2)$.

Then,
\begin{tabbing}
xx\=xx\=\kill
\>\>     $\text{Imp}(S_1'', \text{Imp}({S_1''}^{i}, \text{Use}(e) \cup \text{LVar}(S_1'')))$ \\
\>= \>   $\text{Imp}(S_1'', \text{LVar}(\Delta))$ \\
\>=\>    $\text{Imp}(S_2'', \text{LVar}(\Delta))$  by Lemma~\ref{lmm:sameImpfromEquivCompCond}\\
\>=\>    $\text{Imp}({S_2''}, \text{Imp}({S_2''}^{i}, \text{Use}(e) \cup \text{LVar}(S_2'')))$
\end{tabbing}

In summary, $\text{LVar}(S_1)) = \text{LVar}(S_2)$. The lemma~\ref{lmm:equivLVarFromEquivTerm} holds.
\end{enumerate}

\item $S_1$ and $S_2$ are not both one statement and one of the following holds:

\begin{enumerate}
\item $S_1=S_1';s_1$ and $S_2 = S_2';s_2$ such that all of the following hold:

\begin{itemize}
\item $S_1' \equiv_{H}^S S_2'$;

\item $S_1'$ and $S_2'$ have equivalent computation of $\text{TVar}(s_1) \cup \text{ TVar}(s_2)$;

\item
$s_1 \equiv_{H}^S s_2$ where $s_1$ and $s_2$ are not ``skip";
\end{itemize}

By the hypothesis IH1, $\text{LVar}(S_1') = \text{LVar}(S_2')$, $\text{LVar}(s_1) = \text{LVar}(s_2) = \text{LVar}(\Delta)$.
Besides,

\noindent$\text{Imp}(S_1', \text{LVar}(s_1)) = \text{LVar}(S_2', \text{LVar}(s_2))$ by Lemma~\ref{lmm:sameImpfromEquivCompCond}.
Therefore, $\text{LVar}(S_1) = \text{LVar}(S_2)$ by the definition of loop variables. The lemma~\ref{lmm:equivLVarFromEquivTerm} holds.

\item One last statement is ``skip":
w.l.o.g.,

      \noindent$\big((S_1 \equiv_{H}^S S_2') \, \wedge \,
      (s_2 = ``\text{skip}")\big)$.

By the hypothesis IH1, $\text{LVar}(S_1) = \text{LVar}(S_2')$ .
Besides, $\text{LVar}(s_2) = \emptyset$ by the definition of loop variables.
Therefore, $\text{LVar}(S_1) = \text{LVar}(S_2) = \text{LVar}(S_2') \cup \text{Imp}(S_2', \emptyset)$.
The lemma~\ref{lmm:equivLVarFromEquivTerm} holds.

\item  One last statement is a ``duplicate" statement such that
 one of the following holds:

W.l.o.g.
$S_1 = S_1';s_1';S_1'';s_1$ and all of the following hold:
\begin{itemize}
\item $S_1';s_1';S_1'' \equiv_{H}^S S_2$;

\item $s_1' \equiv_{H}^S s_1$;

\item $\text{Def}(s_1';S_1'') \cap \text{TVar}(s_1) = \emptyset$;

\item $s_2 \neq ``\text{skip}"$;
\end{itemize}

By the hypothesis IH, $\text{LVar}(S_1';s_1';S_1'') = \text{LVar}(S_2)$.

Then, we show that $\text{LVar}(S_1';s_1';S_1'') = \text{LVar}(S_1';s_1';S_1'';s_1)$.
By the induction hypothesis IH, $\text{LVar}(s_1') = \text{LVar}(s_1)$.

\begin{tabbing}
xx\=xx\=\kill
\>\>               $\text{LVar}(S_1';s_1';S_1'';s_1)$ \\
\>= \>             $\text{LVar}(S_1';s_1';S_1'') \cup \text{Imp}(S_1';s_1';S_1'', \text{LVar}(s_1))$\\
\>\>               by the definition of loop variables\\
\\
\>\>              $\text{Imp}(S_1';s_1';S_1'', \text{LVar}(s_1))$  \\
\>=\>             $\text{Imp}(S_1', \text{Imp}(s_1';S_1'', \text{LVar}(s_1)))$ by Lemma~\ref{lmm:ImpPrefixLemma}\\
\>=\>             $\text{Imp}(S_1', \text{LVar}(s_1))$ by $\text{Def}(s_1';S_1'') \cap \text{TVar}(s_1) = \emptyset$\\
\>=\>             $\text{Imp}(S_1', \text{LVar}(s_1'))$ by $\text{LVar}(s_1') = \text{LVar}(s_1)$ \\
\>$\subseteq$\>   $\text{LVar}(S_1';s_1')$ by the definition of loop variables \\
\>$\subseteq$\>   $\text{LVar}(S_1';s_1';S_1'')$ by Lemma~\ref{lmm:prefLVarSubsetOfLVar}.
\end{tabbing}
In conclusion, $\text{LVar}(S_1';s_1';S_1'';s_1)  = \text{LVar}(S_1';s_1';S_1'')$. The lemma holds.

\item $S_1 = S_1';s_1;s_1';$ and $S_2 = S_2';s_2;s_2'$ where $s_1$ and $s_2$ are reordered and all of the following hold:

\begin{itemize}
\item ${S_1'} \equiv_{H}^S {S_2'}$;

\item $S_1'$ and $S_2'$ have equivalent computation of $\text{TVar}(s_1;s_1') \cup \text{TVar}(s_2;s_2')$.

\item ${s_1} \equiv_{H}^S {s_2'}$;

\item ${s_1'} \equiv_{H}^S {s_2}$;

\item $\text{Def}(s_1) \cap \text{TVar}(s_1') = \emptyset$;

\item $\text{Def}(s_2) \cap \text{TVar}(s_2') = \emptyset$;
\end{itemize}

By the hypothesis IH,
$\text{LVar}(S_1') = \text{LVar}(S_2'),
 \noindent\text{LVar}(s_1)  = \text{LVar}(s_2'),
 \text{LVar}(s_1') = \text{LVar}(s_2)$.

In the following, we show $\text{LVar}(S_1) = \text{LVar}(S_2)$ in three steps.
\begin{enumerate}
\item
We show $\text{LVar}(S_1';s_1;s_1') = \text{LVar}(S_1')$

\noindent$\cup \text{Imp}(S_1', \text{LVar}(s_1)) \cup
\text{Imp}(S_1', \text{LVar}(s_1'))$.

\begin{tabbing}
xx\=xx\=\kill
\>\>               $\text{LVar}(S_1';s_1;s_1')$ \\
\>= \>             $\text{LVar}(S_1';s_1) \cup \text{Imp}(S_1';s_1, \text{LVar}(s_1'))$\\
\>\>               by the definition of loop variables\\
\\
\>\>              $\text{LVar}(S_1';s_1)$  \\
\>=\>             $\text{LVar}(S_1') \cup \text{Imp}(S_1', \text{LVar}(s_1))$ (2)\\
\>\>              by the definition of loop variables\\
\\
\>\>              $\text{Imp}(S_1';s_1, \text{LVar}(s_1'))$ \\
\>=\>             $\text{Imp}(S_1',  \text{Imp}(s_1, \text{LVar}(s_1')))$  by Lemma~\ref{lmm:ImpPrefixLemma}\\
\>=\>             $\text{Imp}(S_1',  \text{LVar}(s_1'))$  (3)\\
\>\>              by the condition $\text{Def}(s_1) \cap \text{TVar}(s_1') = \emptyset$
\end{tabbing}

According to (2) and (3),
$\text{LVar}(S_1';s_1;s_1') =  \text{LVar}(S_1') \cup \text{Imp}(S_1', \text{LVar}(s_1)) \cup \text{Imp}(S_1',  \text{LVar}(s_1'))$.

Similarly,
$\text{LVar}(S_2';s_2;s_2') =  \text{LVar}(S_2') \cup$

\noindent$\text{Imp}(S_2', \text{LVar}(s_2)) \cup \text{Imp}(S_2', \text{LVar}(s_2'))$.

\item We show that
$\text{Imp}(S_1', \text{LVar}(s_1)) = \text{Imp}(S_2',  \text{LVar}(s_2'))$.

\noindent$\text{Imp}(S_2', \text{LVar}(s_2'))$.

We need to show that $\text{LVar}(s_1) \subseteq \text{LVar}(s_1;s_1')$ and
$\text{LVar}(s_2') \subseteq \text{LVar}(s_2;s_2')$.
By the definition of loop variables, $\text{LVar}(s_1) \subseteq \text{LVar}(s_1;s_1')$.
By the definition of loop variables again, $\text{LVar}(s_2;s_2') = \text{LVar}(s_2) \cup \text{Imp}(s_2, \text{LVar}(s_2'))$.
Because $\text{Def}(s_2) \cap \text{TVar}(s_2') = \emptyset$,
$\text{Imp}(s_2, \text{LVar}(s_2')) = \text{LVar}(s_2')$.

By the induction hypothesis IH, $\text{LVar}(s_1) = \text{LVar}(s_2')$.
By Lemma~\ref{lmm:sameImpfromEquivCompCond}, $\forall x \in \text{LVar}(s_1) = \text{LVar}(s_2'), \text{Imp}(S_1', \{x\}) = \text{Imp}(S_2', \{x\})$.
By Lemma~\ref{lmm:impVarUnionLemma},
$\text{Imp}(S_1', \text{LVar}(s_1)) = \text{Imp}(S_2',  \text{LVar}(s_2'))$.

\item
We show that
$\text{Imp}(S_1', \text{LVar}(s_1')) = \text{Imp}(S_2',  \text{LVar}(s_2))$.

By the similar argument that $\text{Imp}(S_1', \text{LVar}(s_1)) = \text{Imp}(S_2',  \text{LVar}(s_2'))$.
\end{enumerate}
In conclusion, $\text{LVar}(S_1';s_1;s_1') = \text{LVar}(S_2';s_2;s_2')$.
\end{enumerate}
\end{enumerate}
\end{proof}

\begin{lemma}\label{lmm:equivCVarFromEquivTerm}
If two statement sequences $S_1$ and $S_2$ satisfy the proof rule of termination in the same way, then $S_1$ and $S_2$ have same crash variables, $(S_1 \equiv_{H}^S S_2) => (\text{CVar}(S_1) = \text{ CVar}(S_2))$.
\end{lemma}
By proof similar to those for Lemma~\ref{lmm:equivLVarFromEquivTerm}.


\begin{corollary}\label{coro:sameTermVarFromEquivTerm}
If two statement sequences $S_1$ and $S_2$ satisfy the proof rule of termination in the same way, then $S_1$ and $S_2$ have same termination deciding variables, $(S_1 \equiv_{H}^S S_2) => (\text{TVar}(S_1) = \text{TVar}(S_2))$.
\end{corollary}
By Lemma~\ref{lmm:equivLVarFromEquivTerm}, and~\ref{lmm:equivCVarFromEquivTerm}.

\paragraph{Properties of the input sequence variable}
\begin{lemma}\label{lmm:inputVarNOtInDefStmtSeqWithoutInputStmt}
If there is no input statement in a statement sequence $S$, then the input sequence variable is not in the defined variables of $S$,
$(\nexists ``\text{input }x" \in S) => {id}_{I} \notin \text{Def}(S)$.
\end{lemma}
\begin{proof}
By induction on abstract syntax of $S$.
\end{proof}

\begin{lemma}\label{lmm:inputVarNotInCVarWithoutInputStmt}
If there is no input statement in a statement sequence $S$, then the input sequence variable is not in the crash variables of $S$,
$(\nexists ``\text{input }x" \in S) => ({id}_{I} \notin \text{CVar}(S))$.
\end{lemma}
\begin{proof}
By induction on abstract syntax of $S$.
\end{proof}

\begin{lemma}\label{lmm:inputVarNotInLVarWithoutInputStmt}
If there is no input statement in a statement sequence $S$, then the input sequence variable is in the loop variables of $S$,
$(\nexists ``\text{input }x" \in S) => ({id}_{I} \notin \text{LVar}(S))$.
\end{lemma}
\begin{proof}
By induction on abstract syntax of $S$.
\end{proof}

\begin{corollary}\label{coro:inputVarNotInTVarWithoutInputStmt}
If there is no input statement in a statement sequence $S$, then the input sequence variable is in the termination deciding variables of $S$,
$(\nexists ``\text{input }x" \in S) => ({id}_{I} \notin \text{TVar}(S))$.
\end{corollary}
By Lemma~\ref{lmm:inputVarNotInCVarWithoutInputStmt} and~\ref{lmm:inputVarNotInLVarWithoutInputStmt}.

\begin{lemma}\label{lmm:inputVarInStmtSeqWithInputStmt}
If there is one input statement in a statement sequence $S$, then the input sequence variable is in the crash variables and defined variables of $S$,
$(\exists ``\text{input }x" \in S) => ({id}_{I} \in \text{CVar}(S)) \wedge ({id}_{I} \in \text{Def}(S))$.
\end{lemma}
\begin{proof}
By induction on abstract syntax of $S$.
\end{proof}

\begin{lemma}\label{lmm:impInputVarASubsetOfCVar}
If there is one input statement in a  statement sequence $S$, then the imported variables in $S$ relative to the input sequence variable are a subset of the crash variables of $S$,
$(\exists ``\text{input }x" \in S) => (\text{Imp}(S, \{{id}_{I}\}) \subseteq \text{CVar}(S))$.
\end{lemma}
\begin{proof}
By induction on abstract syntax of $S$.
\begin{enumerate}
\item $S = ``\text{input }x"$.

By the definition of $\text{CVar}(\cdot)$ and $\text{Imp}(\cdot)$, $\text{CVar}(S) = \text{Imp}(S, \{{id}_{I}\}) = \{{id}_{I}\}$.

\item $S = ``\text{If}(e) \text{ then } \{S_t\} \text{ else }\{S_f\}"$.

W.l.o.g., there is input statement in $S_t$, by the induction hypothesis,
$\text{Imp}(S_t, \{{id}_{I}\}) \subseteq \text{CVar}(S_t)$.
There are two subcases regarding if input statement is in $S_f$.
\begin{enumerate}
\item There is input statement in $S_f$.

By the induction hypothesis, $\text{Imp}(S_f, \{{id}_{I}\}) \subseteq \text{CVar}(S_f)$.
Hence, the lemma holds.

\item There is no input statement in $S_f$.

By the definition of imported variables, $\text{Imp}(S_f, \{{id}_{I}\}) = \{{id}_{I}\}$.
By Lemma~\ref{lmm:inputVarInStmtSeqWithInputStmt}, ${id}_{I} \in \text{CVar}(S_t)$.
Therefore, the lemma holds.
\end{enumerate}

\item $S = ``\text{while}_{\langle n\rangle}(e) \{S'\}"$.

By the induction hypothesis, $\text{Imp}(S', \{{id}_{I}\}) \subseteq \text{CVar}(S')$.

By the definition of $\text{Imp}(\cdot)$,
$\text{Imp}(S', \{{id}_{I}\}) = \bigcup_{i\geq 0} \text{Imp}({S'}^i, \{{id}_{I}\} \cup \text{Use}(e))$.
By the definition of $\text{CVar}(\cdot)$,
$\text{CVar}(S') =$

\noindent$\bigcup_{i\geq 0} \text{Imp}({S'}^i, \text{CVar}(S') \cup \text{Use}(e))$.

By induction on $i$, we show that, $\forall i\geq0$,
$\text{Imp}({S'}^i, \{{id}_{I}\} \cup \text{Use}(e)) \subseteq \text{Imp}({S'}^i, \text{CVar}(S') \cup \text{Use}(e))$.

\noindent{Base case} $i=0$.

By notation ${S'}^0 = \text{skip}$.

\begin{tabbing}
xx\=xx\=\kill
\>\>            $\text{Imp}({S'}^0, \{{id}_{I}\} \cup \text{Use}(e))$\\
\>=\>           $\{{id}_{I}\} \cup \text{Use}(e)$ by the definition of imported variables\\
\\
\>\>            $\text{Imp}({S'}^0, \text{CVar}(S') \cup \text{Use}(e))$\\
\>=\>           $\text{CVar}(S') \cup \text{Use}(e)$ by the definition of imported variables\\
\\
\>\>            ${id}_{I} \subseteq \text{CVar}(S')$ (1) by Lemma~\ref{lmm:inputVarInStmtSeqWithInputStmt}\\
\\
\>\>            $\text{Imp}({S'}^0, \{{id}_{I}\} \cup \text{Use}(e))$ \\
\>$\subseteq$\> $\text{Imp}({S'}^0, \text{CVar}(S') \cup \text{Use}(e))$  by Lemma~\ref{lmm:impVarUnionLemma}.
\end{tabbing}

\noindent{Induction step}.

The hypothesis IH1 is that
$\text{Imp}({S'}^i, \{{id}_{I}\} \cup \text{Use}(e)) \subseteq \text{Imp}({S'}^i, \text{CVar}(S') \cup \text{Use}(e))$ for $i>0$.

Then we show that
$\text{Imp}({S'}^{i+1}, \{{id}_{I}\} \cup \text{Use}(e))$

\noindent$\subseteq \text{Imp}({S'}^{i+1}, \text{CVar}(S') \cup \text{Use}(e))$

\begin{tabbing}
xx\=xx\=\kill
\>\>            $\text{Imp}({S'}^i, \{{id}_{I}\} \cup \text{Use}(e))$ \\
\>$\subseteq$\> $\text{Imp}({S'}^i, \text{CVar}(S') \cup \text{Use}(e))$ (1) by the hypothesis IH1\\
\\
\>\>            $\text{Imp}({S'}^{i+1}, \{{id}_{I}\} \cup \text{Use}(e))$\\
\>=\>           $\text{Imp}(S', \text{Imp}({S'}^i, \{{id}_{I}\} \cup \text{Use}(e)))$ (2) by Corollary~\ref{lmm:dupStmtPrefixLemma}\\
\\
\>\>            $\text{Imp}({S'}^{i+1}, \text{CVar}(S') \cup \text{Use}(e))$\\
\>=\>           $\text{Imp}(S', \text{Imp}({S'}^i, \text{CVar}(S') \cup \text{Use}(e)))$ (3) by Corollary~\ref{lmm:dupStmtPrefixLemma}.\\
\end{tabbing}

Combining (1), (2) and (3):
\begin{tabbing}
xx\=xx\=\kill
\>\>            $\text{Imp}(S', \text{Imp}({S'}^i, \{{id}_{I}\} \cup \text{Use}(e)))$\\
\>$\subseteq$\> $\text{Imp}(S', \text{Imp}({S'}^i, \text{CVar}(S') \cup \text{Use}(e)))$ by Lemma~\ref{lmm:impVarUnionLemma}.
\end{tabbing}
Therefore,
$\text{Imp}({S'}^{i+1}, \{{id}_{I}\} \cup \text{Use}(e)) \subseteq \text{Imp}({S'}^{i+1}, \text{CVar}(S') \cup \text{Use}(e))$.

In conclusion, $\text{Imp}(S, \{{id}_{I}\}) \subseteq \text{CVar}(S)$.

\item $S = s_1;...;s_k$, for $k>0$.

By induction on $k$.

\noindent{Base case}. $k = 1$.

By above cases, the lemma holds.

\noindent{Induction step}.

The induction hypothesis IH2 is that the lemma holds when $k > 0$.
We show that the lemma holds when $S = s_1;...;s_{k+1}$.

By the definition of crash variables, $\text{CVar}(s_1;...;s_{k+1}) = \text{CVar}(s_1;...;s_k) \cup \text{Imp}(s_1;...;s_k, \text{CVar}(s_{k+1}))$.There are two possibilities.
\begin{enumerate}
\item $\nexists ``\text{input }x" \in s_{k+1}$.

By Lemma~\ref{lmm:inputVarNOtInDefStmtSeqWithoutInputStmt}, ${id}_{I} \notin \text{Def}(S)$.
\begin{tabbing}
xx\=xx\=\kill
\>\>            $\text{Imp}(s_1;...;s_{k+1}, \{{id}_{I}\})$\\
\>=\>           $\text{Imp}(s_1;...;s_{k}, \{{id}_{I}\})$ by ${id}_{I} \notin \text{Def}(S)$ and\\
\>\>            the definition of imported variables\\
\>$\subseteq$\> $\text{CVar}(s_1;...;s_{k})$ by the hypothesis IH2\\
\>$\subseteq$\> $\text{CVar}(s_1;...;s_{k+1})$ by the definition of crash variables\\
\end{tabbing}

\item $\exists ``\text{input }x" \in s_{k+1}$.

\begin{tabbing}
xx\=xx\=\kill
\>\>            $\text{Imp}(s_1;...;s_{k+1}, \{{id}_{I}\})$ \\
\>=\>           $\text{Imp}(s_1;...;s_{k}, \text{Imp}(s_{k+1}, \{{id}_{I}\}))$\\
\>\>            by the definition of imported variables\\
\\
\>\>            $\text{Imp}(s_{k+1}, \{{id}_{I}\})$\\
\>$\subseteq$\> $\text{CVar}(s_{k+1})$ by the hypothesis IH2\\
\\
\>\>            $\text{Imp}(s_1;...;s_{k}, \text{Imp}(s_{k+1}, \{{id}_{I}\}))$ \\
\>$\subseteq$\> $\text{Imp}(s_1;...;s_{k}, \text{CVar}(s_{k+1}))$ by Lemma~\ref{lmm:impVarUnionLemma} \\
\>$\subseteq$\> $\text{CVar}(s_1;...;s_{k+1})$ by the definition of crash variables.
\end{tabbing}
\end{enumerate}
\end{enumerate}
\end{proof}

\begin{lemma}\label{lmm:TermInSameWayImpliesInputSeqEquivComp}
If two programs $S_1$ and $S_2$ satisfy the proof rule of termination in the same way,
then $S_1$ and $S_2$ satisfy the proof rule of terminating computation in the same way of the input sequence,
$(S_1 \equiv_{H}^S S_2) => (S_1 \equiv_{{id}_{I}}^S S_2)$.
\end{lemma}
\begin{proof}
By induction on $\text{size}(S_1) + \text{size}(S_2)$.

\noindent{Base case}.
$S_1$ and $S_2$ are simple statements.

There are three cases.
\begin{enumerate}
\item $S_1$ and $S_2$ are ``skip":
      $S_1$ = $S_2$ = ``skip";

\item $S_1$ and $S_2$ are input statement: $S_1=``\text{input }{id}_1", S_2 = ``\text{input }{id}_2"$;

\item $s_1$ and $s_2$ are with the same expression:
      $s_1 = ``\text{output }e"$ or $``{id}_1 := e"$, $s_2 = ``\text{output }e"$ or $``{id}_2 := e"$.

By definition of the proof rule of equivalent computation, the lemma holds in above three cases.
\end{enumerate}

\noindent{Induction step}.

The hypothesis IH is that the lemma holds when $\text{size}(S_1) + \text{size}(S_2) = k \geq 2$.

Then, we show that the lemma holds when $\text{size}(S_1) + \text{size}(S_2) = k+1$.
The proof is a case analysis of the cases in the proof rule of termination in the same way.
\begin{enumerate}
\item $S_1$ and $S_2$ are one statement and one of the followings holds.

\begin{enumerate}
\item $S_1$ = ``If($e$) then \{$S_1^t$\} else \{$S_1^f$\}",
      $S_2$ = ``If($e$) then \{$S_2^t$\} else \{$S_2^f$\}" and one of the followings holds:

\begin{enumerate}
\item $S_1^t, S_1^f, S_2^t, S_2^f$ are all sequences of ``skip";

By Lemma~\ref{lmm:inputVarNOtInDefStmtSeqWithoutInputStmt}, ${id}_{I} \notin \text{Def}(S_1) \cap \text{Def}{S_2}$.
The lemma holds.

\item At least one of $S_1^t, S_1^f, S_2^t, S_2^f$ is not a sequence of ``skip" such that:

      $(S_1^t \equiv_{H}^S S_2^t) \land (S_1^f \equiv_{H}^S S_2^f)$;

Because $\text{size}(S_1) = 1 + \text{size}(S_1^t) + \text{size}(S_1^f)$, $\text{size}(S_2) = 1 + \text{size}(S_2^t) + \text{size}(S_2^f)$. Therefore, $\text{size}(S_1^t) + \text{size}(S_2^t) < k, \text{size}(S_1^f) + \text{size}(S_2^f) < k$.
By the induction hypothesis IH, $(S_1^t \equiv_{{id}_{I}}^S S_2^t) \land (S_1^f \equiv_{{id}_{I}}^S S_2^f)$.
Then, the lemma holds by the definition of $S_1 \equiv_{{id}_{I}}^S S_2$.

\end{enumerate}

\item $S_1$ = ``$\text{while}_{\langle n_1\rangle} (e) \{S_1''\}$",
      $S_2$ = ``$\text{while}_{\langle n_2\rangle} (e) \{S_2''\}$" and both of the followings hold:

\begin{itemize}
\item $S_1'' \equiv_{H}^S S_2''$;

\item $S_1''$ and $S_2''$ have equivalent computation of $\text{TVar}(S_1) \cup \text{TVar}(S_2)$;
\end{itemize}

By the induction hypothesis IH, $S_1'' \equiv_{{id}_{I}}^S S_2''$.
In addition, by Corollary~\ref{coro:sameTermVarFromEquivTerm}, $\text{TVar}(S_1) = \text{TVar}(S_2)$.
There are two cases.
\begin{enumerate}
\item ${id}_{I} \in \text{TVar}(S_1) = \text{TVar}(S_2)$.

We show that ${id}_{I} \in \text{Def}(S_1) \cap \text{Def}(S_2)$.
If there is no input statement in $S_1$ or $S_2$, then, by Corollary~\ref{coro:inputVarNotInTVarWithoutInputStmt},
${id}_{I} \notin \text{TVar}(S_1) \cap \text{TVar}(S_2)$. A contradiction.
Thus, there is input statement in $S_1$ and $S_2$, by Lemma~\ref{lmm:inputVarInStmtSeqWithInputStmt},
${id}_{I} \in \text{Def}(S_1) \cap \text{Def}(S_2)$.

By Lemma~\ref{lmm:impInputVarASubsetOfCVar}, $\text{Imp}(S_1, \{{id}_{I}\}) \subseteq \text{CVar}(S_1)$.
Similarly, $\text{Imp}(S_2, \{{id}_{I}\}) \subseteq \text{TVar}(S_2)$.
Hence, loop bodies of $S_1$ and $S_2$ equivalently compute every of the imported variables in $S_1$ and $S_2$ relative to the input sequence variable,
$\forall x \in \text{Imp}(S_1, \{{id}_{I}\}) \cup \text{Imp}(S_2, \{{id}_{I}\})$, $S_1'' \equiv_{x}^S S_2''$.
Thus, the lemma holds.

\item ${id}_{I} \notin \text{TVar}(S_1) = \text{TVar}(S_2)$.

Then there is no input statement in $S_1$ and $S_2$. Otherwise, by Lemma~\ref{lmm:inputVarInStmtSeqWithInputStmt},
$({id}_{I} \in \text{CVar}(S_1)) \vee ({id}_{I} \in \text{CVar}(S_2))$. A contradiction.
Then by Lemma~\ref{lmm:inputVarNOtInDefStmtSeqWithoutInputStmt}, ${id}_{I} \notin (\text{Def}(S_1) \cap \text{Def}(S_2))$.
Hence, the lemma holds.
\end{enumerate}

\end{enumerate}

\item $S_1$ and $S_2$ are not both one statement and one of the followings holds:

\begin{enumerate}
\item $S_1=S_1';s_1$ and $S_2 = S_2';s_2$ and all of the followings hold:

\begin{itemize}
\item $S_1' \equiv_{H}^S S_2'$;

\item $S_1'$ and $S_2'$ have equivalent computation of $\text{TVar}(s_1) \cup \text{TVar}(s_2)$;

\item
$s_1 \equiv_{H}^S s_2$ where $s_1$ and $s_2$ are not ``skip";
\end{itemize}

By the induction hypothesis IH, $s_1 \equiv_{{id}_{I}}^S s_2$.
By Corollary~\ref{coro:sameTermVarFromEquivTerm}, $\text{TVar}(s_1) = \text{TVar}(s_2)$.
There are two cases.
\begin{enumerate}
\item ${id}_I \in \text{TVar}(s_1) = \text{TVar}(s_2)$

Then there is input statement in $s_1$ and $s_2$. Otherwise, by Lemma~\ref{coro:inputVarNotInTVarWithoutInputStmt},
${id}_I \notin \text{TVar}(s_1) = \text{TVar}(s_2)$. A contradiction.
Then, by Lemma~\ref{lmm:inputVarInStmtSeqWithInputStmt}, ${id}_I \in \text{Def}(s_1) \cap \text{Def}(s_2)$.
By Lemma~\ref{lmm:sameImpfromEquivCompCond},
$\text{Imp}(s_1, \{{id}_I\}) = \text{Imp}(s_2, \{{id}_I\})$.

By Lemma~\ref{lmm:impInputVarASubsetOfCVar},
$\text{Imp}(s_1, \{{id}_I\}) \subseteq \text{CVar}(s_1, \{{id}_I\})$.
Therefore, $S_1'$ and $S_2'$ equivalently compute $\text{Imp}(s_1, \{{id}_I\}) \cup \text{Imp}(s_2, \{{id}_I\})$.
The lemma holds.

\item ${id}_I \notin \text{TVar}(s_1) = \text{TVar}(s_2)$.

Then, there is no input statement in $s_1$ and $s_2$.
Otherwise, by Lemma~\ref{lmm:inputVarInStmtSeqWithInputStmt}, ${id}_I \in \text{CVar}(s_1) = \text{CVar}(s_2)$.
A contradiction.
Then, by Lemma~\ref{lmm:inputVarNOtInDefStmtSeqWithoutInputStmt}, ${id}_I \notin \text{Def}(s_1) \cup \text{Def}(s_2)$.
By the induction hypothesis IH, $S_1' \equiv_{{id}_I}^S S_2'$. The lemma holds.
\end{enumerate}

\item One last statement is ``skip":
 W.l.o.g., $\big((S_1' \equiv_{H}^S S_2) \wedge (s_1 = ``\text{skip}")\big)$

By the induction hypothesis, $S_1' \equiv_{{id}_{I}}^S S_2$. By definition, ${id}_{I} \notin \text{Def}(s_1)$.
The lemma holds.

\item  One last statement is a ``duplicate" statement such that
 one of the followings holds:

W.l.o.g., $S_1 = S_1';s_1';S_1'';s_1$ and all of the followings hold:
\begin{itemize}
\item $S_1';s_1';S_1'' \equiv_{H}^S S_2$;

\item $s_1' \equiv_{H}^S s_1$;

\item $\text{Def}(s_1';S_1'') \cap \text{TVar}(s_1) = \emptyset$;

\item $s_2 \neq ``\text{skip}"$.
\end{itemize}

By the induction hypothesis, $S_1';s_1';S_1'' \equiv_{{id}_I}^S S_2$.
In the proof of Theorem~\ref{thm:mainTermSameWaySimpleStmt}, there is no input statement in $s_2$.
Because $\forall x\,:\, ``\text{input }x" \notin s_2$, by Lemma~\ref{lmm:inputVarNOtInDefStmtSeqWithoutInputStmt},
${id}_I \notin \text{Def}(s_1)$.
The lemma holds.

\item $S_1 = S_1';s_1;s_1';$ and $S_2 = S_2';s_2;s_2'$ where $s_1$ and $s_2$ are reordered and all of the followings hold:

\begin{itemize}
\item ${S_1'} \equiv_{H}^S {S_2'}$;

\item $S_1'$ and $S_2'$ have equivalent computation of $\text{TVar}(s_1;s_1') \cup \text{TVar}(s_2;s_2')$.

\item ${s_1} \equiv_{H}^S {s_2'}$;

\item ${s_1'} \equiv_{H}^S {s_2}$;

\item $\text{Def}(s_1) \cap \text{TVar}(s_1') = \emptyset$;

\item $\text{Def}(s_2) \cap \text{TVar}(s_2') = \emptyset$;
\end{itemize}

In the proof of Theorem~\ref{thm:mainTermSameWayLocal}, we showed that $s_1$ and $s_2$ do not both include input statement,
$s_2$ and $s_2'$ do not both include input statement.
There are two subcases.
\begin{enumerate}
\item There are no input statements in both $s_1$ and $s_1'$.

We show that there are no input statements in both $s_2$ and $s_2'$.
By Corollary~\ref{coro:sameTermVarFromEquivTerm}, $\text{TVar}(s_1) = \text{TVar}(s_2')$ and $\text{TVar}(s_1') = \text{TVar}(s_2)$.
By Corollary~\ref{coro:inputVarNotInTVarWithoutInputStmt}, ${id}_I \notin \text{TVar}(s_1) \cup \text{TVar}(s_1')$.
Thus, ${id}_I \notin \text{TVar}(s_2) \cup \text{TVar}(s_2')$.
If there is input statement in $s_2$ or $s_2'$, then, by Lemma~\ref{lmm:inputVarInStmtSeqWithInputStmt},
${id}_I \notin \text{TVar}(s_2) \cup \text{TVar}(s_2')$. A contradiction.
In summary, there are no input statements in both $s_2$ and $s_2'$.

By Lemma~\ref{lmm:inputVarNOtInDefStmtSeqWithoutInputStmt},
${id}_I \notin\text{Def}(s_1;s_1')$ and ${id}_I \notin\text{Def}(s_2;s_2')$.
By the induction hypothesis, ${S_1'} \equiv_{{id}_I}^S {S_2'}$.
Therefore, the lemma holds.

\item W.l.o.g, there are input statements in $s_1$ only.

By similar argument in the proof of Theorem~\ref{thm:mainTermSameWayLocal} that $s_1$ and $s_2$ do not both include input statements, we can show that there is no input statement in $s_2$ and there is input statement in $s_2'$.

In the following, the proof is of two steps.
\begin{enumerate}
\item We show that $s_1;s_1' \equiv_{{id}_I}^S s_2'$.

By the induction hypothesis IH, $s_1 \equiv_{{id}_I}^S s_2'$.
Because there is no input statement in $s_1'$, then by Lemma~\ref{lmm:inputVarNOtInDefStmtSeqWithoutInputStmt},
${id}_I \notin \text{Def}(s_1')$. Thus, $s_1;s_1' \equiv_{{id}_I}^S s_2'$ by definition.

\item We show that $S_1'$ and $S_2';s_2$ equivalently compute $\text{Imp}(s_1;s_1', \{{id}_I\}) \cup \text{Imp}(s_2', \{{id}_I\})$.

The argument is of two parts. First, we need to show that $\text{Def}(s_2) \cap \text{Imp}(s_2', \{{id}_I\}) = \emptyset$.
By Lemma~\ref{lmm:impInputVarASubsetOfCVar}, $\text{Imp}(s_2', \{{id}_I\}) \subseteq \text{CVar}(s_2')$.
Thus, $\text{Imp}(s_2', \{{id}_I\}) \subseteq \text{TVar}(s_2')$.
By assumption, $\text{Def}(s_2) \cap \text{TVar}(s_2') = \emptyset$.
Then, $\text{Def}(s_2) \cap \text{Imp}(s_2', \{{id}_I\}) = \emptyset$.
By Lemma~\ref{lmm:sameImpfromEquivCompCond}, $\text{Imp}(s_1;s_1', \{{id}_I\}) = \text{Imp}(s_2', \{{id}_I\})$.
Thus, $\text{Def}(s_2) \cap \text{Imp}(s_1;s_1', \{{id}_I\}) = \emptyset$.

Second, we show that $\text{Imp}(s_2', \{{id}_I\}) \subseteq \text{TVar}(s_2;s_2')$ and
$\text{Imp}(s_1;s_1', \{{id}_I\}) \subseteq \text{TVar}(s_1;s_1')$.
By Lemma~\ref{lmm:impInputVarASubsetOfCVar}, $\text{Imp}(s_1;s_1', \{{id}_I\}) \subseteq \text{TVar}(s_1;s_1')$ and
$\text{Imp}(s_2', \{{id}_I\}) \subseteq \text{TVar}(s_2')$.
Then we show that $\text{TVar}(s_2') \subseteq \text{TVar}(s_2;s_2')$.
We need to show that $\text{CVar}(s_2') \subseteq \text{CVar}(s_2;s_2')$ and $\text{LVar}(s_2') \subseteq \text{LVar}(s_2;s_2')$.

\begin{tabbing}
xx\=xx\=\kill
\>\>            $\text{CVar}(s_2;s_2')$ \\
\>=\>           $\text{CVar}(s_2) \cup \text{Imp}(s_2, \text{CVar}(s_2'))$\\
\>\>            by the definition of crash variables\\
\\
\>\>            $\text{Imp}(s_2, \text{CVar}(s_2'))$\\
\>=\>           $\text{CVar}(s_2')$ (1) by the assumption\\
\>\>            $\text{Def}(s_2) \cap \text{TVar}(s_2') = \emptyset$\\
\\
\>\>            $\text{CVar}(s_2) \cup \text{Imp}(s_2, \text{CVar}(s_2'))$ \\
\>=\>           $\text{CVar}(s_2) \cup \text{CVar}(s_2')$ by (1)
\end{tabbing}
Similarly, $\text{LVar}(s_2') \subseteq \text{LVar}(s_2;s_2')$.
Thus, $\text{TVar}(s_2') \subseteq \text{TVar}(s_2;s_2')$.

By assumption, $\forall x \in \text{Imp}(s_1;s_1', \{{id}_I\}) \cup \text{Imp}(s_2', \{{id}_I\})\,:\,
                S_1' \equiv_{x}^S S_2'$.
In addition, $\text{Def}(s_2) \cap (\text{Imp}(s_1;s_1', \{{id}_I\}) \cup \text{Imp}(s_2', \{{id}_I\})) = \emptyset$.
Thus, $\forall x \in \text{Imp}(s_1;s_1', \{{id}_I\}) \cup \text{Imp}(s_2', \{{id}_I\})\,:\,
                S_1' \equiv_{x}^S S_2';s_2$.
The lemma holds.
\end{enumerate}

\end{enumerate}

\end{enumerate}
\end{enumerate}
\end{proof}

\begin{lemma}\label{lmm:TermInSameWayConsumeSameInputs}
If two programs $S_1$ and $S_2$ satisfy the proof rule of termination in the same way, and $S_1$ and $S_2$ both terminate when started in their initial states with crash flags not set, $\mathfrak{f}_1 = \mathfrak{f}_2 = 0$, whose value stores agree on values of variables of the termination deciding variables of $S_1$ and $S_2$,
$\forall x \in \text{TVar}(S_1) \cup \text{TVar}(S_2), \vals_1(x) = \vals_2(x)$, and $S_1$ and $S_2$ are fed with the same infinite input sequence, $\vals_1({id}_I) = \vals_2({id}_I)$,
$(S_1, m_1(\mathfrak{f}_1, \vals_1) ->* (\text{skip}, m_1'(\vals_1'))$ and
$(S_2, m_2(\mathfrak{f}_2, \vals_2) ->* (\text{skip}, m_2'(\vals_2'))$,
then the execution of $S_1$ and $S_2$ consume the same number of input values,
$\vals_1'({id}_I) = \vals_2'({id}_I)$.
\end{lemma}
\begin{proof}
By Lemma~\ref{lmm:TermInSameWayImpliesInputSeqEquivComp}, $S_1 \equiv_{{id}_I}^S S_2$.
By Lemma~\ref{lmm:impInputVarASubsetOfCVar}, $\text{Imp}(S_1, {id}_I) \subseteq \text{CVar}(S_1)$ and
$\text{Imp}(S_2, {id}_I) \subseteq \text{CVar}(S_2)$.
By assumption, $\forall x \in \text{Imp}(S_1, {id}_I) \cup \text{Imp}(S_2, {id}_I)\,:\, \vals_1(x) = \vals_2(x)$.
By Theorem~\ref{thm:equivCompMain}, $\vals_1'({id}_I) = \vals_2'({id}_I)$.
\end{proof} 

\paragraph{Theorem of two loop statements terminating in the same way}

\begin{lemma}\label{lmm:loopTermInSameWayIntermediate}
Let $s_1 = ``\text{while}_{\langle n_1\rangle}(e) \{S_1\}"$ and

    \noindent$s_2 = ``\text{while}_{\langle n_2\rangle}(e) \{S_2\}"$ be two while statements with the same set of termination deciding variables in program $P_1$ and $P_2$ respectively,
    whose bodies $S_1$ and $S_2$ satisfy the proof rule of equivalently computation of variables in $\text{TVar}(s)$, and $S_1$ and $S_2$ terminate in the same way when started in states with crash flags not set and agreeing on values of variables in $\text{TVar}(S_1) \cup \text{TVar}(S_2)$:
\begin{itemize}
\item $\text{TVar}(s_1) = \text{TVar}(s_2) = \text{TVar}(s)$;

\item $\forall x \in \text{TVar}(s) \, :\, S_1 \equiv_{x}^S S_2$;

\item $\forall m_{S_1}(\mathfrak{f}_{S_1}, \vals_{S_1})\, m_{S_2}(\mathfrak{f}_{S_2}, \vals_{S_2}):$

$({((\forall z\in \text{TVar}(S_1) \cup \text{TVar}(S_2)), \vals_{S_1}(z) = \vals_{S_2}(z))} \wedge
  (\mathfrak{f}_{S_1} = \mathfrak{f}_{S_2} = 0)) => $

$(S_1, m_{S_1}(\mathfrak{f}_{S_1}, \vals_{S_1})) \equiv_{H} (S_2, m_{S_2}(\mathfrak{f}_{S_2}, \vals_{S_2}))$.
\end{itemize}

If $s_1$ and $s_2$ start in the state
$m_1(\mathfrak{f}_1, \text{loop}_c^1, \vals_1)$ and

\noindent$m_2(\mathfrak{f}_2, \text{loop}_c^2, \vals_2)$ respectively in which
  crash flags are not set, $\mathfrak{f}_1 = \mathfrak{f}_2 = 0$,
  $s_1$ and $s_2$ have not already executed, $\text{loop}_c^1(n_1) = \text{loop}_c^2(n_2) = 0$,
  value stores $\vals_1$ and $\vals_2$ agree on values of variables in $\text{TVar}(s)$,
  $\forall x \in \text{TVar}(s)\,:\, \vals_1(x) = \vals_2(x)$,
then, for any positive integer $i$, one of the following holds:
\begin{enumerate}
\item{The loop counters for $s_1$ and $s_2$ are less than $i$ where $s_1$ and $s_2$ terminate in the same way:}

$\forall m_{1}'\, m_{2}'\,:\,
({s_1}, m_1) ->* ({S_1'}, m_{1}'(\text{loop}_c^{1'}))$ and
$({s_2}, m_2) ->* ({S_2'}, m_{2}'(\text{loop}_c^{2'}))$,
$\text{loop}_c^{1'}(n_1) < i$ and $m_c^{2'}(n_2) < i$ and one of the following holds:
\begin{enumerate}
\item $s_1$ and $s_2$ both terminate:

 $(s_1, m_1) ->* (\text{skip}, m_1'')$ and
 $(s_2, m_2) ->* (\text{skip}, m_2'')$.

\item $s_1$ and $s_2$ both do not terminate:

$\forall k>0\,:\, (s_1, m_1)$ {\kStepArrow [k] } $({S_{1_k}}, m_{1_k})$ and
                 $(s_2, m_2)$ {\kStepArrow [k] } $({S_{2_k}}, m_{2_k})$ in which
                 ${S_{1_k}} \neq \text{skip}, {S_{2_k}} \neq \text{skip}$.
\end{enumerate}

\item{The loop counters for $s_1$ and $s_2$ are less than or equal to $i$ where $s_1$ and $s_2$ do not terminate such that
there are no configurations $(s_1, m_{1_i})$ and $(s_2, m_{2_i})$ reachable from $(s_1, m_1)$ and $(s_2, m_2)$, respectively, in which  crash flags are not set, the loop counters of $s_1$ and $s_2$ are equal to $i$, and value stores agree on the values of variables in $\text{TVar}(s)$:}

\begin{itemize}
\item $\forall m_1'\,m_2' \,:\,       ({s_1}, m_1) ->* (S_{1}, m_1'(\text{loop}_c^{1'})),
                                      ({s_2}, m_2) ->* (S_{2}, m_2'(\text{loop}_c^{2'}))$
 where

 \noindent$\text{loop}_c^{1'}(n_1) \leq i, \text{loop}_c^{2'}(n_2) \leq i$;

\item $\forall k>0 \,:\,$

\noindent$({s_1}, m_1)$ {\kStepArrow [k] } $(S_{1_k}, m_{1_k}),$
$({s_2}, m_2)$ {\kStepArrow [k] } $(S_{2_k}, m_{2_k})$ where

\noindent$S_{1_k} \neq \text{skip}, S_{2_k} \neq \text{skip}$; and

\item  $\nexists (s_1, m_{1_i})\, (s_2, m_{2_i})\,:\,$

\noindent$({s_1}, m_1) ->* ({s_1}, m_{1_i}(\mathfrak{f}_1, \text{loop}_c^{1_i}, \vals_{1_i})) \wedge$

\noindent$({s_2}, m_2) ->* ({s_2}, m_{2_i}(\mathfrak{f}_2, \text{loop}_c^{2_i}, \vals_{2_i}))$ where
      \begin{itemize}
        \item $\mathfrak{f}_1 = \mathfrak{f}_2 = 0$; and

        \item $\text{loop}_c^{1_i}(n_1) = \text{loop}_c^{2_i}(n_2) = i$; and

        \item $\forall x \in \text{TVar}(s)\,:\,
        \vals_{1_i}(x)\allowbreak = \vals_{2_i}(x)$.
      \end{itemize}
\end{itemize}

\item{There are two configurations $(s_1, m_{1_i})$ and $(s_2, m_{2_i})$ reachable from $(s_1, m_1)$ and $(s_2, m_2)$, respectively, in which both crash flags are not set, the loop counters of $s_1$ and $s_2$ are equal to $i$ and value stores agree on the values of variables in $\text{TVar}(s)$, and for every state in execution $(s_1, m_1) ->* (s_1, m_{1_i})$ or
    $(s_2, m_2) ->* (s_2, m_{2_i})$, the loop counters for $s_1$ and $s_2$ are less than or equal to $i$ respectively:}

 $\exists (s_1, m_{1_i})\, (s_2, m_{2_i})\,:\,
 ({s_1}, m_1) ->* ({s_1}, m_{1_i}(\mathfrak{f}_1, \text{loop}_c^{1_i}, \vals_{1_i})) \wedge
 ({s_2}, m_2) ->* ({s_2}, m_{2_i}(\mathfrak{f}_2, \text{loop}_c^{2_i}, \vals_{2_i}))$ where
\begin{itemize}
 \item $\mathfrak{f}_1 = \mathfrak{f}_2 = 0$; and

 \item $\text{loop}_c^{1_i}(n_1) = \text{loop}_c^{2_i}(n_2) = i$; and

 \item $\forall x \in \text{TVar}(s)\,:\,
               \vals_{1_i}(x) = \vals_{2_i}(x)$; and

 \item $\forall m_1'\,:\,
 ({s_1}, m_1) ->* (S_1', m_1'(m_c^{1'})) ->* ({s_1}, m_{1_i}), \;$

\noindent$\text{loop}_c^{1'}(n_1) \leq i$; and

 \item $\forall m_2'\,:\,
 ({s_2}, m_2) ->* (S_2', m_2'(m_c^{2'})) ->* ({s_2}, m_{2_i}), \;$

\noindent$\text{loop}_c^{2'}(n_2) \leq i$;
\end{itemize}
\end{enumerate}
\end{lemma}
\begin{proof}
By induction on $i$.

\noindent{\bf Base case}. $i = 1$.

By assumption, initial loop counters of $s_1$ and $s_2$ are of value zero. Initial value stores $\vals_1$ and $\vals_2$ agree on the values of variables in $\text{TVar}(s)$. Then we show one of the following cases hold:
\begin{enumerate}
\item{The loop counters for $s_1$ and $s_2$ are less than 1, $s_1$ and $s_2$ terminate in the same way:}

$\forall m_{1}'\, m_{2}'$ such that
$({s_1}, m_1) ->* ({S_1'}, m_1'(\text{loop}_c^{1'}))$ and
$({s_2}, m_2) ->* ({S_2'}, m_2'(\text{loop}_c^{2'}))$,

$m_c^{1'}(n_1) < 1$ and
$m_c^{2'}(n_2) < 1$
and one of the following holds:
\begin{enumerate}
\item $s_1$ and $s_2$ both terminate:

 $(s_1, m_1) ->* (\text{skip}, m_1'')$ and $(s_2, m_2) ->* (\text{skip}, m_2'')$.

\item $s_1$ and $s_2$ both do not terminate:

$\forall k>0,
 (s_1, m_1)$ {\kStepArrow [k] } $({S_{1_k}}, m_{1_k})$ and
$(s_2, m_2)$ {\kStepArrow [k] } $({S_{2_k}}, m_{2_k})$ in which
${S_{1_k}} \neq \text{skip}, {S_{2_k}} \neq \text{skip}$.
\end{enumerate}

\item{The loop counters for $s_1$ and $s_2$ are less than or equal to 1, and $s_1$ and $s_2$ do not terminate such that
there are no configurations $(s_1, m_{1_1})$ and $(s_2, m_{2_1})$ reachable from $(s_1, m_1)$ and $(s_2, m_2)$, respectively, in which crash flags are not set, the loop counters of $s_1$ and $s_2$ are equal to 1 and value stores agree on the values of variables in $\text{TVar}(s)$:}

\begin{itemize}
\item $\forall m_1'\,m_2' \,:\,      ({s_1}, m_1) ->* (S_{1}, m_1'(\text{loop}_c^{1'})),
                                     ({s_2}, m_2) ->* (S_{2}, m_2'(\text{loop}_c^{2'}))$ where
      $\text{loop}_c^{1'}(n_1) \leq i, \text{loop}_c^{2'}(n_2) \leq i$;

\item $\forall k>0 \,:\, ({s_1}, m_1)$ {\kStepArrow [k] } $(S_{1_k}, m_{1_k}),$
                        $({s_2}, m_2)$ {\kStepArrow [k] } $(S_{2_k}, m_{2_k})$ where
      $S_{1_k} \neq \text{skip}, S_{2_k} \neq \text{skip}$; and

      $\nexists (s_1, m_{1_1}), (s_2, m_{2_1})\,:\,
      ({s_1}, m_1) ->* ({s_1}, m_{1_1}(\mathfrak{f}_1, \text{loop}_c^{1_1}, \vals_{1_1})) \wedge
      ({s_2}, m_2) ->* ({s_2}, m_{2_1}(\mathfrak{f}_2, \text{loop}_c^{2_1}, \vals_{2_1}))$ where
      \begin{itemize}
       \item $\mathfrak{f}_1 = \mathfrak{f}_2 = 0$; and

       \item $\text{loop}_c^{1_1}(n_1) = \text{loop}_c^{2_1}(n_2) = 1$; and

       \item $\forall x \in \text{TVar}(s)\,:\,
       \vals_{1_1}(x) = \vals_{2_1}(x)$.
      \end{itemize}
\end{itemize}

\item{There are two configurations $(s_1, m_{1_1})$ and $(s_2, m_{2_1})$ reachable from $(s_1, m_1)$ and $(s_2, m_2)$ respectively, in which the loop counters of  $s_1$ and $s_2$ are equal to 1 and value stores agree on the values of variables in $\text{TVar}(s)$ and, for every state in execution, $(s_1, m_1) ->* (s_1, m_{1_1})$ or $(s_2, m_2) ->* (s_2, m_{2_1})$ the loop counters for $s_1$ and $s_2$ are less than or equal to 1 respectively:}

 $\exists (s_1, m_{1_1})\, (s_2, m_{2_1})\,:\,
 ({s_1}, m_1) ->* ({s_1}, m_{1_1}(\mathfrak{f}_1, \text{loop}_c^{1_1}, \vals_{1_1})) \wedge
 ({s_2}, m_2) ->* ({s_2}, m_{2_1}(\mathfrak{f}_2, \text{loop}_c^{2_1}, \vals_{2_1}))$ where
\begin{itemize}
 \item $\mathfrak{f}_1 = \mathfrak{f}_2 = 0$; and

 \item $\text{loop}_c^{1_1}(n_1) = \text{loop}_c^{2_1}(n_1) = 1$; and

 \item $\forall x \in \text{TVar}(s)\,:\,
 \vals_{1_1}(x) = \vals_{2_1}(x)$; and

 \item $\forall m_1'\,:\, ({s_1}, m_1) ->* (S_1', m_1'(\text{loop}_c^{1'})) ->* ({s_1}, m_{1_1}), \;
 \text{loop}_c^{1'}(n_1) \leq 1$; and

 \item $\forall m_2'\,:\, ({s_2}, m_2) ->* (S_2', m_2'(\text{loop}_c^{2'})) ->* ({s_2}, m_{2_1}), \;
 \text{loop}_c^{2'}(n_2) \leq 1$.
\end{itemize}
\end{enumerate}

We show evaluations of the predicate expression of $s_1$ and $s_2$ w.r.t value stores $\vals_1$ and $\vals_2$ produce same value.
By the definition of loop variables,
$\text{LVar}(s_1) = \bigcup_{j\geq0}\text{Imp}(S_1^j, \text{LVar}(S_1) \cup \text{Use}(e))$.
By our notation of $S^0$, $S_1^0 = \text{skip}$.
By the definition of loop variables, $\text{Use}(e) \subseteq \text{LVar}(s) = \text{LVar}(s_1)$. By assumption, value stores $\vals_1$ and $\vals_2$ agree on the values of the variables in $\text{Use}(e)$.
By Lemma~\ref{lmm:expEvalSameVal}, the predicate expression $e$ of $s_1$ and $s_2$ evaluates to same value $v$ w.r.t value stores $\vals_1, \vals_2$,
$\mathcal{E}'\llbracket e\rrbracket\vals_1 = \mathcal{E}'\llbracket e\rrbracket\vals_2$.
Then there are three possibilities.
\begin{enumerate}
\item{$\mathcal{E}'\llbracket e\rrbracket\vals_1 =
       \mathcal{E}'\llbracket e\rrbracket\vals_2 = (\text{error}, *)$}

The execution from $({s_1}, m_1(\mathfrak{f}_1, \text{loop}_c^{1}, \vals_1))$ proceeds as follows.

\begin{tabbing}
xx\=xx\=\kill
\>\>                   $(s_1,m_1(\mathfrak{f}_1, \text{loop}_c^{1}, \vals_1))$\\
\>= \>                 $(\text{while}_{\langle n_1\rangle}(e) \; \{S_1\}, m_1(\mathfrak{f}_1, \text{loop}_c^{1}, \vals_1))$ \\
\>$->$\>               $(\text{while}_{\langle n_1\rangle}((\text{error}, *)) \; \{S_1\},
                        m_1(\mathfrak{f}_1, \text{loop}_c^{1}, \vals_1))$ by the EEval' rule \\
\>$->$\>               $(\text{while}_{\langle n_1\rangle}(0) \; \{S_1\},
                        m_1(1/\mathfrak{f}_1))$ by the ECrash rule \\
\>{\kStepArrow [k] }\> $(\text{while}_{\langle n_1\rangle}(0) \; \{S_1\},
                        m_1(1/\mathfrak{f}_1, \text{loop}_c^{1}, \vals_1))$,\\
\>\>                    for any $k\geq0$, by the Crash rule.
\end{tabbing}

Similarly, the execution of $s_2$ started in the state
$m_2(\mathfrak{f}_2, \text{loop}_c^{2}, \vals_2)$ does not terminate.

The loop counters of $s_1$ and $s_2$ are less than 1:

$\forall m_{1}'\, m_{2}'\,:\,
 ({s_1}, m_1) ->* ({S_1'}, m_1'(\text{loop}_c^{1'}))$ and
$({s_2}, m_2) ->* ({S_2'}, m_2'(\text{loop}_c^{2'}))$ where $\text{loop}_c^{1'}(n_1) < 1$ and $\text{loop}_c^{2'}(n_2) < 1$.

Besides, $s_1$ and $s_2$ both do not terminate when started in states $m_1$ and $m_2$,
$\forall k>0\,:\,
 (s_1, m_1)$ {\kStepArrow [k] } $({S_{1_k}}, m_{1_k})$ and
$(s_2, m_2)$ {\kStepArrow [k] } $({S_{2_k}}, m_{2_k})$ in which
${S_{1_k}} \neq \text{skip}, {S_{2_k}} \neq \text{skip}$.

\item{$\mathcal{E}'\llbracket e\rrbracket\vals_1 =
       \mathcal{E}'\llbracket e\rrbracket\vals_2 = (0, v_\mathfrak{of})$}

The execution of $s_1$ proceeds as follows.

\begin{tabbing}
xx\=xx\=\kill
\>\>     $({s_1}, m_{1}(\text{loop}_c^1,\vals_1))$\\
\>= \>   $(\text{while}_{\langle n_1\rangle} (e) \; \{S_1\}, {m_{1}(\text{loop}_c^{1})})$\\
\>$->$\> $(\text{while}_{\langle n_1\rangle} ((0, v_\mathfrak{of})) \; \{S_1\}, {m_{1}(\text{loop}_c^{1})})$ by the EEval' rule\\
\>$->$\> $(\text{while}_{\langle n_1\rangle} (0) \; \{S_1\}, {m_{1}(\text{loop}_c^{1})})$ by the E-Oflow1 or E-Oflow2 rule\\
\>$->$\> $({\text{skip}}, m_1)$ by the Wh-F1 rule.
\end{tabbing}

Similarly, $({s_2}, m_{2}(\text{loop}_c^{2},\vals_{2}))$ {\kStepArrow [2] } $({\text{skip}}, m_2)$.

The loop counters for $s_1$ and $s_2$ are less than 1:

$\forall m_{1}'\, m_{2}'\,:\,
 ({s_1}, m_1) ->* ({S_1'}, m_1'(\text{loop}_c^{1'}))$ and
$({s_2}, m_2) ->* ({S_2'}, m_2'(\text{loop}_c^{2'}))$ where
$\text{loop}_c^{1'}(n_1) < 1$ and $\text{loop}_c^{2'}(n_2) < 1$.

Besides, $s_1$ and $s_2$ both terminate when started in states $m_1$ and $m_2$:

 $(s_1, m_1) ->* (\text{skip}, m_1'')$ and
 $(s_2, m_2) ->* (\text{skip}, m_2'')$.

\item
$\mathcal{E}'\llbracket e\rrbracket\vals_{1} =
\mathcal{E}'\llbracket e\rrbracket\vals_{2} = (v, v_\mathfrak{of})$
where $v \notin \{0, \text{error}\}$;

The execution from $({s_1}, m_{1}(\text{loop}_c^{1},\vals_{1}))$ proceeds as follows.

\begin{tabbing}
xx\=xx\=\kill
\>\>     $({s_1}, m_{1}(\text{loop}_c^{1}, \vals_{1}))$\\
\>= \>   $(\text{while}_{\langle n_1\rangle} (e) \; \{S_1\}, {m_{1}(\text{loop}_c^{1}, \vals_{1})})$\\
\>$->$\> $(\text{while}_{\langle n_1\rangle} ((v, v_\mathfrak{of})) \; \{S_1\}, {m_{1}(\text{loop}_c^{1}, \vals_{1})})$ by the EEval' rule\\
\>$->$\> $(\text{while}_{\langle n_1\rangle} (v) \; \{S_1\}, {m_{1}(\text{loop}_c^{1}, \vals_{1})})$ \\
\>\>     by rule E-Oflow1 or E-Oflow2\\
\>$->$\> $(S_1;\text{while}_{\langle n_1\rangle} (e) \; \{S_1\}, m_{1}(\text{loop}_c^{1}[1/(n_1)], \vals_{1}))$ \\
\>\>      by the Wh-T1 rule.
\end{tabbing}

Similarly,
$({s_2}, m_{2}(\text{loop}_c^{2}, \vals_{2}))$ {\kStepArrow [2] }
$(S_2;\text{while}_{\langle n_2\rangle} (e) \{S_2\}, \,\, \allowbreak m_{2}(\text{loop}_c^{2}[1/(n_2)], \vals_{2}))$.
After two steps of executions of $s_1$ and $s_2$, crash flags are not set, the loop counter value of $s_1$ and $s_2$ are 1,
value stores $\vals_{1}$ and $\vals_{2}$ agree on values of variables in $\text{TVar}(s)$.

We show that $\text{TVar}(S_1) \subseteq \text{TVar}(s)$.
By definition of loop variables,
$\text{LVar}(s_1) =
\bigcup_{j\geq0}\text{Imp}(S_1^j, \text{LVar}(S_1)\allowbreak \cup \text{Use}(e))$.
By notation of $S^0$, $S^0 = \text{skip}$.
By definition of imported variables,
$\text{Imp}(S_1^0, \text{LVar}(S_1) \cup \text{Use}(e)) =
 \text{LVar}(S_1) \cup \text{Use}(e)$.
Then $\text{LVar}(S_1) \subseteq \text{LVar}(s)$.
By similar argument, we have
$\text{CVar}(S_1) \subseteq \text{CVar}(s)$.
Hence, $\text{TVar}(S_1) \subseteq \text{TVar}(s)$.
Similarly, $\text{TVar}(S_2) \subseteq \text{TVar}(s)$.
By assumption, $S_1$ and $S_2$ either both terminate or both do not terminate when started in state $m_{1}(\text{loop}_c^{1}[1/(n_1)], \vals_{1})$ and
$m_{2}(\text{loop}_c^{2}[1/(n_2)], \vals_{2})$ in which
$\forall y\in\text{TVar}(S_1) \cup \text{TVar}(S_2),
\vals_1(y) = \vals_2(y)$ and crash flags are not set.
Then there are two possibilities:
\begin{enumerate}
\item  $S_1$ and $S_2$ both terminate when started in states
       $m_{1}(\allowbreak \text{loop}_c^{1}[1/(n_1)], \vals_{1})$ and
       $m_{2}(\text{loop}_c^{2}[1/(n_2)], \vals_{2})$ respectively:

    $(S_1, m_{1}(\text{loop}_c^{1}[1/(n_1)], \vals_{1})) ->*
    (\text{skip}, m_{1_1}(\mathfrak{f}_1, \text{loop}_c^{1_1}, \vals_{1_1}))$ and

    $(S_2, m_{2}(\text{loop}_c^{2}[1/(n_2)], \vals_{2})) ->*
    (\text{skip}, m_{2_1}(\mathfrak{f}_2, \text{loop}_c^{2_1}, \vals_{2_1}))$.

%
%

We show that, after the full execution of $S_1$ and $S_2$, the following five properties hold.
\begin{itemize}
\item The crash flags are not set.

By the definition of terminating execution, crash flags are not set, $\mathfrak{f}_1 = \mathfrak{f}_2 = 0$.

\item The loop counter of $s_1$ and $s_2$ are of value $1$,
$\text{loop}_c^{1_{1}}(n_1) = \text{loop}_c^{2_{1}}(n_2) = 1$.

By the assumption of unique loop labels, $s_1 \notin S_1$. Then, the loop counter value of $n_1$ is not redefined in the execution of $S_1$ by corollary~\ref{coro:defExclusion},
$\text{loop}_c^{1}[1/n_1](n_1)$ = $\text{loop}_c^{1_{1}}(n_1) = 1$.
Similarly, the loop counter value of $n_2$ is not redefined in the execution of $S_2$,
$\text{loop}_c^{2}[1/(n_2)](n_2)$ = $\text{loop}_c^{2_{1}}(n_2) = 1$.

\item In any state in the execution $(s_1, m_1) ->* (s_1, m_{1_1}(\text{loop}_c^{1_1}, \vals_{1_1}))$, the loop counter of $s_1$ is less than or equal to 1.

    As is shown above, the loop counter of $s_1$ is zero in any of the two states in the one step execution
    $({s_1}, m_1) -> (\text{while}_{\langle n_1\rangle} (v) \; \{S_1\}, {m_{1}(\text{loop}_c^{1}, \vals_{1})})$, and
    the loop counter of $s_1$ is 1 in any states in the execution
    $(S_1;\text{while}_{\langle n_1\rangle} (e) \; \{S_1\}, m_{1}(\text{loop}_c^{1}[i/(n_1)], \vals_{1})) ->* (s_1, m_{1_1}(\text{loop}_c^{1_1}, \vals_{1_1}))$.

\item In any state in the executions $(s_2, m_2) ->* (s_2, m_{2_1}(\text{loop}_c^{2_1}, \vals_{2_1}))$, the loop counter of $s_2$ is less than or equal to 1.

    By similar argument above.

\item The value stores $\vals_{1_{1}}$ and $\vals_{2_{1}}$ agree on values of the termination deciding variables in $s_1$ and $s_2$:
    $\forall x \in \text{TVar}(s),
    \vals_{1_{1}}(x) = \vals_{2_{1}}(x)$.


We show that the imported variables in $S_1$ relative to those in $\text{LVar}(s)$ are a subset of $\text{LVar}(s)$ and the imported variables in $S_1$ relative to those in $\text{CVar}(s)$ are a subset of $\text{CVar}(s)$.

\begin{tabbing}
xx\=xx\=\kill
\>\>     $\text{LVar}(s_1)$\\
\>=\>    $\bigcup_{j\geq0} \text{Imp}(S_1^j, \text{LVar}(S_1) \cup \text{Use}(e))$  (1)\\
\>\>      by the definition of loop variables.\\
\end{tabbing}

\begin{tabbing}
xx\=xx\=\kill
\>\>              $\text{Imp}(S_1, \text{LVar}(s)) = \text{Imp}(S_1, \text{LVar}(s_1))$ \\
\>=\>             $\text{Imp}(S_1, \text{Imp}(s_1,  \text{Use}(e) \cup \text{LVar}(S_1)))$\\
\>\>              by the definition of $\text{LVar}(s)$\\
\>=\>             $\text{Imp}(S_1, \bigcup_{j\geq0} \text{Imp}(S_1^j, \text{LVar}(S_1) \cup \text{Use}(e)))$ by (1)\\
\>=\>             $\bigcup_{j\geq0}\text{Imp}(S_1, \text{Imp}(S_1^j, \text{LVar}(S_1) \cup \text{Use}(e)))$\\
\>\>               by Lemma~\ref{lmm:impVarUnionLemma}\\
\>=\>             $\bigcup_{j>0}\text{Imp}(S_1^j, \text{LVar}(S_1) \cup \text{Use}(e))$ by Lemma~\ref{lmm:ImpPrefixLemma}\\
\>$\subseteq$\>   $\bigcup_{j\geq0} \text{Imp}(S_1^j, \text{LVar}(S_1) \cup \text{Use}(e))$ \\
\>=\>             $\text{Imp}(s_1, \text{LVar}(S_1) \cup \text{Use}(e)) = \text{LVar}(s_1) = \text{LVar}(s)$.
\end{tabbing}

Similarly, $\text{Imp}(S_1, \text{CVar}(s)) \subseteq \text{CVar}(s)$.
Hence,

\noindent$\text{Imp}(S_1, \text{TVar}(s)) \subseteq  \text{TVar}(s)$. In the same way, we can show that $\text{Imp}(S_2, \text{TVar}(s)) \subseteq \; \text{TVar}(s)$.
Consequently, the value stores $\vals_{1_1}$ and $\vals_{2_1}$ agree on the values of the imported variables in $S_1$ and $S_2$ relative to those in $\text{TVar}(s)$,
$\forall x \in \text{Imp}(S_1, \text{TVar}(s))\allowbreak \cup
               \text{Imp}(S_2, \text{TVar}(s)),
               \vals_{1}(x, \allowbreak) = \vals_{2}(x)$.
Because $S_1$ and $S_2$ have equivalent computation of every variable in $\text{TVar}(s)$ when started in states agreeing on the values of the imported variables relative to $\text{TVar}(s)$, by Theorem~\ref{thm:equivTermCompOfSimpleStmt}, value stores $\vals_{1_{1}}$ and $\vals_{2_{1}}$ agree on the values of the variables $\text{TVar}(s)$,
$\forall x \in \text{TVar}(s),
\vals_{1_{1}}(x) = \vals_{2_{1}}(x)$.
 \end{itemize}

It follows that, by Corollary~\ref{coro:termSeq},

$(S_1;\text{while}_{\langle n_1\rangle} (e)  \{S_1\}, m_{1}(\text{loop}_c^{1}[1/(n_1)], \vals_{1})) ->*$

$(\text{while}_{\langle n_1\rangle} (e) \; \{S_1\}, m_{1_{1}}(\text{loop}_c^{1_{1}}, \vals_{1_{1}})) =
(s_1, m_{1_{1}}(\text{loop}_c^{1_{1}}, \vals_{1_{1}}))$ and

$(S_2;\text{while}_{\langle n_2\rangle} (e) \; \{S_2\}, m_{2}(\text{loop}_c^{2}[1/(n_2)], \vals_{2})) ->*$

$(\text{while}_{\langle n_2\rangle} (e) \; \{S_2\}, m_{2_{1}}(\text{loop}_c^{2_{1}}, \vals_{2_{1}})) =
(s_2, m_{2_{1}}(\text{loop}_c^{2_{1}}, \vals_{2_{1}}))$.

\item $S_1$ and $S_2$ do not terminate when started in states

\noindent$m_{1}(\text{loop}_c^{1}[1/(n_1)], \vals_{1})$ and

\noindent$m_{2}(\text{loop}_c^{2}[1/(n_2)], \vals_{2})$ respectively:

    $\forall k>0, (S_1, m_{1}(\text{loop}_c^{1}[1/(n_1)], \vals_{1}))$ {\kStepArrow [k] }

    \noindent$(S_{1_k}, m_{1_{1_k}}(\text{loop}_c^{1_{1_k}}, \vals_{1_{1_k}}))$ and

    $(S_2, m_{2}(\text{loop}_c^{2}[1/(n_2)], \vals_{2}))$ {\kStepArrow [k] }

    \noindent$(S_{2_k}, m_{2_{1_k}}(\text{loop}_c^{2_{1_k}}, \vals_{2_{1_k}}))$
    in which
    $S_{1_k} \neq \text{skip}, S_{2_k} \neq \text{skip}$.

  By our assumption of unique loop labels, $s_1\notin S_1$.
  Then, $\forall k>0, \text{loop}_c^{1_{1_k}}(n_1) = \text{loop}_c^{1}[1/(n_1)](n_1)\allowbreak = 1$.
  Similarly, $\forall k>0, \text{loop}_c^{2_{1_k}}(n_2)$

  \noindent$ = \text{loop}_c^{2}[1/(n_2)](n_2) = 1$.
  In addition, by Lemma~\ref{lmm:multiStepSeqExec},

  $\forall k>0, (S_1;s_1, m_{1}(\text{loop}_c^{1}[1/(n_1)], \vals_{1}))$ {\kStepArrow [k] }

  \noindent$(S_{k};s_1, m_{1_{k}}(\text{loop}_c^{1_{k}}, \vals_{1_{k}}))$ and

    $(S_2;s_2, m_{2}(\text{loop}_c^{2}[1/(n_2)], \vals_{2}))$ {\kStepArrow [k] } $(S_{2_k};s_2, m_{2_{k}}(\text{loop}_c^{2_{k}}, \vals_{2_{k}}))$
    in which
    $S_{1_k} \neq \text{skip}, S_{2_k} \neq \text{skip}$.

  In summary, loop counters of $s_1$ and $s_2$ are less than or equal to 1, and $s_1$ and $s_2$ do not terminate
  such that
there are no configurations $(s_1, m_{1_1})$ and $(s_2, m_{2_1})$ reachable from $(s_1, m_1)$ and $(s_2, m_2)$, respectively, in which crash flags are not set, the loop counters of $s_1$ and $s_2$ are equal to 1 and value stores agree on the values of variables in $\text{TVar}(s)$.
\end{enumerate}

\end{enumerate}

\myNewLine

\noindent{\bf Induction Step}.

\noindent The induction hypothesis IH is that, for a positive integer $i$, one of the following holds:
\begin{enumerate}
\item{The loop counters for $s_1$ and $s_2$ are less than $i$, and $s_1$ and $s_2$ both terminate in the same way:}

$\forall m_{1}'\, m_{2}'$ such that
$({s_1}, m_1) ->* ({S_1'}, m_{1}'(\text{loop}_c^{1'}))$ and
$({s_2}, m_2) ->* ({S_2'}, m_{2}'(\text{loop}_c^{2'}))$,

$\text{loop}_c^{1'}(n_1) < i$ and $\text{loop}_c^{2'}(n_2) < i$ and one of the following holds:
\begin{enumerate}
\item $s_1$ and $s_2$ both terminate:

 $(s_1, m_1) ->* (\text{skip}, m_1'')$ and $(s_2, m_2) ->* (\text{skip}, m_2'')$.

\item $s_1$ and $s_2$ both do not terminate:

$\forall k>0,
 (s_1, m_1)$ {\kStepArrow [k] } $({S_{1_k}}, m_{1_k})$ and
$(s_2, m_2)$ {\kStepArrow [k] } $({S_{2_k}}, m_{2_k})$ in which
${S_{1_k}} \neq \text{skip}, {S_{2_k}} \neq \text{skip}$.
\end{enumerate}

\item The loop counters for $s_1$ and $s_2$ are less than or equal to $i$, and $s_1$ and $s_2$ do not terminate such that
there are no configurations $(s_1, m_{1_i})$ and $(s_2, m_{2_i})$ reachable from $(s_1, m_1)$ and $(s_2, m_2)$, respectively, in which crash flags are not set, the loop counters of $s_1$ and $s_2$ are equal to $i$ and value stores agree on the values of variables in $\text{TVar}(s)$:

\begin{itemize}
\item $\forall m_1'\,m_2' \,:\,
({s_1}, m_1) ->* (S_{1}, m_1'(\text{loop}_c^{1'})),
({s_2}, m_2) ->* (S_{2}, m_2'(\text{loop}_c^{2'}))$ where
$\text{loop}_c^{1'}(n_1) \leq i, \text{loop}_c^{2'}(n_2) \leq i$;

\item $\forall k>0 \,:\,
 ({s_1}, m_1)$ {\kStepArrow [k] } $(S_{1_k}, m_{1_k}),$
$({s_2}, m_2)$ {\kStepArrow [k] } $(S_{2_k}, m_{2_k})$ where
      $S_{1_k} \neq \text{skip}, S_{2_k} \neq \text{skip}$; and

\item  $\nexists (s_1, m_{1_i}), (s_2, m_{2_i})\,:\,
({s_1}, m_1) ->* ({s_1}, m_{1_i}(\mathfrak{f}_1, \text{loop}_c^{1_i}, \vals_{1_i})) \wedge
({s_2}, m_2) ->* ({s_2}, m_{2_i}(\mathfrak{f}_2, \text{loop}_c^{2_i}, \vals_{2_i}))$ where
      \begin{itemize}
        \item $\mathfrak{f}_1 = \mathfrak{f}_2 = 0$; and

        \item $\text{loop}_c^{1_i}(n_1) = \text{loop}_c^{2_i}(n_2) = i$; and

        \item $\forall x \in \text{TVar}(s)\,:\,
        \vals_{1_i}(x) = \vals_{2_i}(x)$.
      \end{itemize}
\end{itemize}

\item There are two configurations $(s_1, m_{1_i})$ and $(s_2, m_{2_i})$ reachable from $(s_1, m_1)$ and $(s_2, m_2)$, respectively, in which crash flags are not set, the loop counters of $s_1$ and $s_2$ are equal to $i$ and value stores agree on the values of variables in $\text{TVar}(s)$ and, for every state in execution, $(s_1, m_1) ->* (s_1, m_{1_i})$ or $(s_2, m_2) ->* (s_2, m_{2_i})$ the loop counters for $s_1$ and $s_2$ are less than or equal to $i$ respectively:

 $\exists (s_1, m_{1_i})\, (s_2, m_{2_i})\,:\,
 ({s_1}, m_1) ->* ({s_1}, m_{1_i}(\mathfrak{f}_1, \text{loop}_c^{1_i}, \vals_{1_i})) \wedge
 ({s_2}, m_2) ->* ({s_2}, m_{2_i}(\mathfrak{f}_2, \text{loop}_c^{2_i}, \vals_{2_i}))$ where
\begin{itemize}
 \item $\mathfrak{f}_1 = \mathfrak{f}_2 = 0$; and

 \item $\text{loop}_c^{1_i}(n_1) = \text{loop}_c^{2_i}(n_2) = i$; and

 \item $\forall x \in \text{TVar}(s)\,:\,
 \vals_{1_i}(x) = \vals_{2_i}(x, $; and

 \item $\forall m_1'\,:\,
 ({s_1}, m_1) ->* (S_1', m_1'(\text{loop}_c^{1'})) ->* ({s_1}, m_{1_i})$,

 \noindent$\text{loop}_c^{1'}(n_1) \leq i$; and

 \item $\forall m_2'\,:\,
 ({s_2}, m_2) ->* (S_2', m_2'(\text{loop}_c^{2'})) ->* ({s_2}, m_{2_i})$,

 \noindent$\text{loop}_c^{2'}(n_2) \leq i$;
\end{itemize}
\end{enumerate}

\noindent Then we show that, for the positive integer $i+1$, one of the following holds:
\begin{enumerate}
\item{The loop counters for $s_1$ and $s_2$ are less than $i+1$, and $s_1$ and $s_2$ both terminate in the same way:}

$\forall m_{1}'\, m_{2}'$ such that
$({s_1}, m_1) ->* ({S_1'}, m_{1}'(\text{loop}_c^{1'}))$ and
$({s_2}, m_2) ->* ({S_2'}, m_{2}'(\text{loop}_c^{2'}))$,

$\text{loop}_c^{1'}(n_1) < i+1$ and $\text{loop}_c^{2'}(n_2) < i+1$ and one of the following holds:
\begin{enumerate}
\item $s_1$ and $s_2$ both terminate:

 $(s_1, m_1) ->* (\text{skip}, m_1'')$ and
 $(s_2, m_2) ->* (\text{skip}, m_2'')$.

\item $s_1$ and $s_2$ both do not terminate:

$\forall k>0,
 (s_1, m_1)$ {\kStepArrow [k] } $({S_{1_k}}, m_{1_k})$ and
$(s_2, m_2)$ {\kStepArrow [k] } $({S_{2_k}}, m_{2_k})$ in which
${S_{1_k}} \neq \text{skip}, {S_{2_k}} \neq \text{skip}$.
\end{enumerate}

\item{The loop counters for $s_1$ and $s_2$ are less than or equal to $i+1$, and $s_1$ and $s_2$ do not terminate such that
there are no configurations $(s_1, m_{1_{i+1}})$ and $(s_2, m_{2_{i+1}})$ reachable from $(s_1, m_1)$ and $(s_2, m_2)$, respectively, in which crash flags are set, the loop counters of $s_1$ and $s_2$ are equal to $i+1$ and value stores agree on the values of variables in $\text{TVar}(s)$:}

\begin{itemize}
\item $\forall m_1'\,m_2' \,:\,      ({s_1}, m_1) ->* (S_{1}, m_1'(\text{loop}_c^{1'})),
                                     ({s_2}, m_2) ->* (S_{2}, m_2'(\text{loop}_c^{2'}))$ where

      $\text{loop}_c^{1'}(n_1) \leq i+1, \text{loop}_c^{2'}(n_2) \leq i+1$;

\item $\forall k>0 \,:\, ({s_1}, m_1)$ {\kStepArrow [k] } $(S_{1_k}, m_{1_k}),$
                        $({s_2}, m_2)$ {\kStepArrow [k] } $(S_{2_k}, m_{2_k})$ where
      $S_{1_k} \neq \text{skip}, S_{2_k} \neq \text{skip}$; and

\item $\nexists (s_1, m_{1_{i+1}}), (s_2, m_{2_{i+1}})\,:\,$

      $({s_1}, m_1) ->* ({s_1}, m_{1_{i+1}}(\mathfrak{f}_1, \text{loop}_c^{1_{i+1}}, \vals_{1_{i+1}})) \wedge
       ({s_2}, m_2) ->* ({s_2}, m_{2_{i+1}}(\mathfrak{f}_2, \text{loop}_c^{2_{i+1}}, \vals_{2_{i+1}}))$ where
      \begin{itemize}
        \item $\mathfrak{f}_1 = \mathfrak{f}_2 = 0$; and

        \item $\text{loop}_c^{1_{i+1}}(n_1) = \text{loop}_c^{2_{i+1}}(n_2) = i+1$; and

        \item $\forall x \in \text{TVar}(s)\,:\,
        \vals_{1_{i+1}}(x)$ = $\vals_{2_{i+1}}(x)$.
      \end{itemize}
\end{itemize}

\item There are two configurations $(s_1, m_{1_{i+1}})$ and $(s_2, m_{2_{i+1}})$ reachable from $(s_1, m_1)$ and $(s_2, m_2)$, respectively, in which the loop counters of $s_1$ and $s_2$ are equal to $i+1$ and value stores agree on the values of variables in $\text{TVar}(s)$ and, for every state in execution, $(s_1, m_1) ->* (s_1, m_{1_{i+1}})$ or $(s_2, m_2) ->* (s_2, m_{2_{i+1}})$ the loop counters for $s_1$ and $s_2$ are less than or equal to $i+1$ respectively:

 $\exists (s_1, m_{1_{i+1}})\, (s_2, m_{2_{i+1}})\,:\,$

 \noindent$({s_1}, m_1) ->* ({s_1}, m_{1_{i+1}}(\text{loop}_c^{1_{i+1}}, \vals_{1_{i+1}})) \wedge
 ({s_2}, m_2) ->* ({s_2}, m_{2_{i+1}}(\text{loop}_c^{2_{i+1}}, \vals_{2_{i+1}}))$ where
\begin{itemize}
 \item $\text{loop}_c^{1_{i+1}}(n_1) = \text{loop}_c^{2_{i+1}}(n_2) = i+1$; and

 \item $\forall x \in \text{TVar}(s)\,:\,
 \vals_{1_{i+1}}(x) = \vals_{2_{i+1}}(x)$; and

 \item $\forall m_1'\,:\, ({s_1}, m_1) ->* (S_1', m_1'(\text{loop}_c^{1'})) ->* ({s_1}, m_{1_i}), \;$

 \noindent$\text{loop}_c^{1'}(n_1) \leq i+1$; and

 \item $\forall m_2'\,:\, ({s_2}, m_2) ->* (S_2', m_2'(\text{loop}_c^{2'})) ->* ({s_2}, m_{2_i}), \;$

 \noindent$\text{loop}_c^{2'}(n_2) \leq i+1$;
\end{itemize}
\end{enumerate}

\noindent By the hypothesis IH, one of the following holds:
\begin{enumerate}
\item{The loop counters for $s_1$ and $s_2$ are less than $i$:}

$\forall m_{1}'\, m_{2}'$ such that
$({s_1}, m_1) ->* ({S_1'}, m_1'(\text{loop}_c^{1'}))$ and
$({s_2}, m_2) ->* ({S_2'}, m_2'(\text{loop}_c^{2'}))$,

$\text{loop}_c^{1'}(n_1) < i$ and $\text{loop}_c^{2'}(n_2) < i$ and one of the following holds:
\begin{enumerate}
\item $s_1$ and $s_2$ both terminate:

 $(s_1, m_1) ->* (\text{skip}, m_1'')$ and
 $(s_2, m_2) ->* (\text{skip}, m_2'')$.

\item $s_1$ and $s_2$ both do not terminate:

$\forall k>0,
 (s_1, m_1)$ {\kStepArrow [k] } $({S_{1_k}}, m_{1_k})$ and
$(s_2, m_2)$ {\kStepArrow [k] } $({S_{2_k}}, m_{2_k})$ in which
${S_{1_k}} \neq \text{skip}, {S_{2_k}} \neq \text{skip}$.
\end{enumerate}

\myNewLine

When this case holds, then we have the loop counters for $s_1$ and $s_2$ are less than $i+1$, and $s_1$ and $s_2$ both terminate in the same way:

$\forall m_{1}'\, m_{2}'$ such that
$({s_1}, m_1) ->* ({S_1'}, m_1'(\text{loop}_c^{1'}))$ and
$({s_2}, m_2) ->* ({S_2'}, m_2'(\text{loop}_c^{2'}))$,

$\text{loop}_c^{1'}(n_1) < i+1$ and $\text{loop}_c^{2'}(n_2) < i+1$ and one of the following holds:
\begin{enumerate}
\item $s_1$ and $s_2$ both terminate:

 $(s_1, m_1) ->* (\text{skip}, m_1'')$ and $(s_2, m_2) ->* (\text{skip}, m_2'')$.

\item $s_1$ and $s_2$ both do not terminate:

$\forall k>0,
 (s_1, m_1)$ {\kStepArrow [k] } $({S_{1_k}}, m_{1_k})$ and
$(s_2, m_2)$ {\kStepArrow [k] } $({S_{2_k}}, m_{2_k})$ in which
${S_{1_k}} \neq \text{skip}, {S_{2_k}} \neq \text{skip}$.
\end{enumerate}

\item The loop counters for $s_1$ and $s_2$ are less than or equal to $i$, and $s_1$ and $s_2$ both do not terminate such that
there are no configurations $(s_1, m_{1_i})$ and $(s_2, m_{2_i})$ reachable from $(s_1, m_1)$ and $(s_2, m_2)$, respectively, in which the loop counters of $s_1$ and $s_2$ are equal to $i$ and value stores agree on the values of variables in $\text{TVar}(s)$:

\begin{itemize}
\item $\forall m_1'\,m_2' \,:\,      ({s_1}, m_1) ->* (S_{1}, m_1'(\text{loop}_c^{1'})),
                                     ({s_2}, m_2) ->* (S_{2}, m_2'(\text{loop}_c^{2'}))$ where
      $\text{loop}_c^{1'}(n_1) \leq i, \text{loop}_c^{2'}(n_2) \leq i$;

\item $\forall k>0 \,:\, ({s_1}, m_1)$ {\kStepArrow [k] } $(S_{1_k}, m_{1_k}),$
                        $({s_2}, m_2)$ {\kStepArrow [k] } $(S_{2_k}, m_{2_k})$ where
      $S_{1_k} \neq \text{skip}, S_{2_k} \neq \text{skip}$; and

      $\nexists (s_1, m_{1_i}), (s_2, m_{2_i})\,:\,
      ({s_1}, m_1) ->* ({s_1}, m_{1_i}(\mathfrak{f}_1, \text{loop}_c^{1_i}, \vals_{1_i})) \wedge
      ({s_2}, m_2) ->* ({s_2}, m_{2_i}(\mathfrak{f}_2, \text{loop}_c^{2_i}, \vals_{2_i}))$ where
      \begin{itemize}
        \item $\mathfrak{f}_1 = \mathfrak{f}_2 = 0$; and

        \item $\text{loop}_c^{1_i}(n_1) = \text{loop}_c^{2_i}(n_2) = i$; and

        \item $\forall x \in \text{TVar}(s)\,:\,
        \vals_{1_i}(x) = \vals_{2_i}(x)$.
      \end{itemize}
\end{itemize}

\myNewLine

When this case holds, we have the loop counter of $s_1$ and $s_2$ are less than $i+1$, and $s_1$ and $s_2$ both do not terminate:

$\forall m_{1}'\, m_{2}'$ such that
$({s_1}, m_1) ->* ({S_1'}, m_1'(\text{loop}_c^{1'}))$ and
$({s_2}, m_2) ->* ({S_2'}, m_2'(\text{loop}_c^{2'}))$,

$\text{loop}_c^{1'}(n_1) < i+1$ and $\text{loop}_c^{2'}(n_2) < i+1$ and
$s_1$ and $s_2$ both do not terminate:

$\forall k>0,
 (s_1, m_1)$ {\kStepArrow [k] } $({S_{1_k}}, m_{1_k})$ and
$(s_2, m_2)$ {\kStepArrow [k] } $({S_{2_k}}, m_{2_k})$ in which
${S_{1_k}} \neq \text{skip}, {S_{2_k}} \neq \text{skip}$.

\item There are two configurations $(s_1, m_{1_i})$ and $(s_2, m_{2_i})$ reachable from $(s_1, m_1)$ and $(s_2, m_2)$, respectively, in which crash flags are not set, the loop counters of $s_1$ and $s_2$ are equal to $i$ and value stores agree on the values of variables in $\text{TVar}(s)$ and, for every state in executions $(s_1, m_1) ->* (s_1, m_{1_i})$ and $(s_2, m_2) ->* (s_2, m_{2_i})$ the loop counters for $s_1$ and $s_2$ are less than or equal to $i$ respectively:

 $\exists (s_1, m_{1_i})\, (s_2, m_{2_i})\,:\,
 (s_1, m_1) ->* (s_1, m_{1_i}(\mathfrak{f}_1, \text{loop}_c^{1_i}, \vals_{1_i})) \wedge
 (s_2, m_2) ->* (s_2, m_{2_i}(\mathfrak{f}_2, \text{loop}_c^{2_i}, \vals_{2_i}))$ where
\begin{itemize}
 \item $\mathfrak{f}_1 = \mathfrak{f}_2 = 0$; and

 \item $\text{loop}_c^{1_i}(n_1) = \text{loop}_c^{2_i}(n_2) = i$; and

 \item $\forall x \in \text{TVar}(s),
 \vals_{1_i}(x) = \vals_{2_i}(x)$; and

 \item $\forall m_1'\,:\, ({s_1}, m_1) ->* (S_1', m_1'(\text{loop}_c^{1'})) ->* ({s_1}, m_{1_i}), \;$

 \noindent$\text{loop}_c^{1'}(n_1) \leq i$; and

 \item $\forall m_2'\,:\, ({s_2}, m_2) ->* (S_2', m_2'(\text{loop}_c^{2'})) ->* ({s_2}, m_{2_i}), \;$

 \noindent$\text{loop}_c^{2'}(n_2) \leq i$.
\end{itemize}
By similar argument in base case, evaluations of the predicate expression of $s_1$ and $s_2$ w.r.t value stores $\vals_{1_i}$ and $\vals_{2_i}$ produce same value. Then there are three possibilities:
\begin{enumerate}
\item $\mathcal{E}'\llbracket e\rrbracket\vals_{1_i} =
       \mathcal{E}'\llbracket e\rrbracket\vals_{2_i} = (\text{error}, *)$.

Then the execution of $s_1$ proceeds as follows.
\begin{tabbing}
xx\=xx\=\kill
\>\>                   $(s_1,m_{1_i}(\mathfrak{f}_1, \vals_{1_i}))$\\
\>= \>                 $(\text{while}_{\langle n_1\rangle}(e) \; \{S_1\}, m_{1_i}(\mathfrak{f}_1, \vals_{1_i}))$ \\
\>$->$\>               $(\text{while}_{\langle n_1\rangle}((\text{error}, *)) \; \{S_1\}, m_{1_i}(\mathfrak{f}_1, \vals_{1_i}))$ by the EEval' rule \\
\>$->$\>               $(\text{while}_{\langle n_1\rangle}(0) \; \{S_1\}, m_{1_i}(1/\mathfrak{f}_1))$ by the ECrash rule \\
\>{\kStepArrow [k] }\> $(\text{while}_{\langle n_1\rangle}(0) \; \{S_1\}, m_{1_i}(1/\mathfrak{f}_1))$, for any $k\geq0$, by the Crash rule.
\end{tabbing}

Similarly, the execution of $s_2$ started in the state $m_{2_i}(\vals_{2_i})$ does not terminate.

The loop counters for $s_1$ and $s_2$ are less than $i+1$:

$\forall m_{1}'\, m_{2}'$ such that
$({s_1}, m_1) ->* ({S_1'}, m_1'(\text{loop}_c^{1'}))$ and
$({s_2}, m_2) ->* ({S_2'}, m_2'(\text{loop}_c^{2'}))$,

$\text{loop}_c^{1'}(n_1) < i+1$ and $\text{loop}_c^{2'}(n_2) < i+1$.

Besides, $s_1$ and $s_2$ both do not terminate when started in states $m_1$ and $m_2$,

$\forall k>0, (s_1, m_1)$ {\kStepArrow [k] } $({S_{1_k}}, m_{1_k})$ and
             $(s_2, m_2)$ {\kStepArrow [k] } $({S_{2_k}}, m_{2_k})$ in which
             ${S_{1_k}} \neq \text{skip}, {S_{2_k}} \neq \text{skip}$.

\item $\mathcal{E}'\llbracket e\rrbracket\vals_{1_i} =
       \mathcal{E}'\llbracket e\rrbracket\vals_{2_i} = (0, v_\mathfrak{of})$

The execution from $({s_1}, m_{1_i}(\text{loop}_c^{{1_i}},\vals_{1_i}))$ proceeds as follows.

\begin{tabbing}
xx\=xx\=\kill
\>\>     $(s_1, m_{1_i}(\text{loop}_c^{1_i}, \vals_{1_i}))$\\
\>=\>    $(\text{while}_{\langle n_1\rangle} (e) \; \{S_1\}, {m_{1_i}(\text{loop}_c^{1_i}, \vals_{1_i})})$\\
\>$->$\> $(\text{while}_{\langle n_1\rangle} ((0, v_\mathfrak{of})) \; \{S_1\}, {m_{1_i}(\text{loop}_c^{1_i}, \vals_{1_i})})$ by rule  EEval' \\
\>$->$\> $(\text{while}_{\langle n_1\rangle} (0) \; \{S_1\}, {m_{1_i}(\text{loop}_c^{1_i}, \vals_{1_i})})$\\
\>\>     by rule  E-Oflow1 or E-Oflow2 \\
\>$->$\> $({\text{skip}}, m_{1_i}(\text{loop}_c^{1_i}[0/(n_1)], \vals_{1_i}))$ by the Wh-F2 rule.
\end{tabbing}

By the hypothesis IH, the loop counter of $s_1$ and $s_2$ in any configuration in executions
$(s_1, m_1) ->* (s_1, m_{1_i}(\text{loop}_c^{1_i}, \vals_{1_i}))$ and
$(s_2, m_2) ->* (s_2, m_{2_i}(\text{loop}_c^{2_i}, \vals_{2_i}))$
respectively are less than or equal to $i$,

 $\forall m_1'\,:\, ({s_1}, m_1) ->* (S_1', m_1'(\text{loop}_c^{1'})) ->* ({s_1}, m_{1_i}(\text{loop}_c^{1_i},$

 \noindent$\vals_{1_i})), \; \text{loop}_c^{1'}(n_1) \leq i$; and

 $\forall m_2'\,:\, ({s_2}, m_2) ->* (S_2', m_2'(\text{loop}_c^{2'})) ->* ({s_2}, m_{2_i}(\text{loop}_c^{2_i},$

 \noindent$\vals_{2_i})), \; \text{loop}_c^{2'}(n_2) \leq i$.

Therefore, $s_1$ and $s_2$ both terminate and the loop counter of $s_1$ and $s_2$ in any state in executions respectively are less than $i+1$.

\item $\mathcal{E}'\llbracket e\rrbracket\vals_{1_i} =
       \mathcal{E}'\llbracket e\rrbracket\vals_{2_i} = (v, v_\mathfrak{of})$
       where $v \notin \{0, \text{error}\}$;

The execution from $({s_1}, m_{1_i}(\text{loop}_c^{1_i}, \vals_{1_i}))$ proceeds as follows.

\begin{tabbing}
xx\=xx\=\kill
\>\>     $({s_1}, m_{1_i}(\text{loop}_c^{1_i}, \vals_{1_i}))$\\
\>= \>   $(\text{while}_{\langle n_1\rangle} (e) \; \{S_1\}, {m_{1_i}(\text{loop}_c^{1_i}, \vals_{1_i})})$\\
\>$->$\> $(\text{while}_{\langle n_1\rangle} ((v, v_\mathfrak{of})) \; \{S_1\}, {m_{1_i}(\text{loop}_c^{1_i}, \vals_{1_i})})$ by rule EEval'\\
\>$->$\> $(\text{while}_{\langle n_1\rangle} ((v, v_\mathfrak{of})) \; \{S_1\}, {m_{1_i}(\text{loop}_c^{1_i}, \vals_{1_i})})$ by rule EEval'\\
\>$->$\> $(\text{while}_{\langle n_1\rangle} (v) \; \{S_1\}, {m_{1_i}(\text{loop}_c^{1_i}, \vals_{1_i})})$ \\
\>\>     by rule E-Oflow1 or E-Oflow2\\
\>$->$\> $(S_1;\text{while}_{\langle n_1\rangle} (e) \; \{S_1\}, m_{1_i}(\text{loop}_c^{1_i}[(i+1)/(n_1)],$\\
\>\>     $\vals_{1_i}))$ by rule Wh-T.
\end{tabbing}

Similarly,
$({s_2}, m_{2_i}(\text{loop}_c^{2_i},\vals_{2_i}))$ {\kStepArrow [2] }
$(S_2;\text{while}_{\langle n_2\rangle} (e) \{S_2\}, \, \allowbreak m_{2_{i}}(\text{loop}_c^{2_{i}}[(i+1)/(n_2)], \vals_{2_{i}}))$.

By similar argument in base case, the executions of $S_1$ and $S_2$ terminate in the same way when started in states
$m_{1_i}(\text{loop}_c^{1_i}[(i+1)/(n_1)], \vals_{1_i})$ and
$m_{2_i}(\text{loop}_c^{2_i}[(i+1)/(n_2)], \vals_{2_i})$ respectively.
Then there are two possibilities.
\begin{enumerate}
\item $S_1$ and $S_2$ terminate when started in states
$m_{1_i}(\text{loop}_c^{1_i}[(i+1)/(n_1)], \vals_{1_i})$ and

\noindent$m_{2_i}(\text{loop}_c^{2_i}[(i+1)/(n_2)], \vals_{2_i})$ respectively

$(S_1;s_1, m_{1_i}(\text{loop}_c^{1_i}[(i+1)/(n_1)], \vals_{1_i})) ->*$

\noindent$(s_1, m_{1_{i+1}}(\mathfrak{f}_1, \text{loop}_c^{1_{i+1}}, \vals_{1_{i+1}}))$ and

$(S_2;s_1, m_{2_i}(\text{loop}_c^{2_i}[(i+1)/(n_2)], \vals_{2_i})) ->*$

\noindent$(s_2, m_{2_{i+1}}(\mathfrak{f}_2, \text{loop}_c^{2_{i+1}}, \vals_{2_{i+1}}))$ such that all of the following holds:
\begin{itemize}
\item $\mathfrak{f}_1 = \mathfrak{f}_2 = 0$; and

\item
$\text{loop}_c^{1_{i+1}}(n_1) = \text{loop}_c^{2_{i+1}}(n_2) = i+1$; and

\item $\forall y\in\text{TVar}(s),
\vals_{1_{i+1}}(y) = \vals_{2_{i+1}}(y)$, and

\item in any state in the execution

\noindent$(s_{1}, m_{1_i}) ->* (s_1, m_{1_{i+1}}(\text{loop}_c^{1_{i+1}}, \vals_{1_{i+1}}))$, the loop counter of $s_1$ is less than or equal to $i+1$.


\item in any state in the executions

\noindent$(s_2, m_{2_i}) ->* (s_2, m_{2_{i+1}}(\text{loop}_c^{2_{i+1}}, \vals_{2_{i+1}}))$, the loop counter of $s_2$ is less than or equal to $i+1$.

\end{itemize}
With the hypothesis IH, there are two configurations $(s_1, m_{1_{i+1}})$ and $(s_2, m_{2_{i+1}})$ reachable from $(s_1, m_1)$ and $(s_2, m_2)$, respectively, in which
crash flags are not set,
the loop counters of $s_1$ and $s_2$ are equal to $i+1$ and
value stores agree on the values of $\text{TVar}(s)$ and,
for every state in executions $(s_1, m_1) ->* (s_1, m_{1_{i+1}})$ and $(s_2, m_2) ->* (s_2, m_{2_{i+1}})$ the loop counters for $s_1$ and $s_2$ are less than or equal to $i+1$ respectively:

 $\exists (s_1, m_{1_{i+1}})\, (s_2, m_{2_{i+1}})\,:\,$

$(s_1, m_1) ->* (s_1, m_{1_{i+1}}(\mathfrak{f}_1, \text{loop}_c^{1_{i+1}}, \vals_{1_{i+1}})) \wedge
 (s_2, m_2) ->* (s_2, m_{2_{i+1}}(\mathfrak{f}_2, \text{loop}_c^{2_{i+1}}, \vals_{2_{i+1}}))$ where
\begin{itemize}
 \item $\mathfrak{f}_1 = \mathfrak{f}_2 = 0$; and

 \item $\text{loop}_c^{1_{i+1}}(n_1) = \text{loop}_c^{2_{i+1}}(n_2) = i+1$; and

 \item $\forall x \in \text{TVar}(s),
 \vals_{1_{i+1}}(x) = \vals_{2_{i+1}}(x)$; and

 \item $\forall m_1'\,:\,
 ({s_1}, m_1) ->* (S_1', m_1'(\text{loop}_c^{1'})) ->*$

 \noindent$({s_1}, m_{1_{i+1}}(\text{loop}_c^{1_{i+1}}, \vals_{1_{i+1}})), \; \text{loop}_c^{1'}(n_1) \leq i+1$; and

 \item $\forall m_2'\,:\,
 ({s_2}, m_2) ->* (S_2', m_2'(\text{loop}_c^{2'})) ->*$

 \noindent$({s_2}, m_{2_{i+1}}(\text{loop}_c^{2_{i+1}}, \vals_{2_{i+1}})), \; \text{loop}_c^{2'}(n_2) \leq i+1$.
\end{itemize}

\item $S_1$ and $S_2$ do not terminate when started in states

$m_{1_i}(\text{loop}_c^{{1_i}}[(i+1)/(n_1)], \vals_{{1_i}})$ and
$m_{2_i}(\text{loop}_c^{2_i}[(i+1)/(n_2)], \vals_{2_i})$ respectively:

    $\forall k>0,
    (S_1, m_{1_i}(\text{loop}_c^{1_i}[(i+1)/(n_1)], \vals_{1_i}))$ {\kStepArrow [k] }
    $(S_{1_k}, m_{1_{1_k}}(\text{loop}_c^{1_{1_k}}, \vals_{1_{1_k}}))$ and

    $(S_2, m_{2_i}(\text{loop}_c^{2_i}[(i+1)/(n_2)], \vals_{2_i}))$ {\kStepArrow [k] }

    \noindent$(S_{2_k}, m_{2_{1_k}}(\text{loop}_c^{2_{1_k}}, \vals_{2_{1_k}}))$
    in which
    $S_{1_k} \neq \text{skip}, S_{2_k} \neq \text{skip}$.

  By our assumption of unique loop labels, $s_1\notin S_1$.
  Then, $\forall k>0, \text{loop}_c^{1_{1_k}}(n_1) =$

  \noindent$\text{loop}_c^{1_i}[(i+1)/(n_1)](n_1) = i+1$.
  Similarly, $\forall k>0, \text{loop}_c^{2_{1_k}}(n_2) =$

  \noindent$\text{loop}_c^{2_i}[(i+1)/(n_2)](n_2) = i+1$.
  In addition, by Lemma~\ref{lmm:multiStepSeqExec},

  $\forall k>0,
   (S_1;s_1, m_{1_i}(\text{loop}_c^{1_i}[(i+1)/(n_1)], \vals_{1_i}))$ {\kStepArrow [k] }
  $(S_{1_k};s_1, m_{1_{1_k}}(\text{loop}_c^{1_{1_k}}, \vals_{1_{1_k}}))$ and
  $(S_2;s_2, m_{2}(\text{loop}_c^{2}[(i+1)/(n_2)], \vals_{2}))$ {\kStepArrow [k] }
  $(S_{2_k};s_2, m_{2_{1_k}}(\text{loop}_c^{2_{1_k}}, \vals_{2_{1_k}}))$
    in which
    $S_{1_k} \neq \text{skip}, S_{2_k} \neq \text{skip}$.

  In summary, the loop counter of $s_1$ and $s_2$ are less than equal to $i+1$, and $s_1$ and $s_2$ do not terminate
  such that
there are no configurations $(s_1, m_{1_{i+1}})$ and $(s_2, m_{2_{i+1}})$ reachable from $(s_1, m_1)$ and $(s_2, m_2)$, respectively, in which
crash flags are not set,
the loop counters of $s_1$ and $s_2$ are equal to $i+1$ and
value stores agree on values of variables in $\text{TVar}(s)$.
\end{enumerate}
\end{enumerate}
\end{enumerate}
\end{proof}

\begin{corollary}\label{coro:loopTermInSameWay}
Let $s_1 = ``\text{while}_{\langle n_1\rangle}(e) \{S_1\}"$  and

    \noindent$s_2 = ``\text{while}_{\langle n_2\rangle}(e) \{S_2\}"$ be two while statements respectively, with the same set of the termination deciding variables,
     $\text{TVar}(s_1) = \text{TVar}(s_2) = \text{TVar}(s)$,
    whose bodies $S_1$ and $S_2$ satisfy the proof rule of equivalently computation of variables in $\text{TVar}(s)$,
    $\forall x \in \text{TVar}(s) \, :\, (S_1) \equiv_{x}^S (S_2)$,
    and whose bodies $S_1$ and $S_2$ terminate in the same way when started in states with crash flags not set and agreeing on values of variables in $\text{TVar}(S_1) \cup \text{TVar}(S_2)$:

   \noindent$\forall m_{S_1}(\mathfrak{f}_{S_1}, \vals_{S_1}),
            m_{S_2}(\mathfrak{f}_{S_2}, \vals_{S_2}):$

\noindent$(((\forall z\in \text{TVar}(S_1) \cup \text{TVar}(S_2)),
              \vals_{S_1}(z) = \vals_{S_2}(z)) \wedge
  (\mathfrak{f}_{S_1} = \mathfrak{f}_{S_2} = 0)) => $
\noindent$(S_1, m_{S_1}(\mathfrak{f}_{S_1}, \vals_{S_1})) \equiv_{H}
 (S_2, m_{S_2}(\mathfrak{f}_{S_2}, \vals_{S_2}))$.

If $s_1$ and $s_2$ start in the state $m_1(\mathfrak{f}_1, \text{loop}_c^1, \vals_1)$ and $m_2(\mathfrak{f}_2, \text{loop}_c^2,\allowbreak , \vals_2)$ respectively in which
  crash flags are not set, $\mathfrak{f}_1 =$

  \noindent$\mathfrak{f}_2 = 0$,
  $s_1$ and $s_2$ have not already executed, $\text{loop}_c^1(n_1) = \text{loop}_c^2(n_2) = 0$,
  value stores $\vals_1$ and $\vals_2$ agree on values of variables in $\text{TVar}(s)$,
  $\forall x \in \text{TVar}(s),
  \vals_1(x) = \vals_2(x)$,
  then $s_1$ and $s_2$ terminate in the same way:
\begin{enumerate}
\item $s_1$ and $s_2$ both terminate,
$({s_1}, m_1) ->* (\text{skip}, m_1')$,
$({s_2}, m_2) ->* (\text{skip}, m_2')$.

\item $s_1$ and $s_2$ both do not terminate, $\forall k> 0$,
$({s_1}, m_1)$ {\kStepArrow [k] } $({S_{1_k}}, m_{1_k})$,
$({s_2}, m_2)$ {\kStepArrow [k] } $({S_{2_k}}, m_{2_k})$
where $S_{1_k} \neq \text{skip}$, $S_{2_k} \neq \text{skip}$.
\end{enumerate}
\end{corollary}
This is from Lemma~\ref{lmm:loopTermInSameWayIntermediate} immediately.


Lemma\ref{lmm:loopTermInSameWayWithLoopCnt} is necessary only for showing the same I/O sequence in the next section.

\begin{lemma}\label{lmm:loopTermInSameWayWithLoopCnt}
Let $s_1 = ``\text{while}_{\langle n_1\rangle}(e) \{S_1\}"$ and

    \noindent$s_2 = ``\text{while}_{\langle n_2\rangle}(e) \{S_2\}"$ be two while statements in program $P_1$ and $P_2$ respectively with the same set of termination deciding variables,
    $\text{TVar}(s_1) = \text{TVar}(s_2) = \text{TVar}(s)$,
    whose bodies $S_1$ and $S_2$ satisfy the proof rule of equivalently computation of variables in $\text{TVar}(s)$,
    $\forall x \in \text{TVar}(s) \, :\, S_1 \equiv_{x}^S S_2$
    and whose bodies $S_1$ and $S_2$ terminate in the same way in executions when started in states with crash flags not set and agreeing on values of variables in $\text{TVar}(S_1) \cup \text{TVar}(S_2)$:

\noindent$\forall m_{S_1}(\mathfrak{f}_{S_1}, \vals_{S_1})\,
                  m_{S_2}(\mathfrak{f}_{S_2}, \vals_{S_2}):$

\noindent$(((\forall z\in \text{TVar}(S_1) \cup \text{TVar}(S_2)),
  \vals_{S_1}(z) = \vals_{S_2}(z)) \wedge
  (\mathfrak{f}_{S_1} = \mathfrak{f}_{S_2} = 0)$

\noindent$=> (S_1, m_{S_1}(\mathfrak{f}_{S_1}, \vals_{S_1})) \equiv_{H}
 (S_2, m_{S_2}(\mathfrak{f}_{S_2}, \vals_{S_2}))$.

If $s_1$ and $s_2$ start in the state $m_1(\mathfrak{f}_1, \text{loop}_c^1, \vals_1)$ and $m_2(\mathfrak{f}_2, \text{loop}_c^2,\allowbreak , \vals_2)$ respectively in which
  crash flags are not set,
  $\mathfrak{f}_1 = \mathfrak{f}_2 = 0$,
  $s_1$ and $s_2$ have not already executed,
  $\text{loop}_c^1(n_1) = \text{loop}_c^2(n_2) = 0$,
  value stores $\vals_1$ and $\vals_2$ agree on values of variables in $\text{TVar}(s)$,
  $\forall x \in \text{TVar}(s),
  \vals_1(x) = \vals_2(x)$,
  one of the following holds:
\begin{enumerate}
\item $s_1$ and $s_2$ both terminate and the loop counters of $s_1$ and $s_2$ are less than a positive integer $i$ and less than or equal to $i-1$:
$({s_1}, m_1) ->* (\text{skip}, m_1')$,
$({s_2}, m_2) ->* (\text{skip}, m_2')$ where both of the following hold:
\begin{itemize}

\item The loop counters of $s_1$ and $s_2$ are less than a positive integer $i$:

$\exists i>0\, \forall m_{1}'\, m_{2}'\,:\,$

\noindent$({s_1}, m_1) ->* ({S_1'}, m_{1}'(\text{loop}_c^{1'})), \, \text{loop}_c^{1'}(n_1) < i$ and

\noindent$({s_2}, m_2) ->* ({S_2'}, m_{2}'(\text{loop}_c^{2'})), \, \text{loop}_c^{2'}(n_2) < i$.

\item $\forall 0 < j < i$, there are two configurations $(s_1, m_{1_j})$ and $(s_2, m_{2_j})$ reachable from $(s_1, m_1)$ and $(s_2, m_2)$, respectively, in which both crash flags are not set, the loop counters of $s_1$ and $s_2$ are equal to $j$ and value stores agree on the values of variables in $\text{TVar}(s)$, and for every state in execution $(s_1, m_1) ->* (s_1, m_{1_j})$ or $(s_2, m_2) ->* (s_2, m_{2_j})$, the loop counters for $s_1$ and $s_2$ are less than or equal to $j$ respectively:

 $\exists (s_1, m_{1_j})\, (s_2, m_{2_j})\,:\,$

 \noindent$({s_1}, m_1) ->* ({s_1}, m_{1_j}(\mathfrak{f}_1, \text{loop}_c^{1_j}, \vals_{1_j})) \wedge$

 \noindent$({s_2}, m_2) ->* ({s_2}, m_{2_j}(\mathfrak{f}_2, \text{loop}_c^{2_j}, \vals_{2_j}))$ where
\begin{itemize}
 \item $\mathfrak{f}_1 = \mathfrak{f}_2 = 0$; and

 \item $\text{loop}_c^{1_j}(n_1) = \text{loop}_c^{2_j}(n_2) = j$; and

 \item $\forall x \in \text{TVar}(s)\,:\,
 \vals_{1_j}(x) = \vals_{2_j}(x)$; and

 \item $\forall m_1'\,:\, ({s_1}, m_1) ->* (S_1', m_1'(\text{loop}_c^{1'})) ->* ({s_1}, m_{1_j}), \;$

     \noindent$\text{loop}_c^{1'}(n_1) \leq j$; and

 \item $\forall m_2'\,:\, ({s_2}, m_2) ->* (S_2', m_2'(\text{loop}_c^{2'})) ->* ({s_2}, m_{2_j}), \;$

     \noindent$\text{loop}_c^{2'}(n_2) \leq j$.
\end{itemize}

\end{itemize}

\item $s_1$ and $s_2$ both do not terminate, $\forall k> 0$,
$({s_1}, m_1)$ {\kStepArrow [k] } $({S_{1_k}}, m_{1_k})$,
$({s_2}, m_2)$ {\kStepArrow [k] } $({S_{2_k}}, m_{2_k})$
where $S_{1_k} \neq \text{skip}$, $S_{2_k} \neq \text{skip}$ such that one of the following holds:

\begin{enumerate}
\item For any positive integer $i$, there are two configurations $(s_1, m_{1_i})$ and $(s_2, m_{2_i})$ reachable from $(s_1, m_1)$ and $(s_2, m_2)$, respectively, in which both crash flags are not set, the loop counters of $s_1$ and $s_2$ are equal to $i$ and value stores agree on the values of variables in $\text{TVar}(s)$, and for every state in execution $(s_1, m_1) ->* (s_1, m_{1_i})$ or $(s_2, m_2) ->* (s_2, m_{2_i})$, the loop counters for $s_1$ and $s_2$ are less than or equal to $i$ respectively:

 $\forall i>0 \, \exists (s_1, m_{1_i})\, (s_2, m_{2_i})\,:\,$

 \noindent$({s_1}, m_1) ->* ({s_1}, m_{1_i}(\mathfrak{f}_1, \text{loop}_c^{1_i}, \vals_{1_i})) \wedge$

 \noindent$({s_2}, m_2) ->* ({s_2}, m_{2_i}(\mathfrak{f}_2, \text{loop}_c^{2_i}, \vals_{2_i}))$ where
\begin{itemize}
 \item $\mathfrak{f}_1 = \mathfrak{f}_2 = 0$; and

 \item $\text{loop}_c^{1_i}(n_1) = \text{loop}_c^{2_i}(n_2) = i$; and

 \item $\forall x \in \text{TVar}(s)\,:\,
  \vals_{1_i}(x) = \vals_{2_i}(x)$; and

 \item $\forall m_1'\,:\,
 ({s_1}, m_1) ->* (S_1', m_1'(\text{loop}_c^{1'})) ->* ({s_1}, m_{1_i}), \;$

  \noindent$\text{loop}_c^{1'}(n_1) \leq i$; and

 \item $\forall m_2'\,:\, ({s_2}, m_2) ->* (S_2', m_2'(\text{loop}_c^{2'})) ->* ({s_2}, m_{2_i}), \;$

\noindent$\text{loop}_c^{2'}(n_2) \leq i$;
\end{itemize}

\item The loop counters for $s_1$ and $s_2$ are less than a positive integer $i$ and less than or equal to $i-1$ such that all of the following hold:

\begin{itemize}
\item $\exists i>0, \forall m_1'\,m_2' \,:\,
({s_1}, m_1) ->* (S_{1}, m_1'(\text{loop}_c^{1'})),$

\noindent$({s_2}, m_2) ->* (S_{2}, m_2'(\text{loop}_c^{2'}))$ where
      $\text{loop}_c^{1'}(n_1) < i$,
      $\text{loop}_c^{2'}(n_2) < i$;

\item $\forall 0 < j < i$, there are two configurations $(s_1, m_{1_j})$ and $(s_2, m_{2_j})$ reachable from $(s_1, m_1)$ and $(s_2, m_2)$, respectively, in which both crash flags are not set, the loop counters of $s_1$ and $s_2$ are equal to $j$ and value stores agree on the values of variables in $\text{TVar}(s)$, and for every state in execution $(s_1, m_1) ->* (s_1, m_{1_j})$ or $(s_2, m_2) ->* (s_2, m_{2_j})$, the loop counters for $s_1$ and $s_2$ are less than or equal to $j$ respectively:

 $\exists (s_1, m_{1_j})\, (s_2, m_{2_j})\,:\,$

 \noindent$({s_1}, m_1) ->* ({s_1}, m_{1_j}(\mathfrak{f}_1, \text{loop}_c^{1_j}, \vals_{1_j})) \wedge$

 \noindent$({s_2}, m_2) ->* ({s_2}, m_{2_j}(\mathfrak{f}_2, \text{loop}_c^{2_j}, \vals_{2_j}))$ where
\begin{itemize}
 \item $\mathfrak{f}_1 = \mathfrak{f}_2 = 0$; and

 \item $\text{loop}_c^{1_j}(n_1) = \text{loop}_c^{2_j}(n_2) = j$; and

 \item $\forall x \in \text{TVar}(s)\,:\,
   \vals_{1_j}(x) = \vals_{2_j}(x)$; and

 \item $\forall m_1'\,:\,
 ({s_1}, m_1) ->* (S_1', m_1'(\text{loop}_c^{1'})) ->* ({s_1}, m_{1_j}), \;$

  \noindent$\text{loop}_c^{1'}(n_1) \leq j$; and

 \item $\forall m_2'\,:\, ({s_2}, m_2) ->* (S_2', m_2'(\text{loop}_c^{2'})) ->* ({s_2}, m_{2_j}), \;$

  \noindent$\text{loop}_c^{2'}(n_2) \leq j$;
\end{itemize}
\end{itemize}

\item The loop counters for $s_1$ and $s_2$ are less than or equal to some positive integer $i$ such that all of the following hold:

\begin{itemize}
\item $\exists i>0\,\forall m_1'\,m_2' \,:\,
({s_1}, m_1) ->* (S_{1}, m_1'(\text{loop}_c^{1'})),$

\noindent$({s_2}, m_2) ->* (S_{2}, m_2'(\text{loop}_c^{2'}))$ where
      $\text{loop}_c^{1'}(n_1) \leq i$,
      $\text{loop}_c^{2'}(n_2) \leq i$;

\item $\forall 0 < j < i$, there are two configurations $(s_1, m_{1_j})$ and $(s_2, m_{2_j})$ reachable from $(s_1, m_1)$ and $(s_2, m_2)$, respectively, in which both crash flags are not set, the loop counters of $s_1$ and $s_2$ are equal to $j$ and value stores agree on the values of variables in $\text{TVar}(s)$, and for every state in execution $(s_1, m_1) ->* (s_1, m_{1_j})$ or $(s_2, m_2) ->* (s_2, m_{2_j})$, the loop counters for $s_1$ and $s_2$ are less than or equal to $j$ respectively:

 $\exists (s_1, m_{1_j})\, (s_2, m_{2_j})\,:\,$

 \noindent$({s_1}, m_1) ->* ({s_1}, m_{1_j}(\mathfrak{f}_1, \text{loop}_c^{1_j}, \vals_{1_j})) \wedge$

 \noindent$({s_2}, m_2) ->* ({s_2}, m_{2_j}(\mathfrak{f}_2, \text{loop}_c^{2_j}, \vals_{2_j}))$ where
\begin{itemize}
 \item $\mathfrak{f}_1 = \mathfrak{f}_2 = 0$; and

 \item $\text{loop}_c^{1_j}(n_1) = \text{loop}_c^{2_j}(n_2) = j$; and

 \item $\forall x \in \text{TVar}(s)\,:\,
 \vals_{1_j}(x) = \vals_{2_j}(x)$; and

 \item $\forall m_1'\,:\, ({s_1}, m_1) ->* (S_1', m_1'(\text{loop}_c^{1'})) ->* ({s_1}, m_{1_j}), \;$

     \noindent$\text{loop}_c^{1'}(n_1) \leq j$; and

 \item $\forall m_2'\,:\, ({s_2}, m_2) ->* (S_2', m_2'(\text{loop}_c^{2'})) ->* ({s_2}, m_{2_j}), \;$

     \noindent$\text{loop}_c^{2'}(n_2) \leq j$;
\end{itemize}

\item There are no configurations $(s_1, m_{1_i})$ and $(s_2, m_{2_i})$ reachable from $(s_1, m_1)$ and $(s_2, m_2)$, respectively, in which  crash flags are not set, the loop counters of $s_1$ and $s_2$ are equal to $i$, and value stores agree on the values of variables in $\text{TVar}(s)$:

$\nexists (s_1, m_{1_i})\, (s_2, m_{2_i})\,:\,$

\noindent$({s_1}, m_1) ->* ({s_1}, m_{1_i}(\mathfrak{f}_1, \text{loop}_c^{1_i}, \vals_{1_i})) \wedge$

\noindent$({s_2}, m_2) ->* ({s_2}, m_{2_i}(\mathfrak{f}_2, \text{loop}_c^{2_i}, \vals_{2_i}))$ where
      \begin{itemize}
        \item $\mathfrak{f}_1 = \mathfrak{f}_2 = 0$; and

        \item $\text{loop}_c^{1_i}(n_1) = \text{loop}_c^{2_i}(n_2) = i$; and

        \item $\forall x \in \text{TVar}(s)\,:\,
        \vals_{1_i}(x) = \vals_{2_i}(x)$.
      \end{itemize}
\end{itemize}
\end{enumerate}
\end{enumerate}
\end{lemma}
\begin{proof}
From lemma~\ref{lmm:loopTermInSameWayIntermediate}, we have $s_1$ and $s_2$ terminate in the same way when started in states $m_1$ and $m_2$ respectively. Then there are two big cases.
\begin{enumerate}
\item $s_1$ and $s_2$ both terminate.

Let $i$ be the smallest integer such that the loop counters of $s_1$ and $s_2$ are less than $i$ in the executions.
Then there are two possibilities.
\begin{enumerate}

\item $i=1$.

In the proof of Lemma~\ref{lmm:loopTermInSameWayIntermediate}, the evaluation of the loop predicate of $s_1$ and $s_2$ produce zero w.r.t value stores,
$\vals_1$ and $\vals_2$.
Then the execution of $s_1$ proceeds as follows:

\begin{tabbing}
xx\=xx\=\kill
\>\>     $({s_1}, m_{1}(\text{loop}_c^1, \vals_1))$\\
\>= \>   $(\text{while}_{\langle n_1\rangle} (e) \; \{S_1\}, {m_{1}(\text{loop}_c^{1})})$\\
\>$->$\> $(\text{while}_{\langle n_1\rangle} ((0, v_\mathfrak{of})) \; \{S_1\}, {m_{1}(\text{loop}_c^{1})})$ by rule EEval' \\
\>$->$\> $(\text{while}_{\langle n_1\rangle} (0) \; \{S_1\}, {m_{1}(\text{loop}_c^{1})})$\\
\>\>      by rule E-Oflow1 or E-Oflow2 \\
\>$->$\> $({\text{skip}}, m_1(\text{loop}_c^{1}\setminus \{(n_1, *)\}))$\\
\>\>      by rule Wh-F1 or Wh-F2.
\end{tabbing}

Similarly, $({s_2}, m_{2}(\text{loop}_c^{2}, \vals_{2}))$ {\kStepArrow [2] }
$({\text{skip}}, m_2(\text{loop}_c^{2}\setminus \{(n_2,\allowbreak , *)\}))$.
Then the lemma holds because of the initial configuration
$(s_1, m_1(\text{loop}_c^1, \vals_1))$ and $(s_2, m_2(\text{loop}_c^2, \vals_2))$.

\item $i>1$.

Because $s_1$ and $s_2$ terminate, $i$ is the smallest positive integer such that the loop counters of $s_1$ and $s_2$ are less than $i$, by Lemma~\ref{lmm:loopTermInSameWayIntermediate}, $\forall 0<j<i$,
there are two configurations $(s_1, m_{1_j})$ and $(s_2, m_{2_j})$ reachable from $(s_1, m_1)$ and $(s_2, m_2)$, respectively, in which both crash flags are not set, the loop counters of $s_1$ and $s_2$ are equal to $j$ and value stores agree on the values of variables in $\text{TVar}(s)$, and for every state in execution, $(s_1, m_1) ->* (s_1, m_{1_j})$ or $(s_2, m_2) ->* (s_2, m_{2_j})$ the loop counters for $s_1$ and $s_2$ are less than or equal to $j$ respectively.
With the initial configuration $(s_1, m_1)$ and $(s_2, m_2)$, the lemma holds.
\end{enumerate}

\item $s_1$ and $s_2$ both do not terminate. There are three possibilities.
\begin{enumerate}
\item $\forall i>0$, there are two configurations $(s_1, m_{1_i})$ and $(s_2, m_{2_i})$ reachable from $(s_1, m_1)$ and $(s_2, m_2)$, respectively, in which both crash flags are not set, the loop counters of $s_1$ and $s_2$ are equal to $i$ and value stores agree on the values of variables in $\text{TVar}(s)$, and for every state in execution, $(s_1, m_1) ->* (s_1, m_{1_i})$ or $(s_2, m_2) ->* (s_2, m_{2_i})$ the loop counters for $s_1$ and $s_2$ are less than or equal to $i$ respectively.

\item The loop counters of $s_1$ and $s_2$ are less than a positive integer $i$.

Let $i$ be the smallest positive integer such that there is no positive integer $j<i$ that the loop counters of $s_1$ and $s_2$ are less than the positive integer $j$.
This case occurs when $s_1$ and $s_2$ finish the full $(i-1)$th iterations and both executions raise an exception in the evaluation of the loop predicate of $s_1$ and $s_2$ for the $i$th time. There are further two possibilities.
\begin{enumerate}
\item $i=1$.

In proof of Lemma~\ref{lmm:loopTermInSameWayIntermediate}, evaluations of the predicate expression of $s_1$ and $s_2$ raise an exception w.r.t value stores $\vals_1$ and $\vals_2$. The lemma holds.

\item $i>1$.
By the assumption of initial states $(s_1, m_1)$ and $(s_2, m_2)$, when $j=0$, initial states $(s_1, m_1)$ and $(s_2, m_2)$ have crash flags not set, the loop counters of $s_1$ and $s_2$ are zero and value stores agree on values of variables of $s_1$ and $s_2$.

By Lemma~\ref{lmm:loopTermInSameWayIntermediate}, $\forall 0<j<i$,
there are two configurations $(s_1, m_{1_j})$ and $(s_2, m_{2_j})$ reachable from $(s_1, m_1)$ and $(s_2, m_2)$, respectively, in which both crash flags are not set, the loop counters of $s_1$ and $s_2$ are equal to $j$ and value stores agree on the values of variables in $\text{TVar}(s)$, and for every state in execution, $(s_1, m_1) ->* (s_1, m_{1_j})$ or $(s_2, m_2) ->* (s_2, m_{2_j})$ the loop counters for $s_1$ and $s_2$ are less than or equal to $j$ respectively.
With the initial configuration $(s_1, m_1)$ and $(s_2, m_2)$, the lemma holds.

\end{enumerate}

\item The loop counters of $s_1$ and $s_2$ are less than or equal to a positive integer $i$.

Let $i$ be the smallest positive integer such that the loop counters of $s_1$ and $s_2$ are less than or equal to the positive integer $i$.
There are two possibilities.
\begin{enumerate}
\item $i=1$.

In the proof of Lemma~\ref{lmm:loopTermInSameWayIntermediate}, the execution of $s_1$ proceeds as follows:
\begin{tabbing}
xx\=xx\=\kill
\>\>     $({s_1}, m_{1}(\text{loop}_c^{1}, \vals_{1}))$\\
\>= \>   $(\text{while}_{\langle n_1\rangle} (e) \; \{S_1\}, {m_{1}(\text{loop}_c^{1}, \vals_{1})})$\\
\>$->$\> $(\text{while}_{\langle n_1\rangle} ((v, v_\mathfrak{of})) \; \{S_1\}, {m_{1}(\text{loop}_c^{1}, \vals_{1})})$ where $v\neq 0$\\
\>\>      by rule EEval' \\
\>$->$\> $(\text{while}_{\langle n_1\rangle} (v) \; \{S_1\}, {m_{1}(\text{loop}_c^{1}, \vals_{1})})$\\
\>\>     by rule E-Oflow1 or E-Oflow2 \\
\>$->$\> $(S_1;\text{while}_{\langle n_1\rangle} (e) \; \{S_1\}, m_{1}(\text{loop}_c^{1}[1/(n_1)], \vals_{1}))$\\
\>\>      by rule Wh-T .
\end{tabbing}
The execution of $s_2$ proceeds to

\noindent$(S_2;\text{while}_{\langle n_2\rangle} (e) \; \{S_2\}, m_{2}(\text{loop}_c^{2}[1/(n_2)], \vals_{2}))$.
In addition, executions of $S_1$ and $S_2$ do not terminate when started in states
$m_{1}(\text{loop}_c^{1}[1/(n_1)], \vals_{1})$ and

\noindent$m_{2}(\text{loop}_c^{2}[1/(n_2)], \vals_{2})$.
By Lemma~\ref{lmm:multiStepSeqExec}, executions of $s_1$ and $s_2$ do not terminate.

\item $i>1$.

In the proof of Lemma~\ref{lmm:loopTermInSameWayIntermediate}, when started in the state
$({s_1}, m_{1_{i-1}}(\mathfrak{f}_1, \text{loop}_c^{1_{i-1}}, \vals_{1_{i-1}}))$
the execution of $s_1$ proceeds as follows:
\begin{tabbing}
xx\=xx\=\kill
\>\>     $({s_1}, m_{1_{i-1}}(\mathfrak{f}_1, \text{loop}_c^{1_{i-1}}, \vals_{1_{i-1}}))$\\
\>= \>   $(\text{while}_{\langle n_1\rangle} (e) \; \{S_1\}, m_{1_{i-1}}(\mathfrak{f}_1, \text{loop}_c^{1_{i-1}}, \vals_{1_{i-1}}))$\\
\>$->$\> $(\text{while}_{\langle n_1\rangle} ((v, v_\mathfrak{of})) \; \{S_1\}, m_{1_{i-1}}(\mathfrak{f}_1, \text{loop}_c^{1_{i-1}}, \vals_{1_{i-1}}))$\\
\>\>     by rule EEval' \\
\>$->$\> $(\text{while}_{\langle n_1\rangle} (v) \; \{S_1\}, m_{1_{i-1}}(\mathfrak{f}_1, \text{loop}_c^{1_{i-1}}, \vals_{1_{i-1}}))$\\
\>\>      by rule E-Oflow1 or E-Oflow2 \\
\>$->$\> $(S_1;\text{while}_{\langle n_1\rangle} (e) \; \{S_1\},$\\
\>\>     $m_{1_{i-1}}(\mathfrak{f}_1, \text{loop}_c^{1_{i-1}}[i/(n_1)], \vals_{1_{i-1}}))$\\
\>\>     by rule Wh-T1 or Wh-T2.
\end{tabbing}
The execution of $s_2$ proceeds to
$(S_2;\text{while}_{\langle n_2\rangle} (e) \; \{S_2\}, \allowbreak m_{2_{i-1}}(\mathfrak{f}_2, \text{loop}_c^{2_{i-1}}[i/(n_2)], \vals_{2_{i-1}}))$.
In addition, executions of $S_1$ and $S_2$ do not terminate when started in states
$m_{1_{i-1}}(\mathfrak{f}_1, \text{loop}_c^{1_{i-1}}[i/(n_1)], \vals_{1_{i-1}})$ and

\noindent$m_{2_{i-1}}(\mathfrak{f}_2, \text{loop}_c^{2_{i-1}}[i/(n_2)], \vals_{2_{i-1}})$.
By Lemma~\ref{lmm:multiStepSeqExec}, executions of $s_1$ and $s_2$ do not terminate.
\end{enumerate}
\end{enumerate}
\end{enumerate}
\end{proof}

\subsection{Behavioral equivalence}\label{sec:sameOutSeq}
We now propose a proof rule under which two programs produce the same {\it output sequence}, namely the same I/O sequence till any $i$th output value.
We care about the I/O sequence due to the possible crash from the lack of input.
We start by giving the definition of the same output sequence, then we describe the proof rule under which two programs produce the same output sequence, finally we
show that our proof rule ensures same output together with the necessary auxiliary lemmas.
We use the notation $``\text{Out}(\vals)"$ to represent the output sequence in value store $\vals$, the I/O sequence $\vals({id}_{IO})$ till the rightmost output value. Particularly, when there is no output value in the I/O sequence $\vals({id}_{IO})$, $\text{Out}(\vals) = \varnothing$.

\begin{definition} {\bf (Same output sequence)}
Two statement sequences $S_1$ and $S_2$ produce the same output sequence when started in states $m_1$ and $m_2$ respectively, written $(S_1, m_1)  \equiv_{O} (S_2, m_2)$, iff
$\forall m_1'\, m_2'$ such that
$({S_1}, m_1) ->* ({S_1'}, m_1'(\vals_{1}'))$ and
$({S_2}, m_2) ->* ({S_2'}, m_2'(\vals_{2}'))$,
there are states $m_1''\, m_2''$ reachable from initial states $m_1$ and $m_2$,
$({S_1}, m_1) ->* ({S_1''}, m_1''(\vals_{1}''))$ and
$({S_2}, m_2) ->* ({S_2''}, m_2''(\vals_{2}''))$ so that
$\text{Out}(\vals_{2}'') = \text{Out}(\vals_{1}') \,\text{and}\, \text{Out}(\vals_{1}'') = \text{Out}(\vals_{2}')$.
\end{definition}

\subsubsection{Proof rule for behavioral equivalence}
We show the proof rules of the behavioral equivalence.
The output sequence produced in executions of a statement sequence $S$ depends on values of a set of variables in the program, the output deciding variables $\text{OVar}(S)$.
The output deciding variables are of two parts: $\text{TVar}_o(S)$ are variables affecting the termination of executions of a statement sequence;  $\text{Imp}_o(S)$ are variables affecting values of the I/O sequence produced in executions of a statement sequence.
The definitions of $\text{TVar}_o(S)$ and $\text{Imp}_o(S)$ are shown in Definition~\ref{def:ImpVarO} and~\ref{def:TVarO}.

\begin{definition}\label{def:ImpVarO}
{\bf (Imported variables relative to output)} The imported variables in one program $S$ relative to output, written $\text{Imp}_o(S)$, are listed as follows:

\begin{enumerate}
\item $\text{Imp}_o(S)  = \{{id}_{IO}\}$, if $(\forall e\,:\, ``\text{output }e" \notin S)$;

\item $\text{Imp}_o(``\text{output }e") = \{{id}_{IO}\} \cup \text{Use}(e)$;

\item $\text{Imp}_o(``\text{If }(e) \text{ then }\{S_t\}\text{ else }\{S_f\}") =
\text{Use}(e) \cup
\text{Imp}_o(S_t) \cup
\text{Imp}_o(S_f)$ if $(\exists e\,:\, ``\text{output }e" \in S)$;

\item $\text{Imp}_o(``\text{while}_{\langle n\rangle}(e)\{S''\}") = \text{Imp}(``\text{while}_{\langle n\rangle}(e)\{S''\}", \{{id}_{IO}\})$ if $(\exists e\,:\, ``\text{output }e" \in S'')$;

\item For $k>0$, $\text{Imp}_o(s_1;...;s_k;s_{k+1}) =
\text{Imp}(s_1;...;s_k, \text{Imp}_o(s_{k+1}))$ if $(\exists e\,:\, ``\text{output }e" \in s_{k+1})$;

\item For $k>0$, $\text{Imp}_o(s_1;...;s_k;s_{k+1}) = \text{Imp}_o(s_1;...;s_k)$ if $(\forall e\,:\, ``\text{output }e" \notin s_{k+1})$;
\end{enumerate}
\end{definition}

\begin{definition}\label{def:TVarO}
{\bf (Termination deciding variables relative to output)} The termination deciding variables in a statement sequence $S$ relative to output, written $\text{TVar}_o(S)$, are listed as follows:
\begin{enumerate}
\item $\text{TVar}_o(S) = \emptyset$ if $(\forall e\,:\, ``\text{output }e" \notin S)$;

\item $\text{TVar}_o(``\text{output }e") = \text{Err}(e)$;

\item $\text{TVar}_o(``\text{If }(e) \text{ then }\{S_t\}\text{ else }\{S_f\}") =
\text{Use}(e) \cup
\text{TVar}_o(S_t) \cup
\text{TVar}_o(S_f)$ if $(\exists e\,:\, ``\text{output }e" \in S)$;

\item $\text{TVar}_o(``\text{while}_{\langle n\rangle}(e)\{S''\}") = \text{TVar}(``\text{while}_{\langle n\rangle}(e)\{S''\}")$ if $(\exists e\,:\, ``\text{output }e" \in S'')$;

\item For $k>0$, $\text{TVar}_o(s_1;...;s_k;s_{k+1}) =
\text{TVar}(s_1;...;s_k)$

\noindent$\cup \text{Imp}(s_1;...;s_k, \text{TVar}_o(s_{k+1}))$ if $(\exists e\,:\, ``\text{output }e" \in s_{k+1})$;

\item For $k>0$, $\text{TVar}_o(s_1;...;s_k;s_{k+1}) = \text{TVar}_o(s_1;...;s_k)$ if $(\forall e\,:\, ``\text{output }e" \notin s_{k+1})$;
\end{enumerate}
\end{definition}

\begin{definition}\label{def:OVar}
{\bf (Output deciding variables)} The output deciding variables in a statement sequence $S$ are  $\text{Imp}_o(S) \cup \text{TVar}_o(S)$, written $\text{OVar}(S)$.
\end{definition}


The condition of the behavioral equivalence is defined recursively.
The base case is for two same output statements or two statements where the output sequence variable is not defined.
The inductive cases are syntax directed considering the syntax of compound statements and statement sequences.

\begin{definition}
{\bf (proof rule of behavioral equivalence)}
Two statement sequences $S_1$ and $S_2$ satisfy the proof rule of behavioral equivalence,
written $S_1 \equiv_{O}^S S_2$, iff one of the following holds:
\begin{enumerate}
\item $S_1$ and $S_2$ are one statement and one of the following holds:

\begin{enumerate}
\item  $S_1$ and $S_2$ are simple statement and one of the following holds:
\begin{enumerate}
\item $S_1$ and $S_2$ are not output statement, $\forall e_1\, e_2\,:\,$

\noindent$(``\text{output }e_1" \neq S_1) \wedge (``\text{output }e_2" \neq S_2)$; or

\item $S_1 = S_2 = ``\text{output }e"$.
\end{enumerate}

\item
$S_1 = ``\text{If }(e) \text{ then }\{S_1^t\} \text{ else }\{S_1^f\}"$,
$S_2 = ``\text{If }(e) \text{ then }\{S_2^t\} \text{ else}\,\allowbreak \{S_2^f\}"$ and all of the following hold:

\begin{itemize}
\item There is an output statement in $S_1$ and $S_2$,

\noindent$\exists e_1\, e_2\,:\, (``\text{output }e_1" \in S_1) \wedge (``\text{output }e_2" \in S_2)$;

\item $(S_1^t \equiv_{O}^S S_2^t) \wedge (S_1^f \equiv_{O}^S S_2^f)$;
\end{itemize}

\item
$S_1 = ``\text{while}_{\langle n_1 \rangle}(e) \; \{S_1''\}"$ and
$S_2 = ``\text{while}_{\langle n_2 \rangle}(e) \; \{S_2''\}"$ and all of the following hold:

\begin{itemize}
\item There is an output statement in $S_1$ and $S_2$,

\noindent$\exists e_1\, e_2\,:\, (``\text{output }e_1" \in S_1) \wedge (``\text{output }e_2" \in S_2)$;

\item
    $S_1'' \equiv_{O}^S S_2''$;

\item
    $S_1''$ and $S_2''$ have equivalent computation of $\text{OVar}(S_1) \cup \text{OVar}(S_2)$;

\item $S_1''$ and $S_2''$ satisfy the proof rule of termination in the same way,
    $S_1'' \equiv_{H}^S S_2''$;
\end{itemize}

\item Output statements are not in both $S_1$ and $S_2$,

\noindent$\forall e_1\,e_2\,:\, (``\text{output }e_1" \notin S_1) \wedge (``\text{output }e_2" \notin S_2)$.

\end{enumerate}

\item $S_1$ and $S_2$ are not both one statement and one of the following holds:
\begin{enumerate}
\item $S_1=S_1';s_1$ and $S_2=S_2';s_2$, and all of the following hold:
  \begin{itemize}
  \item $S_1' \equiv_{O}^S S_2'$;

  \item $S_1'$ and $S_2'$ have equivalent computation of $\text{OVar}(s_1) \cup \text{OVar}(s_2)$;

  \item $S_1'$ and $S_2'$ satisfy the proof rule of termination in the same way:
   $S_1' \equiv_{H}^S S_2'$;

  \item There is an output statement in both $s_1$ and $s_2$,

  $\exists e_1\, e_2\,:\, (``\text{output }e_1" \in s_1) \wedge (``\text{output }e_2" \in s_2)$;

  \item $s_1 \equiv_{O}^S s_2$;
  \end{itemize}

\item There is no output statement in the last statement in $S_1$ or $S_2$:

$\big((S_1 = S_1';s_1) \wedge (S_1' \equiv_{O}^S S_2) \wedge (\forall e\,:\, ``\text{output }e" \notin s_1)\big)$

$\vee
\big((S_2 = S_2';s_2) \wedge (S_1 \equiv_{O}^S S_2') \wedge (\forall e\,:\, ``\text{output }e" \notin s_2)\big)$;
\end{enumerate}
\end{enumerate}
\end{definition}

\subsubsection{Soundness of the proof rule for behavioral equivalence}
We show that two statement sequences satisfy the proof rule of the behavioral equivalence and their initial states agree on values of their output deciding variables, then the two statement sequences produce the same output sequence when started in their initial states.
\begin{theorem}\label{thm:sameIOtheoremFuncWide}
Two statement sequences $S_1$ and $S_2$ satisfy the proof rule of the behavioral equivalence, $S_1\, {\equiv}_{O}^S \, S_2$.
If $S_1$ and $S_2$ start in states
$m_1(\mathfrak{f}_1, \vals_1)$ and $m_2(\mathfrak{f}_2, \vals_2)$ where both of the following hold:
\begin{itemize}
\item Crash flags are not set, $\mathfrak{f}_1 = \mathfrak{f}_2 = 0$;

\item Value stores $\vals_1$ and $\vals_2$ agree on values of the output deciding variables of $S_1$ and $S_2$,
$\forall id \in \text{OVar}(S_1) \cup \text{OVar}(S_2)\,:\,
\vals_1(id) = \vals_2(id)$;
\end{itemize}
then $S_1$ and $S_2$ produce the same output sequence,

\noindent$(S_1, m_1) \equiv_{O} (S_2, m_2)$.
\end{theorem}
The proof is by induction on the sum of program size of $S_1$ and $S_2$, $\text{size}(S_1) + \text{size}(S_2)$ and is a case analysis based on $S_1\, {\equiv}_{O}^S \, S_2$.

\begin{proof}
The proof is by induction on the sum of program size of $S_1$ and $S_2$, $\text{size}(S_1) + \text{size}(S_2)$ and is a case analysis based on $S_1 {\equiv}_{O}^S S_2$.

\noindent{\bf Base case}.

$S_1$ and $S_2$ are simple statement.
There are two cases according to the proof rule of behavioral equivalence because stacks are not changed in executions of $S_1$ and $S_2$.
\begin{enumerate}
\item $S_1$ and $S_2$ are not output statement, $\forall e_1\, e_2\,:\, (``\text{output }e_1" \neq S_1) \wedge (``\text{output }e_2" \neq S_2)$;

By the definition of imported variables relative to output,
$\text{Imp}_o(S_1) = \text{Imp}_o(S_2) = \{{id}_{IO}\}$.
By assumption, initial value stores $\vals_1$ and $\vals_2$ agree on the value of the I/O sequence variable, $\vals_1({id}_{IO}) = \vals_2({id}_{IO})$. By definition, $\text{Out}(\vals_1) = \text{Out}(\vals_2)$.

By Lemma~\ref{lmm:NoOutputExecMultiSteps}, in any  state $m_1'$ reachable from $m_1$, the output sequence in $m_1'$ is same as that in $m_1$, $\forall m_1'\,:\, ((S_1, m_1(\vals_1)) ->* (S_1', m_1'(\vals_1'))) => (\text{Out}(\vals_1') = \text{Out}(\vals_1))$.
Similarly, for any state $m_2'$ reachable from $m_2$, the output sequence in $m_2'$ is same as that in $m_2$.
The theorem holds.

\item
$S_1 = S_2 = ``\text{output}\, e"$.

We show that the expression $e$ evaluates to the same value w.r.t value stores, $\vals_1$, $\vals_2$.
By the definition of imported variables relative to output, $\text{Imp}_o(S_1) = \text{Imp}_o(S_2) = \text{Use}(e) \cup \{{id}_{IO}\}$.
Then $\forall x \in \text{Use}(e) \cup \{{id}_{IO}\}\,:\,
\vals_1(x) = \vals_2(x)$ by assumption.
By Lemma~\ref{lmm:expEvalSameVal},
$\mathcal{E}\llbracket e\rrbracket\vals_1 = \mathcal{E}\llbracket e\rrbracket\vals_2$.
Then, there are two possibilities.

\begin{enumerate}
\item $\mathcal{E}\llbracket e\rrbracket\vals_1 = \mathcal{E}\llbracket e\rrbracket\vals_2 = (\text{error}, v_\mathfrak{of})$.

The execution of $S_1$ proceeds as follows.
\begin{tabbing}
xx\=xx\=\kill
\>\>                   $(\text{output }e, m_1(\vals_1))$\\
\>= \>                 $(\text{output }(\text{error}, v_\mathfrak{of}), m_1(\vals_1))$ by the rule EEval'\\
\>$->$\>               $(\text{output }0, m_1(1/\mathfrak{f}))$ by the ECrash rule  \\
\>{\kStepArrow [i] }\> $(\text{output }0, m_1(1/\mathfrak{f}))$ for any $i>0$ by the Crash rule.
\end{tabbing}

Similarly, the execution of $S_2$ does not terminate and there is no change to I/O sequence in execution.
Because $\vals_1({id}_{IO}) = \vals_2({id}_{IO})$, then the output sequence in value stores $\vals_1$ and $\vals_2$ are same, $\text{Out}(\vals_1) = \text{Out}(\vals_2)$, the theorem holds.

\item
$\mathcal{E}\llbracket e\rrbracket\vals_1 = \mathcal{E}\llbracket e\rrbracket\vals_2 \neq (\text{error}, v_\mathfrak{of})$

$S_1$ and $S_2$ satisfy the proof rule of equivalent computation of I/O sequence variable and their initial states agree on the values of the imported variables relative to I/O sequence variable. By Theorem~\ref{thm:equivCompMain}, $S_1$ and $S_2$ produce the same output sequence after terminating execution when started in state $m_1(\vals_1)$ and $m_2(\vals_2)$ respectively. The theorem holds.
\end{enumerate}
\end{enumerate}

\noindent{\bf Induction step}.

The hypothesis IH is that Theorem~\ref{thm:sameIOtheoremFuncWide} holds when
$\text{size}(S_1) + \text{size}(S_2) = k \geq 2$.

We show Theorem~\ref{thm:sameIOtheoremFuncWide} holds when  $\text{size}(S_1) + \text{size}(S_2) = k+1$.
The proof is a case analysis according to the cases in the definition of the proof rule of behavioral equivalence.
\begin{enumerate}
\item $S_1$ and $S_2$ are one statement and one of the following holds:

\begin{enumerate}
\item
$S_1 = ``\text{If}(e) \text{ then }\{S_1^t\} \text{ else }\{S_1^f\}"$ and
$S_2 = ``\text{If}(e) \text{ then }\allowbreak\{S_2^t\} \text{ else }\{S_2^f\}"$ and all of the following hold:

\begin{itemize}
\item There is an output statement in $S_1$ and $S_2$:
$\exists e_1\, e_2\,:\, (``\text{output }e_1" \in S_1) \wedge (``\text{output }e_2" \in S_2)$;

\item $S_1^t \equiv_{O}^S S_2^t$;

\item $S_1^f \equiv_{O}^S S_2^f$;
\end{itemize}

By Lemma~\ref{lmm:IOseqInImpOfAnyStmtWithOutStmt},  $\{{id}_{IO}\} \in \text{Imp}_o(S_1)$.
By assumption, value stores $\vals_1$ and $\vals_2$ agree on the value of the I/O sequence variable and the I/O sequence variable, $\vals_1({id}_{IO}) = \vals_2({id}_{IO})$.

We show that the evaluations of the predicate expression of $S_1$ and $S_2$ w.r.t. initial value store $\vals_1$ and $\vals_2$ produce the same value. We need to show that value stores $\vals_1$ and $\vals_2$ agree on values of variables used in the predicate expression $e$ of $S_1$ and $S_2$.
Because the output sequence is defined in $S_1$, by the definition of imported variables relative to output,
$\text{Imp}_o(S_1) = \text{Use}(e) \cup \text{Imp}_o(S_1^t) \cup \text{Imp}_o(S_1^f)$.
Thus, $\text{Use}(e) \subseteq \text{OVar}(S_1)$.
By assumption, value stores $\vals_1$ and $\vals_2$ agree on values of variables used in the predicate expression $e$ of $S_1$ and $S_2$, $\forall x \in \text{Use}(e)\,:\,
\vals_1(x) = \vals_2(x)$.
By Lemma~\ref{lmm:expEvalSameVal}, the evaluations of the predicate expression of $S_1$ and $S_2$ w.r.t. pairs value stores, $\vals_1$ and $\vals_2$ generate the same value,
$\mathcal{E}'\llbracket e\rrbracket\vals_1 = \mathcal{E}'\llbracket e\rrbracket\vals_2$.
Then there are two possibilities.
\begin{enumerate}

\item $\mathcal{E}'\llbracket e\rrbracket\vals_1 = \mathcal{E}'\llbracket e\rrbracket\vals_2 = (\text{error}, v_\mathfrak{of})$.

Then the execution of $S_1$ proceeds as follows:

\begin{tabbing}
xx\=xx\=\kill
\>\>                   $(\text{If}(e) \text{ then }\{S_1^t\} \text{ else }\{S_1^f\}, m_1(\vals_1))$ \\
\>$->$\>               $(\text{If}((\text{error}, v_\mathfrak{of})) \text{ then }\{S_1^t\} \text{ else }\{S_1^f\}, m_1(\vals_1))$\\
\>\>                   by the EEval' rule\\
\>$->$\>               $(\text{If}(0) \text{ then }\{S_1^t\} \text{ else }\{S_1^f\}, m_1(1/\mathfrak{f}))$\\
\>\>                   by the ECrash rule  \\
\>{\kStepArrow [i] }\> $(\text{If}(0) \text{ then }\{S_1^t\} \text{ else }\{S_1^f\}, m_1(1/\mathfrak{f}))$\\
\>\>                   for any $i>0$, by the Crash rule.
\end{tabbing}

Similarly, the execution of $S_2$ does not terminate and does not redefine I/O sequence.
Because $\vals_1({id}_{IO}) = \vals_2({id}_{IO})$, the theorem holds.

\item $\mathcal{E}'\llbracket e\rrbracket\vals_1 = \mathcal{E}'\llbracket e\rrbracket\vals_2 \neq (\text{error}, v_\mathfrak{of})$.

W.l.o.g., $\mathcal{E}'\llbracket e\rrbracket\vals_1 = \mathcal{E}'\llbracket e\rrbracket\vals_2 = (0, v_\mathfrak{of})$.
The execution of $S_1$ proceeds as follows.

\begin{tabbing}
xx\=xx\=\kill
\>\>                   $(\text{If}(e) \text{ then }\{S_1^t\} \text{ else }\{S_1^f\}, m_1(\vals_1))$ \\
\>$->$\>               $(\text{If}((0, v_\mathfrak{of})) \text{ then }\{S_1^t\} \text{ else }\{S_1^f\}, m_1(\vals_1))$\\
\>\>                   by the EEval rule\\
\>$->$\>               $(\text{If}(0) \text{ then }\{S_1^t\} \text{ else }\{S_1^f\}, m_1(\vals_1))$\\
\>\>                   by the E-Oflow1 or E-Oflow2 rule\\
\>$->$\>               $(S_1^f, m_1(\vals_1))$ by the If-F rule.
\end{tabbing}

Similarly, the execution of $S_2$ proceeds to $(S_2^f, m_2(\vals_2))$ after two steps.
By the hypothesis IH, we show that $S_1^f$ and $S_2^f$ produce the same output sequence when started in states $m_1(\vals_1)$ and $m_2(\vals_2)$. We need to show that all required conditions are satisfied.
\begin{itemize}
\item $\text{size}(S_1^f) + \text{size}(S_2^f) \leq k$.

By definition, $\text{size}(S_1)  = 1 + \text{size}(S_1^t) + \text{size}(S_1^f)$.
Therefore, $\text{size}(S_1^f) + \text{size}(S_2^f) < k$.

\item Value stores $\vals_1$  and $\vals_2$ agree on values of the out-deciding variables of $S_1^f$ and $S_2^f$, $\forall x \in \text{OVar}(S_1^f) \cup \text{OVar}(S_2^f)\,:\,
    \vals_1(x) = \vals_2(x)$.

By the definition of imported variables relative to output,
$\text{Imp}_o(S_1^f) \subseteq \text{Imp}_o(S_1)$. Besides, by the definition of $\text{TVar}_o(S_1)$, $\text{TVar}_o(S_1^f) \subseteq \text{TVar}_o(S_1)$.
Then $\text{OVar}(S_1^f) \subseteq \text{OVar}(S_1)$.
Similarly, $\text{OVar}(S_2^f) \subseteq \text{OVar}(S_2)$.
By assumption, the value stores $\vals_1$  and $\vals_2$ agree on the values of the out-deciding variables of $S_1^f$ and $S_2^f$.
\end{itemize}
By the hypothesis IH, $S_1^f$ and $S_2^f$ produce the same output sequence when started from state $m_1(\vals_1)$ and $m_2(\vals_2)$ respectively. The theorem holds.
\end{enumerate}

\item
$S_1 = ``\text{while}_{\langle n_1 \rangle}(e) \; \{S_1''\}"$ and
$S_2 = ``\text{while}_{\langle n_2 \rangle}(e) \; \{S_2''\}"$ and all of the following hold:

\begin{itemize}
\item There is an output statement in $S_1$ and $S_2$:
$\exists e_1\, e_2\,:\, (``\text{output }e_1" \in S_1) \wedge (``\text{output }e_2" \in S_2)$;

\item $S_1'' \equiv_{O}^S S_2''$;

\item Both loop bodies satisfy the proof rule of termination in the same way:
    $S_1'' \equiv_{H}^S S_2''$;

\item $S_1''$ and $S_2''$ have equivalent computation of  $\text{OVar}(S_1) \cup \text{OVar}(S_2)$;
\end{itemize}

By Corollary~\ref{coro:loopSameIOSeq}, we show that $S_1$ and $S_2$ produce the same output sequence when started in states $m_1(\vals_1)$ and $m_2(\vals_2)$ respectively.
We need to show that the required conditions are satisfied.
\begin{itemize}
\item Crash flags are not set, $\mathfrak{f}_1 = \mathfrak{f}_2 = 0$.

\item Value stores $\vals_1$ and $\vals_2$ agree on the values of the out-deciding variables of $S_1$ and $S_2$, $\forall x \in \text{OVar}(S_1) \cup \text{OVar}(S_2)\,:\,
    \vals_1(x) = \vals_2(x)$.

\item The loop counter value of $S_1$ and $S_2$ are zero in initial loop counter, $\text{loop}_c^1(n_1) = \text{loop}_c^2(n_2) = 0$.

\item The loop body of $S_1$ and $S_2$ satisfy the proof rule of termination in the same way,
$S_1'' \equiv_{H}^S S_2''$.

\item The loop body of $S_1$ and $S_2$ satisfy the proof rule of equivalent computation of
$\text{OVar}(S_1) \cup \text{OVar}(S_2)$.

The above five conditions are from assumption.

\item $S_1$ and $S_2$ have same set of termination deciding variables,
$\text{TVar}(S_1) = \text{TVar}(S_2)$.

By the definition of $\text{TVar}_o(S_1)$, $\text{TVar}_o(S_1) = \text{TVar}(S_1)$ and
$\text{TVar}_o(S_2) = \text{TVar}(S_2)$.
By Lemma~\ref{lmm:sameTVarO}, $\text{TVar}_o(S_1) = \text{TVar}_o(S_2)$.
Thus, $\text{TVar}(S_1) = \text{TVar}(S_2)$.

\item $S_1$ and $S_2$ have same set of imported variables relative to the I/O sequence variable,

\noindent$\text{Imp}(S_1, \{{id}_{IO}\}) = \text{Imp}(S_2, \{{id}_{IO}\})$.

By Lemma~\ref{lmm:sameImpOVar}, $\text{Imp}_o(S_1) = \text{Imp}_o(S_2)$.
By definition, $\text{Imp}_o(S_1) = \text{Imp}(S_1, \{{id}_{IO}\})$ and $\text{Imp}_o(S_2)\allowbreak = \text{Imp}(S_2, \{{id}_{IO}\})$.
Thus, $\text{Imp}(S_1, \{{id}_{IO}\}) = \text{Imp}(S_2, \{{id}_{IO}\})$.

\item The loop body of $S_1$ and $S_2$ produce the same output sequence when started in states with crash flags not set and whose value stores agree on values of the out-deciding variables of $S_1''$ and $S_2''$,
    $\forall m_{S_1''}(\mathfrak{f}_1'', \vals_1'')\, m_{S_2''}(\mathfrak{f}_2'', \vals_2'')\,:\,
    ((\forall x \in \text{OVar}(S_1'') \cup \text{OVar}(S_2'')\,:\,
    \vals_1''(x) = \vals_2''(x))
    \wedge
    (\mathfrak{f}_1'' = \mathfrak{f}_2'' = 0)) =>
    (S_1'', m_{S_1''}(\mathfrak{f}_1'', \vals_1'')) \equiv_{O}
    (S_2'', m_{S_2''}(\mathfrak{f}_2'', \vals_2''))$.

Because $\text{size}(S_1) = 1 + \text{size}(S_1''), \text{size}(S_2) = 1 + \text{size}(S_2'')$, then $\text{size}(S_1'') + \text{size}(S_2'') < k$. By the hypothesis IH, the condition is satisfied.
\end{itemize}
By Corollary~\ref{coro:loopSameIOSeq}, $S_1$ and $S_2$ produce the same output sequence when started in states $m_1(\vals_1)$ and $m_2(\vals_2)$ respectively. The theorem holds.

\item Output statements are not in both $S_1$ and $S_2$,
$\forall e_1\,e_2\,:\, (``\text{output }e_1" \notin S_1) \wedge (``\text{output }e_2" \notin S_2)$.

By Lemma~\ref{lmm:IOseqInImpOfAnyStmtWithOutStmt}, $\{{id}_{IO}\} \subseteq \text{Imp}_o(S_1)$.
By assumption, value stores in initial states $m_1, m_2$ agree on  values of the I/O sequence variable,
$\vals_1({id}_{IO}) = \vals_2({id}_{IO})$.
By Lemma~\ref{lmm:NoOutputExecMultiSteps}, the value of output sequence is same in $m_1$ and any state reachable from $m_1$, $\forall m_1'\, m_2'\,:\, (S_1, m_1(\vals_1)) ->* (S_1', m_1'(\vals_1')) \text{ and } (S_2, m_2(\vals_2)) ->* (S_2', m_2'(\vals_2')), \text{Out}(\vals_1') = \text{Out}(\vals_1) = \text{Out}(\vals_2) = \text{Out}(\vals_2')$.
The theorem holds.
\end{enumerate}

\item $S_1=S_1';s_1$ and $S_2=S_2';s_2$ are not both one statement and one of the following holds:

\begin{enumerate}
\item There is an output statement in both $s_1$ and $s_2$,
  $\exists e_1\, e_2\,:\, (``\text{output }e_1" \in s_1) \wedge (``\text{output }e_2" \in s_2)$, and all of the following hold:

  \begin{itemize}
  \item $S_1' \equiv_{O}^S S_2'$;

  \item $S_1'$ and $S_2'$ satisfy the proof rule of termination in the same way:
   $S_1' \equiv_{H}^S S_2'$;

  \item $S_1'$ and $S_2'$ have equivalent computation of $\text{OVar}(s_1) \cup \text{OVar}(s_2)$;

  \item $s_1 \equiv_{O}^S s_2$;
  \end{itemize}

By the hypothesis IH, we show $S_1'$ and $S_2'$ produce the same output sequence when started in states $m_1(\vals_1)$ and $m_2(\vals_2)$ respectively. We need to show that all required conditions are satisfied.
\begin{itemize}
\item $\text{size}(S_1') + \text{size}(S_2') < k$.

By the definition of program size, $\text{size}(s_1) \geq 1, \text{size}(s_2) \geq 1$. Then $\text{size}(S_1') + \text{size}(S_2') < k$.

\item Value stores $\vals_1$ and $\vals_2$ agree on values of the out-deciding variables of $S_1'$ and $S_2'$,
$\forall x \in \text{OVar}(S_1') \cup \text{OVar}(S_2')\,:\,
 \vals_1(x) = \vals_2(x)$.

We show that $\text{TVar}_o(S_1') \subseteq \text{TVar}_o(S_1)$.
\begin{tabbing}
xx\=xx\=\kill
\>\>              $\text{TVar}_o(S_1')$\\
\>$\subseteq$\>   $\text{TVar}(S_1')$ by Lemma~\ref{lmm:TVarOIsSubsetOfTVar}\\
\>$\subseteq$\>   $\text{TVar}_o(S_1)$ by the definition of $\text{TVar}_o(S_1)$\\
\end{tabbing}

%

We show that $\text{Imp}_o(S_1') \subseteq \text{Imp}_o(S_1)$.
\begin{tabbing}
xx\=xx\=\kill
\>\>              $\text{Imp}_o(S_1')$\\
\>$\subseteq$\>   $\text{Imp}(S_1', \{{id}_{IO}\})$ (1) by Lemma~\ref{lmm:ImpOVarIsSubsetOfImpIO}\\
\\
\>\>              $\{{id}_{IO}\} \subseteq \text{Imp}_o(s_{k+1})$ (2) by Lemma~\ref{lmm:IOseqInImpOfAnyStmtWithOutStmt}\\
\\
\>\>              Combining (1) + (2) \\
\>\>              $\text{Imp}(S_1', \{{id}_{IO}\})$ \\
\>$\subseteq$\>   $\text{Imp}(S_1', \text{Imp}_o(s_1))$ by Lemma~\ref{lmm:impVarUnionLemma}\\
\>=\>             $\text{Imp}_o(S_1)$ by the definition of $\text{Imp}_o(S)$.
\end{tabbing}

Similarly, $\text{OVar}(S_2') \subseteq \text{OVar}(S_2)$.
By assumption, value stores $\vals_1$ and $\vals_2$ agree on values of out-deciding variables of $S_1'$ and $S_2'$.
\end{itemize}

By the hypothesis IH, $S_1'$ and $S_2'$ produce the same output sequence when started in state $m_1(\vals_1)$ and $m_2(\vals_2)$ respectively.

We show that $S_1$ and $S_2$ produce the same output sequence if $s_1$ and $s_2$ execute.
We need to show that $S_1'$ and $S_2'$ terminate in the same way when started in states $m_1(\vals_1)$ and $m_2(\vals_2)$ respectively. Specifically, we prove that the value stores $\vals_1$ and $\vals_2$ agree on the values of termination deciding variables of $S_1'$ and $S_2'$.
By definition, the termination deciding variables in $S_1'$ are a subset of the termination deciding variables relative to output, $\text{TVar}(S_1') \subseteq \text{TVar}_o(S_1)$. Similarly, $\text{TVar}(S_2') \subseteq \text{TVar}_o(S_2)$.
By assumption, the value stores $\vals_1$ and $\vals_2$ agree on the values of the termination deciding variables of $S_1'$ and $S_2'$,
$\forall x \in \text{TVar}(S_1') \cup \text{TVar}(S_2')\,:\,
\vals_1(x) = \vals_2(x)$.
By Theorem~\ref{thm:mainTermSameWayLocal}, $S_1'$ and $S_2'$ terminate in the same way when started in state $m_1(\vals_1)$ and $m_2(\vals_2)$ respectively.

If $S_1'$ and $S_2'$ terminate when started in states $m_1(\vals_1)$ and $m_2(\vals_2)$, by Lemma~\ref{lmm:TermInSameWayConsumeSameInputs}, $S_1'$ and $S_2'$ consume same amount of input values.
In addition,
we show that value stores agree on values of the out-deciding variables of $s_1$ and $s_2$ by Theorem~\ref{thm:equivCompMain}.
We need to show that $S_1'$ and $S_2'$ start execution in states agreeing on values of the imported variables in $S_1'$ and $S_2'$ relative to the out-deciding variables of $s_1$ and $s_2$.
\begin{itemize}
\item  $\text{Imp}(\text{TVar}_o(s_1), S_1') \subseteq \text{TVar}_o(S_1)$.

 This is by the definition of $\text{TVar}_o(S_1)$.

\item $\text{Imp}(\text{Imp}_o(s_1), S_1') = \text{Imp}_o(S_1)$.

  This is by the definition of $\text{Imp}_o(S_1)$.
\end{itemize}
Thus, the imported variables in $S_1'$ relative to the out-deciding variables of $s_1$ are a subset of the out-deciding variables of $S_1$, $\text{Imp}(S_1', \text{OVar}(s_1)) \subseteq \text{OVar}(S_1)$. Similarly, $\text{Imp}(S_2', \text{OVar}(s_2)) \subseteq \text{OVar}(S_2)$.
By Corollary~\ref{coro:sameOVar}, $s_1$ and $s_2$ have same out-deciding variables,
$\text{OVar}(s_1) = \text{OVar}(s_2)$.
By assumption, $S_1'$ and $S_2'$ terminate when started in states $m_1(\vals_1)$ and $m_2(\vals_2)$,
$(S_1', m_1(\vals_1))\allowbreak ->* (\text{skip}, m_1'(\vals_1')), (S_2',\allowbreak m_1(\vals_2)) ->* (\text{skip}, m_2'(\vals_2'))$.
By Theorem~\ref{thm:equivCompMain}, value stores $\vals_1'$ and $\vals_2'$ agree on values of the out-deciding variables of $s_1$ and $s_2$.

By the hypothesis IH again, $s_1$ and $s_2$ produce the same output sequence when started in states $m_1'(\vals_1')$ and $m_2'(\vals_2')$ respectively. The theorem holds.

\item There is no output statement in the last statement in $S_1$ or $S_2$:
W.l.o.g., $(\forall e\,:\, ``\text{output }e" \notin s_1) \wedge ((S_1') \equiv_{O}^S (S_2))$.

By the hypothesis IH, we show that $S_1'$ and $S_2$ produce the same output sequence when started in states $m_1(\vals_1)$ and $m_2(\vals_2)$ respectively. We need to show that two required conditions are satisfied.
\begin{itemize}
\item $\text{size}(S_1') + \text{size}(S_2) \leq k$.

$\text{size}(s_1) \geq 1$ by definition. Then $\text{size}(S_1') + \text{size}(S_2) \leq k$.

\item $\forall x \in \text{OVar}(S_1') \cup \text{OVar}(S_2)\,:\,
\vals_1(x) = \vals_2(x)$.

By definition of $\text{TVar}_o(S)/\text{Imp}_o(S)$,
               $\text{TVar}_o(S_1')\allowbreak = \text{TVar}_o(S_1)$, and
               $\text{Imp}_o(S_1') = \text{Imp}_o(S_1)$
Hence, $\forall x \in \text{OVar}(S_1') \cup \text{OVar}(S_2)\,:\,
\vals_1(x) = \vals_2(x)$.
\end{itemize}
Therefore, $S_1'$ and $S_2$ produce the same output sequence when started in state $m_1(\vals_1)$ and $m_2(\vals_2)$ respectively, $(S_1', m_1) \equiv_{O} (S_2, m_2)$ by the hypothesis IH.

When the execution of $S_1'$ terminates, then the output sequence is not changed in the execution of $s_1$ by Lemma~\ref{lmm:NoOutputExecMultiSteps}.
The theorem holds.
\end{enumerate}
\end{enumerate}
\end{proof}

\subsubsection{Supporting lemmas for the soundness proof of  behavioral equivalence}
We listed the lemmas and corollaries used in the proof of Theorem~\ref{thm:sameIOtheoremFuncWide} below.
The supporting lemmas are of two parts. One part is various properties related to the out-deciding variables. The other part is the proof that two loop statements produce the same output sequence.


\begin{lemma}\label{lmm:IOseqInImpOfAnyStmtWithOutStmt}
For any statement sequence $S$, the I/O sequence variable is in imported variable in $S$ relative to output, ${id}_{IO} \in \text{Imp}_o(S)$.
\end{lemma}

\begin{proof}
By structure induction on abstract syntax of $S$.
\end{proof}

\begin{lemma}\label{lmm:ImpOVarIsSubsetOfImpIO}
For any statement sequence $S$, the imported variables in $S$ relative to output are a subset of the imported variables in $S$ relative to the I/O sequence variable, $\text{Imp}_o(S) \subseteq \text{Imp}(S, \{{id}_{IO}\})$.
\end{lemma}
\begin{proof}
By induction on abstract syntax of $S$.
In every case, there are subcases based on if there is output statement in the statement sequence $S$ or not if necessary.
\end{proof}

\begin{lemma}\label{lmm:sameImpOVar}
If two statement sequences $S_1$ and $S_2$ satisfy the proof rule of behavioral equivalence, then $S_1$ and $S_2$ have the same set of imported variables relative to output, $(S_1 \equiv_{O}^S S_2) => (\text{Imp}_o(S_1) = \text{Imp}_o(S_2))$.
\end{lemma}
\begin{proof}
By induction on $\text{size}(S_1) + \text{size}(S_2)$.
\end{proof}

\begin{lemma}\label{lmm:TVarOIsSubsetOfTVar}
For any statement sequence $S$ and any variable $x$, the termination deciding variables in $S$ relative to output is a subset of the termination deciding variables in $S$, $\text{TVar}_o(S) \subseteq \text{TVar}(S)$.
\end{lemma}
\begin{proof}
By induction on abstract syntax of $S$.
In every case, there are subcases based on if there is output statement in the statement sequence $S$ or not if necessary.
\end{proof}

\begin{lemma}\label{lmm:sameTVarO}
If two statement sequences $S_1$ and $S_2$ satisfy the proof rule of behavioral equivalence, then $S_1$ and $S_2$ have the same set of termination deciding variables relative to output, $(S_1 \equiv_{O}^S S_2) => (\text{TVar}_o(S_1) = \text{TVar}_o(S_2))$.
\end{lemma}
\begin{proof}
By induction on $\text{size}(S_1) + \text{size}(S_2)$.
\end{proof}

\begin{corollary}\label{coro:sameOVar}
If two statement sequences $S_1$ and $S_2$ satisfy the proof rule of behavioral equivalence, then $S_1$ and $S_2$ have the same set of out-deciding variables,
$(S_1 \equiv_{O}^S S_2) => \text{OVar}(S_1) = \text{OVar}(S_2)$.
\end{corollary}
\begin{proof}
By Lemma~\ref{lmm:sameImpOVar},
$\text{Imp}_o(S_1) = \text{Imp}_o(S_2)$.
By Lemma~\ref{lmm:sameTVarO}, $\text{TVar}_o(S_1) = \text{TVar}_o(S_2)$.
\end{proof}

\begin{lemma}\label{lmm:NoOutputExec}
In one step execution $(S, m(\vals)) -> (S', m'(\vals'))$, if there is no output statement in $S$, then the output sequence is same in value store $\vals$ and $\vals'$, $\text{Out}(\vals') = \text{Out}(\vals)$.
\end{lemma}
\begin{proof}
By induction on abstract syntax of $S$ and crash flag $\mathfrak{f}$ in state $m$.
%
%
%
%
%
%
%
%
%
%
%
%
%
%
%
%
%
%
%
\end{proof}

\begin{lemma}\label{lmm:NoOutputExecMultiSteps}
If there is no output statement in $S$, then, after the execution $(S, m(\vals)) ->* (S', m'(\vals'))$, the output sequence is same in value store $\vals$ and $\vals'$, $\text{Out}(\vals') = \text{Out}(\vals)$.
\end{lemma}
\begin{proof}
By induction on number $k$ of execution steps in the execution $(S, m(\vals))$ {\kStepArrow [k] } $(S', m'(\vals'))$.
The proof also relies on the fact that if $s \notin S$, then $s \notin S'$.
\end{proof}

\begin{lemma}\label{lmm:ExistenceOfIthIterStart}
One while statement $s = ``\text{while}_{\langle n\rangle}(e) \{S\}"$ starts in a state $m(\mathfrak{f}, \text{loop}_c)$ in which the loop counter of $s$ is zero, $\text{loop}_c(n) = 0$ and the crash flag is not set, $\mathfrak{f} = 0$. For any positive integer $i$, if there is a state $m'(m_c')$ reachable from $m$ in which the loop counter is $i$, $\text{loop}_c'(n) = i$, then there is a configuration $(S;s, m''(\mathfrak{f}'', \text{loop}_c''))$ reachable from the configuration $(s, m)$ in which loop counter of $s$ is $i$, $\text{loop}_c''(n) = i$ and the crash flag is not set, $\mathfrak{f}'' = 0$:

$\forall i>0\,:\, ({((s, m(\mathfrak{f}, m_c)) ->* (S', m'(\mathfrak{f'}, \text{loop}_c')))} \, \wedge \,
                    (\text{loop}_c(n) = 0) \, \wedge \,
                    (\mathfrak{f} = 0) \, \wedge \,
                    (\text{loop}_c'(n) = i)) => $

$(s, m(\mathfrak{f}, \text{loop}_c)) ->* (S;s, m''(\mathfrak{f}, \text{loop}_c''))$ where $\mathfrak{f} = 0$ and $\text{loop}_c''(n) = i$.
\end{lemma}
\begin{proof}
The proof is by induction on $i$.

{\bf Base case} $i=1$.

We show that the evaluation of the loop predicate of $s$ w.r.t the value store $\vals$ in the state $m(\text{loop}_c, \vals)$ produces an nonzero integer value. By our semantic rule, if the evaluation of the predicate expression of $s$ raises an exception, the execution of $s$ proceeds as follows:

\begin{tabbing}
xx\=xx\=\kill
\>\>     $({s}, m(\mathfrak{f}, \text{loop}_c,\vals))$\\
\>= \>   $(\text{while}_{\langle n\rangle} (e) \; \{S\}, {m(\text{loop}_c, \vals)})$\\
\>$->$\> $(\text{while}_{\langle n\rangle} (\text{error}) \; \{S\}, {m(\text{loop}_c, \vals)})$ by the EEval rule\\
\>$->$\> $(\text{while}_{\langle n\rangle} (0) \; \{S\}, {m(1/\mathfrak{f}, \text{loop}_c, \vals)})$ by the ECrash rule\\
\>{\kStepArrow [k] }\> $(\text{while}_{\langle n\rangle} (0) \; \{S\}, m(1/\mathfrak{f}, \text{loop}_c))$ for any $k>0$, by the Crash rule.
\end{tabbing}

Hence, we have a contradiction that there is no configuration in which the loop counter of $s$ is 1.

When the evaluation of the loop predicate expression of $s$ produce zero, the execution of $s$ proceeds as follows:

\begin{tabbing}
xx\=xx\=\kill
\>\>     $({s}, m(\mathfrak{f}, \text{loop}_c,\vals))$\\
\>= \>   $(\text{while}_{\langle n\rangle} (e) \; \{S\}, {m(\text{loop}_c, \vals)})$\\
\>$->$\> $(\text{while}_{\langle n\rangle} (0) \; \{S\}, {m(\text{loop}_c, \vals)})$ by the EEval rule\\
\>$->$\> $({\text{skip}}, m(\text{loop}_c[0/n_1]))$ by the Wh-F rule.
\end{tabbing}

Hence, we have a contradiction that there is no configuration in which the loop counter of $s$ is 1.
The evaluation of the predicate expression of $s$ w.r.t value store $\vals$ produce nonzero value.
The execution of $s$ proceeds as follows:

\begin{tabbing}
xx\=xx\=\kill
\>\>     $({s}, m(\mathfrak{f}, \text{loop}_c,\vals))$\\
\>= \>   $(\text{while}_{\langle n\rangle} (e) \; \{S\}, {m(\mathfrak{f}, \text{loop}_c, \vals)})$\\
\>$->$\> $(\text{while}_{\langle n\rangle} (\mathcal{E}\llbracket e\rrbracket\vals) \; \{S\}, {m(\mathfrak{f}, \text{loop}_c, \vals)})$ by the EEval rule\\
\>$->$\> $({S;\text{while}_{\langle n\rangle} (e) \; \{S\}}, m(\mathfrak{f}, \text{loop}_c[1/n_1], \vals))$ by the Wh-T rule.
\end{tabbing}

The lemma holds.

{\bf Induction step}.

The hypothesis IH is that, if there is a configuration $(S', m_i(\text{loop}_c^i))$ reachable from $(s, m)$ in which the loop counter of $s$ is $i$, $\text{loop}_c^i(n) = i>0$, then there is a reachable configuration $(S;s, m_i(\mathfrak{f}, \text{loop}_c^i))$ from $(s, m)$ where the loop counter of $s$ is $i$ and the crash flag is not set.

Then we show that, if there is a configuration $(S', m_{i+1}(\text{loop}_c^{i+1}))$ reachable from $(s, m)$ in which the loop counter of $s$ is $i+1$, then there is a reachable configuration $(S;s, m_{i+1}(\mathfrak{f}, \text{loop}_c^{i+1}))$ from $(s, m)$ where the loop counter of $s$ is $i+1$ and the crash flag $\mathfrak{f}$ is not set.

By Lemma~\ref{lmm:loopCntStepwiseInc}, the loop counter of $s$ is increasing by one in one step. Hence, there must be one configuration reachable from $(s, m)$ in which the loop counter of $s$ is $i$.
By hypothesis, there is the configuration $(S;s, m_{i}(\mathfrak{f}, \text{loop}_c^i, \vals_i))$ reachable from $(s, m)$ in which the loop counter is $i$, $\text{loop}_c^i(n) = i$, and the crash flag is not set, $\mathfrak{f} = 0$.
By the assumption of unique loop labels, $s\notin S$. Then the loop counter of $s$ is not redefined in the execution of $S$ started in state $m_{i}(\mathfrak{f}, \text{loop}_c^i, \vals_i)$. Because there is a configuration in which the loop counter of  $s$ is $i+1$, then the execution of $S$ when started in the state $m_{i}(\mathfrak{f}, \text{loop}_c^i, \vals_i)$ terminates,
$(S, m_{i}(\mathfrak{f}, \text{loop}_c^i, \vals_i)) ->* (\text{skip}, m_{i+1}(\mathfrak{f}, \text{loop}_c^{i+1}, \vals_{i+1}))$ where
$\mathfrak{f} = 0$ and $\text{loop}_c^{i+1}(n) = i$.
By Corollary~\ref{coro:termSeq},
$(S;s, m_{i}(\mathfrak{f}, \text{loop}_c^i, \vals_i)) ->* (s, m_{i+1}(\mathfrak{f}, \text{loop}_c^{i+1}, \vals_{i+1}))$.
By similar argument in base case, the evaluation of the predicate expression w.r.t the value store $\vals_{i+1}$ produce nonzero integer value. The execution of $s$ proceeds as follows:

\begin{tabbing}
xx\=xx\=\kill
\>\>     $(s, m_{i+1}(\mathfrak{f}, \text{loop}_c^{i+1}, \vals_{i+1}))$\\
\>= \>   $(\text{while}_{\langle n\rangle} (e) \; \{S\}, {m_{i+1}(\mathfrak{f}, \text{loop}_c^{i+1}, \vals_{i+1})})$\\
\>$->$\> $(\text{while}_{\langle n\rangle} (\mathcal{E}\llbracket e\rrbracket\vals_{i+1}) \; \{S\}, {m_{i+1}(\mathfrak{f}, \text{loop}_c^{i+1}, \vals_{i+1})})$\\
\>\>      by the EEval rule\\
\>$->$\> $({S;\text{while}_{\langle n\rangle} (e) \; \{S\}}, {{m_{i+1}(\mathfrak{f}, \text{loop}_c^{i+1}[(i+1)/n], \vals_{i+1})}})$\\
\>\>      by the Wh-T rule.
\end{tabbing}
The lemma holds.
\end{proof}

\begin{lemma}\label{lemma:loopSameIOSeq}
Let $s_1 = ``\text{while}_{\langle n_1\rangle}(e) \; \{S_{1}\}"$ and
    $s_2 = ``\text{while}_{\langle n_2\rangle}(e) \;\allowbreak \{S_{2}\}"$ be two while statements and all of the followings hold:
\begin{itemize}
\item There are output statements in $s_1$ and $s_2$, $\exists e_1\,e_2\,:\, (``\text{output }e_1" \in s_1) \wedge (``\text{output }e_2" \in s_2)$;

\item $s_1$ and $s_2$ have the same set of termination deciding variables relative to output, and the same set of imported variables relative to output, $(\text{TVar}_o(s_1) = \text{TVar}_o(s_2) = \text{TVar}(s)) \wedge
                        (\text{Imp}_o(s_1) = \text{Imp}_o(s_2) = \text{Imp}(io))$;

\item Loop bodies $S_1$ and $S_2$ satisfy the proof rule of equivalent computation of the out-deciding variables of $s_1$ and $s_2$, $\forall x \in \text{OVar}(s) = \text{TVar}(s) \cup \text{Imp}(io)\,:\, S_{1} \equiv_{x}^S S_{2}$;

\item Loop bodies $S_1$ and $S_2$ satisfy the proof rule of termination in the same way, $S_{1} \equiv_{H}^S S_{2}$;

\item Loop bodies $S_1$ and $S_2$ produce the same output sequence when started in states with crash flags not set and whose value stores agree on values of variables in $\text{OVar}(S_{1}) \cup \text{OVar}(S_{2})$,
    $\forall m_{S_{1}}(\mathfrak{f}_1, \vals_{S_{1}})\, m_{S_{2}}(\mathfrak{f}_2, \vals_{S_{2}})\,:$

    $((\mathfrak{f}_1 = \mathfrak{f}_2 = 0) \wedge
     (\forall x \in \text{OVar}(S_{1}) \cup \text{OVar}(S_{2})\,:\, \vals_{S_{1}}(x) =\,\allowbreak \vals_{S_{2}}(x))) =>$

    ${((S_{1}, m_{S_{1}}(\mathfrak{f}_1, \vals_{S_{1}})) \equiv_{O} (S_{2}, m_{S_{2}}(\mathfrak{f}_2, \vals_{S_{2}})))}$.
\end{itemize}

If $s_1$ and $s_2$ start in states $\state_1(\mathfrak{f}_1, \text{loop}_c^1, \vals_1), \state_2(\mathfrak{f}_2, \text{loop}_c^2, \vals_2)$ respectively with crash flags not set $\mathfrak{f}_1 = \mathfrak{f}_2 = 0$ and
in which $s_1$ and $s_2$ have not started execution $(\text{loop}_c^1(n_1) = \text{loop}_c^2(n_2) = 0)$,
  value stores $\vals_1$ and $\vals_2$ agree on values of variables in $\text{OVar}(s)$,
  $\forall x \in \text{OVar}(s)\,:\, \vals_1(x) = \vals_2(x)$,
then one of the followings holds:
\begin{enumerate}
\item $s_1$ and $s_2$ both terminate and produce the same output sequence:

$({s_1}, m_1) ->* (\text{skip}, m_1'(\vals_1'))$,
$({s_2}, m_2) ->* (\text{skip}, m_2'(\vals_2'))$ where $\vals_1'({id}_{IO}) = \vals_2'({id}_{IO})$.
\item $s_1$ and $s_2$ both do not terminate, $\forall k> 0$,
$({s_1}, m_1)$ {\kStepArrow [k] } $({S_{1_k}}, m_{1_k})$,
$({s_2}, m_2)$ {\kStepArrow [k] } $({S_{2_k}}, m_{2_k})$
where $S_{1_k} \neq \text{skip}$, $S_{2_k} \neq \text{skip}$ and one of the followings holds:

\begin{enumerate}
\item{For any positive integer $i$, there are two configurations $(s_1, m_{1_i})$ and $(s_2, m_{2_i})$ reachable from $(s_1, m_1)$ and $(s_2, m_2)$, respectively, in which both crash flags are not set, the loop counters of $s_1$ and $s_2$ are equal to $i$ and value stores agree on values of variables in $\text{OVar}(s)$, and for every state in execution, $(s_1, m_1) ->* (s_1, m_{1_i})$ or $(s_2, m_2) ->* (s_2, m_{2_i})$, loop counters for $s_1$ and $s_2$ are less than or equal to $i$ respectively:}

 $\forall i>0 \, \exists (s_1, m_{1_i}) \, (s_2, m_{2_i})\,:\, ({s_1}, m_1) ->* ({s_1}, m_{1_i}(\mathfrak{f}_1,\;\allowbreak \text{loop}_c^{1_i}, \vals_{1_i})) \wedge
                                              ({s_2}, m_2) ->* ({s_2}, m_{2_i}(\mathfrak{f}_2, \text{loop}_c^{2_i}, \vals_{2_i}))$ where
\begin{itemize}
 \item $\mathfrak{f}_1 = \mathfrak{f}_2 = 0$; and

 \item $\text{loop}_c^{1_i}(n_1) = \text{loop}_c^{2_i}(n_2) = i$; and

 \item $\forall x \in \text{OVar}(s)\,:\, \vals_{1_i}(x)$ = $\vals_{2_i}(x)$.

 \item $\forall m_1'\,:\, ({s_1}, m_1) ->* (S_1', m_1'(\text{loop}_c^{1'})) ->* ({s_1}, m_{1_i}(\text{loop}_c^{1_i},\allowbreak \vals_{1_i}))$,
     \noindent$\text{loop}_c^{1'}(n_1) \leq i$; and

 \item $\forall m_2'\,:\, ({s_2}, m_2) ->* (S_2', m_2'(\text{loop}_c^{2'})) ->* ({s_2}, m_{2_i}(\text{loop}_c^{2_i},\allowbreak \vals_{2_i}))$,
 \noindent$\text{loop}_c^{2'}(n_2) \leq i$;
\end{itemize}

\item{The loop counters for $s_1$ and $s_2$ are less than a smallest positive integer $i$ and all of the followings hold:}

\begin{itemize}
\item $\exists i>0\, \forall m_1',m_2' \,:\,       ({s_1}, m_1) ->* (S_{1}', m_1'(\text{loop}_c^{1'}))$,

                                     \noindent$({s_2}, m_2) ->* (S_{2}', m_2'(\text{loop}_c^{2'}))$ where
      $\text{loop}_c^{1'}(n_1) < i$,

      \noindent$\text{loop}_c^{2'}(n_2) < i$;

\item{$\forall 0 < j < i$, there are two configurations $(s_1, m_{1_j})$ and $(s_2, m_{2_j})$ reachable from $(s_1, m_1)$ and $(s_2, m_2)$, respectively, in which both crash flags are not set, loop counters of $s_1$ and $s_2$ are equal to $j$ and value stores agree on values of variables in $\text{OVar}(s)$:}

 $\exists (s_1, m_{1_j}), (s_2, m_{2_j})\,:\, ({s_1}, m_1) ->* ({s_1}, m_{1_j}(\mathfrak{f}_1, \text{loop}_c^{1_j},\allowbreak \vals_{1_j})) \wedge
                                              ({s_2}, m_2) ->* ({s_2}, m_{2_j}(\mathfrak{f}_2, \text{loop}_c^{2_j},\allowbreak \vals_{2_j}))$ where
\begin{itemize}
 \item $\mathfrak{f}_1 = \mathfrak{f}_2 = 0$; and

 \item $\text{loop}_c^{1_j}(n_1) = \text{loop}_c^{2_j}(n_2) = j$; and

 \item $\forall x \in \text{OVar}(s)\,:\, \vals_{1_j}(x)$ = $\vals_{2_j}(x)$.

%
\end{itemize}

\item If $i=1$, then the I/O sequence is not redefined in any states reachable from
$({s_1}, m_1)$ and $({s_2}, m_2)$.

\begin{itemize}
\item
$\forall m_1''\,:\, $
$ ({s_1}, m_1(\text{loop}_c^1, \vals_1)) ->* ({S_1''}, {m_1''(\vals_1'')})$ where
 $\vals_1''({id}_{IO}) = \vals_1({id}_{IO})$.

\item
$\forall m_2''\,:\, $
$({s_2}, m_2(\text{loop}_c^2, \vals_2)) ->* ({S_2''}, {m_2''(\vals_2'')})$ where
 $\vals_2''({id}_{IO}) = \vals_2({id}_{IO})$.

\end{itemize}

\item If $i>1$, then the I/O sequence is not redefined in any states reachable from
$({s_1}, m_{1_{i-1}})$ and $({s_2}, m_{2_{i-1}})$.

\begin{itemize}
\item
$\forall m_1''\,:\, $
$ ({s_1}, m_{1_{i-1}}(\text{loop}_c^{1_{i-1}}, \vals_{1_{i-1}})) ->* ({S_1''}, {m_1''(\vals_1'')})$ where
 $\vals_1''({id}_{IO}) = \vals_{1_{i-1}}({id}_{IO})$.

\item
$\forall m_2''\,:\, $
$({s_2}, m_{2_{i-1}}(\text{loop}_c^{2_{i-1}}, \vals_{2_{i-1}})) ->* ({S_2''}, {m_2''(\vals_2'')})$ where
 $\vals_2''({id}_{IO}) = \vals_{2_{i-1}}({id}_{IO})$.
\end{itemize}

\end{itemize}

\item{The loop counters for $s_1$ and $s_2$ are less than or equal to a smallest positive integer $i$ and all of the followings hold:}

\begin{itemize}
\item $\exists i>0\,\forall m_1',m_2' \,:\,       ({s_1}, m_1) ->* (S_{1}', m_1'(\text{loop}_c^{1'})),
                                     ({s_2},\allowbreak m_2) ->* (S_{2}', m_2'(\text{loop}_c^{2'}))$ where
      $\text{loop}_c^{1'}(n_1) \leq i$,
      \noindent$\text{loop}_c^{2'}(n_2) \leq i$;

\item There are no configurations $(s_1, m_{1_i})$ and $(s_2, m_{2_i})$ reachable from $(s_1, m_1)$ and $(s_2, m_2)$, respectively, in which  crash flags are not set, the loop counters of $s_1$ and $s_2$ are equal to $i$, and value stores agree on values of variables in $\text{OVar}(s)$:

$\nexists (s_1, m_{1_i}), (s_2, m_{2_i})\,:\, ({s_1}, m_1) ->* ({s_1}, m_{1_i}(\mathfrak{f}_1, \text{loop}_c^{1_i},\allowbreak \vals_{1_i})) \wedge
       ({s_2}, m_2) ->* ({s_2}, m_{2_i}(\mathfrak{f}_2, \text{loop}_c^{2_i}, \vals_{2_i}))$ where
      \begin{itemize}
        \item $\mathfrak{f}_1 = \mathfrak{f}_2 = 0$; and

        \item $\text{loop}_c^{1_i}(n_1) = \text{loop}_c^{2_i}(n_2) = i$; and

        \item $\forall x \in \text{OVar}(s)\,:\, \vals_{1_i}(x)$ = $\vals_{2_i}(x)$.
      \end{itemize}

\item{$\forall 0 < j < i$, there are two configurations $(s_1, m_{1_j})$ and $(s_2, m_{2_j})$ reachable from $(s_1, m_1)$ and $(s_2, m_2)$, respectively, in which crash flags are not set, the loop counters of $s_1$ and $s_2$ are equal to $j$ and value stores agree on values of variables in $\text{OVar}(s)$:}

 $\exists (s_1, m_{1_j}), (s_2, m_{2_j})\,:\, ({s_1}, m_1) ->* ({s_1}, m_{1_j}(\mathfrak{f}_1, \text{loop}_c^{1_j},\allowbreak \vals_{1_j})) \wedge
                                              ({s_2}, m_2) ->* ({s_2}, m_{2_j}(\mathfrak{f}_2, \text{loop}_c^{2_j}, \vals_{2_j}))$ where
\begin{itemize}
 \item $\mathfrak{f}_1 = \mathfrak{f}_2 = 0$; and

 \item $\text{loop}_c^{1_j}(n_1) = \text{loop}_c^{2_j}(n_2) = j$; and

 \item $\forall x \in \text{OVar}(s)\,:\, \vals_{1_j}(x)$ = $\vals_{2_j}(x)$.

%
\end{itemize}

\item If $i=1$, then executions from $({s_1}, m_1)$ and $({s_2}, m_2)$ produce the same output sequence:

$({s_1}, m_1(\text{loop}_c^1, \vals_1)) \equiv_{O} ({s_2}, m_2(\text{loop}_c^2, \vals_2))$.

\item If $i>1$, then executions from
 $({s_1}, m_{1_{i-1}})$ and $({s_2}, m_{2_{i-1}})$
 produce the same output sequence:

 $({s_1}, m_{1_{i-1}}(\text{loop}_c^{1_{i-1}}, \vals_{1_{i-1}})) \equiv_{O} ({s_2}, m_{2_{i-1}}(\text{loop}_c^{2_{i-1}},\allowbreak \vals_{2_{i-1}}))$.
\end{itemize}

\end{enumerate}
\end{enumerate}
\end{lemma}
\begin{proof}
We show that $s_1$ and $s_2$ terminate in the same way when started in states
$m_1(\mathfrak{f}_1, \text{loop}_c^1, \vals_1)$ and $m_2(\mathfrak{f}_2, \text{loop}_c^2, \vals_2)$ respectively,  $(s_1, m_1) \equiv_{H} (s_2, m_2)$. In addition, we show that $s_1$ and $s_2$ produce the same output sequence in every possibilities of termination in the same way, $(s_1, m_1) \equiv_{O} (s_2,m_2)$.

By definition, $s_1$ and $s_2$ satisfy the proof rule of termination in the same way because
\begin{itemize}
\item Loop bodies $S_1$ and $S_2$ satisfy the proof rule of termination in the same way;

By assumption.

\item Loop bodies $S_1$ and $S_2$ satisfy the proof rule of equivalent computation of those in the termination deciding variables of $s_1$ and $s_2$, $\forall x \in \text{TVar}(s_1) \cup \text{TVar}(s_2)\,:\, S_{1} \equiv_{x}^S S_{2}$;

By the definition of $\text{OVar}(s)$, $\text{TVar}_o(s_1) \subseteq \text{OVar}(s_1)$ and $\text{TVar}_o(s_2) \subseteq \text{OVar}(s_2)$.
By the definition of $\text{TVar}_o$, $\text{TVar}_o(s_1) = \text{TVar}(s_1)$ and $\text{TVar}_o(s_2) = \text{TVar}(s_2)$.
\end{itemize}
By Lemma~\ref{lmm:loopTermInSameWayWithLoopCnt}, we show $s_1$ and $s_2$ terminate in the same way when started in states
$m_1(\mathfrak{f}_1, \text{loop}_c^1, \vals_1)$ and $m_2(\mathfrak{f}_2, \text{loop}_c^2, \vals_2)$.
We need to show that all the required conditions are satisfied.
\begin{itemize}
\item Crash flags are not set, $\mathfrak{f}_1 = \mathfrak{f}_2 = 0$;

\item Loop counters of $s_1$ and $s_2$ are initially zero, $\text{loop}_c^1(n_1) = \text{loop}_c^2(n_2) = 0$;

\item $s_1$ and $s_2$ have same set of termination deciding variables, $\text{TVar}(s_1) = \text{TVar}(s_2) = \text{TVar}(s)$;

\item Value stores $\vals_1$ and $\vals_2$ agree on values of variables in $\text{TVar}(s_1) = \text{TVar}(s_2)$,
$\forall x \in \text{TVar}(s)\,:\, \vals_1(x) = \vals_2(x)$;

  The above four conditions are from assumption.

\item Loop bodies $S_1$ and $S_2$ terminate in the same way when started in states with crash flags not set and whose value stores agree on values of variables in $\text{TVar}(S_1) \cup \text{TVar}(S_2)$;

    By Theorem~\ref{thm:mainTermSameWayLocal}.
\end{itemize}
Therefore, by Lemma~\ref{lmm:loopTermInSameWayWithLoopCnt}, we have one of the followings holds:
\begin{enumerate}
\item $s_1$ and $s_2$ both terminate and the loop counters of $s_1$ and $s_2$ are less than a positive integer $i$ such that the loop counters of $s_1$ and $s_2$ are less than or equal to $i-1$:

$({s_1}, m_1) ->* (\text{skip}, m_1')$,
$({s_2}, m_2) ->* (\text{skip}, m_2')$.

We show that, when $s_1$ and $s_2$ terminate, value stores of $s_1$ and $s_2$ agree on the value of the I/O sequence variable
 by Lemma~\ref{lmm:sameFinalValueX}. We need to show all the required conditions hold.
\begin{itemize}
\item $\forall x \in \text{Imp}(io)\,:\, \vals_{1}(x) = \vals_{2}(x)$;

\item $\text{loop}_c^1(n_1) = \text{loop}_c^2(n_2) = 0$;

  The above two conditions are from assumption.

\item ${id}_{IO} \in \text{Def}(s_1) \cap \text{Def}(s_2)$;

Because there are output statements in $s_1$ and $s_2$. By the definition of $\text{Def}(\cdot)$, the I/O sequence variable is defined in $s_1$ and $s_2$.

\item $\text{Imp}(s_1, \{{id}_{IO}\}) = \text{Imp}(s_2, \{{id}_{IO}\}) = \text{Imp}(io)$;

By the definition of $\text{Imp}_o(\cdot)$, $\text{Imp}_o(s_1) = \text{Imp}(s_1, \{{id}_{IO}\}),\allowbreak
\text{Imp}_o(s_2) = \text{Imp}(s_2, \{{id}_{IO}\})$.

\item $\forall y \in \text{Imp}(io), \forall m_{S_1}(\vals_{S_1})\, m_{S_2}(\vals_{S_2}):$

$((\forall z\in \text{Imp}(S_1, \text{Imp}(io)) \cup \text{Imp}(S_2, \text{Imp}(io)), \vals_{S_1}(z) = \vals_{S_2}(z))
=> (S_1, m_{S_1}(\vals_{S_1})) \equiv_{y} (S_2, m_{S_2}(\vals_{S_2})))$.

By Theorem~\ref{thm:equivCompMain}.
\end{itemize}

In addition, by the semantic rules, the I/O sequence is appended at most by one value in one step.
Hence, $s_1$ and $s_2$ produce the same output sequence when started in states $m_1$ and $m_2$ respectively.

\item $s_1$ and $s_2$ both do not terminate, $\forall k> 0$,
$({s_1}, m_1)$ {\kStepArrow [k] } $({S_{1_k}}, m_{1_k})$,
$({s_2}, m_2)$ {\kStepArrow [k] } $({S_{2_k}}, m_{2_k})$
where $S_{1_k} \neq \text{skip}$, $S_{2_k} \neq \text{skip}$ and one of the followings holds:

\begin{enumerate}
\item{$\forall i>0$, there are two configurations $(s_1, m_{1_i})$ and $(s_2, m_{2_i})$ reachable from $(s_1, m_1)$ and $(s_2, m_2)$, respectively, in which both crash flags are not set, the loop counters of $s_1$ and $s_2$ are equal to $i$ and value stores agree on the values of variables in $\text{TVar}(s)$:}

 $\forall i>0\, \exists (s_1, m_{1_i})\, (s_2, m_{2_i})\,:\, ({s_1}, m_1) ->* ({s_1}, m_{1_i}(\mathfrak{f}_1, \,\allowbreak \text{loop}_c^{1_i}, \vals_{1_i})) \wedge
                                              ({s_2}, m_2) ->* ({s_2}, m_{2_i}(\mathfrak{f}_2, \text{loop}_c^{2_i}, \vals_{2_i}))$ where
\begin{itemize}
 \item $\mathfrak{f}_1 = \mathfrak{f}_2 = 0$; and

 \item $\text{loop}_c^{1_i}(n_1) = \text{loop}_c^{2_i}(n_2) = i$; and

 \item $\forall x \in \text{TVar}(s)\,:\, \vals_{1_i}(x)$ = $\vals_{2_i}(x)$.

 \item $\forall m_1'\,:\, ({s_1}, m_1) ->* (S_1', m_1'(\text{loop}_c^{1'})) ->* ({s_1}, m_{1_i}(\text{loop}_c^{1_i},\allowbreak \vals_{1_i})), \; \text{loop}_c^{1'}(n_1) \leq i$; and

 \item $\forall m_2'\,:\, ({s_2}, m_2) ->* (S_2', m_2'(\text{loop}_c^{2'})) ->* ({s_2}, m_{2_i}(\text{loop}_c^{2_i},\allowbreak \vals_{2_i})), \; \text{loop}_c^{2'}(n_2) \leq i$;
\end{itemize}

We show that, for any positive integer $i$, value stores $\vals_{1_i}$ and $\vals_{2_i}$ agree on values of variables in $\text{Imp}(io)$ by the proof of Lemma~\ref{lmm:equivTermCompSameLoopIteration}.
We need to show that all the required conditions are satisfied.
\begin{itemize}
\item $\forall x \in \text{Imp}(io)\,:\, \vals_{1}(x) = \vals_{2}(x)$;

\item $\text{loop}_c^1(n_1) = \text{loop}_c^2(n_2) = 0$;

The above two conditions are by assumption.

\item ${id}_{IO} \in \text{Def}(s_1) \cap \text{Def}(s_2)$;


\item $\text{Imp}(s_1, \{{id}_{IO}\}) = \text{Imp}(s_2, \{{id}_{IO}\}) = \text{Imp}(io)$;

The above two conditions are obtained by similar argument in the case that $s_1$ and $s_2$ both terminate.

\item $\forall y \in \text{Imp}(io), \forall m_{S_1}(\vals_{S_1})\, m_{S_2}(\vals_{S_2}):$

$((\forall z\in \text{Imp}(S_1, \text{Imp}(io)) \cup \text{Imp}(S_2, \text{Imp}(io)), \vals_{S_1}(z) = \vals_{S_2}(z))
=>$

$(S_1, m_{S_1}(\vals_{S_1})) \equiv_{y} (S_2, m_{S_2}(\vals_{S_2})))$.

By Theorem~\ref{thm:equivCompMain}.
\end{itemize}
We cannot apply Lemma~\ref{lmm:equivTermCompSameLoopIteration} directly because $s_1$ and $s_2$ do not terminate. But we can still have the proof closely similar to that of Lemma~\ref{lmm:equivTermCompSameLoopIteration}
by using the fact that there exists a configuration of arbitrarily large loop counters of $s_1$ and $s_2$ and in which crash flags are not set.

Then, $\forall i>0, \forall x \in \text{Imp}(io)\,:\, \vals_{1_i}(x) = \vals_{2_i}(x)$.
In addition, by the semantic rules, the I/O sequence is appended at most by one value in one step.
The lemma holds.

\item{The loop counters for $s_1$ and $s_2$ are less than a smallest positive integer $i$ and all of the followings hold:}

\begin{itemize}
\item $\exists i>0\,\forall m_1',m_2' \,:\,       ({s_1}, m_1) ->* (S_{1}', m_1'(\text{loop}_c^{1'})),
                                     ({s_2},\allowbreak m_2) ->* (S_{2}', m_2'(\text{loop}_c^{2'}))$ where
      $\text{loop}_c^{1'}(n_1) < i$,

      \noindent$\text{loop}_c^{2'}(n_2) < i$;

\item{$\forall 0 < j < i$, there are two configurations $(s_1, m_{1_j})$ and $(s_2, m_{2_j})$ reachable from $(s_1, m_1)$ and $(s_2, m_2)$, respectively, in which both crash flags are not set, the loop counters of $s_1$ and $s_2$ are equal to $j$ and value stores agree on the values of variables in $\text{TVar}(s)$:}

 $\exists (s_1, m_{1_j}), (s_2, m_{2_j})\,:\, ({s_1}, m_1) ->* ({s_1}, m_{1_j}(\mathfrak{f}_1, \text{loop}_c^{1_j},\allowbreak \vals_{1_j})) \wedge
                                              ({s_2}, m_2) ->* ({s_2}, m_{2_j}(\mathfrak{f}_2, \text{loop}_c^{2_j}, \vals_{2_j}))$ where
\begin{itemize}
 \item $\mathfrak{f}_1 = \mathfrak{f}_2 = 0$; and

 \item $\text{loop}_c^{1_j}(n_1) = \text{loop}_c^{2_j}(n_2) = j$; and

 \item $\forall x \in \text{TVar}(s)\,:\, \vals_{1_j}(x)$ = $\vals_{2_j}(x)$.

%
\end{itemize}
\end{itemize}

This case corresponds to the situation that the $i$th evaluations of the predicate expression of $s_1$ and $s_2$ raise an exception. There are two possibilities regarding the value of $i$.
\begin{enumerate}
\item
$i=1$.

 $s_1$ and $s_2$ raise an exception in the 1st evaluation of their predicate expression because loop counters of $s_1$ and $s_2$ are less than 1.
 In the proof of Lemma~\ref{lmm:loopTermInSameWayWithLoopCnt}, value stores $\vals_1$ and $\vals_2$ in states $m_1$ and $m_2$ respectively are not modified in the 1st evaluation of the predicate expression of $s_1$ and $s_2$.
 In addition, value stores $\vals_1$ and $\vals_2$ are not modified after $s_1$ and $s_2$ both crash according to the rule Crash.
We have the corresponding initial state in which value stores $\vals_1$ and $\vals_2$ agree on values of variables in $\text{OVar}(s)$. Thus, $\vals_1({id}_{IO}) = \vals_2({id}_{IO})$. The lemma holds.

\item $i>1$.

We show that, for any positive integer $0<j<i$, value stores $\vals_{1_j}$ and $\vals_{2_j}$ agree on values of variables in $\text{Imp}(s)$ by the proof of Lemma~\ref{lmm:equivTermCompSameLoopIteration}.
We need to show that all the required conditions are satisfied.
\begin{itemize}
\item $\forall x \in \text{Imp}(io)\,:\, \vals_{1}(x) = \vals_{2}(x)$;

\item $\text{loop}_c^1(n_1) = \text{loop}_c^2(n_2) = 0$;

The above two conditions are from assumption.

\item ${id}_{IO} \in \text{Def}(s_1) \cap \text{Def}(s_2)$;

\item $\text{Imp}(s_1, \{{id}_{IO}\}) = \text{Imp}(s_2, \{{id}_{IO}\}) = \text{Imp}(io)$;

The above two conditions are obtained by the same argument in the case that $s_1$ and $s_2$ terminate.

\item $\forall y \in \text{Imp}(io), \forall m_{S_1}(\vals_{S_1})\, m_{S_2}(\vals_{S_2}):$

$((\forall z\in \text{Imp}(S_1, \text{Imp}(io)) \cup \text{Imp}(S_2, \text{Imp}(io)), \vals_{S_1}(z) = \vals_{S_2}(z))
=>$

$(S_1, m_{S_1}(\vals_{S_1})) \equiv_{y} (S_2, m_{S_2}(\vals_{S_2})))$.

By Theorem~\ref{thm:equivCompMain}.
\end{itemize}
We cannot apply Lemma~\ref{lmm:equivTermCompSameLoopIteration} directly because $s_1$ and $s_2$ do not terminate. But we can still have the proof closely similar to that of Lemma~\ref{lmm:equivTermCompSameLoopIteration}
by using the fact that there are reachable configurations $(s_1, m_{1_{i-1}})$ and $(s_2, m_{2_{i-1}})$ with the loop counters of $s_1$ and $s_2$ of value $i-1$ and crash flags not set.

 By assumption, there is configuration $(s_1, m_{1_{i-1}}(\mathfrak{f}_1,\,\,\allowbreak \text{loop}_c^{1_{i-1}}, \vals_{1_{i-1}}))$ reachable from $(s_1, m_1)$ in which the loop counter of $s_1$ is $i-1$ and the crash flag is not set;
there is also a configuration $(s_2, m_{2_{i-1}}(\mathfrak{f}_2, \text{loop}_c^{2_{i-1}}, \vals_{2_{i-1}}))$ of $s_2$ reachable from $(s_2, m_2)$ in which the loop counter is $i-1$ and the crash flag is not set.
In addition, value stores $\vals_{1_{i-1}}$ and $\vals_{2_{i-1}}$ agree on values of variables in
$\text{Imp}(io)$.
In the proof of Lemma~\ref{lmm:loopTermInSameWayIntermediate}, the $i$th evaluations of the predicate expression of $s_1$ and $s_2$ must raise an exception because loop counters of $s_1$ and $s_2$ are less than $i$.
Then the I/O sequence is not redefined in any state reachable from
$(s_1, m_{1_{i-1}}(\mathfrak{f}_1, \text{loop}_c^{1_{i-1}}, \vals_{1_{i-1}}))$ and
$(s_2, m_{2_{i-1}}(\mathfrak{f}_2, \text{loop}_c^{2_{i-1}}, \vals_{2_{i-1}}))$ respectively.
In addition, by the semantic rules, the I/O sequence is appended at most by one value in one step.
The lemma holds.
\end{enumerate}

\item{The loop counters for $s_1$ and $s_2$ are less than or equal to a smallest positive integer $i$ and all of the followings hold:}

\begin{itemize}
\item $\exists i>0\,\forall m_1',m_2' \,:\,       ({s_1}, m_1) ->* (S_{1}', m_1'(\text{loop}_c^{1'})),
                                     ({s_2}, m_2) ->* (S_{2}', m_2'(\text{loop}_c^{2'}))$ where
      $\text{loop}_c^{1'}(n_1) \leq i$,

      \noindent$\text{loop}_c^{2'}(n_2) \leq i$;

\item There are no configurations $(s_1, m_{1_i})$ and $(s_2, m_{2_i})$ reachable from $(s_1, m_1)$ and $(s_2, m_2)$, respectively, in which  crash flags are not set, the loop counters of $s_1$ and $s_2$ are equal to $i$, and value stores agree on values of variables in $\text{TVar}(s)$:

$\nexists (s_1, m_{1_i}), (s_2, m_{2_i})\,:\, ({s_1}, m_1) ->* ({s_1}, m_{1_i}(\mathfrak{f}_1, \text{loop}_c^{1_i}, \vals_{1_i})) \wedge
       ({s_2}, m_2) ->* ({s_2}, m_{2_i}(\mathfrak{f}_2, \text{loop}_c^{2_i}, \vals_{2_i}))$ where
      \begin{itemize}
        \item $\mathfrak{f}_1 = \mathfrak{f}_2 = 0$; and

        \item $\text{loop}_c^{1_i}(n_1) = \text{loop}_c^{2_i}(n_2) = i$; and

        \item $\forall x \in \text{TVar}(s)\,:\, \vals_{1_i}(x)$ = $\vals_{2_i}(x)$.
      \end{itemize}

\item{$\forall 0 < j < i$, there are two configurations $(s_1, m_{1_j})$ and $(s_2, m_{2_j})$ reachable from $(s_1, m_1)$ and $(s_2, m_2)$, respectively, in which both crash flags are not set, the loop counters of $s_1$ and $s_2$ are equal to $j$ and value stores agree on values of variables in $\text{TVar}(s)$:}

 $\exists (s_1, m_{1_j}), (s_2, m_{2_j})\,:\, ({s_1}, m_1) ->* ({s_1}, m_{1_j}(\mathfrak{f}_1, \text{loop}_c^{1_j}, \vals_{1_j})) \wedge
                                              ({s_2}, m_2) ->* ({s_2}, m_{2_j}(\mathfrak{f}_2, \text{loop}_c^{2_j}, \vals_{2_j}))$ where
\begin{itemize}
 \item $\mathfrak{f}_1 = \mathfrak{f}_2 = 0$; and

 \item $\text{loop}_c^{1_j}(n_1) = \text{loop}_c^{2_j}(n_2) = j$; and

 \item $\forall x \in \text{TVar}(s)\,:\, \vals_{1_j}(x)$ = $\vals_{2_j}(x)$.

%
\end{itemize}
\end{itemize}

This case corresponds to the situation that, the $i$th evaluation of the predicate expression of $s_1$ and $s_2$ produce same nonzero integer value and loop bodies $S_1$ and $S_2$ do not terminate after the $i$th evaluation of the predicate expression of $s_1$ and $s_2$.
There are two possibilities regarding the value of $i$.
\begin{enumerate}
\item $i=1$.

By assumption, we have the initial value stores $\vals_1$ and $\vals_2$ agree on values of variables in $\text{OVar}(s)$.
In the proof of Lemma~\ref{lmm:loopTermInSameWayIntermediate}, the execution of $s_1$ proceeds as follows:

\begin{tabbing}
xx\=xx\=\kill
\>\>     $({s_1}, m_{1}(\text{loop}_c^{1},\vals_{1}))$\\
\>= \>   $(\text{while}_{\langle n_1\rangle} (e) \; \{S_1\}, {m_{1}(\text{loop}_c^{1}, \vals_{1})})$\\
\>$->$\> $(\text{while}_{\langle n_1\rangle} (v) \; \{S_1\}, {m_{1}(\text{loop}_c^{1}, \vals_{1})})$ by the EEval rule\\
\>$->$\> $(S_1;\text{while}_{\langle n_1\rangle} (e) \; \{S_1\}, m_{1}(\text{loop}_c^{1}[1/n_1], \vals_{1}))$\\
\>\>     by the Wh-T rule.
\end{tabbing}

The execution of $s_2$ proceeds to $(S_2;\text{while}_{\langle n_2\rangle} (e) \; \{S_2\},\allowbreak m_{2}(\text{loop}_c^{2}[1/n_2], \vals_{2}))$.
Then the execution of $S_1$ and $S_2$ do not terminate when started in states
$m_{1}(\text{loop}_c^{1}[1/n_1], \vals_{1})$ and $m_{2}(\text{loop}_c^{2}[1/n_2], \vals_{2})$.
By assumption, value stores $\vals_1$ and $\vals_2$ agree on values of the out-deciding variables of $s_1$ and $s_2$,
$\forall x \in \text{OVar}(s_1) \cup \text{OVar}(s_2)\,:\, \vals_1(x) = \vals_2(x)$.
By definition,
$\text{Imp}(S_1, \{{id}_{IO}\}) \subseteq \text{Imp}(s_1, \{{id}_{IO}\}),
 \text{CVar}(S_1) \subseteq \text{CVar}(s_1)$ and
$\text{LVar}(S_1) \subseteq \text{LVar}(s_1)$.
Thus, $\text{TVar}(S_1) \subseteq \text{TVar}(s_1)$.
By Lemma~\ref{lmm:TVarOIsSubsetOfTVar}, $\text{TVar}_o(S_1) \subseteq \text{TVar}(S_1)$.
By Lemma~\ref{lmm:ImpOVarIsSubsetOfImpIO}, $\text{Imp}_o(S_1) \subseteq \text{Imp}(S_1, {id}_{IO})$.
In conclusion, $\text{OVar}(S_1) \subseteq \text{OVar}(s_1)$. Similarly, $\text{OVar}(S_2) \subseteq \text{OVar}(s_2)$.
Thus, $\forall x \in \text{OVar}(S_1) \cup \text{OVar}(S_2)\,:\, \vals_1(x) = \vals_2(x)$.
Then executions of $S_1$ and $S_2$ when started from states

\noindent$m_{1}(\text{loop}_c^{1}[1/n_1], \vals_{1})$ and $m_{2}(\text{loop}_c^{2}[1/n_2], \vals_{2})$ produce the same output sequence:

\noindent$(S_1, m_{1}(\text{loop}_c^{1}[1/n_1], \vals_{1})) \equiv_{O} (S_2, m_{2}(\text{loop}_c^{2}[1/n_2], \vals_{2}))$.
 In addition, by the semantic rules, the I/O sequence
is appended at most by one value in one step.
The lemma holds.

\item $i>1$.

We show that, for any positive integer $0<j<i$, value stores $\vals_{1_j}$ and $\vals_{2_j}$ agree on values of variables in $\text{Imp}(io)$ by the proof of Lemma~\ref{lmm:equivTermCompSameLoopIteration}.
We need to show that all the required conditions are satisfied.
\begin{itemize}
\item $\forall x \in \text{Imp}(io)\,:\, \vals_{1}(x) = \vals_{2}(x)$;

\item $\text{loop}_c^1(n_1) = \text{loop}_c^2(n_2) = 0$;

The above two conditions are from assumption.

\item ${id}_{IO} \in \text{Def}(s_1) \cap \text{Def}(s_2)$;

\item $\text{Imp}(s_1, \{{id}_{IO}\}) = \text{Imp}(s_2, \{{id}_{IO}\}) = \text{Imp}(io)$;

The above two conditions are obtained by the same argument in the case that $s_1$ and $s_2$ terminate.

\item $\forall y \in \text{Imp}(io), \forall m_{S_1}(\vals_{S_1})\, m_{S_2}(\vals_{S_2}):$

$((\forall z\in \text{Imp}(S_1, \text{Imp}(io)) \cup \text{Imp}(S_2, \text{Imp}(io)), \vals_{S_1}(z) = \vals_{S_2}(z))
=>$

$(S_1, m_{S_1}(\vals_{S_1})) \equiv_{y} (S_2, m_{S_2}(\vals_{S_2})))$.

By Theorem~\ref{thm:equivCompMain}.
\end{itemize}
We cannot apply Lemma~\ref{lmm:equivTermCompSameLoopIteration} directly because $s_1$ and $s_2$ do not terminate. But we can still have the proof closely similar to that of Lemma~\ref{lmm:equivTermCompSameLoopIteration}.
The reason is that there are reachable configurations $(S_1;s_1, m_{1}')$ and $(S_2;s_2, m_{2}')$ with loop counters of $s_1$ and $s_2$ of value $i$ and crash flags not set.
This is by Lemma~\ref{lmm:ExistenceOfIthIterStart} because there are configurations reachable from $(s_1, m_1)$ and $(s_2, m_2)$
respectively with loop counters of $s_1$ and $s_2$ of $i$.

There are configurations $(s_1, m_{1_{i-1}})$ reachable from $(s_1, m_1)$ and $(s_2, m_{2_{i-1}})$ reachable from $(s_2, m_2)$ in which loop counters of $s_1$ and $s_2$ are $i-1$ and crash flags are not set and value stores agree on values of variables in $\text{Imp}(io)$.
Because loop counters of $s_1$ and $s_2$ are less than or equal to $i$. Then the execution of $s_1$ proceeds as follows:

\begin{tabbing}
xx\=xx\=\kill
\>\>     $({s_1}, m_{1_{i-1}}(\text{loop}_c^{1_{i-1}},\vals_{1_{i-1}}))$\\
\>= \>   $(\text{while}_{\langle n_1\rangle} (e) \; \{S_1\}, {m_{1_{i-1}}(\text{loop}_c^{1_{i-1}}, \vals_{1_{i-1}})})$\\
\>$->$\> $(\text{while}_{\langle n_1\rangle} (v) \; \{S_1\}, {m_{1_{i-1}}(\text{loop}_c^{1_{i-1}}, \vals_{1_{i-1}})})$\\
\>\>     by the EEval rule\\
\>$->$\> $(S_1;\text{while}_{\langle n_1\rangle} (e) \; \{S_1\}, m_{1_{i-1}}(\text{loop}_c^{1_{i-1}}[i/n_1],$\\
\>\>     $\vals_{1_{i-1}}))$ by the Wh-T rule.
\end{tabbing}

The execution of $s_2$ proceeds to
$(S_2;\text{while}_{\langle n_2\rangle} (e) \; \{S_2\},\allowbreak m_{2_{i-1}}(\text{loop}_c^{2_{i-1}}[i/n_2], \vals_{2_{i-1}}))$.
By similar argument in the case $i=1$, $S_2$ and $S_1$ produce the same output sequence when started in states
$m_{1_{i-1}}(\text{loop}_c^{1_{i-1}}[i/n_1], \vals_{1_{i-1}})$ and $m_{2_{i-1}}(\text{loop}_c^{2_{i-1}}[i/n_2], \vals_{2_{i-1}})$ respectively.
In addition, by the semantic rules, the I/O sequence is appended at most by one value in one step.
The lemma holds.
\end{enumerate}

\end{enumerate}
\end{enumerate}
\end{proof}

\begin{corollary}\label{coro:loopSameIOSeq}
Let $s_1 = ``\text{while}_{\langle n_1\rangle}(e) \; \{S_{1}\}"$ and
    $s_2 = ``\text{while}_{\langle n_2\rangle}(e) \; \{S_{2}\}"$ be two while statements such that all of the followings hold
\begin{itemize}
\item There are output statements in $s_1$ and $s_2$, $\exists e_1\, e_2\,:\, (``\text{output }e_1" \in s_1)\wedge (``\text{output }e_2" \in s_2)$;

\item $s_1$ and $s_2$ have same set of termination deciding variables and same set of imported variables relative to the I/O sequence variable, $(\text{TVar}(s_1) = \text{TVar}(s_2) = \text{TVar}(s)) \wedge
                        (\text{Imp}(s_1, \{{id}_{IO}\}) = \text{ Imp}(s_2, \{{id}_{IO}\}) = \text{Imp}(io))$;

\item Loop bodies $S_1$ and $S_2$ satisfy the proof rule of equivalent computation of those in out-deciding variables of $s_1$ and $s_2$, $\forall x \in \text{OVar}(s) = \text{TVar}(s) \cup \text{Imp}(io)\,:\, S_{1} \equiv_{x}^S S_{2}$;

\item Loop bodies $S_1$ and $S_2$ satisfy the proof rule of termination in the same way, $S_{1} \equiv_{H}^S S_{2}$;

\item Loop bodies $S_1$ and $S_2$ produce the same output sequence when started in states with crash flags not set and agreeing on values of variables in $\text{OVar}(S_{1}) \cup \text{OVar}(S_{2})$,
    $\forall m_{S_{1}}(\mathfrak{f}_1, \vals_{S_{1}})\, m_{S_{2}}(\mathfrak{f}_2,\allowbreak \vals_{S_{2}}):$

    $({(\mathfrak{f}_1 = \mathfrak{f}_2 = 0)} \wedge
     (\forall x \in \text{OVar}(S_{1}) \cup \text{OVar}(S_{2})\,:\, \vals_{S_{1}}(x) =\,\,\,\allowbreak \vals_{S_{2}}(x))) =>$
    ${((S_{1}, m_{S_{1}}(\mathfrak{f}_1, \vals_{S_{1}})) \equiv_{O} (S_{2}, m_{S_{2}}(\mathfrak{f}_2, \vals_{S_{2}})))}$.
\end{itemize}

If $s_1$ and $s_2$ start in states $\state_1(\mathfrak{f}_1, \text{loop}_c^1, \vals_1), \state_2(\mathfrak{f}_2, \text{loop}_c^2, \vals_2)$ respectively with crash flags not set $\mathfrak{f}_1 = \mathfrak{f}_2 = 0$ and
in which $s_1$ and $s_2$ have not started execution ($\text{loop}_c^1(n_1) = \text{loop}_c^2(n_2) = 0$),
  value stores $\vals_1$ and $\vals_2$ agree on values of variables in $\text{OVar}(s)$,
  $\forall x \in \text{OVar}(s)\,:\, \vals_1(x) = \vals_2(x)$, then $s_1$ and $s_2$ produce the same output sequence:
$(s_1, m_1) \equiv_{O} (s_2, m_2)$.
\end{corollary}
This is from lemma~\ref{lemma:loopSameIOSeq}.

%
%

\subsection{Backward compatible DSU based on program equivalence}
Based on the equivalence result above, we show that there exists backward compatible DSU.
We need to show there exists a mapping of old program configurations and new program configurations and the hybrid execution obtained from the configuration mapping is backward compatible.
We do not provide a practical algorithm to calculate the state mapping. Instead we only show that there exists new program configurations corresponding to some old program configurations via a simulation. The treatment in this section is informal.

The idea is to map a configuration just before an output is produced to a corresponding configuration.
Based on the proof rule of same output sequences, not every statement of the old
program can correspond to  a statement of the new program, but every output statemet of the old
program should correspond to an output statement of the new program.
Consider configuration $C_1$ of the old program where the leftmost statement (next statement to execute) is an output statement.
We can define a corresponding statement of the new program by {\em simulating} the execution of the new program on
the input consumed so far in $C_1$.
There are two cases. When the leftmost statement in $C_1$ is not included in a loop statement,
then it is easy to know when to stop simulation.
Otherwise, we have the bijection of loop statements including output statements based on the condition of same output sequences. Therefore, it is easy to know how many iterations of the loop statements including the output statement shall be carried out based on the loop counters in the old program configuration $C_1$.
Based on Theorem~\ref{thm:sameIOtheoremFuncWide}, there must be a configuration $C_2$ corresponding to $C_1$.
Moreover, the executions starting from configurations $C_1$ and $C_2$ produce the same output sequence based on
Theorem~\ref{thm:sameIOtheoremFuncWide}. In conclusion, we obtain a backward compatible hybrid execution where the state mapping is from $C_1$ to $C_2$.

\section{Real world backward compatible update classes: proof rules}\label{sec:realupdateclasses}
We propose our formal treatment for real world update classes. For each update class, we show how the old program and new program produce the same I/O sequence which guarantees backward compatible DSU.

\subsection{Proof rule for specializing new configuration variables}
New configuration variables can be introduced to generalize functionality.
\begin{figure}
\begin{small}
\begin{tabbing}
xxxxxx\=xxx\=xxx\=xxx\=xxxxxxxxxxxxxxxx\=xxxx\=xxx\=xxxxx\= \kill
\>1: \>\>\>                                           \>1': \> {\bf If} $(b)$  {\bf then}       \> \\
\>2: \>\>\>                                           \>2': \> \> \text{output} $a * 2$         \> \\
\>3: \>\>\>                                           \>3': \> {\bf else}                        \> \\
\>4: \> \text{output} $a + 2$ \>\>                             \>4': \> \> \text{output} $a + 2$    \> \\
\>\\
\> \> old\>\>                                         \>  \> new \>
\end{tabbing}
\end{small}
\caption{Specializing new configuration variables}\label{fig:newParamExample}
\end{figure}
Figure~\ref{fig:newParamExample} shows an example of how a new configuration variable introduces new code.
The two statement sequences in Figure~\ref{fig:newParamExample} are equivalent when the new variable $b$ is specialized to 0.

Our generalized formal definition of ``specializing new configuration variables" is defined as follows.
\begin{definition}\label{def:localAdditionalParam}
{\bf (Specializing new configuration variables)}
A statement sequence $S_2$ includes updates of specializing new configuration variables compared with $S_1$ w.r.t a mapping $\rho$ of new configuration variables in $S_2$,
$\rho : \{id\} \mapsto \{0, 1\}$, denoted $S_2\, {\approx}_{\rho}^S \, S_1$, iff one of the following holds:
\begin{enumerate}
\item $S_2 = ``\text{If}(id) \, \text{then} \{S_2^t\}\, \text{else} \{S_2^f\}"$
where one of the following holds:
\begin{enumerate}
\item $(\rho(id) = 0) \wedge (S_2^f\, {\approx}_{\rho}^S \, S_1)$;

\item $(\rho(id) = 1) \wedge (S_2^t \, {\approx}_{\rho}^S \, S_1)$;
\end{enumerate}

\item $S_1$ and $S_2$ produce the same output sequence, $S_1 \approx_{O}^S S_2$;

\item
$S_1 = ``\text{If}(e) \, \text{then}\{S_1^t\} \, \text{else} \{S_1^f\}"$,
$S_2 = ``\text{If}(e) \, \text{then}\{S_2^t\} \, \text{else} \{S_2^f\}"$ where
$(S_2^t\, {\approx}_{\rho}^S\, S_1^t) \, \wedge (S_2^f\, {\approx}_{\rho}^S\, S_1^f)$;

\item
$S_1 = ``\text{while}_{\langle n_1\rangle}(e) \, \{S_1'\}"$,
$S_2 = ``\text{while}_{\langle n_2\rangle}(e) \, \{S_2'\}"$ where

\noindent$S_2'\, {\approx}_{\rho}^S\, S_1'$;

\item $S_1 = S_1';s_1$ and $S_2 = S_2';s_2$ where
\noindent$(S_2'\, \approx_{\rho}^S\, S_1') \, \wedge \,
          (S_2'\, \approx_{H}^S\, S_1')  \, \wedge \,
          \big(\forall x \in \text{Imp}(s_1, {id}_{IO})\cup\text{Imp}(s_1, {id}_{IO})\,:\, (S_2'\, \approx_{x}^S\, S_1')\big)  \, \wedge \,
          (s_2\, \approx_{\rho}^S\, s_1)$.
\end{enumerate}
\end{definition}

Then we show that executions of two statement sequences produce the same I/O sequence if there are updates of specializing new configuration variables between the two.
\begin{lemma}\label{lmm:additionalParamsFuncWide}
Let $S_1$ and $S_2$ be two different statement sequences where there are updates of ``specializing new configuration variables" in $S_2$ compared with $S_1$ w.r.t a mapping of new configuration variables $\rho$, $S_2\, {\approx}_{\rho}^S\, S_1$.
If executions of $S_2$ and $S_1$ start in states
$m_2(\mathfrak{f}_2, {\vals}_2)$ and $m_1(\mathfrak{f}_1, {\vals}_1)$ respectively where all of the following hold:
\begin{itemize}
\item Crash flags $\mathfrak{f}_2, \mathfrak{f}_1$ are not set, $\mathfrak{f}_2 = \mathfrak{f}_1 = 0$;

\item Value stores $\vals_1$ and $\vals_2$ agree on output deciding variables in both $S_1$ and $S_2$ including the input and I/O sequence variable,

\noindent$\forall id \in (\text{OVar}(S_1) \cap \text{OVar}(S_2)) \cup \{{id}_I, {id}_{IO}\}$\,:\,
\noindent$\vals_1(id) = \vals_2(id)$;

\item Values of new configuration variables in the value store $\vals_2$ are matching those in $\rho$, $\forall id \in \text{Dom}(\rho)\,:\, \rho(id) = \vals_2(id)$;

\item Values of new configuration variables are not defined in the statement sequence $S_2$,
$\text{Dom}(\rho) \cap \text{Def}(S_2) = \emptyset$;
\end{itemize}
then $S_2$ and $S_1$ satisfy all of the following:
\begin{itemize}
\item $(S_1, m_1) \equiv_{H} (S_2, m_2)$;

\item $(S_1, m_1) \equiv_{O} (S_2, m_2)$;

\item $\forall x \in \{{id}_I,  {id}_{IO}\}$\,:\,
\noindent$(S_1, m_1) \equiv_{x} (S_2, m_2)$;
\end{itemize}
\end{lemma}
\begin{proof}
The proof of Lemma~\ref{lmm:additionalParamsFuncWide} is by induction on the sum of program sizes of $S_1$ and $S_2$ and is a case analysis based on Definition~\ref{def:localAdditionalParam}.

\noindent{Base case}.

\item $S_1$ is a simple statement $s$, $S_2 = ``\text{If}(id) \, \text{then} \{s_2^t\}\, \text{else} \{s_2^f\}"$
where $s_2^t, s_2^f$ are simple statement and one of the following holds:
\begin{enumerate}
\item $(\rho(id) = 0) \wedge (s_2^f = s)$;

\item $(\rho(id) = 1) \wedge (s_2^t = s)$;
\end{enumerate}

W.l.o.g., we assume that $\rho(id) = 0$.
By assumption, $\vals_2(id) = \rho(id) = 0$.
Then the execution of $S_2$ proceeds as follows:
\begin{tabbing}
xx\=xx\=\kill
\>\>                   $(\text{If}(id) \, \text{then}\{s_2^t\} \, \text{else} \{s_2^f\}, m_2(\vals_2))$\\
\>$->$\>               $(\text{If}(0) \, \text{then}\{s_2^t\} \, \text{else} \{s_2^f\}, m_2(\vals_2))$\\
\>\>                   by the rule Var\\
\>$->$\>               $(s_2^f, m_2(\vals_2))$ by the If-F rule.
\end{tabbing}
By Theorem~\ref{thm:sameIOtheoremFuncWide} and Theorem~\ref{thm:mainTermSameWayLocal}, this lemma holds.

\noindent{Induction step}.

The induction hypothesis (IH) is that Lemma~\ref{lmm:additionalParamsFuncWide} holds when the sum of the program size of $S_1$ and $S_2$ is at least 4, $\text{size}(S_1) + \text{size}(S_2) = k \geq 4$.

Then we show that the lemma holds when $\text{size}(S_1) + \text{size}(S_2) = k + 1$.
There are cases to consider.
\begin{enumerate}
\item $S_1$ and $S_2$ satisfy the condition of same output sequence, $S_1 \equiv_{O}^S S_2$.

By Theorem~\ref{thm:sameIOtheoremFuncWide}, the lemma~\ref{lmm:additionalParamsFuncWide} holds.

\item $S_1$ and $S_2$ are both ``If" statement:

$S_1 = ``\text{If}(e) \, \text{then}\{S_1^t\} \, \text{else} \{S_1^f\}"$,
$S_2 = ``\text{If}(e) \, \text{then}\{S_2^t\} \, \text{else} \{S_2^f\}"$ where both of the following hold
\begin{itemize}
\item $S_2^t {\approx}_{\rho}^S S_1^t$;

\item $S_2^f {\approx}_{\rho}^S S_1^f$;
\end{itemize}

By the definition of $\text{Use}(S_1)$, variables used in the predicate expression $e$ are a subset of  used variables in $S_1$ and $S_2$,
$\text{Use}(e) \subseteq \text{Use}(S_1) \cap \text{Use}(S_2)$.
By assumption, corresponding variables used in $e$ are of same value in value stores $\vals_1$ and $\vals_2$. By Lemma~\ref{lmm:expEvalSameVal}, the expression evaluates to the same value w.r.t value stores $\vals_1$ and $\vals_2$. There are three possibilities.
\begin{enumerate}
\item The evaluation of $e$ crashes,
$\mathcal{E}'\llbracket e\rrbracket \vals_1 =
 \mathcal{E}'\llbracket e\rrbracket \vals_2 = (\text{error}, v_{\mathfrak{of}})$.

The execution of $S_1$ continues as follows:
\begin{tabbing}
xx\=xx\=\kill
\>\>                   $(\text{If}(e) \, \text{then}\{S_1^t\} \, \text{else} \{S_1^f\}, m_1(\vals_1))$\\
\>$->$\>               $(\text{If}((\text{error}, v_{\mathfrak{of}})) \, \text{then}\{S_1^t\} \, \text{else} \{S_1^f\}, m_1(\vals_1))$\\
\>\>                   by the rule EEval'\\
\>$->$\>               $(\text{If}(0) \, \text{then}\{S_1^t\} \, \text{else} \{S_1^f\}, m_1(1/\mathfrak{f}))$\\
\>\>                   by the ECrash rule  \\
\>{\kStepArrow [i] }\> $(\text{If}(0) \, \text{then}\{S_1^t\} \, \text{else} \{S_1^f\}, m_1(1/\mathfrak{f}))$ for any $i>0$\\
\>\>                   by the Crash rule.
\end{tabbing}

Similarly, the execution of $S_2$ started from the state $m_2(\vals_2)$ crashes.
The lemma holds.

\item The evaluation of $e$ reduces to zero,
$\mathcal{E}'\llbracket e\rrbracket \vals_1 =
 \mathcal{E}'\llbracket e\rrbracket \vals_2 = (0, v_{\mathfrak{of}})$.

The execution of $S_1$ continues as follows.
\begin{tabbing}
xx\=xx\=\kill
\>\>                   $(\text{If}(e) \, \text{then}\{S_1^t\} \, \text{else} \{S_1^f\}, m_1(\vals_1))$\\
\>= \>                 $(\text{If}((0, v_{\mathfrak{of}})) \, \text{then}\{S_1^t\} \, \text{else} \{S_1^f\}, m_1(\vals_1))$\\
\>\>                   by the rule EEval'\\
\>$->$\>               $(\text{If}(0) \, \text{then}\{S_1^t\} \, \text{else} \{S_1^f\}, m_1(\vals_1))$\\
\>\>                   by the E-Oflow1 or E-Oflow2 rule  \\
\>$->$\>               $(S_1^f, m_1(\vals_1))$ by the If-F rule.
\end{tabbing}

Similarly, the execution of $S_2$ gets to the configuration $(S_2^f, m_2(\vals_2))$.

By the hypothesis IH, we show the lemma holds.
We need to show that all conditions are satisfied for the application of the hypothesis IH.
\begin{itemize}
\item $(S_2^f {\approx}_{\rho}^S S_1^f)$

By assumption.

\item The sum of  the program size of $S_1^f$ and $S_2^f$ is less than $k$,
$\text{size}(S_1^f) + \text{size}(S_2^f) < k$.

By definition, $\text{size}(S_1) = 1 + \text{size}(S_1^t) + \text{size}(S_1^f)$.
Then, $\text{size}(S_1^f) + \text{size}(S_2^f) < k + 1 - 2 = k - 1$.

\item Value stores $\vals_1$ and $\vals_2$ agree on values of  used variables in $S_1^f$ and $S_2^f$ as well as the input, I/O sequence variable.

By definition, $\text{Use}(S_1^f)\allowbreak \subseteq \text{Use}(S_1)$.
So are the cases to $S_2^f$ and $S_2$.
In addition, value stores $\vals_1$ and $\vals_2$ are not changed in the evaluation of the predicate expression $e$. The condition holds.

\item Values of new configuration variables are consistent in the value store $\vals_2$ and the specialization $\rho$, $\forall id \in \text{Dom}(\rho)\,:\, \vals_2(id) = \rho(id)$.

By assumption.
\end{itemize}

By the hypothesis IH, the lemma holds.

\item The evaluation of $e$ reduces to the same nonzero integer value,
$\mathcal{E}'\llbracket e\rrbracket \vals_1 =
 \mathcal{E}'\llbracket e\rrbracket \vals_2 = (v, v_{\mathfrak{of}})$ where $v \neq 0$.

By arguments similar to the second subcase above.
\end{enumerate}

\item $S_1$ and $S_2$ are both ``while" statements:

$S_1 = ``\text{while}_{\langle n\rangle}(e) \, \{S_1'\}"$,
$S_2 = ``\text{while}_{\langle n\rangle}(e) \, \{S_2'\}"$ where
$S_2' {\approx}_{\rho}^S S_1'$;

By Lemma~\ref{lmm:additionalParamLoopStmt}, we show this lemma holds.
We need to show that all required conditions are satisfied for the application of Lemma~\ref{lmm:additionalParamLoopStmt}.
\begin{itemize}
\item $S_1$ and $S_2$ have same set of  output deciding variables,
$\text{OVar}(S_1) = \text{OVar}(S_2) = \text{OVar}(S)$;

By Lemma~\ref{lmm:additionalParamFuncWideSimilarUseDef} and Corollary~\ref{coro:sameTermVarFromEquivTerm}.

\item When started in states $m_1'(\vals_1'), m_2'(\vals_1')$ where
value stores $\vals_1'$ and $\vals_2'$ agree on values of output deciding variables in both $S_1$ and $S_2$ as well as the input sequence variable and the I/O sequence variable,
then $S_1'$ and $S_2'$ terminate in the same way, produce the same output sequence, and have equivalent computation of defined variables in both $S_1$ and $S_2$.

By the induction hypothesis IH. This is because the sum of the program size of $S_1'$ and $S_2'$ is less than $k$. By definition, $\text{size}(S_1) = 1 + \text{size}(S_1')$.
\end{itemize}
By Lemma~\ref{lmm:additionalParamLoopStmt}, this lemma holds.

\item $S_2 = ``\text{If}(id) \, \text{then} \, \{S_2^t\} \, \text{else} \, \{S_2^f\}"$ where one of the following holds:
  \begin{enumerate}
    \item $(\rho(id) = 0) \wedge
           (S_2^f {\approx}_{\rho}^S S_1)$;

    \item $(\rho(id) = 1) \wedge
           (S_2^t {\approx}_{\rho}^S S_1)$;
  \end{enumerate}

W.l.o.g, we assume $(\rho(id) = 0) \wedge
           (S_2^f {\approx}_{\rho}^S S_1)$;

Then the execution of $S_2$ proceeds as follows:
\begin{tabbing}
xx\=xx\=\kill
\>\>                   $(\text{If}(id) \, \text{then}\{S_2^t\} \, \text{else} \{S_2^f\}, m_2(\vals_2))$\\
\>= \>                 $(\text{If}(0) \, \text{then}\{S_2^t\} \, \text{else} \{S_2^f\}, m_2(\vals_2))$\\
\>\>                   by the Var rule\\
\>$->$\>               $(S_2^f, m_2(\vals_2))$\\
\>\>                   by the If-F rule
\end{tabbing}

By the induction hypothesis, we show that the lemma holds.
We need to show the required conditions are satisfied for the application of the hypothesis.
\begin{itemize}
\item $S_2^f {\approx}_{\rho}^S S_1$

By assumption.

\item The sum of  the program size of $S_1^f$ and $S_2^f$ is less than $k$,
$\text{size}(S_1^f) + \text{size}(S_2^f) < k$.

By definition, $\text{size}(S_2) = 1 + \text{size}(S_2^t) + \text{size}(S_2^f)$.
Then, $\text{size}(S_2^f) + \text{size}(S_1) < k + 1 - 1 - \text{size}(S_2^t) < k $.

\item Value stores $\vals_1$ and $\vals_2$ agree on values of  used  variables in $S_2^f$ and $S_1$ as well as the input, I/O sequence variable.

By definition, $\text{Use}(S_2^f)\allowbreak \subseteq \text{Use}(S_2)$.
In addition, the value store $\vals_2$ is not changed in the evaluation of the predicate expression $e$. The condition holds.
\end{itemize}
By the hypothesis IH, the lemma holds.

\item $S_1$ and $S_2$ are same, $S_1 = S_2$;

By definition, used variables in $S_1$ and $S_2$ are same;
defined variables in $S_1$ and $S_2$ are same.
By semantic rules, $S_1$ and $S_2$ terminate in the same way, produce the same output sequence and have equivalent computation of defined variables in $S_1$ and $S_2$.
This lemma holds.

\item $S_1 = S_1';s_1$ and $S_2 = S_2';s_2$ where both of the following hold:

\begin{itemize}
\item $S_2' \approx_{\rho}^S S_1'$;

\item $s_2 \approx_{\rho}^S s_1$;
\end{itemize}

By Theorem~\ref{thm:mainTermSameWayLocal} and the hypothesis IH, we show $S_2'$ and $S_1'$ terminate in the same way and produce the same output sequence and when $S_2'$ and $S_1'$ both terminate, $S_2'$ and $S_1'$ have equivalent terminating computation of variables  used or defined in $S_2'$ and $S_1'$.

We show all the required conditions are satisfied for the application of the hypothesis IH.
\begin{itemize}
\item $S_2' {\approx}_{\rho}^S S_1'$.

By assumption.

\item The sum of  the program size of $S_1'$ and $S_2'$ is less than $k$,
$\text{size}(S_1') + \text{size}(S_2') < k$.

By definition, $\text{size}(S_2) = \text{size}(s_2) + \text{size}(S_2')$ where $\text{size}(s_2) < 1$.
Then, $\text{size}(S_2') + \text{size}(S_1') < k + 1 - \text{size}(s_2) - \text{size}(s_1) < k$.

\item Value stores $\vals_1$ and $\vals_2$ agree on values of  output deciding variables in $S_2'$ and $S_1'$ including the input, I/O sequence variable.

By definition of $\text{TVar}_o$ and $\text{Imp}_o$, $\text{OVar}(S_2')\allowbreak \subseteq \text{OVar}(S_2)$.
The condition holds.

\item Values of new configuration variables are consistent in the value store $\vals_2$ and the specialization $\rho$, $\forall id \in \text{Dom}(\rho)\,:\,
    \vals_2(id) = \rho(id)$.

By assumption.
\end{itemize}
By the hypothesis IH, one of the following holds:
\begin{enumerate}
\item $S_1'$ and $S_2'$ both do not terminate.

By Lemma~\ref{lmm:multiStepSeqExec}, executions of $S_1 = S_1';s_1$ and $S_2 = S_2';s_2$ both do not terminate and produce the same output sequence.

\item $S_1'$ and $S_2'$ both terminate.

By assumption,
$(S_2', m_2(\vals_2)) ->* (\text{skip}, m_2'(\vals_2'))$,

\noindent$(S_1', m_1(\vals_1)) ->* (\text{skip}, m_1'(\vals_1'))$.

By Corollary~\ref{coro:termSeq},
$(S_2';s_2, m_2(\vals_2)) ->* (s_2, m_2'(\vals_2'))$,
$(S_1';s_1, m_1(\vals_1)) ->* (s_1, m_1'(\vals_1'))$.

By the hypothesis IH, we show that $s_2$ and $s_1$ terminate in the same way, produce the same output sequence and when $s_2$ and $s_1$ both terminate, $s_2$ and $s_1$ have equivalent computation of variables  used or defined in $s_1$ and $s_2$ and the input, and I/O sequence variables.

We need to show that all conditions are satisfied for the application of the hypothesis IH.
\begin{itemize}
\item There are updates of ``new configuration variables" between $s_2$ and $s_1$;

By assumption, $s_2 {\approx}_{\rho}^S s_1$.

\item The sum of the program size $s_2$ and $s_1$ is less than or equals to $k$;

By definition, $\text{size}(S_2') \geq 1,  \text{size}(S_1') \geq 1$.
Therefore, $\text{size}(s_2) + \text{size}(s_1) < k+1 - \text{size}(S_2') - \text{size}(S_1') \leq k$.

\item Value stores $\vals_1'$ and $\vals_2'$ agree on values of  output deciding variables in $s_2$ and $s_1$ as well as the input, I/O sequence variable.

By induction hypothesis IH,
$\text{OVar}(s_1) \subseteq \text{OVar}(s_2)$, then $\text{Use}(s_2)\,\allowbreak   \cap \text{Use}(s_1) = \text{Use}(s_1)$.
For any variable $id$ in $\text{OVar}(s_1)$, if $id$ is in $\text{OVar}(S_1')$,
then the value of $id$ is same after the execution of $S_1'$ and $S_2'$,
$\vals_1'(id) = \vals_1(id) = \vals_2(id)\allowbreak = \vals_2'(id)$.
Otherwise, the variable $id$ is defined in the execution of $S_1'$ and $S_2'$, by assumption, $\vals_1'(id) = \vals_2'(id)$.
The condition holds.

\item Values of new configuration variables are consistent in the value store $\vals_2'$ and the specialization $\rho$, $\forall id \in \text{Dom}(\rho)\,:\,
    \vals_2'(id) = \rho(id)$.

By assumption, $\text{Dom}(\rho) \cap \text{Def}(S_2)$.
By Corollary~\ref{coro:defExclusion}, values of new configuration variables are not changed in the execution of $S_2'$,
$\forall id \in \text{Dom}(\rho)\,:\,
    \vals_2'(id) = \vals_2(id) = \rho(id)$.
\end{itemize}
By the hypothesis IH, the lemma holds.
\end{enumerate}

\end{enumerate}
\end{proof}

We list properties of the update of new configuration variables and the proof of backward compatibility for the case of loop statement as follows.
We present one auxiliary lemma used in the proof of Lemma~\ref{lmm:additionalParamsFuncWide}.
\begin{lemma}\label{lmm:additionalParamFuncWideSimilarUseDef}
Let $S_2$ be a statement sequence and $S_1$ where there are updates of ``specializing new configuration variables" w.r.t a mapping of new configuration variables $\rho$, $S_2 {\approx}_{\rho}^S S_1$.
Then the output deciding variables in $S_1$ are a subset of the union of those in $S_2$,
 $\text{OVar}(S_1) \subseteq \text{OVar}(S_2)$.
\end{lemma}
\begin{proof}
By induction on the sum of the program size of $S_1$ and $S_2$.
\end{proof}

\begin{lemma}\label{lmm:additionalParamLoopStmt}
Let $S_1 = \text{while}_{\langle n_1\rangle}(e) \, \{S_1'\}$ and
    $S_2 = \text{while}_{\langle n_2\rangle}(e) \,\allowbreak \{S_2'\}$ be two loop statements where all of the following hold:
\begin{itemize}
\item $S_2'$ includes updates of ``specializing new configuration variables" compared to $S_1'$, $S_2' \approx_{\rho}^S S_1'$ where $\text{Dom}(\rho) \cap \text{Def}(S_2') = \emptyset$.

\item the output deciding variables in $S_1$ are a subset of those in $S_2$,

\noindent$\text{OVar}(S_1) \subseteq \text{OVar}(S_2)$;

\item When started in states agreeing on values of output deciding variables in $S_1$ and $S_2$ including the input sequence variable and the I/O sequence variable,
    $\forall x \in \text{OVar}(S_1) \cup \text{OVar}(S_2) \cup \{{id}_I, {id}_{IO}\}\, \forall m_1'(\vals_1')\, m_2'(\vals_2')\,:\, \, \allowbreak
    (\vals_1'(x) = \vals_2'(x))$,
$S_1'$ and $S_2'$ terminate in the same way, produce the same output sequence, and have equivalent computation of defined variables in $S_1'$ and $S_2'$ as well as the input sequence variable and the I/O sequence variable
    $((S_1', m_1) \equiv_{H} (S_2', m_2)) \wedge
     ((S_1', m_1) \equiv_{O} (S_2', m_2)) \wedge
     (\forall x \in \text{OVar}(S_1) \cup \text{OVar}(S_2) \cup \{{id}_I, \allowbreak {id}_{IO}\}\,:\,
     (S_1', m_1) \equiv_{x} (S_2', m_2))$;
\end{itemize}

If $S_1$ and $S_2$ start in states $m_1(\text{loop}_c^1, \vals_1), m_2(\text{loop}_c^2, \vals_2)$ respectively, with loop counters of $S_1$ and $S_2$ not initialized ($S_1, S_2$ have not executed yet), value stores agree on values of output deciding variables in $S_1$ and $S_2$, then, for any positive integer $i$, one of the following holds:
\begin{enumerate}
\item Loop counters for $S_1$ and $S_2$ are always less than $i$ if any is present,
$\forall m_1'(\text{loop}_c^{1'})\, m_2'(\text{loop}_c^{2'})\,:\,
(S_1, m_1(\text{loop}_c^1, \vals_1)) ->* (S_1'', m_1'(\text{loop}_c^{1'})),
\text{loop}_c^{1'}(n_1) < i,
(S_2, m_2(\text{loop}_c^2,\allowbreak \vals_2)) ->* (S_2'', m_2'(\text{loop}_c^{2'})),
\text{loop}_c^{2'}(n_2) < i$,
$S_1$ and $S_2$ terminate in the same way, produce the same output sequence, and have equivalent computation of output deciding variables in both $S_1$ and $S_2$ and the input sequence variable, the I/O sequence variable,
$(S_1, m_1)\allowbreak \equiv_{H} (S_2, m_2)$ and
$(S_1, m_1) \equiv_{O} (S_2, m_2)$ and
$\forall x \in (\text{OVar}(S_1) \cap \text{OVar}(S_2)) \cup
               \{{id}_I,  {id}_{IO}\}\,:\,
               (S_1, m_1) \equiv_{x} (S_2, m_2)$;

\item The loop counter of $S_1$ and $S_2$ are of value less than or equal to $i$,
and there are no reachable configurations
$(S_1, m_1(\text{loop}_c^{1_i},\allowbreak \vals_{1_i}))$ from $(S_1, m_1(\vals_1))$,
$(S_2, m_2(\text{loop}_c^{2_i}, \vals_{2_i}))$ from $(S_2, \,\allowbreak m_2(\vals_2))$ where all of the following hold:
\begin{itemize}
\item The loop counters of $S_1$ and $S_2$ are of value $i$,
$\text{loop}_c^{1_i}(n_1)\allowbreak = \text{loop}_c^{2_i}(n_2) = i$.

\item Value stores $\vals_{1_i}$ and $\vals_{2_i}$ agree on values of output deciding variables in both $S_1$ and $S_2$ as well as the input sequence variable and the I/O sequence variable,
$\forall x \in \, (\text{OVar}(S_1) \cap \text{OVar}(S_2)) \cup \{{id}_I,  {id}_{IO}\}\,:\,
\vals_{1_i}(x) = \vals_{2_i}(x)$.
\end{itemize}

\item There are reachable configurations
$(S_1, m_1(\text{loop}_c^{1_i}, \vals_{1_i}))$ from $(S_1, m_1(\vals_1))$,
$(S_2, m_2(\text{loop}_c^{2_i}, \vals_{2_i}))$ from $(S_2, \,\allowbreak m_2(\vals_2))$ where all of the following hold:
\begin{itemize}
\item The loop counter of $S_1$ and $S_2$ are of value $i$,
$\text{loop}_c^{1_i}(n_1)\allowbreak = \text{loop}_c^{2_i}(n_2) = i$.

\item Value stores $\vals_{1_i}$ and $\vals_{2_i}$ agree on values of output deciding variables in both $S_1$ and $S_2$ including the input sequence variable and the I/O sequence variable,
$\forall x \in \, (\text{OVar}(S_1) \cap \text{OVar}(S_2)) \cup \{{id}_I,  {id}_{IO}\}\,:\,
\vals_{1_i}(x) = \vals_{2_i}(x)$.
\end{itemize}
\end{enumerate}
\end{lemma}

\begin{proof}
By induction on $i$.

\noindent{Base case}.

We show that, when $i=1$, one of the following holds:
\begin{enumerate}
\item Loop counters for $S_1$ and $S_2$ are always less than 1 if any is present,
$\forall m_1'(\text{loop}_c^{1'})\, m_2'(\text{loop}_c^{2'})\,:\,
(S_1, m_1(\text{loop}_c^1, \vals_1)) ->* (S_1'', m_1'(\text{loop}_c^{1'})),
\text{loop}_c^{1'}(n_1) < i,
(S_2, m_2(\text{loop}_c^2,\allowbreak \vals_2)) ->* (S_2'', m_2'(\text{loop}_c^{2'})),
\text{loop}_c^{2'}(n_2) < i$,
$S_1$ and $S_2$ terminate in the same way, produce the same output sequence, and have equivalent computation of used/defined variables in both $S_1$ and $S_2$ and the input sequence variable, the I/O sequence variable,
$(S_1, m_1)\allowbreak \equiv_{H} (S_2, m_2)$ and
$(S_1, m_1) \equiv_{O} (S_2, m_2)$ and
$\forall x \in (\text{OVar}(S_1) \cap \text{OVar}(S_2)) \cup
               \{{id}_I,  {id}_{IO}\}\,:\,
               (S_1, m_1) \equiv_{x} (S_2, m_2)$;

\item Loop counters of $S_1$ and $S_2$ are of value less than or equal to 1 but there are no reachable configurations
$(S_1, m_1(\text{loop}_c^{1_1}, \vals_{1_i}))$ from $(S_1, m_1(\vals_1))$,
$(S_2, m_2(\text{loop}_c^{2_1}, \vals_{2_i}))$ from $(S_2, \,\allowbreak m_2(\vals_2))$ where all of the following hold:
\begin{itemize}
\item The loop counter of $S_1$ and $S_2$ are of value 1,
$\text{loop}_c^{1_1}(n_1)\allowbreak = \text{loop}_c^{2_1}(n_2) = 1$.

\item Value stores $\vals_{1_1}$ and $\vals_{2_1}$ agree on values of used variables in both $S_1$ and $S_2$ as well as the input sequence variable and the I/O sequence variable,
$\forall x \in \, (\text{OVar}(S_1) \cap \text{OVar}(S_2)) \cup \{{id}_I,  {id}_{IO}\}\,:\,
\vals_{1_1}(x) = \vals_{2_1}(x)$.
\end{itemize}

\item There are reachable configuration
$(S_1, m_1(\text{loop}_c^{1_1}, \vals_{1_i}))$ from $(S_1, m_1(\vals_1))$,
$(S_2, m_2(\text{loop}_c^{2_1}, \vals_{2_i}))$ from $(S_2, \,\allowbreak m_2(\vals_2))$ where all of the following hold:
\begin{itemize}
\item The loop counter of $S_1$ and $S_2$ are of value 1,
$\text{loop}_c^{1_1}(n_1)\allowbreak = \text{loop}_c^{2_1}(n_2) = 1$.

\item Value stores $\vals_{1_1}$ and $\vals_{2_1}$ agree on values of used variables in both $S_1$ and $S_2$ as well as the input sequence variable and the I/O sequence variable,
$\forall x \in \, (\text{OVar}(S_1) \cap \text{OVar}(S_2)) \cup \{{id}_I,  {id}_{IO}\}\,:\,
\vals_{1_1}(x) = \vals_{2_1}(x)$.
\end{itemize}
\end{enumerate}

By definition, variables used in the predicate expression $e$ of $S_1$ and $S_2$ are used in $S_1$ and $S_2$, $\text{Use}(e) \subseteq \text{OVar}(S_1) \cap \text{OVar}(S_2)$.
By assumption, value stores $\vals_1$ and $\vals_2$ agree on values of variables in $\text{Use}(e)$,
the predicate expression $e$ evaluates to the same value w.r.t value stores $\vals_1$ and $\vals_2$.
There are three possibilities.
\begin{enumerate}
\item The evaluation of $e$ crashes,

\noindent$\mathcal{E}'\llbracket e\rrbracket \vals_1 =
 \mathcal{E}'\llbracket e\rrbracket \vals_2 = (\text{error}, v_{\mathfrak{of}})$.

The execution of $S_1$ continues as follows:
\begin{tabbing}
xx\=xx\=\kill
\>\>                   $(\text{while}_{\langle n_1\rangle}(e) \, \{S_1'\}, m_1(\vals_1))$\\
\>$->$\>               $(\text{while}_{\langle n_1\rangle}((\text{error}, v_{\mathfrak{of}})) \, \{S_1'\}, m_1(\vals_1))$\\
\>\>                   by the rule EEval'\\
\>$->$\>               $(\text{while}_{\langle n_1\rangle}(0) \, \{S_1'\}, m_1(1/\mathfrak{f}))$\\
\>\>                   by the ECrash rule  \\
\>{\kStepArrow [i] }\> $(\text{while}_{\langle n_1\rangle}(0) \, \{S_1'\}, m_1(1/\mathfrak{f}))$ for any $i>0$\\
\>\>                   by the Crash rule.
\end{tabbing}

Similarly, the execution of $S_2$ started from the state $m_2(\vals_2)$ crashes.
Therefore $S_1$ and $S_2$ terminate in the same way when started from $m_1$ and $m_2$ respectively.
Because $\vals_1({id}_{IO}) = \vals_2({id}_{IO})$, the lemma holds.

\item The evaluation of $e$ reduces to zero,
$\mathcal{E}'\llbracket e\rrbracket \vals_1 =
 \mathcal{E}'\llbracket e\rrbracket \vals_2 = (0, v_{\mathfrak{of}})$.

The execution of $S_1$ continues as follows.
\begin{tabbing}
xx\=xx\=\kill
\>\>                   $(\text{while}_{\langle n_1\rangle}(e) \, \{S_1'\}, m_1(\vals_1))$\\
\>= \>                 $(\text{while}_{\langle n_1\rangle}((0, v_{\mathfrak{of}})) \, \{S_1'\}, m_1(\vals_1))$\\
\>\>                   by the rule EEval'\\
\>$->$\>               $(\text{while}_{\langle n_1\rangle}(0) \, \{S_1'\}, m_1(\vals_1))$\\
\>\>                   by the E-Oflow1 or E-Oflow2 rule  \\
\>$->$\>               $(\text{skip}, m_1(\vals_1))$ by the Wh-F rule.
\end{tabbing}

Similarly, the execution of $S_2$ gets to the configuration $(\text{skip}, m_2(\vals_2))$.
Loop counters of $S_1$ and $S_2$ are less than 1 and value stores agree on values of used/defined variables in both $S_1$ and $S_2$ as well as the input sequence variable and the I/O sequence variable.

\item The evaluation of $e$ reduces to the same nonzero integer value,
$\mathcal{E}'\llbracket e\rrbracket \vals_1 =
 \mathcal{E}'\llbracket e\rrbracket \vals_2 = (0, v_{\mathfrak{of}})$.

Then the execution of $S_1$ proceeds as follows:
\begin{tabbing}
xx\=xx\=\kill
\>\>                   $(\text{while}_{\langle n_1\rangle}(e) \, \{S_1'\}, m_1(\vals_1))$\\
\>= \>                 $(\text{while}_{\langle n_1\rangle}((v, v_{\mathfrak{of}})) \, \{S_1'\}, m_1(\vals_1))$\\
\>\>                   by the rule EEval'\\
\>$->$\>               $(\text{while}_{\langle n_1\rangle}(v) \, \{S_1'\}, m_1(\vals_1))$\\
\>\>                   by the E-Oflow1 or E-Oflow2 rule  \\
\>$->$\>               $(S_1';\text{while}_{\langle n_1\rangle}(e) \, \{S_1'\}, m_1($\\
\>\>                   $\text{loop}_c^1 \cup \{(n_1) \mapsto 1\}, \vals_1))$ by the Wh-T rule.
\end{tabbing}

Similarly, the execution of $S_2$ proceeds to the configuration
$(S_2';\text{while}_{\langle n_2\rangle}(e) \, \{S_2'\}, m_2(\text{loop}_c^2 \cup \{n_2 \mapsto 1\}, \vals_2))$.

By the hypothesis IH, we show that $S_1'$ and $S_2'$ terminate in the same way and produce the same output sequence when started in the state $m_1(\text{loop}_c^{1_1}, \vals_1)$ and $m_2(\text{loop}_c^{2_1}, \vals_2)$, and $S_1'$ and $S_2'$ have equivalent computation of variables used or defined in both statement sequences if both terminate.
We need to show that all conditions are satisfied for the application of the hypothesis IH.
\begin{itemize}
\item variables in the domain of $\rho$ are not redefined in the execution of $S_2'$.


The above three conditions are by assumption.

By definition, $\text{size}(S_1) = 1 + \text{size}(S_1')$.
Then, $\text{size}(S_1') + \text{size}(S_2') = k + 1 - 2 = k - 1$.

\item Value stores $\vals_1$ and $\vals_2$ agree on values of used variables in $S_1'$ and $S_2'$ as well as the input, I/O sequence variable.

By definition, $\text{OVar}(S_1')\allowbreak \subseteq \text{OVar}(S_1)$.
So are the cases to $S_2'$ and $S_2$.
In addition, value stores $\vals_1$ and $\vals_2$ are not changed in the evaluation of the predicate expression $e$. The condition holds.

\item Values of new configuration variables are consistent in the value store $\vals_2$ and the specialization $\rho$, $\forall id \in \text{Dom}(\rho)\,:\,
    \vals_2(id) = \rho(id)$.

By assumption.
\end{itemize}
By  assumption, $S_1'$ and $S_2'$ terminate in the same way and produce the same output sequence when started in states
$m_1(\text{loop}_c', \vals_1)$ and $m_2(\text{loop}_c', \vals_2)$.
In addition, $S_1'$ and $S_2'$ have equivalent computation of variables used or defined in $S_1'$ and $S_2'$ when started in states
$m_1(\text{loop}_c', \vals_1)$ and $m_2(\text{loop}_c', \vals_2)$.

Then there are two cases.
\begin{enumerate}
\item $S_1'$ and $S_2'$ both do not terminate and produce the same output sequence.

By Lemma~\ref{lmm:multiStepSeqExec}, $S_1';S_1$ and $S_2';S_2$ both do not terminate and produce the same output sequence.

\item $S_1'$ and $S_2'$ both terminate and have equivalent computation of variables used or defined in $S_1'$ and $S_2'$.

By assumption, $(S_1', m_1(\text{loop}_c', \vals_1)) ->*
                (\text{skip}, \,\allowbreak m_1'(\text{loop}_c'', \vals_1'))$;
               $(S_2', m_2(\text{loop}_c', \vals_2)) ->*
                (\text{skip},\allowbreak m_2'(\text{loop}_c'', \vals_2'))$
where $\forall x \in (\text{OVar}(S_1') \cap \text{OVar}(S_2')) \cup
                     \{{id}_I,  {id}_{IO}\},
\allowbreak\vals_1'(x) = \vals_2'(x)$.

Because $S_1$ and $S_2$ have the same predicate expression,
variables used in the predicate expression of $S_1$ and $S_2$ are not in the domain of $\rho$.
By assumption,
$\text{OVar}(S_1') \subseteq \text{OVar}(S_2') \subseteq \text{OVar}(S_2') \cup \text{Dom}(\rho)$ and
$\text{OVar}(S_1') \subseteq \text{OVar}(S_2')$.
Then variables used in the predicate expression of $S_1$ and $S_2$ are either in variables used or defined in both $S_1'$ and $S_2'$ or not.
Therefore value stores $\vals_2'$ and $\vals_1'$ agree on values of variables used in the expression $e$ and even variables used or defined in $S_1$ and $S_2$.
\end{enumerate}
\end{enumerate}

\noindent{Induction step on iterations}

The induction hypothesis (IH) is that, when $i\geq 1$, one of the following holds:
\begin{enumerate}
\item Loop counters for $S_1$ and $S_2$ are always less than $i$ if any is present,
$\forall m_1'(\text{loop}_c^{1'})\, m_2'(\text{loop}_c^{2'})\,:\,
(S_1, m_1(\text{loop}_c^1, \vals_1)) ->* (S_1'', m_1'(\text{loop}_c^{1'})),
\text{loop}_c^{1'}(n_1) < i,
(S_2, m_2(\text{loop}_c^2,\allowbreak \vals_2)) ->* (S_2'', m_2'(\text{loop}_c^{2'})),
\text{loop}_c^{2'}(n_2) < i$,
$S_1$ and $S_2$ terminate in the same way, produce the same output sequence, and have equivalent computation of used/defined variables in both $S_1$ and $S_2$ and the input sequence variable, the I/O sequence variable,
$(S_1, m_1)\allowbreak \equiv_{H} (S_2, m_2)$ and
$(S_1, m_1) \equiv_{O} (S_2, m_2)$ and
$\forall x \in (\text{OVar}(S_1) \cap \text{OVar}(S_2)) \cup
               \{{id}_I,  {id}_{IO}\}\,:\,
               (S_1, m_1) \equiv_{x} (S_2, m_2)$;

\item The loop counter of $S_1$ and $S_2$ are of value less than or equal to $i$,
and there are no reachable configurations
$(S_1, m_1(\text{loop}_c^{1_i}, \vals_{1_i}))$ from $(S_1, m_1(\vals_1))$,
$(S_2, m_2(\text{loop}_c^{2_i}, \vals_{2_i}))$ from $(S_2, \,\allowbreak m_2(\vals_2))$ where all of the following hold:
\begin{itemize}
\item The loop counters of $S_1$ and $S_2$ are of value $i$,
$\text{loop}_c^{1_i}(n_1)\allowbreak = \text{loop}_c^{2_i}(n_2) = i$.

\item Value stores $\vals_{1_i}$ and $\vals_{2_i}$ agree on values of used variables in both $S_1$ and $S_2$ as well as the input sequence variable and the I/O sequence variable,
$\forall x \in \, (\text{OVar}(S_1) \cap \text{OVar}(S_2)) \cup \{{id}_I,  {id}_{IO}\}\,:\,
\vals_{1_i}(x) = \vals_{2_i}(x)$.
\end{itemize}

\item There are reachable configurations
$(S_1, m_1(\text{loop}_c^{1_i}, \vals_{1_i}))$ from $(S_1, m_1(\vals_1))$,
$(S_2, m_2(\text{loop}_c^{2_i}, \vals_{2_i}))$ from $(S_2, \,\allowbreak m_2(\vals_2))$ where all of the following hold:
\begin{itemize}
\item The loop counter of $S_1$ and $S_2$ are of value $i$,
$\text{loop}_c^{1_i}(n_1)\allowbreak = \text{loop}_c^{2_i}(n_2) = i$.

\item Value stores $\vals_{1_i}$ and $\vals_{2_i}$ agree on values of used variables in both $S_1$ and $S_2$ as well as the input sequence variable and the I/O sequence variable,
$\forall x \in \, (\text{OVar}(S_1) \cap \text{OVar}(S_2)) \cup \{{id}_I,  {id}_{IO}\}\,:\,
\vals_{1_i}(x) = \vals_{2_i}(x)$.
\end{itemize}
\end{enumerate}

Then we show that, when $i+1$, one of the following holds:
The induction hypothesis (IH) is that, when $i\geq 1$, one of the following holds:
\begin{enumerate}
\item Loop counters for $S_1$ and $S_2$ are always less than $i+1$ if any is present,
$\forall m_1'(\text{loop}_c^{1'})\, m_2'(\text{loop}_c^{2'})\,:\,
(S_1, m_1(\text{loop}_c^1, \vals_1)) ->* (S_1'', m_1'(\text{loop}_c^{1'})),
\text{loop}_c^{1'}(n_1) < i+1,
(S_2, m_2(\text{loop}_c^2,\allowbreak \vals_2)) ->* (S_2'', m_2'(\text{loop}_c^{2'})),
\text{loop}_c^{2'}(n_2) < i+1$,
$S_1$ and $S_2$ terminate in the same way, produce the same output sequence, and have equivalent computation of used/defined variables in both $S_1$ and $S_2$ and the input sequence variable, the I/O sequence variable,
$(S_1, m_1)\allowbreak \equiv_{H} (S_2, m_2)$ and
$(S_1, m_1) \equiv_{O} (S_2, m_2)$ and
$\forall x \in (\text{OVar}(S_1) \cap \text{OVar}(S_2)) \cup
               \{{id}_I,  {id}_{IO}\}\,:\,
               (S_1, m_1) \equiv_{x} (S_2, m_2)$;

\item The loop counter of $S_1$ and $S_2$ are of value less than or equal to $i+1$,
and there are no reachable configurations
$(S_1, m_1(\text{loop}_c^{1_{i+1}}, \vals_{1_{i+1}}))$ from $(S_1, m_1(\vals_1))$,
$(S_2, m_2(\text{loop}_c^{2_{i+1}}, \vals_{2_{i+1}}))$ from $(S_2, \,\allowbreak m_2(\vals_2))$ where all of the following hold:
\begin{itemize}
\item The loop counters of $S_1$ and $S_2$ are of value $i+1$,
$\text{loop}_c^{1_{i+1}}(n_1)\allowbreak = \text{loop}_c^{2_{i+1}}(n_2) = i+1$.

\item Value stores $\vals_{1_{i+1}}$ and $\vals_{2_{i+1}}$ agree on values of used variables in both $S_1$ and $S_2$ as well as the input sequence variable and the I/O sequence variable,
$\forall x \in \, (\text{OVar}(S_1) \cap \text{OVar}(S_2)) \cup \{{id}_I,  {id}_{IO}\}\,:\,
\vals_{1_{i+1}}(x) = \vals_{2_{i+1}}(x)$.
\end{itemize}

\item There are reachable configurations
$(S_1, m_1(\text{loop}_c^{1_{i+1}}, \vals_{1_i}))$ from $(S_1, m_1(\vals_1))$,
$(S_2, m_2(\text{loop}_c^{2_{i+1}}, \vals_{2_i}))$ from $(S_2, \,\allowbreak m_2(\vals_2))$ where all of the following hold:
\begin{itemize}
\item The loop counter of $S_1$ and $S_2$ are of value $i$,
$\text{loop}_c^{1_{i+1}}(n_1)\allowbreak = \text{loop}_c^{2_{i+1}}(n_2) = i+1$.

\item Value stores $\vals_{1_{i+1}}$ and $\vals_{2_{i+1}}$ agree on values of used variables in both $S_1$ and $S_2$ as well as the input sequence variable and the I/O sequence variable,
$\forall x \in \, (\text{OVar}(S_1) \cap \text{OVar}(S_2)) \cup \{{id}_{I},  {id}_{IO}\}\,:\,
\vals_{1_{i+1}}(x) = \vals_{2_{i+1}}(x)$.
\end{itemize}
\end{enumerate}

By hypothesis IH, there is no configuration where loop counters of $S_1$ and $S_2$ are of value $i+1$ when any of the following holds:
\begin{enumerate}
\item Loop counters for $S_1$ and $S_2$ are always less than $i$ if any is present,
$\forall m_1'(\text{loop}_c^{1'})\, m_2'(\text{loop}_c^{2'})\,:\,
(S_1, m_1(\text{loop}_c^1, \vals_1)) ->* (S_1'', m_1'(\text{loop}_c^{1'})),
\text{loop}_c^{1'}(n_1) < i,
(S_2, m_2(\text{loop}_c^2,\allowbreak \vals_2)) ->* (S_2'', m_2'(\text{loop}_c^{2'})),
\text{loop}_c^{2'}(n_2) < i$,
$S_1$ and $S_2$ terminate in the same way, produce the same output sequence, and have equivalent computation of used/defined variables in both $S_1$ and $S_2$ and the input sequence variable, the I/O sequence variable,
$(S_1, m_1)\allowbreak \equiv_{H} (S_2, m_2)$ and
$(S_1, m_1) \equiv_{O} (S_2, m_2)$ and
$\forall x \in (\text{OVar}(S_1) \cap \text{OVar}(S_2)) \cup
               \{{id}_I,  {id}_{IO}\}\,:\,
               (S_1, m_1) \equiv_{x} (S_2, m_2)$;

\item The loop counter of $S_1$ and $S_2$ are of value less than or equal to $i$,
and there are no reachable configurations
$(S_1, m_1(\text{loop}_c^{1_i}, \vals_{1_i}))$ from $(S_1, m_1(\vals_1))$,
$(S_2, m_2(\text{loop}_c^{2_i}, \vals_{2_i}))$ from $(S_2, \,\allowbreak m_2(\vals_2))$ where all of the following hold:
\begin{itemize}
\item The loop counters of $S_1$ and $S_2$ are of value $i$,
$\text{loop}_c^{1_i}(n_1)\allowbreak = \text{loop}_c^{2_i}(n_2) = i$.

\item Value stores $\vals_{1_i}$ and $\vals_{2_i}$ agree on values of used variables in both $S_1$ and $S_2$ as well as the input sequence variable, and the I/O sequence variable,
$\forall x \in \, (\text{OVar}(S_1) \cap \text{OVar}(S_2)) \cup \{{id}_I,  {id}_{IO}\}\,:\,
\vals_{1_i}(x) = \vals_{2_i}(x)$.
\end{itemize}
\end{enumerate}

When there are  reachable configurations
$(S_1, m_1(\text{loop}_c^{1_i}, \vals_{1_i}))$ from $(S_1, m_1(\vals_1))$,
$(S_2, m_2(\text{loop}_c^{2_i}, \vals_{2_i}))$ from $(S_2, \,\allowbreak m_2(\vals_2))$ where all of the following hold:
\begin{itemize}
\item The loop counter of $S_1$ and $S_2$ are of value $i$,
$\text{loop}_c^{1_i}(n_1)\allowbreak = \text{loop}_c^{2_i}(n_2) = i$.

\item The loop counter of $S_1$ and $S_2$ are of value $i$,
$\text{loop}_c^{1_i}(n_1)\allowbreak = \text{loop}_c^{2_i}(n_2) = i$.

\item Value stores $\vals_{1_i}$ and $\vals_{2_i}$ agree on values of used variables in both $S_1$ and $S_2$ as well as the input sequence variable and the I/O sequence variable,
$\forall x \in \, (\text{OVar}(S_1) \cap \text{OVar}(S_2)) \cup \{{id}_I,  {id}_{IO}\}\,:\,
\vals_{1_i}(x) = \vals_{2_i}(x)$.
\end{itemize}

By similar argument in base case, we have one of the following holds:
\begin{enumerate}
\item Loop counters for $S_1$ and $S_2$ are always less than $i+1$ if any is present,
$\forall m_1'(\text{loop}_c^{1'})\, m_2'(\text{loop}_c^{2'})\,:\,
(S_1, m_1(\text{loop}_c^1, \vals_1)) ->* (S_1'', m_1'(\text{loop}_c^{1'})),
\text{loop}_c^{1'}(n_1) < i+1,
(S_2, m_2(\text{loop}_c^2,\allowbreak \vals_2)) ->* (S_2'', m_2'(\text{loop}_c^{2'})),
\text{loop}_c^{2'}(n_2) < i+1$,
$S_1$ and $S_2$ terminate in the same way, produce the same output sequence, and have equivalent computation of used/defined variables in both $S_1$ and $S_2$ and the input sequence variable, the I/O sequence variable,
$(S_1, m_1)\allowbreak \equiv_{H} (S_2, m_2)$ and
$(S_1, m_1) \equiv_{O} (S_2, m_2)$ and
$\forall x \in (\text{OVar}(S_1) \cap \text{OVar}(S_2)) \cup
               \{{id}_I,  {id}_{IO}\}\,:\,
               (S_1, m_1) \equiv_{x} (S_2, m_2)$;

\item The loop counter of $S_1$ and $S_2$ are of value less than or equal to $i+1$,
and there are no reachable configurations
$(S_1, m_1(\text{loop}_c^{1_{i+1}}, \vals_{1_{i+1}}))$ from $(S_1, m_1(\vals_1))$,
$(S_2, m_2(\text{loop}_c^{2_{i+1}},\allowbreak \vals_{2_{i+1}}))$ from $(S_2, \,\allowbreak m_2(\vals_2))$ where all of the following hold:
\begin{itemize}
\item The loop counters of $S_1$ and $S_2$ are of value $i$,
$\text{loop}_c^{1_{i+1}}(n_1)\allowbreak = \text{loop}_c^{2_{i+1}}(n_2) = i$.

\item Value stores $\vals_{1_{i+1}}$ and $\vals_{2_{i+1}}$ agree on values of used variables in both $S_1$ and $S_2$ as well as the input sequence variable and the I/O sequence variable,
$\forall x \in \, (\text{OVar}(S_1) \cap \text{OVar}(S_2)) \cup \{{id}_{I+1},  {id}_{IO}\}\,:\,
\vals_{1_{i+1}}(x) = \vals_{2_{i+1}}(x)$.
\end{itemize}

\item There are reachable configurations
$(S_1, m_1(\text{loop}_c^{1_{i+1}}, \vals_{1_{i+1}}))$ from $(S_1, m_1(\vals_1))$,
$(S_2, m_2(\text{loop}_c^{2_{i+1}}, \vals_{2_{i+1}}))$ from $(S_2, \,\allowbreak m_2(\vals_2))$ where all of the following hold:
\begin{itemize}
\item The loop counter of $S_1$ and $S_2$ are of value $i$,
$\text{loop}_c^{1_{i+1}}(n_1)\allowbreak = \text{loop}_c^{2_{i+1}}(n_2) = i+1$.

\item Value stores $\vals_{1_{i+1}}$ and $\vals_{2_{i+1}}$ agree on values of used variables in both $S_1$ and $S_2$ as well as the input sequence variable and the I/O sequence variable,
$\forall x \in \, (\text{OVar}(S_1) \cap \text{OVar}(S_2)) \cup \{{id}_{I+1},  {id}_{IO}\}\,:\,
\vals_{1_{i+1}}(x) = \vals_{2_{i+1}}(x)$.
\end{itemize}
\end{enumerate}
\end{proof} 

\subsection{Proof rule for enumeration type extension}

Enumeration types allow developers to list similar items. New code is usually accompanied with the introduction of new enumeration labels. Figure~\ref{fig:enumTypeExtExample} shows an example of the update. The new enum label $o_2$ gives a new option for matching the value of the variable $a$, which introduce the new code $b := 3 + c$.
\begin{figure}
\begin{small}
\begin{tabbing}
xxxxxx\=xxx\=xxx\=xxx\=xxxxxxxxxxxxxxxx\=xxxx\=xxx\=xxxxx\= \kill
\>1: \>{\bf enum} ${id}$ \{$o_1$\}\>\>          \>1': \> {\bf enum} ${id}$ \{$o_1, o_2$\} \> \\
\>2: \>a : enum $id$\>\>                        \>2': \> a : enum $id$                    \> \\
\>3: \>{\bf If} $(a == o_1)$ {\bf then}\>\>     \>3': \> {\bf If} $(a == o_1)$ {\bf then} \> \\
\>4: \> \>output $2 + c$ \>                       \>4': \> \>output $2 + c$               \> \\
\>5: \>\>\>                                     \>5': \> {\bf If} $(a == o_2)$ {\bf then} \> \\
\>6: \>\>\>                                     \>6': \> \>output $3 + c$                 \> \\
\>\\
\> \> old\>\>                                         \>  \> new \>
\end{tabbing}
\end{small}
\caption{Enumeration type extension}\label{fig:enumTypeExtExample}
\end{figure}
To show updates ``enumeration type extension" to be backward compatible,
we assume that values of enum variables, used in the If-predicate introducing the new code, are only from inputs that cannot be translated to new enum labels.

In order to have a general definition of the update class, we show a relation between two sequences of enumeration type definitions, called proper subset.
\begin{definition}
{\bf (Extension relation of enumeration types)}
Let $\text{EN}_1, \text{EN}_2$ be two different sequences of enumeration type definitions.
$\text{EN}_1$ is a subset of $\text{EN}_2$, written $\text{EN}_1 \subset \text{EN}_2$, iff one of the following holds:
\begin{enumerate}
\item $\text{EN}_1 = ``\text{enum}\, id \, \{{el}_1\}",
       \text{EN}_2 = ``\text{enum}\, id \, \{{el}_2\}"$ where labels in type ``enum $id$" in $\text{EN}_1$ are a subset of those in $\text{EN}_2$,
    ${el}_2 = {el}_1, {el}$ and ${el} \neq \varnothing$;

\item $\text{EN}_1, \text{EN}_2$ include more than one enumeration type definitions
 $\text{EN}_1 = \, ``\text{enum}\, id \, \{{el}_1\}, \text{EN}_1'",
  \text{EN}_2 = \, ``\text{enum}\, id \, \{{el}_2\}, \text{EN}_2'"$
 where one of the following holds:
\begin{enumerate}
\item $(\text{EN}_1' \subset \text{EN}_2')$ and $({el}_1 = {el}_2) \, \vee \, ({el}_2 = {el}_1,el)$;

\item $(\text{EN}_1' \subset \text{EN}_2') \, \vee \, (\text{EN}_1' = \text{EN}_2')$ and
$``\text{enum}\, id \, \{{el}_1\}" \, \subset \, ``\text{enum}\, id \, \{{el}_2\}"$.
\end{enumerate}
\end{enumerate}
\end{definition}

\begin{definition}\label{def:enumTypeExtFuncWide}
{\bf (Enumeration type extension)}
Let $P_1, P_2$ be two programs where enumeration type definitions $\text{EN}_1$ in $P_1$ are a subset of $\text{EN}_2$ in $P_2$, $\text{EN}_1 \subset \text{EN}_2$ and $E$ are new enum labels in $P_2$.
A statement sequence $S_2$ in a program $P_2$ includes updates of enumeration type extension compared with a statement sequence $S_1$ in $P_1$, written $S_2\, {\approx}_{E}^S\, S_1$, iff one of the following holds:
\begin{enumerate}
\item $S_2 = ``\text{If}(id \text{==} l)\, \text{then} \{S_2^t\} \, \text{else}\{S_2^f\}"$ and all of the following hold:
\begin{itemize}
\item $l \in E$;

\item The variable $id$ is not lvalue in an assignment statement,
 $``id := \text{e}" \, \notin \, P_2$;

\item $S_2^f\, {\approx}_{E}^S\, S_1$;
\end{itemize}

\item
$S_1 = ``\text{If}(e) \, \text{then}\{S_1^t\} \, \text{else} \{S_1^f\}"$,
$S_2 = ``\text{If}(e) \, \text{then}\{S_2^t\} \, \text{else} \{S_2^f\}"$ where
$(S_2^t\, {\approx}_{E}^S\, S_1^t) \, \wedge \, (S_2^f\, {\approx}_{E}^S\, S_1^f)$;

\item
$S_1 = ``\text{while}_{\langle n_1\rangle}(e) \, \{S_1'\}"$,
$S_2 = ``\text{while}_{\langle n_2\rangle}(e) \, \{S_2'\}"$ where

\noindent$S_2'\, {\approx}_{E}^S\, S_1'$;

\item $S_1 \approx_{O}^S S_2$;

\item $S_1 = S_1';s_1$ and $S_2 = S_2';s_2$ where
$(S_2'\, \approx_{E}^S\, S_1') \, \wedge \,
 (S_2'\, \approx_{H}^S\, S_1') \, \wedge \,
 \big(\forall x \in \text{Imp}(s_1, {id}_{IO}) \cup \text{Imp}(s_1, {id}_{IO})\,:\, (S_2'\, \approx_{x}^S\, S_1')\big) \, \wedge \,
 (s_2\, \approx_{E}^S\, s_1)$.
\end{enumerate}
\end{definition}

We show that two programs terminate in the same way, produce the same output sequence, and have equivalent computation of variables defined in both of them in executions if there are updates of enumeration type extension between them.
\begin{lemma}\label{lmm:enumTypeExtFuncWideAppend}
Let $S_1$ and $S_2$ be two statement sequences in programs $P_1$ and $P_2$ respectively where there are updates of enumeration type extensions in $S_2$  of $P_2$ compared with $S_1$ of $P_1$, $S_2\, {\approx}_{E}^S\, S_1$.
If $S_1$ and $S_2$ start in states $m_1(\vals_1)$ and $m_2(\vals_2)$ such that both of the following hold:
\begin{itemize}
\item Value stores $\vals_1$ and $\vals_2$ agree on values of output deciding variables in both $S_1$ and $S_2$ including the input sequence variable and the I/O sequence variable,
$\forall x \in (\text{OVar}(S_1)\cup \text{OVar}(S_2)) \cup \{{id}_I, {id}_{IO}\}\,:\,
 \vals_1(x) = \vals_2(x)$;

\item No variables  used in $S_2$ are of initial value of enum labels in $E$,
$\forall x \in \text{Use}(S_2)\,:\, (\vals_2(x) \notin E)$;

\item No inputs are translated to any label in $E$ during the execution of $S_2$;
\end{itemize}
then $S_1$ and $S_2$ terminate in the same way, produce the same output sequence, and when $S_1$ and $S_2$ both terminate, they have equivalent computation of used variables and defined variables,
\begin{itemize}
\item $(S_1, m_1) \equiv_{H} (S_2, m_2)$;

\item $(S_1, m_1) \equiv_{O} (S_2, m_2)$;

\item $\forall x \in \text{OVar}(S_1)\cup \text{OVar}(S_2)\,:\,
 \allowbreak (S_1, m_1) \equiv_{x} (S_2, m_2)$;
\end{itemize}
\end{lemma}
\begin{proof}
By induction on the sum of the program size of $S_1$ and $S_2$, $\text{size}(S_1) + \text{size}(S_1)$.

\noindent{Base case}.
$S_1$ is a simple statement $s$, and $S_2 = ``\text{If}(id \text{==} l)\, \text{then} \{s_2^t\} \,\allowbreak \text{else} \{s_2^f\}"$ where all of the following hold:
\begin{itemize}
\item $l \in E$;

\item $s_2^t, s_2^f$ are two simple statements;

\item $s_2^f = s$;
\end{itemize}

We informally argue that the value of the variable $id$ in the predicate expression of $S_2$ only coming from an input value or the initial value.
There are three ways a scalar variable is defined: the execution of an assignment statement, the execution of an input statement or the initial value.
Because $id$ is not lvalue in an assignment statement, then the value of $id$ is only from the execution of an input statement or the initial value.

In addition, by assumption, any output deciding variable is not of the initial value of enum label in $E$; no input values are translated into an enum label in $E$. Then the execution of $S_2$ proceeds as follows:
\begin{tabbing}
xx\=xx\=\kill
\>\>                   $(\text{If}(id \text{==} l) \, \text{then}\{s_2^t\} \, \text{else} \{s_2^f\}, m_2(\vals_2))$\\
\>$->$\>               $(\text{If}(0) \, \text{then}\{s_2^t\} \, \text{else} \{s_2^f\}, m_2(\vals_2))$\\
\>\>                   by the rule Eq-F\\
\>$->$\>               $(s_2^f, m_2(\vals_2))$ by the If-F rule.
\end{tabbing}

The value store $\vals_2$ is not updated in the execution of $S_2$ so far.
By assumption, value stores $\vals_1$ and $\vals_2$ agree on values of  output deciding variables in both $S_1$ and $S_2$.

By Theorem~\ref{thm:equivCompMain} and~\ref{thm:mainTermSameWayLocal}, $S_1$ and $S_2$ terminate in the same way, produce the same I/O sequence. The lemma holds.

\noindent{Induction step}.

The hypothesis is that this lemma holds when the sum $k$ of the program size of $S_1$ and $S_2$ are great than or equal to 4, $k\geq 4$.

We then show that this lemma holds when the sum of the program size of $S_1$ and $S_2$ is $k+1$.
There are cases regarding $S_2 {\approx}_{E}^S S_1$.
\begin{enumerate}
\item $S_1$ and $S_2$ are both ``If" statement:

$S_1 = ``\text{If}(e) \, \text{then}\{S_1^t\} \, \text{else} \{S_1^f\}"$,
$S_2 = ``\text{If}(e) \, \text{then}\{S_2^t\} \, \text{else} \{S_2^f\}"$ where both of the following hold
\begin{itemize}
\item $S_2^t {\approx}_{E}^S S_1^t$;

\item $S_2^f {\approx}_{E}^S S_1^f$;
\end{itemize}

By the definition of $\text{Imp}_o(S_1)$, variables used in the predicate expression $e$ are a subset of  output deciding variables in $S_1$ and $S_2$,
$\text{Use}(e) \subseteq \text{OVar}(S_1) \cap \text{OVar}(S_2)$.
By assumption, corresponding variables used in $e$ are of same value in value stores $\vals_1$ and $\vals_2$. By Lemma~\ref{lmm:expEvalSameVal}, the expression evaluates to the same value w.r.t value stores $\vals_1$ and $\vals_2$. There are three possibilities.
\begin{enumerate}
\item The evaluation of $e$ crashes,
$\mathcal{E}'\llbracket e\rrbracket \vals_1 =
 \mathcal{E}'\llbracket e\rrbracket \vals_2 = (\text{error}, v_{\mathfrak{of}})$.

The execution of $S_1$ continues as follows:
\begin{tabbing}
xx\=xx\=\kill
\>\>                   $(\text{If}(e) \, \text{then}\{S_1^t\} \, \text{else} \{S_1^f\}, m_1(\vals_1))$\\
\>$->$\>               $(\text{If}((\text{error}, v_{\mathfrak{of}})) \, \text{then}\{S_1^t\} \, \text{else} \{S_1^f\}, m_1(\vals_1))$\\
\>\>                   by the rule EEval'\\
\>$->$\>               $(\text{If}(0) \, \text{then}\{S_1^t\} \, \text{else} \{S_1^f\}, m_1(1/\mathfrak{f}))$\\
\>\>                   by the ECrash rule  \\
\>{\kStepArrow [i] }\> $(\text{If}(0) \, \text{then}\{S_1^t\} \, \text{else} \{S_1^f\}, m_1(1/\mathfrak{f}))$ for any $i>0$\\
\>\>                   by the Crash rule.
\end{tabbing}

Similarly, the execution of $S_2$ started from the state $m_2(\vals_2)$ crashes.
The lemma holds.

\item The evaluation of $e$ reduces to zero,
$\mathcal{E}'\llbracket e\rrbracket \vals_1 =
 \mathcal{E}'\llbracket e\rrbracket \vals_2 = (0, v_{\mathfrak{of}})$.

The execution of $S_1$ continues as follows.
\begin{tabbing}
xx\=xx\=\kill
\>\>                   $(\text{If}(e) \, \text{then}\{S_1^t\} \, \text{else} \{S_1^f\}, m_1(\vals_1))$\\
\>= \>                 $(\text{If}((0, v_{\mathfrak{of}})) \, \text{then}\{S_1^t\} \, \text{else} \{S_1^f\}, m_1(\vals_1))$\\
\>\>                   by the rule EEval'\\
\>$->$\>               $(\text{If}(0) \, \text{then}\{S_1^t\} \, \text{else} \{S_1^f\}, m_1(\vals_1))$\\
\>\>                   by the E-Oflow1 or E-Oflow2 rule  \\
\>$->$\>               $(S_1^f, m_1(\vals_1))$ by the If-F rule.
\end{tabbing}

Similarly, the execution of $S_2$ gets to the configuration $(S_2^f, m_2(\vals_2))$.

By the hypothesis IH, we show the lemma holds.
We need to show that all conditions are satisfied for the application of the hypothesis IH.
\begin{itemize}
\item $S_2^f {\approx}_{E}^S S_1^f$

By assumption.

\item The sum of  the program size of $S_1^f$ and $S_2^f$ is less than $k$,
$\text{size}(S_1^f) + \text{size}(S_2^f) < k$.

By definition, $\text{size}(S_1) = 1 + \text{size}(S_1^t) + \text{size}(S_1^f)$.
Then, $\text{size}(S_1^f) + \text{size}(S_2^f) < k + 1 - 2 = k - 1$.

\item Value stores $\vals_1$ and $\vals_2$ agree on values of  output deciding variables in $S_1^f$ and $S_2^f$ including the input, I/O sequence variable.

By definition, $\text{OVar}(S_1^f)\allowbreak \subseteq \text{OVar}(S_1)$.
So are the cases to $S_2^f$ and $S_2$.
In addition, value stores $\vals_1$ and $\vals_2$ are not changed in the evaluation of the predicate expression $e$. The condition holds.

\item There are no inputs translated to enum labels in $E$ in $S_2^f$'s execution.

By assumption.
\end{itemize}

By the hypothesis IH, the lemma holds.

\item The evaluation of $e$ reduces to the same nonzero integer value,
$\mathcal{E}'\llbracket e\rrbracket \vals_1 =
 \mathcal{E}'\llbracket e\rrbracket \vals_2 = (0, v_{\mathfrak{of}})$.

By similar to the second subcase above.
\end{enumerate}

\item $S_1$ and $S_2$ are both ``while" statements:

$S_1 = ``\text{while}_{\langle n_1\rangle}(e) \, \{S_1'\}"$,
$S_2 = ``\text{while}_{\langle n_2\rangle}(e) \, \{S_2'\}"$ where

\noindent$S_2' {\approx}_{E}^S S_1'$;

By Lemma~\ref{lmm:enumTypeExtLoopStmt}, we show the lemma holds.
We need to show all the required conditions for the application of Lemma~\ref{lmm:enumTypeExtLoopStmt} holds
\begin{enumerate}
\item No variables are of initial values as new enum labels in $E$;

\item Value stores $\vals_1$ and $\vals_2$ agree on values of variables used in both $S_1$ and $S_2$;

\item Enumeration types in $P_1$ are a subset of those in $P_2$;

The above three conditions are by assumption.

\item The output deciding variables in $S_1'$ are a subset of those in $S_2'$;


The above condition is by Lemma~\ref{lmm:enumTypeExtSameDefAndUse}.

\item $S_1'$ and $S_2'$ produce the same output sequence, terminate in the same way and have equivalent computation of defined variables in both $S_1'$ and $S_2'$ when started in states agreeing on values of variables used in both $S_1'$ and $S_2'$;

Because $\text{size}(S_1) = \text{size}(S_1') + 1$, then the condition holds by the induction hypothesis.
\end{enumerate}
By Lemma~\ref{lmm:enumTypeExtLoopStmt}, the lemma holds.

\item $S_2 = ``\text{If}(id \text{==} l) \, \text{then} \, \{S_2^t\} \, \text{else} \, \{S_2^f\}"$ such that both of the following hold:
  \begin{itemize}
    \item The label $l$ is in $E$, $l \in E$;

    \item The variable $id$ is not lvalue in an assignment statement,
    $\nexists ``id := \text{e}" \, \text{in} \, S_2$;

    \item There are updates of enumeration type extension from $S_2^f$ to $S_1$,
    $S_2^f {\approx}_{E}^S S_1$;
  \end{itemize}
By Lemma~\ref{lmm:enumTypeExtLoopStmt}, we show this lemma holds.
We need to show all the conditions are satisfied for the application of Lemma~\ref{lmm:enumTypeExtLoopStmt}.
\begin{itemize}
\item $S_1'$ and $S_2'$ have same set of  output deciding variables,
$\text{OVar}(S_1') = \text{OVar}(S_2') = \text{OVar}(S)$;

\item The  output deciding variables in $S_1'$ are a subset of those in $S_2'$,
$\text{OVar}(S_1') \subseteq \text{OVar}(S_2')$;

By Lemma~\ref{lmm:enumTypeExtSameDefAndUse}.

\item There are no inputs translated to enum labels in the set $E$.

By assumption.

\item When started in states $m_1'(\vals_1'), m_2'(\vals_1')$ where
value stores $\vals_1'$ and $\vals_2'$ agree on values of  output deciding variables in both $S_1'$ and $S_2'$ as well as the input sequence variable and the I/O sequence variable, and there are no inputs translated to enum labels in $E$,
then $S_1'$ and $S_2'$ produce the same output sequence.

By the induction hypothesis IH. This is because the sum of the program size of $S_1'$ and $S_2'$ is less than $k$. By definition, $\text{size}(S_1) = 1 + \text{size}(S_1')$.
\end{itemize}

\item $S_2 = ``\text{If}(id \text{==} l) \, \text{then} \, \{S_2^t\} \, \text{else} \, \{S_2^f\}"$ such that both of the following hold:
  \begin{itemize}
    \item The label $l$ is in $E$, $l \in E$;

    \item The variable $id$ is not lvalue in an assignment statement,
    $``id := \text{e}" \, \notin \, S_2$;

    \item There are updates of enumeration type extension from $S_2^f$ to $S_1$,
    $S_2^f {\approx}_{E}^S S_1$;
  \end{itemize}

We informally argue that the value of the variable $id$ in the predicate expression of $S_2$ only coming from an input value or the initial value.
There are several ways a scalar variable is defined: the execution of an assignment statement, the execution of an input statement or the initial value.
Because $id$ is not lvalue in an assignment statement, then the value of $id$ is only from the execution of an input statement or initial value. In addition, by assumption, any  used variable is not of initial value of enum label in $E$; no input values are translated into an enum label in $E$.

Then the expression $id \text{==} l$ evaluates to 0. The execution of $S_2$ proceeds as follows.
\begin{tabbing}
xx\=xx\=\kill
\>\>                   $(\text{If}(id \text{==} l) \, \text{then}\{S_2^t\} \, \text{else} \{S_2^f\}, m_2(\vals_2))$\\
\>$->$\>               $(\text{If}(v \text{==} l) \, \text{then}\{S_2^t\} \, \text{else} \{S_2^f\}, m_2(\vals_2))$ where $v \neq l$\\
\>\>                   by the rule Var\\
\>$->$\>               $(\text{If}(0) \, \text{then}\{S_2^t\} \, \text{else} \{S_2^f\}, m_1(\vals_2))$\\
\>\>                   by the Eq-F rule  \\
\>$->$\>               $(S_2^f, m_2(\vals_2))$ by the If-F rule.
\end{tabbing}

By the hypothesis IH, we show the lemma holds.
We need to show the conditions are satisfied for the application of the hypothesis IH.
\begin{itemize}
\item $S_2^f {\approx}_{E}^S S_1$

\item There are no inputs translated to enum labels in $E$ in $S_2^f$'s execution.

\item Initial values of  used variables in $S_2$ are not enum labels in $E$.

The above three conditions are by assumption.

\item The sum of  the program size of $S_1$ and $S_2^f$ is less than $k$,
$\text{size}(S_1) + \text{size}(S_2^f) < k$.

By definition, $\text{size}(S_1) = 1 + \text{size}(S_1^t) + \text{size}(S_1^f)$.
Then, $\text{size}(S_1) + \text{size}(S_2^f) < k + 1 - 1 - \text{size}(S_2^t) < k$.

\item Value stores $\vals_1$ and $\vals_2$ agree on values of  used variables in both $S_1$ and $S_2^f$ as well as the input, I/O sequence variable.

By definition, $\text{OVar}(S_2^f)\allowbreak \subseteq \text{OVar}(S_2)$.
In addition, value stores $\vals_1$ and $\vals_2$ are not changed in the evaluation of the predicate expression $e$. The condition holds.
\end{itemize}
By the hypothesis IH, this lemma holds.

\item $S_1 = S_1';s_1$ and $S_2 = S_2';s_2$ such that both of the following hold:

\begin{itemize}
\item $S_2' \approx_{E}^S S_1'$;

\item $s_2 \approx_{E}^S s_1$;
\end{itemize}

By the hypothesis IH, we show $S_2'$ and $S_1'$ terminate in the same way, produce the same output sequence, and have equivalent computation of  defined variables in $S_1$ and $S_2$.
We need to show that all the conditions are satisfied for the application of the hypothesis IH.
\begin{itemize}
\item $S_2' {\approx}_{E}^S S_1'$;

\item There are no inputs translated to enum labels in $E$ in $S_2'$'s execution.

\item Initial values of  used variables in $S_2'$ are not enum labels in $E$.

The above three conditions are by assumption.

\item The sum of the program size of $S_1'$ and $S_2'$ is less than $k$,
$\text{size}(S_1') + \text{size}(S_2') < k$.

By definition, $\text{size}(S_1) = \text{size}(S_1') + \text{size}(s_1)$.
Then, $\text{size}(S_1') + \text{size}(S_2') < k + 1 - \text{size}(s_2) - \text{size}(s_1) < k$.

\item Value stores $\vals_1$ and $\vals_2$ agree on values of  used variables in both $S_1'$ and $S_2'$ as well as the input, output, I/O sequence variable.

By definition, $\text{OVar}(S_2')\allowbreak \subseteq \text{OVar}(S_2),
                \text{OVar}(S_1')\allowbreak \subseteq \text{OVar}(S_1)$.
In addition, value stores $\vals_1$ and $\vals_2$ are not changed in the evaluation of the predicate expression $e$. The condition holds.
\end{itemize}

By the hypothesis IH, one of the following holds:
\begin{enumerate}
\item $S_1'$ and $S_2'$ both do not terminate.

By Lemma~\ref{lmm:multiStepSeqExec}, executions of $S_1 = S_1';s_1$ and $S_2 = S_2';s_2$ both do not terminate and produce the same output sequence.

\item $S_1'$ and $S_2'$ both terminate.

By assumption,
$(S_2', m_2(\vals_2)) ->* (\text{skip}, m_2'(\vals_2'))$,
$(S_1', m_1(\vals_1)) ->* (\text{skip}, m_1'(\vals_1'))$.

By Corollary~\ref{coro:termSeq},
$(S_2';s_2, m_2(\vals_2)) ->* (s_2, m_2'(\vals_2'))$,
$(S_1';s_1, m_1(\vals_1)) ->* (s_1, m_1'(\vals_1'))$.

By the hypothesis IH, we show that $s_2$ and $s_1$ terminate in the same way, produce the same output sequence and when $s_2$ and $s_1$ both terminate, $s_2$ and $s_1$ have equivalent computation of variables  used or defined in $s_1$ and $s_2$ and the input, output, and I/O sequence variables.

We need to show that all conditions are satisfied for the application of the hypothesis IH.
\begin{itemize}
\item There are updates of ``enumeration type extension" between $s_2$ and $s_1$;

\item There are no input values translated into enum labels in $E$ in the execution of $s_2$;

The above two conditions are by assumption.

\item The sum of the program size $s_2$ and $s_1$ is less than or equals to $k$;

By definition, $\text{size}(S_2') \geq 1,  \text{size}(S_1') \geq 1$.
Therefore, $\text{size}(s_2) + \text{size}(s_1) < k+1 - \text{size}(S_2') - \text{size}(S_1') \leq k$.

\item Value stores $\vals_1'$ and $\vals_2'$ agree on values of  used variables in $s_2$ and $s_1$ as well as the input, output, I/O sequence variable.

By Lemma~\ref{lmm:enumTypeExtSameDefAndUse},
$\text{OVar}(s_1) \subseteq \text{OVar}(s_2)$, then $\text{OVar}(s_2)\,\allowbreak   \cap \text{OVar}(s_1) = \text{OVar}(s_1)$.
Similarly, by Lemma~\ref{lmm:enumTypeExtSameDefAndUse}, $\text{OVar}(S_1') \subseteq \text{OVar}(S_2')$. For any variable $id$ in $\text{OVar}(s_1)$, if $id$ is not in $\text{OVar}(S_1')$,
then the value of $id$ is not changed in the execution of $S_1'$ and $S_2'$,
$\vals_1'(id) = \vals_1(id) = \vals_2(id) = \vals_2'(id)$.
Otherwise, the variable $id$ is defined in the execution of $S_1'$ and $S_2'$, by assumption, $\vals_1'(id) = \vals_2'(id)$.
The condition holds.

\item Values of  used variables in $s_2$ are not of value as enum labels in $E$,
$\forall id \in \text{OVar}(s_2)\,:\, \vals_2'(id) \in E$.

By assumption, initial values of  used variables in $s_2$ are not of values as enum labels in $E$. $S_2'$ and $S_1'$ have equivalent computation of  defined variables in $S_2'$ and $S_1'$.
Because enum labels are not defined in $P_1$, defined variables in the execution of $S_2'$ and $S_1'$ are not of values as enum labels in $E$.
\end{itemize}
By the hypothesis IH, the lemma holds.
\end{enumerate}

\end{enumerate}
\end{proof}

We show a auxiliary lemma telling that the two programs with updates of enumeration type extension have same set of used variables and the same set of  defined variables.
\begin{lemma}\label{lmm:enumTypeExtSameDefAndUse}
If there are updates of enumeration type extension in a statement sequence $S_2$ against a statement sequence $S_1$, $S_2 {\approx}_{E}^S S_1$,
then the output deciding variables in $S_1$ are a subset of those in $S_2$,
$\text{OVar}(S_1) \subseteq \text{OVar}(S_2)$.
\end{lemma}
\begin{proof}
By induction on the sum of the program size of $S_1$ and $S_2$.
\end{proof}

\begin{lemma}\label{lmm:enumTypeExtLoopStmt}
Let $S_1 = \text{while}_{\langle n_1\rangle}(e) \, \{S_1'\}$ and
    $S_2 = \text{while}_{\langle n_2\rangle}(e) \,\allowbreak \{S_2'\}$ be two loop statements in programs $P_1$ and $P_2$ respectively where all of the following hold:
\begin{itemize}
\item Enumeration types ${EN}_1$ in $P_1$ are a proper subset of ${EN}_2$ in $P_2$, ${EN}_1 \subset {EN}_2$, such that there are a set of enum labels $E$ only defined in $P_2$;

\item When started in states agreeing on values of  output deciding variables in both $S_1'$ and $S_2'$ as well as the input sequence variable and the I/O sequence variable,
initial values of  used variables in $S_2'$ are not enum labels in $E$,
and there are no inputs in $S_2$'s execution translated into any label in $E$,
    $\forall x \in \text{OVar}(S_1') \cup \{{id}_I,  {id}_{IO}\}\,\, \forall m_1(\vals_1)\,\,\allowbreak m_2(\vals_2)\,:\,
    \vals_1(x) = \vals_2(x)$, and
$S_1'$ and $S_2'$ terminate in the same way, produce the same output sequence, and have equivalent computation of defined variables in $S_1'$ and $S_2'$ as well as the input sequence variable and the I/O sequence variable
    $((S_1', m_1) \equiv_{H} (S_2', m_2)) \wedge
     ((S_1', m_1) \equiv_{O} (S_2', m_2)) \wedge
     (\forall x \in \text{OVar}(S_1) \cup \text{OVar}(S_2) \cup \{{id}_I, {id}_{IO}\}\,:\,$

\noindent$(S_1', \allowbreak m_1) \equiv_{x} (S_2', m_2))$;
\end{itemize}
If $S_1$ and $S_2$ start in states $m_1(\text{loop}_c^1, \vals_1), m_2(\text{loop}_c^2, \vals_2)$, with loop counters of $S_1$ and $S_2$ not initialized ($S_1, S_2$ have not executed yet), value stores agree on values of  output deciding variables in $S_1$ and $S_2$ as well as the input sequence variable,  the I/O sequence variable,
initial values of  used variables in $S_2$ are not of values as enum labels in $E$,
no inputs are translated into enum labels in $E$,
  then, for any positive integer $i$, one of the following holds:
\begin{enumerate}
\item Loop counters for $S_1$ and $S_2$ are always less than $i$ if any is present,
$\forall m_1'(\text{loop}_c^{1'})\, m_2'(\text{loop}_c^{2'})\,:\,
(S_1, m_1(\text{loop}_c^1, \vals_1)) ->* (S_1'', m_1'(\text{loop}_c^{1'})),
\text{loop}_c^{1'}(n_1) < i,
(S_2, m_2(\text{loop}_c^2,\allowbreak \vals_2)) ->* (S_2'', m_2'(\text{loop}_c^{2'})),
\text{loop}_c^{2'}(n_2) < i$,
$S_1$ and $S_2$ terminate in the same way, produce the same output sequence, and have equivalent computation of  output deciding variables in both $S_1$ and $S_2$ and the input sequence variable,  the I/O sequence variable,
$(S_1, m_1)\allowbreak \equiv_{H} (S_2, m_2)$ and
$(S_1, m_1) \equiv_{O} (S_2, m_2)$ and
$\forall x \in (\text{OVar}(S_1) \cup \text{OVar}(S_2)) \cup
               \{{id}_I,  {id}_{IO}\}\,:\,
               (S_1, m_1) \equiv_{x} (S_2, m_2)$;

\item The loop counter of $S_1$ and $S_2$ are of value less than or equal to $i$,
and there are no reachable configurations
$(S_1, m_1(\text{loop}_c^{1_i}, \vals_{1_i}))$ from $(S_1, m_1(\vals_1))$,
$(S_2, m_2(\text{loop}_c^{2_i}, \vals_{2_i}))$ from $(S_2, \,\allowbreak m_2(\vals_2))$ where all of the following hold:
\begin{itemize}
\item The loop counters of $S_1$ and $S_2$ are of value $i$,
$\text{loop}_c^{1_i}(n_1)\allowbreak = \text{loop}_c^{2_i}(n_2) = i$.

\item Value stores $\vals_{1_i}$ and $\vals_{2_i}$ agree on values of  output deciding variables in both $S_1$ and $S_2$ as well as the input sequence variable  and the I/O sequence variable,
$\forall x \in \, (\text{OVar}(S_1) \cap \text{OVar}(S_2)) \cup \{{id}_I,   {id}_{IO}\}\,:\,
\vals_{1_i}(x) = \vals_{2_i}(x)$.
\end{itemize}

\item There are reachable configurations
$(S_1, m_1(\text{loop}_c^{1_i}, \vals_{1_i}))$ from $(S_1, m_1(\vals_1))$,
$(S_2, m_2(\text{loop}_c^{2_i}, \vals_{2_i}))$ from $(S_2, \,\allowbreak m_2(\vals_2))$ where all of the following hold:
\begin{itemize}
\item The loop counter of $S_1$ and $S_2$ are of value $i$,
$\text{loop}_c^{1_i}(n_1)\allowbreak = \text{loop}_c^{2_i}(n_2) = i$.

\item Value stores $\vals_{1_i}$ and $\vals_{2_i}$ agree on values of  output deciding variables in both $S_1$ and $S_2$ as well as the input sequence variable and the I/O sequence variable,
$\forall x \in \, (\text{OVar}(S_1) \cap \text{OVar}(S_2)) \cup \{{id}_I,   {id}_{IO}\}\,:\,
\vals_{1_i}(x) = \vals_{2_i}(x)$.
\end{itemize}
\end{enumerate}
\end{lemma}

\begin{proof}
By induction on $i$.

\noindent{Base case}.

We show that, when $i=1$, one of the following holds:
\begin{enumerate}
\item Loop counters for $S_1$ and $S_2$ are always less than 1 if any is present,
$\forall m_1'(\text{loop}_c^{1'})\, m_2'(\text{loop}_c^{2'})\,:\,
(S_1, m_1(\text{loop}_c^1, \vals_1)) ->* (S_1'', m_1'(\text{loop}_c^{1'})),
\text{loop}_c^{1'}(n_1) < i,
(S_2, m_2(\text{loop}_c^2,\allowbreak \vals_2)) ->* (S_2'', m_2'(\text{loop}_c^{2'})),
\text{loop}_c^{2'}(n_2) < i$,
$S_1$ and $S_2$ terminate in the same way, produce the same output sequence, and have equivalent computation of  defined variables in both $S_1$ and $S_2$ and the input sequence variable,  the I/O sequence variable,
$(S_1, m_1)\allowbreak \equiv_{H} (S_2, m_2)$ and
$(S_1, m_1) \equiv_{O} (S_2, m_2)$ and
$\forall x \in (\text{Def}(S_1) \cap \text{Def}(S_2)) \cup
               \{{id}_I,   {id}_{IO}\}\,:\,
               (S_1, m_1) \equiv_{x} (S_2, m_2)$;

\item Loop counters of $S_1$ and $S_2$ are of value less than or equal to 1 but there are no reachable configurations
$(S_1, m_1(\text{loop}_c^{1_1}, \vals_{1_i}))$ from $(S_1, m_1(\vals_1))$,
$(S_2, m_2(\text{loop}_c^{2_1}, \vals_{2_i}))$ from $(S_2, \,\allowbreak m_2(\vals_2))$ where all of the following hold:
\begin{itemize}
\item The loop counter of $S_1$ and $S_2$ are of value 1,
$\text{loop}_c^{1_1}(n_1)\allowbreak = \text{loop}_c^{2_1}(n_2) = 1$.

\item Value stores $\vals_{1_1}$ and $\vals_{2_1}$ agree on values of  used variables in both $S_1$ and $S_2$ as well as the input sequence variable,  and the I/O sequence variable,
$\forall x \in \, (\text{Use}(S_1) \cap \text{Use}(S_2)) \cup \{{id}_I,   {id}_{IO}\}\,:\,
\vals_{1_1}(x) = \vals_{2_1}(x)$.
\end{itemize}

\item There are reachable configuration
$(S_1, m_1(\text{loop}_c^{1_1}, \vals_{1_i}))$ from $(S_1, m_1(\vals_1))$,
$(S_2, m_2(\text{loop}_c^{2_1}, \vals_{2_i}))$ from $(S_2, \,\allowbreak m_2(\vals_2))$ where all of the following hold:
\begin{itemize}
\item The loop counter of $S_1$ and $S_2$ are of value 1,
$\text{loop}_c^{1_1}(n_1)\allowbreak = \text{loop}_c^{2_1}(n_2) = 1$.

\item Value stores $\vals_{1_1}$ and $\vals_{2_1}$ agree on values of  used variables in both $S_1$ and $S_2$ as well as the input sequence variable,  and the I/O sequence variable,
$\forall x \in \, (\text{Use}(S_1) \cap \text{Use}(S_2)) \cup \{{id}_I,   {id}_{IO}\}\,:\,
\vals_{1_1}(x) = \vals_{2_1}(x)$.
\end{itemize}
\end{enumerate}

By definition, variables used in the predicate expression $e$ of $S_1$ and $S_2$ are  used in $S_1$ and $S_2$, $\text{Use}(e) \subseteq \text{Use}(S_1) \cap \text{Use}(S_2)$.
By assumption, value stores $\vals_1$ and $\vals_2$ agree on values of variables in $\text{Use}(e)$,
the predicate expression $e$ evaluates to the same value w.r.t value stores $\vals_1$ and $\vals_2$ by Lemma~\ref{lmm:expEvalSameTerm}.
There are three possibilities.
\begin{enumerate}
\item The evaluation of $e$ crashes,

\noindent$\mathcal{E}'\llbracket e\rrbracket \vals_1 =
 \mathcal{E}'\llbracket e\rrbracket \vals_2 = (\text{error}, v_{\mathfrak{of}})$.

The execution of $S_1$ continues as follows:
\begin{tabbing}
xx\=xx\=\kill
\>\>                   $(\text{while}_{\langle n_1\rangle}(e) \, \{S_1'\}, m_1(\vals_1))$\\
\>$->$\>               $(\text{while}_{\langle n_1\rangle}((\text{error}, v_{\mathfrak{of}})) \, \{S_1'\}, m_1(\vals_1))$\\
\>\>                   by the rule EEval'\\
\>$->$\>               $(\text{while}_{\langle n_1\rangle}(0) \, \{S_1'\}, m_1(1/\mathfrak{f}))$\\
\>\>                   by the ECrash rule  \\
\>{\kStepArrow [i] }\> $(\text{while}_{\langle n_1\rangle}(0) \, \{S_1'\}, m_1(1/\mathfrak{f}))$ for any $i>0$\\
\>\>                   by the Crash rule.
\end{tabbing}

Similarly, the execution of $S_2$ started from the state $m_2(\vals_2)$ crashes.
Therefore $S_1$ and $S_2$ terminate in the same way when started from $m_1$ and $m_2$ respectively.
Because $\vals_1({id}_{IO}) = \vals_2({id}_{IO})$, the lemma holds.

\item The evaluation of $e$ reduces to zero,
$\mathcal{E}'\llbracket e\rrbracket \vals_1 =
 \mathcal{E}'\llbracket e\rrbracket \vals_2 = (0, v_{\mathfrak{of}})$.

The execution of $S_1$ continues as follows.
\begin{tabbing}
xx\=xx\=\kill
\>\>                   $(\text{while}_{\langle n_1\rangle}(e) \, \{S_1'\}, m_1(\vals_1))$\\
\>= \>                 $(\text{while}_{\langle n_1\rangle}((0, v_{\mathfrak{of}})) \, \{S_1'\}, m_1(\vals_1))$\\
\>\>                   by the rule EEval'\\
\>$->$\>               $(\text{while}_{\langle n_1\rangle}(0) \, \{S_1'\}, m_1(\vals_1))$\\
\>\>                   by the E-Oflow1 or E-Oflow2 rule  \\
\>$->$\>               $(\text{skip}, m_1(\vals_1))$ by the Wh-F rule.
\end{tabbing}

Similarly, the execution of $S_2$ gets to the configuration $(\text{skip},\,\allowbreak m_2(\vals_2))$.
Loop counters of $S_1$ and $S_2$ are less than 1 and value stores agree on values of  used/defined variables in both $S_1$ and $S_2$ as well as the input sequence variable,  and the I/O sequence variable.

\item The evaluation of $e$ reduces to the same nonzero integer value,
$\mathcal{E}'\llbracket e\rrbracket \vals_1 =
 \mathcal{E}'\llbracket e\rrbracket \vals_2 = (0, v_{\mathfrak{of}})$.

Then the execution of $S_1$ proceeds as follows:
\begin{tabbing}
xx\=xx\=\kill
\>\>                   $(\text{while}_{\langle n_1\rangle}(e) \, \{S_1'\}, m_1(\vals_1))$\\
\>= \>                 $(\text{while}_{\langle n_1\rangle}((v, v_{\mathfrak{of}})) \, \{S_1'\}, m_1(\vals_1))$\\
\>\>                   by the rule EEval'\\
\>$->$\>               $(\text{while}_{\langle n_1\rangle}(v) \, \{S_1'\}, m_1(\vals_1))$\\
\>\>                   by the E-Oflow1 or E-Oflow2 rule  \\
\>$->$\>               $(S_1';\text{while}_{\langle n_1\rangle}(e) \, \{S_1'\}, m_1($\\
\>\>                   $\text{loop}_c^1[1/n_1], \vals_1))$ by the Wh-T rule.
\end{tabbing}

Similarly, the execution of $S_2$ proceeds to the configuration
$(S_2';\text{while}_{\langle n_2\rangle}(e) \, \{S_2'\}, m_2(\text{loop}_c^2[1/n_1], \vals_2))$.

By the assumption, we show that $S_1'$ and $S_2'$ terminate in the same way and produce the same output sequence when started in the state $m_1(\text{loop}_c^{1_1}, \vals_1)$ and $m_2(\text{loop}_c^{2_1}, \vals_2)$ respectively, and $S_1'$ and $S_2'$ have equivalent computation of variables  defined in both statement sequences if both terminate.
We need to show that all conditions are satisfied for the application of the assumption.
\begin{itemize}
\item There are no inputs translated into enum labels in $E$ in the execution of $S_2'$.

The above condition is by assumption.

\item Initial values of  used variables in $S_2'$ are not enum labels in $E$.

By the definition of  used variables, $\text{Use}(S_2') \subseteq \text{Use}(S_2)$.
By assumption, initial values of  used variables in $S_2$ are not enum labels in $E$.
The condition holds.

\item Value stores $\vals_1$ and $\vals_2$ agree on values of  used variables in $S_1'$ and $S_2'$ as well as the input, output, I/O sequence variable.

By definition, $\text{Use}(S_1')\allowbreak \subseteq \text{Use}(S_1)$.
So are the cases to $S_2'$ and $S_2$.
In addition, value stores $\vals_1$ and $\vals_2$ are not changed in the evaluation of the predicate expression $e$. The condition holds.
\end{itemize}
By  assumption, $S_1'$ and $S_2'$ terminate in the same way and produce the same output sequence when started in states
$m_1(\text{loop}_c', \vals_1)$ and $m_2(\text{loop}_c', \vals_2)$.
In addition, $S_1'$ and $S_2'$ have equivalent computation of variables  used or defined in $S_1'$ and $S_2'$ when started in states
$m_1(\text{loop}_c', \vals_1)$ and $m_2(\text{loop}_c', \vals_2)$.

Then there are two cases.
\begin{enumerate}
\item $S_1'$ and $S_2'$ both do not terminate and produce the same output sequence.

By Lemma~\ref{lmm:multiStepSeqExec}, $S_1';S_1$ and $S_2';S_2$ both do not terminate and produce the same output sequence.

\item $S_1'$ and $S_2'$ both terminate and have equivalent computation of variables  defined in $S_1'$ and $S_2'$.

By assumption, $(S_1', m_1(\text{loop}_c', \vals_1)) ->*
                (\text{skip}, \,\allowbreak m_1'(\text{loop}_c'', \vals_1'))$;
               $(S_2', m_2(\text{loop}_c', \vals_2)) ->*
                (\text{skip},\allowbreak m_2'(\text{loop}_c'', \vals_2'))$
where $\forall x \in (\text{Def}(S_1') \cap \text{Def}(S_2')) \cup
                     \{{id}_I,   {id}_{IO}\},
\allowbreak\vals_1'(x) = \vals_2'(x)$.

By assumption,
$\text{Use}(S_1') \subseteq \text{Use}(S_2')$ and
$\text{Def}(S_1') = \text{Def}(S_2')$.
Then variables used in the predicate expression of $S_1$ and $S_2$ are either in variables  used or defined in both $S_1'$ and $S_2'$ or not.
Therefore value stores $\vals_2'$ and $\vals_1'$ agree on values of variables used in the expression $e$ and even variables  used or defined in $S_1$ and $S_2$.
\end{enumerate}
\end{enumerate}

\noindent{Induction step on iterations}

The induction hypothesis (IH) is that, when $i\geq 1$, one of the following holds:
\begin{enumerate}
\item Loop counters for $S_1$ and $S_2$ are always less than $i$ if any is present,
$\forall m_1'(\text{loop}_c^{1'})\, m_2'(\text{loop}_c^{2'})\,:\,
(S_1, m_1(\text{loop}_c^1, \vals_1)) ->* (S_1'', m_1'(\text{loop}_c^{1'})),
\text{loop}_c^{1'}(n_1) < i,
(S_2, m_2(\text{loop}_c^2,\allowbreak \vals_2)) ->* (S_2'', m_2'(\text{loop}_c^{2'})),
\text{loop}_c^{2'}(n_2) < i$,
$S_1$ and $S_2$ terminate in the same way, produce the same output sequence, and have equivalent computation of  used/defined variables in both $S_1$ and $S_2$ and the input sequence variable,  the I/O sequence variable,
$(S_1, m_1)\allowbreak \equiv_{H} (S_2, m_2)$ and
$(S_1, m_1) \equiv_{O} (S_2, m_2)$ and
$\forall x \in (\text{Def}(S_1) \cap \text{Def}(S_2)) \cup
               \{{id}_I,   {id}_{IO}\}\,:\,
               (S_1, m_1) \equiv_{x} (S_2, m_2)$;

\item The loop counter of $S_1$ and $S_2$ are of value less than or equal to $i$,
and there are no reachable configurations
$(S_1, m_1(\text{loop}_c^{1_i}, \vals_{1_i}))$ from $(S_1, m_1(\vals_1))$,
$(S_2, m_2(\text{loop}_c^{2_i}, \vals_{2_i}))$ from $(S_2, \,\allowbreak m_2(\vals_2))$ where all of the following hold:
\begin{itemize}
\item The loop counters of $S_1$ and $S_2$ are of value $i$,
$\text{loop}_c^{1_i}(n_1)\allowbreak = \text{loop}_c^{2_i}(n_2) = i$.

\item Value stores $\vals_{1_i}$ and $\vals_{2_i}$ agree on values of  used variables in both $S_1$ and $S_2$ as well as the input sequence variable,  and the I/O sequence variable,
$\forall x \in \, (\text{Use}(S_1) \cap \text{Use}(S_2)) \cup \{{id}_I,   {id}_{IO}\}\,:\,
\vals_{1_i}(x) = \vals_{2_i}(x)$.
\end{itemize}

\item There are reachable configurations
$(S_1, m_1(\text{loop}_c^{1_i}, \vals_{1_i}))$ from $(S_1, m_1(\vals_1))$,
$(S_2, m_2(\text{loop}_c^{2_i}, \vals_{2_i}))$ from $(S_2, \,\allowbreak m_2(\vals_2))$ where all of the following hold:
\begin{itemize}
\item The loop counter of $S_1$ and $S_2$ are of value $i$,
$\text{loop}_c^{1_i}(n_1)\allowbreak = \text{loop}_c^{2_i}(n_2) = i$.

\item Value stores $\vals_{1_i}$ and $\vals_{2_i}$ agree on values of  used variables in both $S_1$ and $S_2$ as well as the input sequence variable,  and the I/O sequence variable,
$\forall x \in \, (\text{Use}(S_1) \cap \text{Use}(S_2)) \cup \{{id}_I,   {id}_{IO}\}\,:\,
\vals_{1_i}(x) = \vals_{2_i}(x)$.
\end{itemize}
\end{enumerate}

Then we show that, when $i+1$, one of the following holds:
The induction hypothesis (IH) is that, when $i\geq 1$, one of the following holds:
\begin{enumerate}
\item Loop counters for $S_1$ and $S_2$ are always less than $i+1$ if any is present,
$\forall m_1'(\text{loop}_c^{1'})\, m_2'(\text{loop}_c^{2'})\,:\,
(S_1, m_1(\text{loop}_c^1, \vals_1)) ->* (S_1'', m_1'(\text{loop}_c^{1'})),
\text{loop}_c^{1'}(n_1) < i+1,
(S_2, m_2(\text{loop}_c^2,\allowbreak \vals_2)) ->* (S_2'', m_2'(\text{loop}_c^{2'})),
\text{loop}_c^{2'}(n_2) < i+1$,
$S_1$ and $S_2$ terminate in the same way, produce the same output sequence, and have equivalent computation of  used/defined variables in both $S_1$ and $S_2$ and the input sequence variable,  the I/O sequence variable,
$(S_1, m_1)\allowbreak \equiv_{H} (S_2, m_2)$ and
$(S_1, m_1) \equiv_{O} (S_2, m_2)$ and
$\forall x \in (\text{Def}(S_1) \cap \text{Def}(S_2)) \cup
               \{{id}_I,   {id}_{IO}\}\,:\,
               (S_1, m_1) \equiv_{x} (S_2, m_2)$;

\item The loop counter of $S_1$ and $S_2$ are of value less than or equal to $i+1$,
and there are no reachable configurations
$(S_1, m_1(\text{loop}_c^{1_{i+1}}, \vals_{1_{i+1}}))$ from $(S_1, m_1(\vals_1))$,

\noindent$(S_2, m_2(\text{loop}_c^{2_{i+1}}, \vals_{2_{i+1}}))$ from $(S_2, \,\allowbreak m_2(\vals_2))$ where all of the following hold:
\begin{itemize}
\item The loop counters of $S_1$ and $S_2$ are of value $i+1$,
$\text{loop}_c^{1_{i+1}}(n_1)\allowbreak = \text{loop}_c^{2_{i+1}}(n_2) = i+1$.

\item Value stores $\vals_{1_{i+1}}$ and $\vals_{2_{i+1}}$ agree on values of used variables in both $S_1$ and $S_2$ as well as the input sequence variable,  and the I/O sequence variable,
$\forall x \in \, (\text{Use}(S_1) \cap \text{Use}(S_2)) \cup \{{id}_I,   {id}_{IO}\}\,:\,
\vals_{1_{i+1}}(x) = \vals_{2_{i+1}}(x)$.
\end{itemize}

\item There are reachable configurations
$(S_1, m_1(\text{loop}_c^{1_{i+1}}, \vals_{1_i}))$ from $(S_1, m_1(\vals_1))$,
$(S_2, m_2(\text{loop}_c^{2_{i+1}}, \vals_{2_i}))$ from $(S_2, \,\allowbreak m_2(\vals_2))$ where all of the following hold:
\begin{itemize}
\item The loop counter of $S_1$ and $S_2$ are of value $i$,
$\text{loop}_c^{1_{i+1}}(n_1)\allowbreak = \text{loop}_c^{2_{i+1}}(n_2) = i+1$.

\item Value stores $\vals_{1_{i+1}}$ and $\vals_{2_{i+1}}$ agree on values of  used variables in both $S_1$ and $S_2$ as well as the input sequence variable,  and the I/O sequence variable,
$\forall x \in \, (\text{Use}(S_1) \cap \text{Use}(S_2)) \cup \{{id}_{I},   {id}_{IO}\}\,:\,
\vals_{1_{i+1}}(x) = \vals_{2_{i+1}}(x)$.
\end{itemize}
\end{enumerate}

By hypothesis IH, there is no configuration where loop counters of $S_1$ and $S_2$ are of value $i+1$ when any of the following holds:
\begin{enumerate}
\item Loop counters for $S_1$ and $S_2$ are always less than $i$ if any is present,
$\forall m_1'(\text{loop}_c^{1'})\, m_2'(\text{loop}_c^{2'})\,:\,
(S_1, m_1(\text{loop}_c^1, \vals_1)) ->* (S_1'', m_1'(\text{loop}_c^{1'})),
\text{loop}_c^{1'}(n_1) < i,
(S_2, m_2(\text{loop}_c^2,\allowbreak \vals_2)) ->* (S_2'', m_2'(\text{loop}_c^{2'})),
\text{loop}_c^{2'}(n_2) < i$,
$S_1$ and $S_2$ terminate in the same way, produce the same output sequence, and have equivalent computation of  used/defined variables in both $S_1$ and $S_2$ and the input sequence variable,  the I/O sequence variable,
$(S_1, m_1)\allowbreak \equiv_{H} (S_2, m_2)$ and
$(S_1, m_1) \equiv_{O} (S_2, m_2)$ and
$\forall x \in (\text{Def}(S_1) \cap \text{Def}(S_2)) \cup
               \{{id}_I,   {id}_{IO}\}\,:\,
               (S_1, m_1) \equiv_{x} (S_2, m_2)$;

\item The loop counter of $S_1$ and $S_2$ are of value less than or equal to $i$,
and there are no reachable configurations
$(S_1, m_1(\text{loop}_c^{1_i}, \vals_{1_i}))$ from $(S_1, m_1(\vals_1))$,
$(S_2, m_2(\text{loop}_c^{2_i}, \vals_{2_i}))$ from $(S_2, \,\allowbreak m_2(\vals_2))$ where all of the following hold:
\begin{itemize}
\item The loop counters of $S_1$ and $S_2$ are of value $i$,
$\text{loop}_c^{1_i}(n_1)\allowbreak = \text{loop}_c^{2_i}(n_2) = i$.

\item Value stores $\vals_{1_i}$ and $\vals_{2_i}$ agree on values of  used variables in both $S_1$ and $S_2$ as well as the input sequence variable,  and the I/O sequence variable,
$\forall x \in \, (\text{Use}(S_1) \cap \text{Use}(S_2)) \cup \{{id}_I,   {id}_{IO}\}\,:\,
\vals_{1_i}(x) = \vals_{2_i}(x)$.
\end{itemize}
\end{enumerate}

When there are  reachable configurations
$(S_1, m_1(\text{loop}_c^{1_i}, \vals_{1_i}))$ from $(S_1, m_1(\vals_1))$,
$(S_2, m_2(\text{loop}_c^{2_i}, \vals_{2_i}))$ from $(S_2, \,\allowbreak m_2(\vals_2))$ where all of the following hold:
\begin{itemize}
\item The loop counter of $S_1$ and $S_2$ are of value $i$,
$\text{loop}_c^{1_i}(n_1)\allowbreak = \text{loop}_c^{2_i}(n_2) = i$.

\item Value stores $\vals_{1_i}$ and $\vals_{2_i}$ agree on values of  used variables in both $S_1$ and $S_2$ as well as the input sequence variable,  and the I/O sequence variable,
$\forall x \in \, (\text{Use}(S_1) \cap \text{Use}(S_2)) \cup \{{id}_I,   {id}_{IO}\}\,:\,
\vals_{1_i}(x) = \vals_{2_i}(x)$.
\end{itemize}

By similar argument in base case, we have one of the following holds:
\begin{enumerate}
\item Loop counters for $S_1$ and $S_2$ are always less than $i+1$ if any is present,
$\forall m_1'(\text{loop}_c^{1'})\, m_2'(\text{loop}_c^{2'})\,:\,
(S_1, m_1(\text{loop}_c^1, \vals_1)) ->* (S_1'', m_1'(\text{loop}_c^{1'})),
\text{loop}_c^{1'}(n_1) < i+1,
(S_2, m_2(\text{loop}_c^2,\allowbreak \vals_2)) ->* (S_2'', m_2'(\text{loop}_c^{2'})),
\text{loop}_c^{2'}(n_2) < i$,
$S_1$ and $S_2$ terminate in the same way, produce the same output sequence, and have equivalent computation of  used/defined variables in both $S_1$ and $S_2$ and the input sequence variable,  the I/O sequence variable,
$(S_1, m_1)\allowbreak \equiv_{H} (S_2, m_2)$ and
$(S_1, m_1) \equiv_{O} (S_2, m_2)$ and
$\forall x \in (\text{Def}(S_1) \cap \text{Def}(S_2)) \cup
               \{{id}_I,   {id}_{IO}\}\,:\,
               (S_1, m_1) \equiv_{x} (S_2, m_2)$;

\item The loop counter of $S_1$ and $S_2$ are of value less than or equal to $i+1$,
and there are no reachable configurations
$(S_1, m_1(\text{loop}_c^{1_i}, \vals_{1_i}))$ from $(S_1, m_1(\vals_1))$,
$(S_2, m_2(\text{loop}_c^{2_i}, \vals_{2_i}))$ from $(S_2, \,\allowbreak m_2(\vals_2))$ where all of the following hold:
\begin{itemize}
\item The loop counters of $S_1$ and $S_2$ are of value $i$,
$\text{loop}_c^{1_{i+1}}(n_1)\allowbreak = \text{loop}_c^{2_{i+1}}(n_2) = i+1$.

\item Value stores $\vals_{1_{i+1}}$ and $\vals_{2_{i+1}}$ agree on values of  used variables in both $S_1$ and $S_2$ as well as the input sequence variable,  and the I/O sequence variable,
$\forall x \in \, (\text{Use}(S_1) \cap \text{Use}(S_2)) \cup \{{id}_I,   {id}_{IO}\}\,:\,
\vals_{1_{i+1}}(x) = \vals_{2_{i+1}}(x)$.
\end{itemize}

\item There are reachable configurations
$(S_1, m_1(\text{loop}_c^{1_{i+1}}, \vals_{1_{i+1}}))$ from $(S_1, m_1(\vals_1))$,
$(S_2, m_2(\text{loop}_c^{2_{i+1}}, \vals_{2_{i+1}}))$ from $(S_2, \,\allowbreak m_2(\vals_2))$ where all of the following hold:
\begin{itemize}
\item The loop counter of $S_1$ and $S_2$ are of value $i$,
$\text{loop}_c^{1_{i+1}}(n_1)\allowbreak = \text{loop}_c^{2_{i+1}}(n_2) = i$.

\item The loop counter of $S_1$ and $S_2$ are of value $i$,
$\text{loop}_c^{1_{i+1}}(n_1)\allowbreak = \text{loop}_c^{2_{i+1}}(n_2) = i$.

\item Value stores $\vals_{1_{i+1}}$ and $\vals_{2_{i+1}}$ agree on values of  used variables in both $S_1$ and $S_2$ as well as the input sequence variable and the I/O sequence variable,
$\forall x \in \, (\text{Use}(S_1) \cap \text{Use}(S_2)) \cup \{{id}_{I+1},   {id}_{IO}\}\,:\,
\vals_{1_{i+1}}(x) = \vals_{2_{i+1}}(x)$.
\end{itemize}
\end{enumerate}
\end{proof}

\subsection{Proof rule for variable type weakening}

In programs, variable types are changed either to allow for larger ranges (weakening).
For example,
an integer variable might be changed to become a long variable
to avoid integer overflow.
Adding a new enumeration value can is also type weakening.
Increasing array size is another example of weakening.
Allowing for type weakening is essentially an assumption about the intent behind the update.
The kinds of weakening that should be allowed are application dependent and would need to be defined by the user in general.
The type weakening considered are either changes of type Int to Long or increase of array size.
These updates fix integer overflow or array index out of bound.
In order to prove the update of variable type weakening to be backward compatible, we assume that there are no integer overflow and array index out of bound in execution of the old program and the updated program.
In conclusion, the old program and the new program produce the same output sequence because the integer overflow and index out of bound errors fixed by the new program do not occur.

We formalize the update of variable type weakening, then we show that the updated program produce the same output sequence as the old program in executions if there are no integer overflow or index out of bound exceptions related to variables with type changes.
First, we define a relation between variable definitions showing the type weakening.
\begin{definition}\label{def:typeRelaxTypeRelation} {\bf (Cases of type weakening)}
 We say there is type weakening from a sequence of variable definitions $V_1$ to $V_2$,
 written $V_1 \nearrow_{\tau} V_2$, iff one of the following holds:
\begin{enumerate}
\item $V_1 = ``\text{Int} \, id", V_2 = ``\text{Long} \, id"$;

\item $V_1 = ``\tau \, id[n_2]", V_2 = ``\tau \, id[n_1]"$ where $n_2 > n_1$;

\item $V_1 = V_1',``\tau_1 \, {id}_1", V_2 = V_2',``\tau_2 \, {id}_2"$ where
$(V_1' \nearrow_{\tau} V_2') \wedge (``\tau_1 \, {id}_1" \nearrow_{\tau} ``\tau_2 \, {id}_2")$;

\item $V_1 = V_1',``\tau_1 \, {id}_1[n_1]", V_2 = V_2',``\tau_2 \, {id}_2[n_2]"$ where
$(V_1' \nearrow_{\tau} V_2') \wedge (``\tau_1 \, {id}_1[n_1]" \nearrow_{\tau} ``\tau_2 \, {id}_2[n_2]")$;
\end{enumerate}
\end{definition}

The following is the generalized definition of variable type weakening.
\begin{definition}\label{def:typeRelaxFuncWide}
{\bf (Variable type weakening)}
We say that there are updates of variable type weakening in the program $P_2= Pmpt;EN;V_2;S_{entry}$ compared with the program $P_1 = Pmpt;$ $EN;V_1;S_{entry}$, written $P_2 {\approx}_{\tau}^S P_1$, iff $V_{1} \nearrow_{\tau} V_{2}$.
\end{definition}

We show that two programs terminate in the same way, produce the same output sequence, and have equivalent computation of defined variables in both programs in valid executions if there are updates of variable type weakening between them.
\begin{lemma}\label{lmm:typeRelaxFuncWide}
Let $P_1 = EN;V_1;S_{entry}$ and $P_2  = EN;V_2;S_{entry}$ be two programs where there are updates of variable type weakening,
$P_2 {\approx}_{\tau}^S P_1$.
If the programs $P_1$ and $P_2$ start in states $m_1(\vals_1)$ and $m_2(\vals_2)$ such that both of the following hold:
\begin{itemize}
\item Value stores $\vals_1$ and $\vals_2$ agree on values of variables  used in $S_{entry}$ as well as the input sequence variable, the I/O sequence variable,
$\forall x \in \text{Use}(S_{entry}) \cup \{{id}_I, {id}_{IO}\}\,:\,
 \vals_1(x) = \vals_2(x)$;

 \item There is no integer overflow or index out of bound exceptions related to variables of type change;
\end{itemize}
then $S_{entry}$ in the program $P_1$ and $P_2$ terminate in the same way, produce the same output sequence, and when $S_{entry}$ both terminate, they have equivalent computation of  defined variables in $S_{entry}$ in both programs as well as the input sequence variable, the I/O sequence variable,
\begin{itemize}
\item $(S_{entry}, m_1) \equiv_{H} (S_{entry}, m_2)$;

\item $(S_{entry}, m_1) \equiv_{O} (S_{entry}, m_2)$;

\item $\forall x \in \text{Def}(S) \cup \{{id}_I, {id}_{IO}\}\,:\,$

\noindent$\allowbreak (S_{entry}, m_1) \equiv_{x} (S_{entry}, m_2)$;
\end{itemize}
\end{lemma}
Because $S_{entry}$ are the exactly same in both programs $P_1$ and $P_2$, we omit the straightforward proof.
Instead, we show that, if there is no array index out of bound and integer overflow in executions of the old program, then there is no array index out of bound or integer overflow in executions of updated program due to the increase of array index and change of type Int to Long.
\begin{proof}
The proof is straightforward because the statement sequence $S$ is same in programs $P_1$ and $P_2$.
The only point is that if there is array index out of bound or integer overflow in execution of $S$ in $P_1$, then there is no array index out of bound or integer overflow in execution of $S$ in $P_2$.
To show the point, we present the argument for the array index out of bound and integer overflow separately.
\begin{enumerate}
\item We show that, as to one expression ${id}_1[{id}_2]$, there is no array index out of bound in $P_2$ if there is no array index out of bound in $P_1$ when $P_1$ and $P_2$ are in states agreeing on values of used variables in $P_1$ and $P_2$;

\begin{tabbing}
xx\=xx\=\kill
\>\>                   $({id}_1[{id}_2], m_1(\vals_1))$\\
\>$->$\>               $({id}_1[v], m_1(\vals_1))$ by the rule Var
\end{tabbing}

Similarly, $({id}_1[{id}_2], m_2(\vals_2)) -> ({id}_1[v], m_2(\vals_2))$.
By Definition~\ref{def:typeRelaxTypeRelation}, the array bound of ${id}_1$ in $P_2$ is no less than that in $P_1$, then there is no array out of bound exception in evaluation of ${id}_1[{id}_2]$ in $P_2$ if there is no array out of bound exception in evaluation of ${id}_1[{id}_2]$ in $P_1$.

\item We show that, as to one expression $e$, there is no integer overflow in evaluation of $e$ in $P_2$ if there is no integer overflow in evaluation of $e$ in $P_1$;

\begin{tabbing}
xx\=xx\=\kill
\>\>                   $(e, m_1(\vals_1))$\\
\>$->$\>               $((v_e, v_{\mathfrak{of}}), m_1(\vals_1))$ by the rule EEval'
\end{tabbing}

When every used variable in the expression $e$ is of same type in $P_1$ and $P_2$, then the evaluation of the expression $e$ in $P_2$ is of the same result $(v_e, v_{\mathfrak{of}})$ in $P_1$.
When every used variable in the expression $e$ is of type Int in $P_1$ and of type Long $P_2$, then there is no integer overflow in the evaluation of the expression $e$ in $P_2$ if there is no integer overflow in the evaluation of the expression $e$ in $P_1$. This is because the values of type Long are a superset of those of type $Int$.
\end{enumerate}
\end{proof} 

\subsection{Proof rule for exit on errors}

Another bugfix is called ``exit-on-error", which causes the program to exit in observation of application-semantic-dependent errors.
\begin{figure}
\begin{small}
\begin{tabbing}
xxxxxx\=xxx\=xxx\=xxx\=xxxxxxxxxxxxxxxx\=xxxx\=xxx\=xxxxx\= \kill
\>1: \>\>\>                                    \>1': \> {\bf If} $(1/(a-5))$ {\bf then}\> \\
\>2: \>\>\>                                    \>2': \> \> {\bf skip} \> \\
\>3: \> {\bf output} $a$\>\>                   \>3': \> {\bf output} $a$  \> \\
\>\\
\> \> old\>\>                                         \>  \> new \>
\end{tabbing}
\end{small}
\caption{Exit-on-error}\label{fig:exitOnErrExample}
\end{figure}
Figure~\ref{fig:exitOnErrExample} shows an example of exit-on-error update.
In the example, the fixed bugs refer to the program semantic error that $a = 5$.
Instead of using an ``exit" statement, we rely on the crash from expression evaluations to formalize the update class.
In order to prove the update of exit-on-error to be backward compatible, we assume that there are no application related error in executions of the old program. Therefore, the two programs produce the same output sequence because the extra check does not cause the new program's execution to crash.

The following is the generalized definition of the update class ``exit-on-error".
\begin{definition}\label{def:exitOnErrFuncWide}
{\bf (Exit on error)}
We say a statement sequence $S_2$ includes updates of exit-on-err from a statement sequence $S_1$, written $S_2\, {\approx}_{\text{Exit}}^S\, S_1$, iff
one of the following holds:
\begin{small}
\begin{enumerate}
\item $S_2 = ``\text{If}(e) \, \text{then}\{\text{skip}\}\, \text{else}\{\text{skip}\}";S_1$;

\item
$S_1 = ``\text{If}(e) \, \text{then}\{S_1^t\} \, \text{else} \{S_1^f\}"$,
$S_2 = ``\text{If}(e) \, \text{then}\{S_2^t\} \, \text{else} \{S_2^f\}"$ where both of the following hold
\begin{itemize}
\item $S_2^t\, {\approx}_{\text{Exit}}^S\, S_1^t$;

\item $S_2^f\, {\approx}_{\text{Exit}}^S\, S_1^f$;
\end{itemize}

\item
$S_1 = ``\text{while}_{\langle n_1\rangle}(e) \, \{S_1'\}"$,
$S_2 = ``\text{while}_{\langle n_2\rangle}(e) \, \{S_2'\}"$ where

\noindent$S_2' {\approx}_{\text{Exit}}^S S_1'$;

\item $S_1 \approx_{O}^S S_2$;

\item $S_1 = S_1';s_1$ and $S_2 = S_2';s_2$ such that both of the following hold:
\begin{itemize}
\item $S_2' \approx_{\text{Exit}}^S S_1'$;
\item $S_2' \approx_{H}^S S_1'$;
\item $\forall x \in \text{Imp}(s_1, {id}_{IO}) \cup \text{Imp}(s_1, {id}_{IO})\,:\, S_2' \approx_{x}^S S_1'$;
\item $s_2 \approx_{\text{Exit}}^S s_1$;
\end{itemize}
\end{enumerate}
\end{small}
\end{definition}
Though the bugfix in Definition~\ref{def:exitOnErrFuncWide} is not in rare execution in the first case, the definition shows the basic form of bugfix clearly.

We show that two programs terminate in the same way, produce the same output sequence, and have equivalent computation of defined variables in both programs in valid executions if there are updates of exit-on-error between them.
\begin{lemma}\label{lmm:exitOnErrFuncWide}
Let $S_1$ and $S_2$ be two statement sequences respectively where there are updates of exit-on-error in $S_2$ against $S_1$,
$S_2\, {\approx}_{\text{Exit}}^S\, S_1$.
If $S_1$ and $S_2$ start in states $m_1(\vals_1)$ and $m_2(\vals_2)$ such that both of the following hold:
\begin{itemize}
\item Value stores $\vals_1$ and $\vals_2$ agree on values of variables  used in both $S_1$ and $S_2$ as well as the input sequence variable and the I/O sequence variable,
$\forall x \in (\text{Use}(S_1)\cap \text{Use}(S_2)) \cup \{{id}_I,  {id}_{IO}\}\,:\,
 \vals_1(x) = \vals_2(x)$;

\item There are no program semantic errors related to the extra check in the update of exit-on-error in the execution of $S_1$;
\end{itemize}
then $S_1$ and $S_2$ terminate in the same way, produce the same output sequence, and when $S_1$ and $S_2$ both terminate, they have equivalent computation of  defined variables in both $S_1$ and $S_2$ as well as the input sequence variable and the I/O sequence variable,
\begin{itemize}
\item $(S_1, m_1) \equiv_{H} (S_2, m_2)$;

\item $(S_1, m_1) \equiv_{O} (S_2, m_2)$;

\item $\forall x \in (\text{Def}(S_1)\cap \text{Def}(S_2)) \cup \{{id}_I,  {id}_{IO}\}\,:\,$

\noindent$\allowbreak (S_1, m_1) \equiv_{x} (S_2, m_2)$;
\end{itemize}
\end{lemma}

\begin{proof}
By induction on the sum of the program size of $S_1$ and $S_2$, $\text{size}(S_1) + \text{size}(S_2)$.

\noindent{Base case}.
$S_1 = s$ and $S_2 = ``\text{If}(e) \, \text{then} \{\text{skip}\} \, \text{else}\{\text{skip}\}";s$;

By assumption, there is no program semantic error related to the update of exit-on-error. Then the evaluation of the predicate expression $e$ in the first statement of $S_2$ does not crash. W.l.o.g., the expression $e$ evaluates to zero.
Then the execution of $S_2$ proceeds as follows.
\begin{tabbing}
xx\=xx\=\kill
\>\>                   $(\text{If}(e) \, \text{then} \{\text{skip}\} \, \text{else}\{\text{skip}\};s, m_2(\vals_2))$\\
\>$->$\>               $(\text{If}((0, v_{\mathfrak{of}})) \, \text{then} \{\text{skip}\} \, \text{else}\{\text{skip}\};s, m_2(\vals_2))$\\
\>\>                   by the rule EEval'\\
\>$->$\>               $(\text{If}(0) \, \text{then} \{\text{skip}\} \, \text{else}\{\text{skip}\};s, m_2(\vals_2))$\\
\>\>                   by the rule E-Oflow1 or E-Oflow2.\\
\>$->$\>               $(\text{skip};s, m_2(\vals_2))$ by the rule If-F\\
\>$->$\>               $(s, m_2(\vals_2))$ by the rule Seq.
\end{tabbing}

Value stores $\vals_2$ are not changed in the execution of $(S_2, m_2(\vals_2)) ->* (s, m_2(\vals_2))$. By assumption, $\vals_1$ and $\vals_2$ agree on values of  used variables as well as the input sequence variable, and the I/O sequence variable,
$\forall x \in \text{Use}(S_2) \cap \text{Use}(S_1) \cup \{{id}_I,  {id}_{IO}\}\,:\,
\vals_1(id) = \vals_2(id)$.
By semantics, $S_1$ and $S_2$ terminate in the same way, produce the same output sequence, and have equivalent computation of  defined variables in both $S_1$ and $S_2$ as well as the input sequence variable, and the I/O sequence variable. Then this lemma holds.

\noindent{Induction step}.

The hypothesis is that this lemma holds when the sum $k$ of the program size of $S_1$ and $S_2$ are great than or equal to 4, $k\geq 4$.

We then show that this lemma holds when the sum of the program size of $S_1$ and $S_2$ is $k+1$.
There are cases to consider.
\begin{enumerate}
\item $S_1$ and $S_2$ are both ``If" statement:

$S_1 = ``\text{If}(e) \, \text{then}\{S_1^t\} \, \text{else} \{S_1^f\}"$,
$S_2 = ``\text{If}(e) \, \text{then}\{S_2^t\} \, \text{else} \{S_2^f\}"$ where both of the following hold
\begin{itemize}
\item $(S_2^t {\approx}_{\text{Exit}}^S S_1^t)$;

\item $(S_2^f {\approx}_{\text{Exit}}^S S_1^f)$;
\end{itemize}

By the definition of $\text{Use}(S_1)$, variables used in the predicate expression $e$ are a subset of  used variables in $S_1$ and $S_2$,
$\text{Use}(e) \subseteq \text{Use}(S_1) \cap \text{Use}(S_2)$.
By assumption, corresponding variables used in $e$ are of same value in value stores $\vals_1$ and $\vals_2$. By Lemma~\ref{lmm:expEvalSameVal}, the expression evaluates to the same value w.r.t value stores $\vals_1$ and $\vals_2$. There are three possibilities.
\begin{enumerate}
\item The evaluation of $e$ crashes,
$\mathcal{E}'\llbracket e\rrbracket\vals_1 =
 \mathcal{E}'\llbracket e\rrbracket\vals_2 = (\text{error}, v_{\mathfrak{of}})$.

The execution of $S_1$ continues as follows:
\begin{tabbing}
xx\=xx\=\kill
\>\>                   $(\text{If}(e) \, \text{then}\{S_1^t\} \, \text{else} \{S_1^f\}, m_1(\vals_1))$\\
\>$->$\>               $(\text{If}((\text{error}, v_{\mathfrak{of}})) \, \text{then}\{S_1^t\} \, \text{else} \{S_1^f\}, m_1(\vals_1))$\\
\>\>                   by the rule EEval'\\
\>$->$\>               $(\text{If}(0) \, \text{then}\{S_1^t\} \, \text{else} \{S_1^f\}, m_1(1/\mathfrak{f}))$\\
\>\>                   by the ECrash rule  \\
\>{\kStepArrow [i] }\> $(\text{If}(0) \, \text{then}\{S_1^t\} \, \text{else} \{S_1^f\}, m_1(1/\mathfrak{f}))$ for any $i>0$\\
\>\>                   by the Crash rule.
\end{tabbing}

Similarly, the execution of $S_2$ started from the state $m_2(\vals_2)$ crashes.
The lemma holds.

\item The evaluation of $e$ reduces to zero,
$\mathcal{E}'\llbracket e\rrbracket\vals_1 =
 \mathcal{E}'\llbracket e\rrbracket\vals_2 = (0, v_{\mathfrak{of}})$.

The execution of $S_1$ continues as follows.
\begin{tabbing}
xx\=xx\=\kill
\>\>                   $(\text{If}(e) \, \text{then}\{S_1^t\} \, \text{else} \{S_1^f\}, m_1(\vals_1))$\\
\>= \>                 $(\text{If}((0, v_{\mathfrak{of}})) \, \text{then}\{S_1^t\} \, \text{else} \{S_1^f\}, m_1(\vals_1))$\\
\>\>                   by the rule EEval'\\
\>$->$\>               $(\text{If}(0) \, \text{then}\{S_1^t\} \, \text{else} \{S_1^f\}, m_1(\vals_1))$\\
\>\>                   by the E-Oflow1 or E-Oflow2 rule  \\
\>$->$\>               $(S_1^f, m_1(\vals_1))$ by the If-F rule.
\end{tabbing}

Similarly, the execution of $S_2$ gets to the configuration $(S_2^f, m_2(\vals_2))$.

By the hypothesis IH, we show the lemma holds.
We need to show that all conditions are satisfied for the application of the hypothesis IH.
\begin{itemize}
\item $(S_2^f {\approx}_{\text{Exit}}^S S_1^f)$

By assumption.

\item The sum of the program size of $S_1^f$ and $S_2^f$ is less than $k$,
$\text{size}(S_1^f) + \text{size}(S_2^f) < k$.

By definition, $\text{size}(S_1) = 1 + \text{size}(S_1^t) + \text{size}(S_1^f)$.
Then, $\text{size}(S_1^f) + \text{size}(S_2^f) < k + 1 - 2 = k - 1$.

\item Value stores $\vals_1$ and $\vals_2$ agree on values of  used variables in $S_1^f$ and $S_2^f$ as well as the input, I/O sequence variable.

By definition, $\text{Use}(S_1^f)\allowbreak \subseteq \text{Use}(S_1)$.
So are the cases to $S_2^f$ and $S_2$.
In addition, value stores $\vals_1$ and $\vals_2$ are not changed in the evaluation of the predicate expression $e$. The condition holds.

\item There are no program semantic error related to the extra check in the update of exit-on-error in the execution of $S_2$.

By assumption.
\end{itemize}

By the hypothesis IH, the lemma holds.

\item The evaluation of $e$ reduces to the same nonzero integer value,
$\mathcal{E}'\llbracket e\rrbracket \vals_1 =
 \mathcal{E}'\llbracket e\rrbracket \vals_2 = (v, v_{\mathfrak{of}})$ where $v \neq 0$.

By similar to the second subcase above.
\end{enumerate}

\item $S_1$ and $S_2$ are both ``while" statements:

$S_1 = ``\text{while}_{\langle n\rangle}(e) \, \{S_1'\}"$,
$S_2 = ``\text{while}_{\langle n\rangle}(e) \, \{S_2'\}"$ where
$(S_2' {\approx}_{\text{exit}}^S S_1')$;

By Lemma~\ref{lmm:exitOnErrLoopStmt}, we show this lemma holds.
We need to show that all required conditions are satisfied for the application of Lemma~\ref{lmm:exitOnErrLoopStmt}.
\begin{itemize}

\item The output deciding variables in $S_1'$ are a subset of those in $S_2'$,
$\text{OVar}(S_1') = \text{OVar}(S_2')$;

By Lemma~\ref{lmm:exitOnErrFuncWideSimilarUseDef}.

\item When started in states $m_1'(\vals_1'), m_2'(\vals_1')$ where
value stores $\vals_1'$ and $\vals_2'$ agree on values of  used variables in both $S_1'$ and $S_2'$ as well as the input sequence variable,  and the I/O sequence variable,
then $S_1'$ and $S_2'$ terminate in the same way, produce the same output sequence, and have equivalent computation of defined variables in both $S_1$ and $S_2$ as well as the input sequence variable,  and the I/O sequence variable.

By the induction hypothesis IH. This is because the sum of the program size of $S_1'$ and $S_2'$ is less than $k$. By definition, $\text{size}(S_1) = 1 + \text{size}(S_1')$.
\end{itemize}
By Lemma~\ref{lmm:exitOnErrLoopStmt}, this lemma holds.

\item $S_1 = S_1';s_1$ and $S_2 = S_2';s_2$ where both of the following hold:

\begin{itemize}
\item $(S_2' \approx_{\text{Exit}}^S S_1')$;

\item $(s_2 \approx_{\text{Exit}}^S s_1)$;
\end{itemize}

By the hypothesis IH, we show $S_2'$ and $S_1'$ terminate in the same way and produce the same output sequence and when $S_2'$ and $S_1'$ both terminate, $S_2'$ and $S_1'$ have equivalent terminating computation of variables  used or defined in $S_2'$ and $S_1'$ as well as the input sequence variable,  and the I/O sequence variable.

We show all the required conditions are satisfied for the application of the hypothesis IH.
\begin{itemize}
\item $(S_2' {\approx}_{\text{Exit}}^S S_1')$.

By assumption.

\item The sum of  the program size of $S_1'$ and $S_2'$ is less than $k$,
$\text{size}(S_1') + \text{size}(S_2') < k$.

By definition, $\text{size}(S_2) = \text{size}(s_2) + \text{size}(S_2')$ where $\text{size}(s_2) < 1$.
Then, $\text{size}(S_2') + \text{size}(S_1') < k + 1 - \text{size}(s_2) - \text{size}(s_1) < k$.

\item Value stores $\vals_1$ and $\vals_2$ agree on values of  used variables in both $S_2'$ and $S_1'$ as well as the input, I/O sequence variable.

By definition, $\text{Use}(S_2')\allowbreak \subseteq \text{Use}(S_2)$,
$\text{Use}(S_1')\allowbreak \subseteq \text{Use}(S_1)$.
The condition holds.
\end{itemize}
By the hypothesis IH, one of the following holds:
\begin{enumerate}
\item $S_1'$ and $S_2'$ both do not terminate.

By Lemma~\ref{lmm:multiStepSeqExec}, executions of $S_1 = S_1';s_1$ and $S_2 = S_2';s_2$ both do not terminate and produce the same output sequence.

\item $S_1'$ and $S_2'$ both terminate.

By assumption,
$(S_2', m_2(\vals_2)) ->* (\text{skip}, m_2'(\vals_2'))$,
$(S_1', m_1(\vals_1)) ->* (\text{skip}, m_1'(\vals_1'))$.

By Corollary~\ref{coro:termSeq},
$(S_2';s_2, m_2(\vals_2)) ->* (s_2, m_2'(\vals_2'))$,
$(S_1';s_1, m_1(\vals_1)) ->* (s_1, m_1'(\vals_1'))$.

By the hypothesis IH, we show that $s_2$ and $s_1$ terminate in the same way, produce the same output sequence and when $s_2$ and $s_1$ both terminate, $s_2$ and $s_1$ have equivalent computation of variables  defined in both $s_1$ and $s_2$ and the input, and I/O sequence variables.

We need to show that all conditions are satisfied for the application of the hypothesis IH.
\begin{itemize}
\item There are updates of ``exit-on-error" between $s_2$ and $s_1$,
$s_2 {\approx}_{\text{Exit}}^S s_1$;

By assumption, $s_2 {\approx}_{\text{Exit}}^S s_1$.

\item The sum of the program size $s_2$ and $s_1$ is less than or equals to $k$;

By definition, $\text{size}(S_2') \geq 1,  \text{size}(S_1') \geq 1$.
Therefore, $\text{size}(s_2) + \text{size}(s_1) < k+1 - \text{size}(S_2') - \text{size}(S_1') \leq k$.

\item Value stores $\vals_1'$ and $\vals_2'$ agree on values of  output deciding variables in $s_2$ and $s_1$ as well as the input, I/O sequence variable.

By Lemma~\ref{lmm:exitOnErrFuncWideSimilarUseDef},
$\text{OVar}(s_1) \subseteq \text{OVar}(s_2)$, then $\text{OVar}(s_2)\,\allowbreak   \cap \text{OVar}(s_1) = \text{OVar}(s_1)$.
For any variable $id$ in $\text{Use}(s_1)$, if $id$ is not in $\text{OVar}(S_1')$,
then the value of $id$ is not changed in the execution of $S_1'$ and $S_2'$,
$\vals_1'(id) = \vals_1(id) = \vals_2(id)\allowbreak = \vals_2'(id)$.
Otherwise, the variable $id$ is defined in the execution of $S_1'$ and $S_2'$, by assumption, $\vals_1'(id) = \vals_2'(id)$.
The condition holds.

\item There are no program semantic errors related to the extra check in the update of exit-on-error in the execution of $S_2$.

By assumption.
\end{itemize}
By the hypothesis IH, the lemma holds.
\end{enumerate}
\end{enumerate}
\end{proof}

We list the auxiliary lemmas below. One lemma shows that, if there are updates of exit-on-error between two statement sequences, then there are same set of  defined variables in the two statement sequences, and the  used variables in the update program are the superset of those in the old program.
\begin{lemma}\label{lmm:exitOnErrFuncWideSimilarUseDef}
Let $S_2$ be a statement sequence and $S_1$ where there are updates of exit-on-error, $S_2 {\approx}_{\text{Exit}}^S S_1$.
Then output deciding variables in $S_1$ are a subset of those in $S_2$,
$\text{OVar}(S_1) \subseteq \text{OVar}(S_2)$.
\end{lemma}
\begin{proof}
By induction on the sum of the program size of $S_1$ and $S_2$.
\end{proof}

\begin{lemma}\label{lmm:exitOnErrLoopStmt}
Let $S_1 = \text{while}_{\langle n_1\rangle}(e) \, \{S_1'\}$ and
    $S_2 = \text{while}_{\langle n_2\rangle}(e) \,\allowbreak \{S_2'\}$ be two loop statements where all of the following hold:
\begin{itemize}
\item the output deciding variables in $S_1'$ are a subset of those in $S_2'$,

\noindent$\text{OVar}(S_1') \subseteq \text{OVar}(S_2') = \text{OVar}(S)$;

\item When started in states $m_1'(\vals_1'), m_2'(\vals_2')$ where
    \begin{itemize}
    \item Value stores agree on  values of  output deciding variables in both $S_1'$ and $S_2'$ as well as the input sequence variable,  and the I/O sequence variable,
     $\forall x \in \text{OVar}(S_2') \cup \{{id}_I,  {id}_{IO}\}\,
    \forall m_1'(\vals_1')\, m_2'(\vals_2')\,:\, \, \allowbreak
    \vals_1'(x) = \vals_2'(x)$;

    \item There are no program semantic errors related to the extra check in the update of exit-on-error in executions of $S_1'$ and $S_2'$;
    \end{itemize}
then $S_1'$ and $S_2'$ terminate in the same way, produce the same output sequence, and have equivalent computation of defined variables in $S_1'$ and $S_2'$ as well as the input sequence variable,  and the I/O sequence variable
    $((S_1', m_1) \equiv_{H} (S_2', m_2)) \wedge
     ((S_1', m_1) \equiv_{O} (S_2', m_2)) \wedge
     (\forall x \in \text{OVar}(S) \cup \{{id}_I,\allowbreak {id}_{IO}\}\,:\,$

     \noindent$(S_1', m_1) \equiv_{x} (S_2', m_2))$;
\end{itemize}

If $S_1$ and $S_2$ start in states $m_1(\text{loop}_c^1, \vals_1), m_2(\text{loop}_c^2, \vals_2)$ respectively, with loop counters of $S_1$ and $S_2$ not initialized ($S_1, S_2$ have not executed yet), value stores agree on values of  used variables in $S_1$ and $S_2$, and there are no program semantic errors related to the extra check in the update of exit-on-error, then, for any positive integer $i$, one of the following holds:
\begin{enumerate}
\item Loop counters for $S_1$ and $S_2$ are always less than $i$ if any is present,
$\forall m_1'(\text{loop}_c^{1'})\, m_2'(\text{loop}_c^{2'})\,:\,
(S_1, m_1(\text{loop}_c^1, \vals_1)) ->* (S_1'', m_1'(\text{loop}_c^{1'})),
\text{loop}_c^{1'}(n_1) < i,
(S_2, m_2(\text{loop}_c^2,\allowbreak \vals_2)) ->* (S_2'', m_2'(\text{loop}_c^{2'})),
\text{loop}_c^{2'}(n_2) < i$,
$S_1$ and $S_2$ terminate in the same way, produce the same output sequence, and have equivalent computation of  output deciding variables in $S_1$ and $S_2$ and the input sequence variable,  the I/O sequence variable,
$(S_1, m_1)\allowbreak \equiv_{H} (S_2, m_2)$ and
$(S_1, m_1) \equiv_{O} (S_2, m_2)$ and
$\forall x \in (\text{OVar}(S_1) \cup \text{OVar}(S_2)) \cup
               \{{id}_I, {id}_{IO}\}\,:\,$

               \noindent$(S_1, m_1) \equiv_{x} (S_2, m_2)$;

\item The loop counter of $S_1$ and $S_2$ are of value less than or equal to $i$,
and there are no reachable configurations
$(S_1, m_1(\text{loop}_c^{1_i},\allowbreak \vals_{1_i}))$ from $(S_1, m_1(\vals_1))$,
$(S_2, m_2(\text{loop}_c^{2_i}, \vals_{2_i}))$ from $(S_2, \,\allowbreak m_2(\vals_2))$ where all of the following hold:
\begin{itemize}
\item The loop counters of $S_1$ and $S_2$ are of value $i$,
$\text{loop}_c^{1_i}(n_1)\allowbreak = \text{loop}_c^{2_i}(n_2) = i$.

\item Value stores $\vals_{1_i}$ and $\vals_{2_i}$ agree on values of output deciding variables in $S_1$ and $S_2$ as well as the input sequence variable,  and the I/O sequence variable,
$\forall x \in \, (\text{OVar}(S_1) \cup \text{OVar}(S_2)) \cup \{{id}_I, {id}_{IO}\}\,:\,
\vals_{1_i}(x) = \vals_{2_i}(x)$.
\end{itemize}

\item There are reachable configurations
$(S_1, m_1(\text{loop}_c^{1_i}, \vals_{1_i}))$ from $(S_1, m_1(\vals_1))$,
$(S_2, m_2(\text{loop}_c^{2_i}, \vals_{2_i}))$ from $(S_2, \,\allowbreak m_2(\vals_2))$ where all of the following hold:
\begin{itemize}
\item The loop counter of $S_1$ and $S_2$ are of value $i$,
$\text{loop}_c^{1_i}(n_1)\allowbreak = \text{loop}_c^{2_i}(n_2) = i$.

\item Value stores $\vals_{1_i}$ and $\vals_{2_i}$ agree on values of output deciding variables in $S_1$ and $S_2$ including the input sequence variable and the I/O sequence variable,
$\forall x \in \, (\text{OVar}(S_1) \cup \text{OVar}(S_2)) \cup \{{id}_I, {id}_{IO}\}\,:\,
\vals_{1_i}(x) = \vals_{2_i}(x)$.
\end{itemize}
\end{enumerate}
\end{lemma}

\begin{proof}
By induction on $i$.

\noindent{Base case}.

We show that, when $i = 1$, one of the following holds:
\begin{enumerate}
\item Loop counters for $S_1$ and $S_2$ are always less than 1 if any is present,
$\forall m_1'(\text{loop}_c^{1'})\, m_2'(\text{loop}_c^{2'})\,:\,
(S_1, m_1(\text{loop}_c^1, \vals_1)) ->* (S_1'', m_1'(\text{loop}_c^{1'})),
\text{loop}_c^{1'}(n_1) < i,
(S_2, m_2(\text{loop}_c^2,\allowbreak \vals_2)) ->* (S_2'', m_2'(\text{loop}_c^{2'})),
\text{loop}_c^{2'}(n_2) < i$,
$S_1$ and $S_2$ terminate in the same way, produce the same output sequence, and have equivalent computation of  output deciding variables in $S_1$ and $S_2$ including the input sequence variable,  the I/O sequence variable,
$(S_1, m_1)\allowbreak \equiv_{H} (S_2, m_2)$ and
$(S_1, m_1) \equiv_{O} (S_2, m_2)$ and
$\forall x \in (\text{OVar}(S_1) \cup \text{OVar}(S_2)) \cup
               \{{id}_I, {id}_{IO}\}\,:\,
               (S_1, m_1) \equiv_{x} (S_2, m_2)$;

\item Loop counters of $S_1$ and $S_2$ are of values less than or equal to 1 but there are no reachable configurations
$(S_1, m_1(\text{loop}_c^{1_1}, \vals_{1_i}))$ from $(S_1, m_1(\vals_1))$,
$(S_2, m_2(\text{loop}_c^{2_1}, \vals_{2_i}))$ from $(S_2, \,\allowbreak m_2(\vals_2))$ where all of the following hold:
\begin{itemize}
\item The loop counter of $S_1$ and $S_2$ are of value 1,
$\text{loop}_c^{1_1}(n_1)\allowbreak = \text{loop}_c^{2_1}(n_2) = 1$.

\item Value stores $\vals_{1_1}$ and $\vals_{2_1}$ agree on values of  output deciding variables in $S_1$ and $S_2$ including the input sequence variable and the I/O sequence variable,
$\forall x \in \, (\text{OVar}(S_1) \cup \text{OVar}(S_2)) \cup \{{id}_I, {id}_{IO}\}\,:\,
\vals_{1_1}(x) = \vals_{2_1}(x)$.
\end{itemize}

\item There are reachable configuration
$(S_1, m_1(\text{loop}_c^{1_1}, \vals_{1_i}))$ from $(S_1, m_1(\vals_1))$,
$(S_2, m_2(\text{loop}_c^{2_1}, \vals_{2_i}))$ from $(S_2, \,\allowbreak m_2(\vals_2))$ where all of the following hold:
\begin{itemize}
\item The loop counter of $S_1$ and $S_2$ are of value 1,
$\text{loop}_c^{1_1}(n_1)\allowbreak = \text{loop}_c^{2_1}(n_2) = 1$.

\item Value stores $\vals_{1_1}$ and $\vals_{2_1}$ agree on values of  output deciding variables in $S_1$ and $S_2$ including the input sequence variable and the I/O sequence variable,
$\forall x \in \, (\text{OVar}(S_1) \cup \text{OVar}(S_2)) \cup \{{id}_I, {id}_{IO}\}\,:\,
\vals_{1_1}(x) = \vals_{2_1}(x)$.
\end{itemize}
\end{enumerate}

By definition, variables used in the predicate expression $e$ of $S_1$ and $S_2$ are in output deciding variables in $S_1$ and $S_2$, $\text{Use}(e) \subseteq \text{OVar}(S_1) \cup \text{OVar}(S_2)$.
By assumption, value stores $\vals_1$ and $\vals_2$ agree on values of variables in $\text{Use}(e)$,
the predicate expression $e$ evaluates to the same value w.r.t value stores $\vals_1$ and $\vals_2$ by Lemma~\ref{lmm:expEvalSameTerm}.
There are three possibilities.
\begin{enumerate}
\item The evaluation of $e$ crashes,

\noindent$\mathcal{E}'\llbracket e\rrbracket \vals_1 =
 \mathcal{E}'\llbracket e\rrbracket \vals_2 = (\text{error}, v_{\mathfrak{of}})$.

The execution of $S_1$ continues as follows:
\begin{tabbing}
xx\=xx\=\kill
\>\>                   $(\text{while}_{\langle n_1\rangle}(e) \, \{S_1'\}, m_1(\vals_1))$\\
\>$->$\>               $(\text{while}_{\langle n_1\rangle}((\text{error}, v_{\mathfrak{of}})) \, \{S_1'\}, m_1(\vals_1))$\\
\>\>                   by the rule EEval'\\
\>$->$\>               $(\text{while}_{\langle n_1\rangle}(0) \, \{S_1'\}, m_1(1/\mathfrak{f}))$\\
\>\>                   by the ECrash rule  \\
\>{\kStepArrow [i] }\> $(\text{while}_{\langle n_1\rangle}(0) \, \{S_1'\}, m_1(1/\mathfrak{f}))$ for any $i>0$\\
\>\>                   by the Crash rule.
\end{tabbing}

Similarly, the execution of $S_2$ started from the state $m_2(\vals_2)$ crashes.
Therefore $S_1$ and $S_2$ terminate in the same way when started from $m_1$ and $m_2$ respectively.
Because $\vals_1({id}_{IO}) = \vals_2({id}_{IO})$, the lemma holds.

\item The evaluation of $e$ reduces to zero,
$\mathcal{E}'\llbracket e\rrbracket \vals_1 =
 \mathcal{E}'\llbracket e\rrbracket \vals_2 = (0, v_{\mathfrak{of}})$.

The execution of $S_1$ continues as follows.
\begin{tabbing}
xx\=xx\=\kill
\>\>                   $(\text{while}_{\langle n_1\rangle}(e) \, \{S_1'\}, m_1(\vals_1))$\\
\>= \>                 $(\text{while}_{\langle n_1\rangle}((0, v_{\mathfrak{of}})) \, \{S_1'\}, m_1(\vals_1))$\\
\>\>                   by the rule EEval'\\
\>$->$\>               $(\text{while}_{\langle n_1\rangle}(0) \, \{S_1'\}, m_1(\vals_1))$\\
\>\>                   by the E-Oflow1 or E-Oflow2 rule  \\
\>$->$\>               $(\text{skip}, m_1(\vals_1))$ by the Wh-F rule.
\end{tabbing}

Similarly, the execution of $S_2$ gets to the configuration $(\text{skip}, m_2(\vals_2))$.
Loop counters of $S_1$ and $S_2$ are less than 1 and value stores agree on values of  output deciding variables in $S_1$ and $S_2$ including the input sequence variable and the I/O sequence variable.

\item The evaluation of $e$ reduces to the same nonzero integer value,
$\mathcal{E}'\llbracket e\rrbracket \vals_1 =
 \mathcal{E}'\llbracket e\rrbracket \vals_2 = (0, v_{\mathfrak{of}})$.

Then the execution of $S_1$ proceeds as follows:
\begin{tabbing}
xx\=xx\=\kill
\>\>                   $(\text{while}_{\langle n_1\rangle}(e) \, \{S_1'\}, m_1(\vals_1))$\\
\>= \>                 $(\text{while}_{\langle n_1\rangle}((v, v_{\mathfrak{of}})) \, \{S_1'\}, m_1(\vals_1))$\\
\>\>                   by the rule EEval'\\
\>$->$\>               $(\text{while}_{\langle n_1\rangle}(v) \, \{S_1'\}, m_1(\vals_1))$\\
\>\>                   by the E-Oflow1 or E-Oflow2 rule  \\
\>$->$\>               $(S_1';\text{while}_{\langle n_1\rangle}(e) \, \{S_1'\}, m_1($\\
\>\>                   $\text{loop}_c^1 \cup \{(n_1) \mapsto 1\}, \vals_1))$ by the Wh-T rule.
\end{tabbing}

Similarly, the execution of $S_2$ proceeds to the configuration
$(S_2';\text{while}_{\langle n_2\rangle}(e) \, \{S_2'\}, m_2(\text{loop}_c^2 \cup \{(n_2) \mapsto 1\}, \vals_2))$.

By the assumption, we show that $S_1'$ and $S_2'$ terminate in the same way and produce the same output sequence when started in the state $m_1(\text{loop}_c^{1_1}, \vals_1)$ and $m_2(\text{loop}_c^{2_1}, \vals_2)$ respectively, and $S_1'$ and $S_2'$ have equivalent computation of variables  defined in both statement sequences if both terminate.
We need to show that all conditions are satisfied for the application of the assumption.
\begin{itemize}
\item There are no program semantic errors related to the extra check in the update of exit-on-error in executions of $S_2'$ and $S_1'$.

The above two conditions are by assumption.

\item Value stores $\vals_1$ and $\vals_2$ agree on values of output deciding variables in $S_1'$ and $S_2'$ including the input, I/O sequence variable.

By definition, $\text{OVar}(S_1')\allowbreak \subseteq \text{OVar}(S_1)$.
So are the cases to $S_2'$ and $S_2$.
In addition, value stores $\vals_1$ and $\vals_2$ are not changed in the evaluation of the predicate expression $e$. The condition holds.
\end{itemize}
By  assumption, $S_1'$ and $S_2'$ terminate in the same way and produce the same output sequence when started in states
$m_1(\text{loop}_c', \vals_1)$ and $m_2(\text{loop}_c', \vals_2)$.
In addition, $S_1'$ and $S_2'$ have equivalent computation of output deciding variables in $S_1'$ and $S_2'$ when started in states
$m_1(\text{loop}_c', \vals_1)$ and $m_2(\text{loop}_c', \vals_2)$.

Then there are two cases.
\begin{enumerate}
\item $S_1'$ and $S_2'$ both do not terminate and produce the same output sequence.

By Lemma~\ref{lmm:multiStepSeqExec}, $S_1';S_1$ and $S_2';S_2$ both do not terminate and produce the same output sequence.

\item $S_1'$ and $S_2'$ both terminate and have equivalent computation of output deciding variables in $S_1'$ and $S_2'$.

By assumption, $(S_1', m_1(\text{loop}_c', \vals_1)) ->*
                (\text{skip}, \,\allowbreak m_1'(\text{loop}_c'', \vals_1'))$;
               $(S_2', m_2(\text{loop}_c', \vals_2)) ->*
                (\text{skip},\allowbreak m_2'(\text{loop}_c'', \vals_2'))$
where $\forall x \in (\text{OVar}(S_1') \cup \text{OVar}(S_2')) \cup
                     \{{id}_I, {id}_{IO}\},
\allowbreak\vals_1'(x) = \vals_2'(x)$.

By Lemma~\ref{lmm:exitOnErrFuncWideSimilarUseDef},
$\text{OVar}(S_1') \subseteq \text{OVar}(S_2')$.
Then variables used in the predicate expression of $S_1$ and $S_2$ are either in output deciding variables in both $S_1'$ and $S_2'$ or not.
Therefore value stores $\vals_2'$ and $\vals_1'$ agree on values of variables used in the expression $e$ and even output deciding variables in $S_1$ and $S_2$.

\end{enumerate}
\end{enumerate}

\noindent{Induction step on iterations}

The induction hypothesis (IH) is that, when $i\geq 1$, one of the following holds:
\begin{enumerate}
\item Loop counters for $S_1$ and $S_2$ are always less than $i$ if any is present,
$\forall m_1'(\text{loop}_c^{1'})\, m_2'(\text{loop}_c^{2'})\,:\,
(S_1, m_1(\text{loop}_c^1, \vals_1)) ->* (S_1'', m_1'(\text{loop}_c^{1'})),
\text{loop}_c^{1'}(n_1) < i,
(S_2, m_2(\text{loop}_c^2,\allowbreak \vals_2)) ->* (S_2'', m_2'(\text{loop}_c^{2'})),
\text{loop}_c^{2'}(n_2) < i$,
$S_1$ and $S_2$ terminate in the same way, produce the same output sequence, and have equivalent computation of  output deciding variables in both $S_1$ and $S_2$ as well as the input sequence variable,  the I/O sequence variable,
$(S_1, m_1)\allowbreak \equiv_{H} (S_2, m_2)$ and
$(S_1, m_1) \equiv_{O} (S_2, m_2)$ and
$\forall x \in (\text{OVar}(S_1) \cup \text{OVar}(S_2)) \cup
               \{{id}_I, {id}_{IO}\}\,:\,
               (S_1, m_1) \equiv_{x} (S_2, m_2)$;

\item The loop counter of $S_1$ and $S_2$ are of value less than or equal to $i$,
and there are no reachable configurations
$(S_1, m_1(\text{loop}_c^{1_i}, \vals_{1_i}))$ from $(S_1, m_1(\vals_1))$,
$(S_2, m_2(\text{loop}_c^{2_i}, \vals_{2_i}))$ from $(S_2, \,\allowbreak m_2(\vals_2))$ where all of the followings hold:
\begin{itemize}
\item The loop counters of $S_1$ and $S_2$ are of value $i$,
$\text{loop}_c^{1_i}(n_1)\allowbreak = \text{loop}_c^{2_i}(n_2) = i$.

\item Value stores $\vals_{1_i}$ and $\vals_{2_i}$ agree on values of  output deciding variables in $S_1$ and $S_2$ including the input sequence variable,  and the I/O sequence variable,
$\forall x \in \, (\text{OVar}(S_1) \cup \text{OVar}(S_2)) \cup \{{id}_I, {id}_{IO}\}\,:\,
\vals_{1_i}(x) = \vals_{2_i}(x)$.
\end{itemize}

\item There are reachable configurations
$(S_1, m_1(\text{loop}_c^{1_i}, \vals_{1_i}))$ from $(S_1, m_1(\vals_1))$,
$(S_2, m_2(\text{loop}_c^{2_i}, \vals_{2_i}))$ from $(S_2, \,\allowbreak m_2(\vals_2))$ where all of the following hold:
\begin{itemize}
\item The loop counter of $S_1$ and $S_2$ are of value $i$,
$\text{loop}_c^{1_i}(n_1)\allowbreak = \text{loop}_c^{2_i}(n_2) = i$.

\item Value stores $\vals_{1_i}$ and $\vals_{2_i}$ agree on values of  output deciding variables in $S_1$ and $S_2$ including the input sequence variable,  and the I/O sequence variable,
$\forall x \in \, (\text{OVar}(S_1) \cup \text{OVar}(S_2)) \cup \{{id}_I, {id}_{IO}\}\,:\,
\vals_{1_i}(x) = \vals_{2_i}(x))$.
\end{itemize}
\end{enumerate}

Then we show that, when $i+1$, one of the following holds:
\begin{enumerate}
\item Loop counters for $S_1$ and $S_2$ are always less than $i+1$ if any is present,
$\forall m_1'(\text{loop}_c^{1'})\, m_2'(\text{loop}_c^{2'})\,:\,
(S_1, m_1(\text{loop}_c^1, \vals_1)) ->* (S_1'', m_1'(\text{loop}_c^{1'})),
\text{loop}_c^{1'}(n_1) < i+1,
(S_2, m_2(\text{loop}_c^2,\allowbreak \vals_2)) ->* (S_2'', m_2'(\text{loop}_c^{2'})),
\text{loop}_c^{2'}(n_2) < i+1$,
$S_1$ and $S_2$ terminate in the same way, produce the same output sequence, and have equivalent computation of output deciding variables in $S_1$ and $S_2$ including the input sequence variable and  the I/O sequence variable,
$(S_1, m_1)\allowbreak \equiv_{H} (S_2, m_2)$ and
$(S_1, m_1) \equiv_{O} (S_2, m_2)$ and
$\forall x \in (\text{OVar}(S_1) \cup \text{OVar}(S_2)) \cup
               \{{id}_I, {id}_{IO}\}\,:\,
               (S_1, m_1) \equiv_{x} (S_2, m_2)$;

\item The loop counter of $S_1$ and $S_2$ are of value less than or equal to $i+1$,
and there are no reachable configurations
$(S_1, m_1(\text{loop}_c^{1_{i+1}}, \vals_{1_{i+1}}))$ from $(S_1, m_1(\vals_1))$,
$(S_2, m_2(\text{loop}_c^{2_{i+1}},\,\allowbreak \vals_{2_{i+1}}))$ from $(S_2, \,\allowbreak m_2(\vals_2))$ where all of the following hold:
\begin{itemize}
\item The loop counters of $S_1$ and $S_2$ are of value $i+1$,
$\text{loop}_c^{1_{i+1}}(n_1)\allowbreak = \text{loop}_c^{2_{i+1}}(n_2) = i+1$.

\item Value stores $\vals_{1_{i+1}}$ and $\vals_{2_{i+1}}$ agree on values of  output deciding variables in $S_1$ and $S_2$ including the input sequence variable,  and the I/O sequence variable,
$\forall x \in \, (\text{OVar}(S_1) \cup \text{OVar}(S_2)) \cup \{{id}_I, {id}_{IO}\}\,:\,
\vals_{1_{i+1}}(x) = \vals_{2_{i+1}}(x)$.
\end{itemize}

\item There are reachable configurations
$(S_1, m_1(\text{loop}_c^{1_{i+1}}, \vals_{1_i}))$ from $(S_1, m_1(\vals_1))$,
$(S_2, m_2(\text{loop}_c^{2_{i+1}}, \vals_{2_i}))$ from $(S_2, \,\allowbreak m_2(\vals_2))$ where all of the following hold:
\begin{itemize}
\item The loop counter of $S_1$ and $S_2$ are of value $i$,
$\text{loop}_c^{1_{i+1}}(n_1)\allowbreak = \text{loop}_c^{2_{i+1}}(n_2) = i+1$.

\item Value stores $\vals_{1_{i+1}}$ and $\vals_{2_{i+1}}$ agree on values of  output deciding variables in $S_1$ and $S_2$ including the input sequence variable,  and the I/O sequence variable,
$\forall x \in \, (\text{OVar}(S_1) \cup \text{OVar}(S_2)) \cup \{{id}_{I}, {id}_{IO}\}\,:\,
\vals_{1_{i+1}}(x) = \vals_{2_{i+1}}(x)$.
\end{itemize}
\end{enumerate}

By hypothesis IH and theorem~\ref{thm:mainTermSameWayLocal} and~\ref{thm:sameIOtheoremFuncWide}, there is no configuration where loop counters of $S_1$ and $S_2$ are of value $i+1$ when any of the following holds:
\begin{enumerate}
\item Loop counters for $S_1$ and $S_2$ are always less than $i$ if any is present,
$\forall m_1'(\text{loop}_c^{1'})\, m_2'(\text{loop}_c^{2'})\,:\,
(S_1, m_1(\text{loop}_c^1, \vals_1)) ->* (S_1'', m_1'(\text{loop}_c^{1'})),
\text{loop}_c^{1'}(n_1) < i,
(S_2, m_2(\text{loop}_c^2,\allowbreak \vals_2)) ->* (S_2'', m_2'(\text{loop}_c^{2'})),
\text{loop}_c^{2'}(n_2) < i$,
$S_1$ and $S_2$ terminate in the same way, produce the same output sequence, and have equivalent computation of output deciding variables in $S_1$ and $S_2$ including the input sequence variable and  the I/O sequence variable,
$(S_1, m_1)\allowbreak \equiv_{H} (S_2, m_2)$ and
$(S_1, m_1) \equiv_{O} (S_2, m_2)$ and
$\forall x \in (\text{OVar}(S_1) \cup \text{OVar}(S_2)) \cup
               \{{id}_I, {id}_{IO}\}\,:\,
               (S_1, m_1) \equiv_{x} (S_2, m_2)$;

\item The loop counter of $S_1$ and $S_2$ are of value less than or equal to $i$,
and there are no reachable configurations
$(S_1, m_1(\text{loop}_c^{1_i}, \vals_{1_i}))$ from $(S_1, m_1(\vals_1))$,
$(S_2, m_2(\text{loop}_c^{2_i}, \vals_{2_i}))$ from $(S_2, \,\allowbreak m_2(\vals_2))$ where all of the following hold:
\begin{itemize}
\item The loop counters of $S_1$ and $S_2$ are of value $i$,
$\text{loop}_c^{1_i}(n_1)\allowbreak = \text{loop}_c^{2_i}(n_2) = i$.

\item Value stores $\vals_{1_i}$ and $\vals_{2_i}$ agree on values of  output deciding variables in both $S_1$ and $S_2$ as well as the input sequence variable,  and the I/O sequence variable,
$\forall x \in \, (\text{OVar}(S_1) \cup \text{OVar}(S_2)) \cup \{{id}_I, {id}_{IO}\}\,:\,
\vals_{1_i}(x) = \vals_{2_i}(x)$.
\end{itemize}
\end{enumerate}

When there are  reachable configurations
$(S_1, m_1(\text{loop}_c^{1_i}, \vals_{1_i}))$ from $(S_1, m_1(\vals_1))$,
$(S_2, m_2(\text{loop}_c^{2_i}, \vals_{2_i}))$ from $(S_2, \,\allowbreak m_2(\vals_2))$ where all of the following hold:
\begin{itemize}
\item The loop counter of $S_1$ and $S_2$ are of value $i$,
$\text{loop}_c^{1_i}(n_1)\allowbreak = \text{loop}_c^{2_i}(n_2) = i$.

\item Value stores $\vals_{1_i}$ and $\vals_{2_i}$ agree on values of  output deciding variables in $S_1$ and $S_2$ including the input sequence variable and the I/O sequence variable,
$\forall x \in \, (\text{OVar}(S_1) \cup \text{OVar}(S_2)) \cup \{{id}_I, {id}_{IO}\}\,:\,
\vals_{1_i}(x) = \vals_{2_i}(x)$.
\end{itemize}

By similar argument in base case, we have one of the following holds:
\begin{enumerate}
\item Loop counters for $S_1$ and $S_2$ are always less than $i+1$ if any is present,
$\forall m_1'(\text{loop}_c^{1'})\, m_2'(\text{loop}_c^{2'})\,:\,
(S_1, m_1(\text{loop}_c^1, \vals_1)) ->* (S_1'', m_1'(\text{loop}_c^{1'})),
\text{loop}_c^{1'}(n_1) < i+1,
(S_2, m_2(\text{loop}_c^2,\allowbreak \vals_2)) ->* (S_2'', m_2'(\text{loop}_c^{2'})),
\text{loop}_c^{2'}(n_2) < i$,
$S_1$ and $S_2$ terminate in the same way, produce the same output sequence, and have equivalent computation of  output deciding variables in $S_1$ and $S_2$ including the input sequence variable,  the I/O sequence variable,
$(S_1, m_1)\allowbreak \equiv_{H} (S_2, m_2)$ and
$(S_1, m_1) \equiv_{O} (S_2, m_2)$ and
$\forall x \in (\text{OVar}(S_1) \cup \text{OVar}(S_2)) \cup
               \{{id}_I, {id}_{IO}\}\,:\,
               (S_1, m_1) \equiv_{x} (S_2, m_2)$;

\item The loop counter of $S_1$ and $S_2$ are of value less than or equal to $i+1$,
and there are no reachable configurations
$(S_1, m_1(\text{loop}_c^{1_i}, \vals_{1_i}))$ from $(S_1, m_1(\vals_1))$,
$(S_2, m_2(\text{loop}_c^{2_i}, \vals_{2_i}))$ from $(S_2, \,\allowbreak m_2(\vals_2))$ where all of the following hold:
\begin{itemize}
\item The loop counters of $S_1$ and $S_2$ are of value $i$,
$\text{loop}_c^{1_{i+1}}(n_1)\allowbreak = \text{loop}_c^{2_{i+1}}(n_2) = i+1$.

\item Value stores $\vals_{1_{i+1}}$ and $\vals_{2_{i+1}}$ agree on values of  output deciding variables in $S_1$ and $S_2$ including the input sequence variable and the I/O sequence variable,
$\forall x \in \, (\text{OVar}(S_1) \cup \text{OVar}(S_2)) \cup \{{id}_I, {id}_{IO}\}\,:\,
\vals_{1_{i+1}}(x) = \vals_{2_{i+1}}(x)$.
\end{itemize}

\item There are reachable configurations
$(S_1, m_1(\text{loop}_c^{1_{i+1}}, \vals_{1_{i+1}}))$ from $(S_1, m_1(\vals_1))$,
$(S_2, m_2(\text{loop}_c^{2_{i+1}}, \vals_{2_{i+1}}))$ from $(S_2, \,\allowbreak m_2(\vals_2))$ where all of the following hold:
\begin{itemize}
\item The loop counter of $S_1$ and $S_2$ are of value $i$,
$\text{loop}_c^{1_{i+1}}(n_1)\allowbreak = \text{loop}_c^{2_{i+1}}(n_2) = i$.

\item Value stores $\vals_{1_{i+1}}$ and $\vals_{2_{i+1}}$ agree on values of output deciding variables in $S_1$ and $S_2$ including the input sequence variable,  and the I/O sequence variable,
$\forall x \in \, (\text{OVar}(S_1) \cup \text{OVar}(S_2)) \cup \{{id}_{I+1}, {id}_{IO}\}\,:\,
\vals_{1_{i+1}}(x) = \vals_{2_{i+1}}(x)$.
\end{itemize}
\end{enumerate}
\end{proof}

\subsection{Proof rule for improved prompt message}

If the only difference between two programs are the constant messages that the user receives, we consider that the two programs to be equivalent. We realize that in general it is possible to introduce new semantics even by changing constant strings.
An old version might have incorrectly labeled output: ``median value = 5" instead of ``average value = 5”, for example.
We rule out such possibilities because all non-constant values are guaranteed to be exactly same.
In practice, outputs could be classified into prompt outputs and actual outputs. Prompt outputs are those asking clients for inputs, which are constants hardcoded in the output statement. Actual outputs are dynamic messages produced by evaluation of non-constant expression in execution.
The changes of prompt outputs are equivalent only for interactions with human clients.
In order to prove the update of improved prompt messages to be backward compatible, we assume that the different prompt outputs produced in executions of the old program and the updated program, due to the different constants in output statements, are equivalent.
Because the old program and the new program are exactly same except some output statements with different constants as expression $e$, we could show two programs produce the ``equivalent" output sequence under the assumption of equivalent prompt outputs.


We formalize the generalized update of improved prompt messages, then we show that the updated program produce the same I/O sequence as the old program in executions without program semantic errors.
The following is the definition of the update class of improved prompt messages.

\begin{definition}\label{def:outputConstChangeFuncWide}
{\bf (Improved user messages)}
A program $P_2 = {Pmpt}_2;EN;$ $V;S_{entry}$ includes updates of improved prompt messages compared with a program $P_1 = {Pmpt}_1;EN;V;S_{entry}$, written $P_2\, {\approx}_{\text{Out}}^S\, P_1$, iff
${Pmpt}_2 \neq {Pmpt}_1$.
\end{definition}


We give the lemma that two programs terminate in the same way, produce the equivalent output sequence, and have equivalent computation of defined variables in both programs in valid executions if there are updates of improved prompt messages between them.
\begin{lemma}\label{lmm:outputConstChangeFuncWide}
Let $P_1 = {Pmpt}_1;EN;V;S_{entry}$ and $P_2 = {Pmpt}_2;EN;V;S_{entry}$ be two programs where there are updates of improved prompt messages in $P_2$ compared with $P_1$.
If $S_1$ and $S_2$ start in states $m_1(\vals_1)$ and $m_2(\vals_2)$ such that both of the following hold:
\begin{itemize}
\item Value stores $\vals_1$ and $\vals_2$ agree on values of variables used in $S_{entry}$ in both programs as well as the input sequence variable,
$\forall x \in \text{Use}(S_{entry}) \cup \{{id}_I\}\,:\,
 \vals_1(x) = \vals_2(x)$;

\item Value stores $\vals_1$ and $\vals_2$ have ``equivalent" I/O sequence,
$\vals_1({id}_{IO}) \equiv \vals_2({id}_{IO})$;

\item The different prompt outputs in the update of improved prompt messages are equivalent;
\end{itemize}
then $S_1$ and $S_2$ terminate in the same way, produce the equivalent output sequence, and when $S_1$ and $S_2$ both terminate, they have equivalent computation of  defined variables in $S_{entry}$ in both programs as well as the input sequence variable, $S_{entry}$ in the two programs produce the equivalent  I/O sequence variable,
\begin{itemize}
\item $(S_{entry}, m_1) \equiv_{H} (S_{entry}, m_2)$;

\item $\forall x \in (\text{Def}(S_1)\cap \text{Def}(S_2)) \cup \{{id}_I\}\,:\,$
\noindent$\allowbreak (S_{entry}, m_1) \equiv_{x} (S_{entry}, m_2)$;

\item The produced output sequences in executions of $S_{entry}$ in both programs are ``equivalent",
$\vals_1({id}_{IO}) \equiv \vals_2({id}_{IO})$.
\end{itemize}
\end{lemma}
The difference between prompt types in $P_1$ and $P_2$ can be either addition/removal of labels as well as the change of the mapping of labels with constants.
The proof is straightforward because programs $P_1$ and $P_2$ have the same entry statement sequence and we have the assumption that different prompt outputs due to the difference of the prompt type are equivalent.
\begin{proof}
By induction on the sum of the program size of $S_1$ and $S_2$, $\text{size}(S_1) + \text{size}(S_2)$.

\noindent{Base case}.
$S_1 = ``\text{output}\, v_1"$ and $S_2 = ``\text{output}\, v_2"$;

Then the execution of $S_2$ proceeds as follows.
\begin{tabbing}
xx\=xx\=\kill
\>\>                   $(\text{output}\, v_2, m_2(\vals_2))$\\
\>$->$\>               $(\text{skip}, m_2(\vals_2[``\vals_2({id}_{IO}) \cdot \bar{v_2}"/{id}_{IO}]))$\\
\>\>                   by the rule Out-1 or Out-2
\end{tabbing}
Similarly, $(\text{output}\, v_1, m_1(\vals_1)) ->*
(\text{skip}, m_1(\vals_1[``\vals_1({id}_{IO}) \cdot \bar{v_1}"/{id}_{IO}]))$.

By assumption, $\vals_2({id}_{IO}) \equiv \vals_1({id}_{IO})$.
In addition, by assumption, $\bar{v_2} \equiv  \bar{v_1}$.
Therefore, $S_1$ and $S_2$ terminate in the same way, produce the same output sequence and have equivalent computation of defined variables in $S_1$ and $S_2$. This lemma holds.


\noindent{Induction step}.

The hypothesis is that this lemma holds when the sum $k$ of the program size of $S_1$ and $S_2$ are great than or equal to 2, $k\geq 2$.

We then show that this lemma holds when the sum of the program size of $S_1$ and $S_2$ is $k+1$.
There are cases to consider.
\begin{enumerate}
\item $S_1$ and $S_2$ are both ``If" statement:

$S_1 = ``\text{If}(e) \, \text{then}\{S_1^t\} \, \text{else} \{S_1^f\}"$,
$S_2 = ``\text{If}(e) \, \text{then}\{S_2^t\} \, \text{else} \{S_2^f\}"$ where both of the following hold
\begin{itemize}
\item $S_2^t {\approx}_{\text{Out}}^S S_1^t$;

\item $S_2^f {\approx}_{\text{Out}}^S S_1^f$;
\end{itemize}

By the definition of $\text{Use}(S_1)$, variables used in the predicate expression $e$ are a subset of  used variables in $S_1$ and $S_2$,
$\text{Use}(e) \subseteq \text{Use}(S_1) \cap \text{Use}(S_2)$.
By assumption, corresponding variables used in $e$ are of same value in value stores $\vals_1$ and $\vals_2$. By Lemma~\ref{lmm:expEvalSameVal}, the expression evaluates to the same value w.r.t value stores $(\vals_1$ and $\vals_2$. There are three possibilities.
\begin{enumerate}
\item The evaluation of $e$ crashes,
$\mathcal{E}'\llbracket e\rrbracket \vals_1 =
 \mathcal{E}'\llbracket e\rrbracket \vals_2 = (\text{error}, v_{\mathfrak{of}})$.

The execution of $S_1$ continues as follows:
\begin{tabbing}
xx\=xx\=\kill
\>\>                   $(\text{If}(e) \, \text{then}\{S_1^t\} \, \text{else} \{S_1^f\}, m_1(\vals_1))$\\
\>$->$\>               $(\text{If}((\text{error}, v_{\mathfrak{of}})) \, \text{then}\{S_1^t\} \, \text{else} \{S_1^f\}, m_1(\vals_1))$\\
\>\>                   by the rule EEval'\\
\>$->$\>               $(\text{If}(0) \, \text{then}\{S_1^t\} \, \text{else} \{S_1^f\}, m_1(1/\mathfrak{f}))$\\
\>\>                   by the ECrash rule  \\
\>{\kStepArrow [i] }\> $(\text{If}(0) \, \text{then}\{S_1^t\} \, \text{else} \{S_1^f\}, m_1(1/\mathfrak{f}))$ for any $i>0$\\
\>\>                   by the Crash rule.
\end{tabbing}

Similarly, the execution of $S_2$ started from the state $m_2(\vals_2)$ crashes.
The lemma holds.

\item The evaluation of $e$ reduces to zero,
$\mathcal{E}'\llbracket e\rrbracket \vals_1 =
 \mathcal{E}'\llbracket e\rrbracket \vals_2 = (0, v_{\mathfrak{of}})$.

The execution of $S_1$ continues as follows.
\begin{tabbing}
xx\=xx\=\kill
\>\>                   $(\text{If}(e) \, \text{then}\{S_1^t\} \, \text{else} \{S_1^f\}, m_1(\vals_1))$\\
\>= \>                 $(\text{If}((0, v_{\mathfrak{of}})) \, \text{then}\{S_1^t\} \, \text{else} \{S_1^f\}, m_1(\vals_1))$\\
\>\>                   by the rule EEval'\\
\>$->$\>               $(\text{If}(0) \, \text{then}\{S_1^t\} \, \text{else} \{S_1^f\}, m_1(\vals_1))$\\
\>\>                   by the E-Oflow1 or E-Oflow2 rule  \\
\>$->$\>               $(S_1^f, m_1(\vals_1))$ by the If-F rule.
\end{tabbing}

Similarly, the execution of $S_2$ gets to the configuration $(S_2^f, m_2(\vals_2))$.

By the hypothesis IH, we show the lemma holds.
We need to show that all conditions are satisfied for the application of the hypothesis IH.
\begin{itemize}
\item $S_2^f {\approx}_{\text{Out}}^S S_1^f$

By assumption.

\item The sum of  the program size of $S_1^f$ and $S_2^f$ is less than $k$,
$\text{size}(S_1^f) + \text{size}(S_2^f) < k$.

By definition, $\text{size}(S_1) = 1 + \text{size}(S_1^t) + \text{size}(S_1^f)$.
Then, $\text{size}(S_1^f) + \text{size}(S_2^f) < k + 1 - 2 = k - 1$.

\item Value stores $\vals_1$ and $\vals_2$ agree on values of  used variables in $S_1^f$ and $S_2^f$ as well as the input, I/O sequence variable.

By definition, $\text{Use}(S_1^f)\allowbreak \subseteq \text{Use}(S_1)$.
So are the cases to $S_2^f$ and $S_2$.
In addition, value stores $\vals_1$ and $\vals_2$ are not changed in the evaluation of the predicate expression $e$. The condition holds.

\item Different constants used in output statements are equivalent as output values.

By assumption.
\end{itemize}

By the hypothesis IH, the lemma holds.

\item The evaluation of $e$ reduces to the same nonzero integer value,
$\mathcal{E}'\llbracket e\rrbracket \vals_1 =
 \mathcal{E}'\llbracket e\rrbracket \vals_2 = (v, v_{\mathfrak{of}})$ where $v \neq 0$.

By argument similar to the second subcase above.
\end{enumerate}

\item $S_1$ and $S_2$ are both ``while" statements:

$S_1 = ``\text{while}_{\langle n\rangle}(e) \, \{S_1'\}"$,
$S_2 = ``\text{while}_{\langle n\rangle}(e) \, \{S_2'\}"$ where
$S_2' {\approx}_{\text{Out}}^S S_1'$;

By Lemma~\ref{lmm:outputConstChangeLoopStmt}, we show this lemma holds.
We need to show that all required conditions are satisfied for the application of Lemma~\ref{lmm:outputConstChangeLoopStmt}.
\begin{itemize}
\item $S_1'$ and $S_2'$ have same set of  defined variables,
$\text{Def}(S_1') = \text{Def}(S_2') = \text{Def}(S)$;

\item The  used variables in $S_1'$ are a subset of those in $S_2'$,
$\text{Use}(S_1') = \text{Use}(S_2')$;

By Lemma~\ref{lmm:outputConstChangeFuncWideSimilarUseDef}.

\item When started in states $m_1'(\vals_1'), m_2'(\vals_1')$ where
value stores $\vals_1'$ and $\vals_2'$ agree on values of  used variables in both $S_1'$ and $S_2'$ as well as the input sequence variable, , and the I/O sequence variable,
then $S_1'$ and $S_2'$ terminate in the same way, produce the same output sequence, and have equivalent computation of defined variables in both $S_1$ and $S_2$ as well as the input sequence variable and the I/O sequence variable.

By the induction hypothesis IH. This is because the sum of the program size of $S_1'$ and $S_2'$ is less than $k$. By definition, $\text{size}(S_1) = 1 + \text{size}(S_1')$.
\end{itemize}
By Lemma~\ref{lmm:outputConstChangeLoopStmt}, this lemma holds.

\item $S_1 = S_1';s_1$ and $S_2 = S_2';s_2$ where both of the following hold:

\begin{itemize}
\item $S_2' \approx_{\text{Out}}^S S_1'$;

\item $s_2 \approx_{\text{Out}}^S s_1$;
\end{itemize}

By the hypothesis IH, we show $S_2'$ and $S_1'$ terminate in the same way and produce the equivalent output sequence and when $S_2'$ and $S_1'$ both terminate, $S_2'$ and $S_1'$ have equivalent terminating computation of variables defined in $S_2'$ and $S_1'$ as well as the input sequence variable.
By assumption, the different value of the I/O sequence in executions of $S_1$ and $S_2$ are equivalent.

We show all the required conditions are satisfied for the application of the hypothesis IH.
\begin{itemize}
\item $S_2' {\approx}_{\text{Out}}^S S_1'$;

\item The I/O sequence variable in executions of $S_1$ and $S_2$ are equivalent,
    $\vals_1({id}_{IO}) \equiv \vals_2({id}_{IO})$;

By assumption.

\item The sum of  the program size of $S_1'$ and $S_2'$ is less than $k$,
$\text{size}(S_1') + \text{size}(S_2') < k$.

By definition, $\text{size}(S_2) = \text{size}(s_2) + \text{size}(S_2')$ where $\text{size}(s_2) < 1$.
Then, $\text{size}(S_2') + \text{size}(S_1') < k + 1 - \text{size}(s_2) - \text{size}(s_1) < k$.

\item Value stores $\vals_1$ and $\vals_2$ agree on values of  used variables in both $S_2'$ and $S_1'$ as well as the input sequence variable.

By definition, $\text{Use}(S_2')\allowbreak \subseteq \text{Use}(S_2)$,
$\text{Use}(S_1')\allowbreak \subseteq \text{Use}(S_1)$.
The condition holds.
\end{itemize}

By the hypothesis IH, one of the following holds:
\begin{enumerate}
\item $S_1'$ and $S_2'$ both do not terminate.

By Lemma~\ref{lmm:multiStepSeqExec}, executions of $S_1 = S_1';s_1$ and $S_2 = S_2';s_2$ both do not terminate and produce the same output sequence.

\item $S_1'$ and $S_2'$ both terminate.

By assumption,
$(S_2', m_2(\vals_2)) ->* (\text{skip}, m_2'(\vals_2'))$,
$(S_1', m_1(\vals_1)) ->* (\text{skip}, m_1'(\vals_1'))$.

By Corollary~\ref{coro:termSeq},
$(S_2';s_2, m_2(\vals_2)) ->* (s_2, m_2'(\vals_2'))$,
$(S_1';s_1, m_1(\vals_1)) ->* (s_1, m_1'(\vals_1'))$.

By the hypothesis IH, we show that $s_2$ and $s_1$ terminate in the same way, produce the ``equivalent" output sequence and when $s_2$ and $s_1$ both terminate, $s_2$ and $s_1$ have equivalent computation of variables  defined in both $s_1$ and $s_2$ and the input sequence variable;
$s_2$ and $s_1$ produce ``equivalent" output sequence.

We need to show that all conditions are satisfied for the application of the hypothesis IH.
\begin{itemize}
\item There are updates of ``improved prompt messages" in $s_2$ compared with $s_1$,
$s_2 {\approx}_{\text{Out}}^S s_1$;

By assumption, $s_2 {\approx}_{\text{Out}}^S s_1$.

\item The sum of the program size $s_2$ and $s_1$ is less than or equals to $k$;

By definition, $\text{size}(S_2') \geq 1,  \text{size}(S_1') \geq 1$.
Therefore, $\text{size}(s_2) + \text{size}(s_1) < k+1 - \text{size}(S_2') - \text{size}(S_1') \leq k$.

\item Value stores $\vals_1'$ and $\vals_2'$ agree on values of  used variables in $s_2$ and $s_1$ as well as the input sequence variable;

By Lemma~\ref{lmm:exitOnErrFuncWideSimilarUseDef},
$\text{Use}(s_1) = \text{Use}(s_2)$, then $\text{Use}(s_2)\,\allowbreak  = \text{Use}(s_1) = \text{Use}(s)$.
Similarly, by Lemma~\ref{lmm:exitOnErrFuncWideSimilarUseDef}, $\text{Def}(S_1') = \text{Def}(S_2')$. For any variable $id$ in $\text{Use}(s_1)$, if $id$ is not in $\text{Def}(S_1')$,
then the value of $id$ is not changed in the execution of $S_1'$ and $S_2'$,
$\vals_1'(id) = \vals_1(id) = \vals_2(id)\allowbreak = \vals_2'(id)$.
Otherwise, the variable $id$ is defined in the execution of $S_1'$ and $S_2'$, by assumption, $\vals_1'(id) = \vals_2'(id)$.
The condition holds.

\item Values of , the I/O sequence variable in value stores $\vals_1'$ and $\vals_2'$ are equivalent.

By assumption.
\end{itemize}
By the hypothesis IH, the lemma holds.
\end{enumerate}
\end{enumerate}
\end{proof}

We list the auxiliary lemmas below. One lemma shows that, if there are updates of improved prompt messages between two statement sequences, then there are same set of  defined variables and used variables in the two statement sequences.
The second lemma shows that, if there are updates of improved prompt messages between two loop statements, then the two  loop statement terminate in the same way, produce the equivalent output sequence, and have equivalent computation of  defined variables in both the old and updated programs as well as the input sequence variable.
\begin{lemma}\label{lmm:outputConstChangeFuncWideSimilarUseDef}
Let $S_2$ be a statement sequence  and $S_1$  where there are updates of ``improved prompt messages",
$S_2 {\approx}_{\text{Out}}^S S_1$.
Then  used variables in $S_2$ are the same of  used variables in $S_1$,
$\text{Use}(S_1) = \text{Use}(S_2)$,
 defined variables in $S_2$ are the same as  used variables in $S_1$,
$\text{Def}(S_1) = \text{Def}(S_2)$.
\end{lemma}
\begin{proof}
By induction on the sum of the program size of $S_1$ and $S_2$.
\end{proof}
\begin{lemma}\label{lmm:outputConstChangeLoopStmt}
Let $S_1 = \text{while}_{\langle n_1\rangle}(e) \, \{S_1'\}$ and
    $S_2 = \text{while}_{\langle n_2\rangle}(e) \,\allowbreak \{S_2'\}$ be two loop statements where all of the following hold:
\begin{itemize}
\item There are updates of improved prompt messages in $S_2'$ compared with $S_1'$,
$S_2' {\approx}_{\text{Out}}^S S_1'$;

\item $S_1'$ and $S_2'$ have same set of  defined variables,

\noindent$\text{Def}(S_1') = \text{Def}(S_2') = \text{Def}(S)$;

\item $S_1'$ and $S_2'$ have same set of  used variables,
$\text{Use}(S_1') = \text{Use}(S_2')$;

\item When started in states $m_1'(\vals_1'), m_2'(\vals_2')$ where
    \begin{itemize}
    \item Value stores agree on  values of  used variables in both $S_1'$ and $S_2'$ as well as the input sequence variable,
     $\forall x \in \text{Use}(S_1') \cup \{{id}_I\}\,
    \forall m_1'(\vals_1')\, m_2'(\vals_2')\,:\, \, \allowbreak
    \vals_1'(x) = \vals_2'(x)$;

    \item Values of the I/O sequence variable in value stores $\vals_1', \vals_2'$ are equivalent,
        $\vals_1'({id}_{IO}) \equiv \vals_2'({id}_{IO}))$;
    \end{itemize}
then $S_1'$ and $S_2'$ terminate in the same way, produce the ``equivalent" output sequence, and have equivalent computation of  defined variables in $S_1'$ and $S_2'$ as well as the input sequence variable,
    $((S_1', m_1) \equiv_{H} (S_2', m_2)) \wedge
     (\forall x \in \text{Def}(S) \cup \{{id}_I\}\,:\,
     (S_1', m_1) \equiv_{x} (S_2', m_2))$;
\end{itemize} 

If $S_1$ and $S_2$ start in states $m_1(\text{loop}_c^1, \vals_1), m_2(\text{loop}_c^2, \vals_2)$ respectively, with loop counters of $S_1$ and $S_2$ not initialized ($S_1, S_2$ have not executed yet), value stores agree on values of  used variables in $S_1$ and $S_2$, and there are no program semantic errors, then, for any positive integer $i$, one of the following holds:
\begin{enumerate}
\item Loop counters for $S_1$ and $S_2$ are always less than $i$ if any is present,
$\forall m_1'(\text{loop}_c^{1'})\, m_2'(\text{loop}_c^{2'})\,:\,
(S_1, m_1(\text{loop}_c^1, \vals_1)) ->* (S_1'', m_1'(\text{loop}_c^{1'})),
\text{loop}_c^{1'}( n_1) < i,
(S_2, m_2(\text{loop}_c^2,\allowbreak \vals_2)) ->* (S_2'', m_2'(\text{loop}_c^{2'})),
\text{loop}_c^{2'}(n_2) < i$,
$S_1$ and $S_2$ terminate in the same way, produce the equivalent output sequence, and have equivalent computation of  defined variables in both $S_1$ and $S_2$ and the input sequence variable,
$(S_1, m_1)\allowbreak \equiv_{H} (S_2, m_2)$ and
$\forall x \in (\text{Def}(S_1) \cap \text{Def}(S_2)) \cup
               \{{id}_I\}\,:\,
               (S_1, m_1) \equiv_{x} (S_2, m_2)$;
$S_1$ and $S_2$ produce the ``equivalent"  I/O sequence;

\item The loop counter of $S_1$ and $S_2$ are of value less than or equal to $i$,
and there are no reachable configurations
$(S_1, m_1(\text{loop}_c^{1_i},\allowbreak \vals_{1_i}))$ from $(S_1, m_1(\vals_1))$,
$(S_2, m_2(\text{loop}_c^{2_i}, \vals_{2_i}))$ from $(S_2, \,\allowbreak m_2(\vals_2))$ where all of the following hold:
\begin{itemize}
\item The loop counters of $S_1$ and $S_2$ are of value $i$,
$\text{loop}_c^{1_i}(n_1)\allowbreak = \text{loop}_c^{2_i}(n_2) = i$.

\item Value stores $\vals_{1_i}$ and $\vals_{2_i}$ agree on values of  used variables in both $S_1$ and $S_2$ as well as the input sequence variable,
$\forall x \in \, (\text{Use}(S_1) \cap \text{Use}(S_2)) \cup \{{id}_I\}\,:\,
\vals_{1_i}(x) = \vals_{2_i}(x)$.

\item Values of the I/O sequence variable in value stores
$\vals_{1_i}({id}_{IO}) \equiv \vals_{2_i}({id}_{IO})$;
\end{itemize}

\item There are reachable configurations
$(S_1, m_1(\text{loop}_c^{1_i}, \vals_{1_i}))$ from $(S_1, m_1(\vals_1))$,
$(S_2, m_2(\text{loop}_c^{2_i}, \vals_{2_i}))$ from $(S_2, \,\allowbreak m_2(\vals_2))$ where all of the following hold:
\begin{itemize}
\item The loop counter of $S_1$ and $S_2$ are of value $i$,
$\text{loop}_c^{1_i}(n_1)\allowbreak = \text{loop}_c^{2_i}(n_2) = i$.

\item Value stores $\vals_{1_i}$ and $\vals_{2_i}$ agree on values of  used variables in both $S_1$ and $S_2$ as well as the input sequence variable,
$\forall x \in \, (\text{Use}(S_1) \cap \text{Use}(S_2)) \cup \{{id}_I\}\,:\,
\vals_{1_i}(x) = \vals_{2_i}(x)$.

\item Values of the I/O sequence variable in value stores $\vals_{1_i}, \vals_{2_i}$ are equivalent,
$\vals_{1_i}({id}_{IO}) \equiv \vals_{2_i}({id}_{IO})$;
\end{itemize}
\end{enumerate}
\end{lemma}

\begin{proof}
By induction on $i$.

\noindent{Base case}.

We show that, when $i=1$, one of the followings holds:
\begin{enumerate}
\item Loop counters for $S_1$ and $S_2$ are always less than 1 if any is present,
$\forall m_1'(\text{loop}_c^{1'})\, m_2'(\text{loop}_c^{2'})\,:\,
(S_1, m_1(\text{loop}_c^1, \vals_1)) ->* (S_1'', m_1'(\text{loop}_c^{1'})),
\text{loop}_c^{1'}(n_1) < i,
(S_2, m_2(\text{loop}_c^2,\allowbreak \vals_2)) ->* (S_2'', m_2'(\text{loop}_c^{2'})),
\text{loop}_c^{2'}(n_2) < i$,
$S_1$ and $S_2$ terminate in the same way, produce the equivalent I/O sequence, and have equivalent computation of  defined variables in both $S_1$ and $S_2$ and the input sequence variable, the I/O sequence variable,
$(S_1, m_1)\allowbreak \equiv_{H} (S_2, m_2)$ and
$\forall x \in \text{Def}(S) \cup
               \{{id}_I\}\,:\,
               (S_1, m_1) \equiv_{x} (S_2, m_2)$;

\item $S_1$ and $S_2$ produce the equivalent output sequence and the equivalent I/O sequence;

\item Loop counters of $S_1$ and $S_2$ are of values less than or equal to 1 but there are no reachable configurations
$(S_1, m_1(\text{loop}_c^{1_1}, \vals_{1_i}))$ from $(S_1, m_1(\vals_1))$,
$(S_2, m_2(\text{loop}_c^{2_1}, \vals_{2_i}))$ from $(S_2, \,\allowbreak m_2(\vals_2))$ where all of the following hold:
\begin{itemize}
\item The loop counter of $S_1$ and $S_2$ are of value 1,
$\text{loop}_c^{1_1}(n_1)\allowbreak = \text{loop}_c^{2_1}(n_2) = 1$.

\item Value stores $\vals_{1_1}$ and $\vals_{2_1}$ agree on values of  used variables in both $S_1$ and $S_2$ as well as the input sequence variable, and the I/O sequence variable,
$\forall x \in \, (\text{Use}(S_1) \cap \text{Use}(S_2)) \cup \{{id}_I\}\,:\,
\vals_{1_1}(x) = \vals_{2_1}(x)$.

\item Values of the I/O sequence variable in value stores $\vals_{1_1}$ and $\vals_{2_1}$ are equivalent,
    $\vals_{1_1}({id}_{IO}) \equiv \vals_{2_1}({id}_{IO})$;
\end{itemize}

\item There are reachable configuration
$(S_1, m_1(\text{loop}_c^{1_1}, \vals_{1_i}))$ from $(S_1, m_1(\vals_1))$,
$(S_2, m_2(\text{loop}_c^{2_1}, \vals_{2_i}))$ from $(S_2, \,\allowbreak m_2(\vals_2))$ where all of the following hold:
\begin{itemize}
\item The loop counter of $S_1$ and $S_2$ are of value 1,
$\text{loop}_c^{1_1}( n_1)\allowbreak = \text{loop}_c^{2_1}( n_2) = 1$.

\item Value stores $\vals_{1_1}$ and $\vals_{2_1}$ agree on values of  used variables in both $S_1$ and $S_2$ as well as the input sequence variable, and the I/O sequence variable,
$\forall x \in \, (\text{Use}(S_1) \cap \text{Use}(S_2)) \cup \{{id}_I\}\,:\,
\vals_{1_1}(x) = \vals_{2_1}(x)$.

\item Values of the I/O sequence variable in value stores $\vals_{1_1}$ and $\vals_{2_1}$ are equivalent,
    $\vals_{1_1}({id}_{IO}) \equiv \vals_{2_1}({id}_{IO})$;
\end{itemize}
\end{enumerate}

By definition, variables used in the predicate expression $e$ of $S_1$ and $S_2$ are  used in $S_1$ and $S_2$, $\text{Use}(e) \subseteq \text{Use}(S_1) \cap \text{Use}(S_2)$.
By assumption, value stores $\vals_1$ and $\vals_2$ agree on values of variables in $\text{Use}(e)$,
the predicate expression $e$ evaluates to the same value w.r.t value stores $\vals_1$ and $\vals_2$ by Lemma~\ref{lmm:expEvalSameTerm}.
There are three possibilities.
\begin{enumerate}
\item The evaluation of $e$ crashes,

\noindent$\mathcal{E}'\llbracket e\rrbracket \vals_1 =
 \mathcal{E}'\llbracket e\rrbracket \vals_2 = (\text{error}, v_{\mathfrak{of}})$.

The execution of $S_1$ continues as follows:
\begin{tabbing}
xx\=xx\=\kill
\>\>                   $(\text{while}_{\langle n_1\rangle}(e) \, \{S_1'\}, m_1(\vals_1))$\\
\>$->$\>               $(\text{while}_{\langle n_1\rangle}((\text{error}, v_{\mathfrak{of}})) \, \{S_1'\}, m_1(\vals_1))$\\
\>\>                   by the rule EEval'\\
\>$->$\>               $(\text{while}_{\langle n_1\rangle}(0) \, \{S_1'\}, m_1(1/\mathfrak{f}))$\\
\>\>                   by the ECrash rule  \\
\>{\kStepArrow [i] }\> $(\text{while}_{\langle n_1\rangle}(0) \, \{S_1'\}, m_1(1/\mathfrak{f}))$ for any $i>0$\\
\>\>                   by the Crash rule.
\end{tabbing}

Similarly, the execution of $S_2$ started from the state $m_2(\vals_2)$ crashes.
Therefore $S_1$ and $S_2$ terminate in the same way when started from $m_1$ and $m_2$ respectively.
Because $\vals_1({id}_{IO}) \equiv \vals_2({id}_{IO})$, the lemma holds.

\item The evaluation of $e$ reduces to zero,
$\mathcal{E}'\llbracket e\rrbracket \vals_1 =
 \mathcal{E}'\llbracket e\rrbracket \vals_2 = (0, v_{\mathfrak{of}})$.

The execution of $S_1$ continues as follows.
\begin{tabbing}
xx\=xx\=\kill
\>\>                   $(\text{while}_{\langle n_1\rangle}(e) \, \{S_1'\}, m_1(\vals_1))$\\
\>= \>                 $(\text{while}_{\langle n_1\rangle}((0, v_{\mathfrak{of}})) \, \{S_1'\}, m_1(\vals_1))$\\
\>\>                   by the rule EEval'\\
\>$->$\>               $(\text{while}_{\langle n_1\rangle}(0) \, \{S_1'\}, m_1(\vals_1))$\\
\>\>                   by the E-Oflow1 or E-Oflow2 rule  \\
\>$->$\>               $(\text{skip}, m_1(\vals_1))$ by the Wh-F rule.
\end{tabbing}

Similarly, the execution of $S_2$ gets to the configuration $(\text{skip}, m_2(\vals_2))$.
Loop counters of $S_1$ and $S_2$ are less than 1 and value stores agree on values of  used/defined variables in both $S_1$ and $S_2$ as well as the input sequence variable and the I/O sequence variable.

\item The evaluation of $e$ reduces to the same nonzero integer value,
$\mathcal{E}'\llbracket e\rrbracket \vals_1 =
 \mathcal{E}'\llbracket e\rrbracket \vals_2 = (v, v_{\mathfrak{of}})$ where $v \neq 0$.

Then the execution of $S_1$ proceeds as follows:
\begin{tabbing}
xx\=xx\=\kill
\>\>                   $(\text{while}_{\langle n_1\rangle}(e) \, \{S_1'\}, m_1(\vals_1))$\\
\>= \>                 $(\text{while}_{\langle n_1\rangle}((v, v_{\mathfrak{of}})) \, \{S_1'\}, m_1(\vals_1))$\\
\>\>                   by the rule EEval'\\
\>$->$\>               $(\text{while}_{\langle n_1\rangle}(v) \, \{S_1'\}, m_1(\vals_1))$\\
\>\>                   by the E-Oflow1 or E-Oflow2 rule  \\
\>$->$\>               $(S_1';\text{while}_{\langle n_1\rangle}(e) \, \{S_1'\}, m_1(\text{loop}_c^1[1/n_1], \vals_1))$\\
\>\>                   by the Wh-T rule.
\end{tabbing}

Similarly, the execution of $S_2$ proceeds to the configuration
$(S_2';\text{while}_{\langle n_2\rangle}(e) \, \{S_2'\}, m_2(\text{loop}_c^2 [1/n_2], \vals_2))$.

By the assumption, we show that $S_1'$ and $S_2'$ terminate in the same way and produce the equivalent I/O sequence when started in the state $m_1(\text{loop}_c^{1_1}, \vals_1)$ and $m_2(\text{loop}_c^{2_1}, \vals_2)$ respectively, and $S_1'$ and $S_2'$ have equivalent computation of variables  defined in both statement sequences if both terminate.
We need to show that all conditions are satisfied for the application of the assumption.
\begin{itemize}
\item Values of the I/O sequence variable in value stores $\vals_1$ and $\vals_2$ are equivalent,
    $\vals_1({id}_{IO}) \equiv \vals_2({id}_{IO})$;

The above two conditions are by assumption.

\item Value stores $\vals_1$ and $\vals_2$ agree on values of  used variables in $S_1'$ and $S_2'$ as well as the input sequence variable.

By definition, $\text{Use}(S_1')\allowbreak \subseteq \text{Use}(S_1)$.
So are the cases to $S_2'$ and $S_2$.
In addition, value stores $\vals_1$ and $\vals_2$ are not changed in the evaluation of the predicate expression $e$. The condition holds.
\end{itemize}
By  assumption, $S_1'$ and $S_2'$ terminate in the same way and produce the equivalent output sequence when started in states
$m_1(\text{loop}_c', \vals_1)$ and $m_2(\text{loop}_c', \vals_2)$.
In addition, $S_1'$ and $S_2'$ have equivalent computation of variables  used or defined in $S_1'$ and $S_2'$ when started in states
$m_1(\text{loop}_c', \vals_1)$ and $m_2(\text{loop}_c', \vals_2)$.

Then there are two cases.
\begin{enumerate}
\item $S_1'$ and $S_2'$ both do not terminate and produce the equivalent I/O sequence.

By Lemma~\ref{lmm:multiStepSeqExec}, $S_1';S_1$ and $S_2';S_2$ both do not terminate and produce the equivalent I/O sequence.

\item $S_1'$ and $S_2'$ both terminate and have equivalent computation of variables  defined in $S_1'$ and $S_2'$.

By assumption, $(S_1', m_1(\text{loop}_c', \vals_1)) ->*
                (\text{skip}, \,\allowbreak m_1'(\text{loop}_c'', \vals_1'))$;
               $(S_2', m_2(\text{loop}_c', \vals_2)) ->*
                (\text{skip},\allowbreak m_2'(\text{loop}_c'', \vals_2'))$
where $\forall x \in (\text{Def}(S_1') \cap \text{Def}(S_2')) \cup
                     \{{id}_I\},
\vals_1'(x) = \vals_2'(x)$.

By assumption,
$\text{Use}(S_1') = \text{Use}(S_2')$ and
$\text{Def}(S_1') = \text{Def}(S_2')$.
Then variables used in the predicate expression of $S_1$ and $S_2$ are either in variables  used or defined in both $S_1'$ and $S_2'$ or not.
Therefore value stores $\vals_2'$ and $\vals_1'$ agree on values of variables used in the expression $e$ and even variables  used or defined in $S_1$ and $S_2$.

By assumption, $S_1$ and $S_2$ produce the equivalent output sequence.
\end{enumerate}
\end{enumerate}

\noindent{Induction step on iterations}

The induction hypothesis (IH) is that, when $i\geq 1$, one of the following holds:
\begin{enumerate}
\item Loop counters for $S_1$ and $S_2$ are always less than $i$ if any is present,
$\forall m_1'(\text{loop}_c^{1'})\, m_2'(\text{loop}_c^{2'})\,:\,
(S_1, m_1(\text{loop}_c^1, \vals_1)) ->* (S_1'', m_1'(\text{loop}_c^{1'})),
\text{loop}_c^{1'}( n_1) < i,
(S_2, m_2(\text{loop}_c^2,\allowbreak \vals_2)) ->* (S_2'', m_2'(\text{loop}_c^{2'})),
\text{loop}_c^{2'}( n_2) < i$,
$S_1$ and $S_2$ terminate in the same way, and have equivalent computation of  defined variables in both $S_1$ and $S_2$ and the input sequence variable,
$(S_1, m_1)\allowbreak \equiv_{H} (S_2, m_2)$
and
$\forall x \in \text{Def}(S) \cup
               \{{id}_I\}\,:\,
               (S_1, m_1) \equiv_{x} (S_2, m_2)$;
$S_1$ and $S_2$ produce the equivalent I/O sequence;

\item The loop counter of $S_1$ and $S_2$ are of value less than or equal to $i$,
and there are no reachable configurations
$(S_1, m_1(\text{loop}_c^{1_i}, \vals_{1_i}))$ from $(S_1, m_1(\vals_1))$,
$(S_2, m_2(\text{loop}_c^{2_i}, \vals_{2_i}))$ from $(S_2, \,\allowbreak m_2(\vals_2))$ where all of the following hold:
\begin{itemize}
\item The loop counters of $S_1$ and $S_2$ are of value $i$,
$\text{loop}_c^{1_i}(n_1)\allowbreak = \text{loop}_c^{2_i}(n_2) = i$.

\item Value stores $\vals_{1_i}$ and $\vals_{2_i}$ agree on values of  used variables in both $S_1$ and $S_2$ as well as the input sequence variable,
$\forall x \in \, (\text{Use}(S_1) \cap \text{Use}(S_2)) \cup \{{id}_I\}\,:\,
\vals_{1_i}(x) = \vals_{2_i}(x)$.

\item Values of the I/O sequence variable in value stores $\vals_{1_i}$ and $\vals_{2_i}$,
    $\vals_{1_i}({id}_{IO}) \equiv \vals_{2_i}({id}_{IO})$;
\end{itemize}

\item There are reachable configurations
$(S_1, m_1(\text{loop}_c^{1_i}, \vals_{1_i}))$ from $(S_1, m_1(\vals_1))$,
$(S_2, m_2(\text{loop}_c^{2_i}, \vals_{2_i}))$ from $(S_2, \,\allowbreak m_2(\vals_2))$ where all of the following hold:
\begin{itemize}
\item The loop counter of $S_1$ and $S_2$ are of value $i$,
$\text{loop}_c^{1_i}(n_1)\allowbreak = \text{loop}_c^{2_i}(n_2) = i$.

\item Value stores $\vals_{1_i}$ and $\vals_{2_i}$ agree on values of  used variables in both $S_1$ and $S_2$ as well as the input sequence variable,
$\forall x \in \, (\text{Use}(S_1) \cap \text{Use}(S_2)) \cup \{{id}_I\}\,:\,
\vals_{1_i}(x) = \vals_{2_i}(x)$.

\item Values of the I/O sequence variable in value stores $\vals_{1_i}$ and $\vals_{2_i}$ are equivalent,
    $\vals_{1_i}({id}_{IO}) \equiv \vals_{2_i}({id}_{IO})$;
\end{itemize}
\end{enumerate}

Then we show that, when $i+1$, one of the following holds:
The induction hypothesis (IH) is that, when $i\geq 1$, one of the following holds:
\begin{enumerate}
\item Loop counters for $S_1$ and $S_2$ are always less than $i+1$ if any is present,
$\forall m_1'(\text{loop}_c^{1'})\, m_2'(\text{loop}_c^{2'})\,:\,
(S_1, m_1(\text{loop}_c^1, \vals_1)) ->* (S_1'', m_1'(\text{loop}_c^{1'})),
\text{loop}_c^{1'}( n_1) < i+1,
(S_2, m_2(\text{loop}_c^2,\allowbreak \vals_2)) ->* (S_2'', m_2'(\text{loop}_c^{2'})),
\text{loop}_c^{2'}( n_2) < i+1$,
$S_1$ and $S_2$ terminate in the same way, produce the equivalent I/O sequence, and have equivalent computation of  defined variables in both $S_1$ and $S_2$ and the input sequence variable,
$(S_1, m_1)\allowbreak \equiv_{H} (S_2, m_2)$ and
$(S_1, m_1) \equiv_{O} (S_2, m_2)$ and
$\forall x \in (\text{Def}(S) \cup
               \{{id}_I\}\,:\,
               (S_1, m_1) \equiv_{x} (S_2, m_2)$;
$S_1$ and $S_2$ produce the ``equivalent" I/O sequence variable;

\item The loop counter of $S_1$ and $S_2$ are of value less than or equal to $i+1$,
and there are no reachable configurations
$(S_1, m_1(\text{loop}_c^{1_{i+1}}, \vals_{1_{i+1}}))$ from $(S_1, m_1(\vals_1))$,
$(S_2, m_2(\text{loop}_c^{2_{i+1}},\,\allowbreak \vals_{2_{i+1}}))$ from $(S_2,  m_2(\vals_2))$ where all of the following hold:
\begin{itemize}
\item The loop counters of $S_1$ and $S_2$ are of value $i+1$,
$\text{loop}_c^{1_{i+1}}(n_1)\allowbreak = \text{loop}_c^{2_{i+1}}(n_2) = i+1$.

\item Value stores $\vals_{1_{i+1}}$ and $\vals_{2_{i+1}}$ agree on values of  used variables in both $S_1$ and $S_2$ as well as the input sequence variable,
$\forall x \in \, (\text{Use}(S_1) \cap \text{Use}(S_2)) \cup \{{id}_I\}\,:\,
\vals_{1_{i+1}}(x) = \vals_{2_{i+1}}(x)$.

\item Values of the I/O sequence variable in value stores $\vals_{1_{i+1}}$ and $\vals_{2_{i+1}}$  are equivalent,
    $\vals_{1_{i+1}}({id}_{IO}) \equiv \vals_{2_{i+1}}({id}_{IO})$;
\end{itemize}

\item There are reachable configurations
$(S_1, m_1(\text{loop}_c^{1_{i+1}}, \vals_{1_i}))$ from $(S_1, m_1(\vals_1))$,
$(S_2, m_2(\text{loop}_c^{2_{i+1}}, \vals_{2_i}))$ from $(S_2, \,\allowbreak m_2(\vals_2))$ where all of the following hold:
\begin{itemize}
\item The loop counter of $S_1$ and $S_2$ are of value $i$,
$\text{loop}_c^{1_{i+1}}( n_1)\allowbreak = \text{loop}_c^{2_{i+1}}( n_2) = i+1$.

\item Value stores $\vals_{1_{i+1}}$ and $\vals_{2_{i+1}}$ agree on values of  used variables in both $S_1$ and $S_2$ as well as the input sequence variable,
$\forall x \in \, (\text{Use}(S_1) \cap \text{Use}(S_2)) \cup \{{id}_{I}\}\,:\,
\vals_{1_{i+1}}(x) = \vals_{2_{i+1}}(x)$;

\item Values of the I/O sequence variable in value stores $\vals_{1_{i+1}}$ and $\vals_{2_{i+1}}$ are equivalent,
    $\vals_{1_{i+1}}({id}_{IO}) \equiv \vals_{2_{i+1}}({id}_{IO})$;
\end{itemize}
\end{enumerate}

By hypothesis IH, there is no configuration where loop counters of $S_1$ and $S_2$ are of value $i+1$ when any of the following holds:
\begin{enumerate}
\item Loop counters for $S_1$ and $S_2$ are always less than $i$ if any is present,
$\forall m_1'(\text{loop}_c^{1'})\, m_2'(\text{loop}_c^{2'})\,:\,
(S_1, m_1(\text{loop}_c^1, \vals_1)) ->* (S_1'', m_1'(\text{loop}_c^{1'})),
\text{loop}_c^{1'}( n_1) < i,
(S_2, m_2(\text{loop}_c^2,\allowbreak \vals_2)) ->* (S_2'', m_2'(\text{loop}_c^{2'})),
\text{loop}_c^{2'}( n_2) < i$,
$S_1$ and $S_2$ terminate in the same way, produce the equivalent I/O sequence, and have equivalent computation of  used/defined variables in both $S_1$ and $S_2$ and the input sequence variable, , the I/O sequence variable,
$(S_1, m_1)\allowbreak \equiv_{H} (S_2, m_2)$ and
$(S_1, m_1) \equiv_{O} (S_2, m_2)$ and
$\forall x \in (\text{Def}(S_1) \cap \text{Def}(S_2)) \cup
               \{{id}_I\}\,:\,
               (S_1, m_1) \equiv_{x} (S_2, m_2)$;
$S_1$ and $S_2$ produce , and the I/O sequence variable;

\item The loop counter of $S_1$ and $S_2$ are of value less than or equal to $i$,
and there are no reachable configurations
$(S_1, m_1(\text{loop}_c^{1_i}, \vals_{1_i}))$ from $(S_1, m_1(\vals_1))$,
$(S_2, m_2(\text{loop}_c^{2_i}, \vals_{2_i}))$ from $(S_2, \,\allowbreak m_2(\vals_2))$ where all of the following hold:
\begin{itemize}
\item The loop counters of $S_1$ and $S_2$ are of value $i$,
$\text{loop}_c^{1_i}( n_1)\allowbreak = \text{loop}_c^{2_i}( n_2) = i$.

\item Value stores $\vals_{1_i}$ and $\vals_{2_i}$ agree on values of  used variables in both $S_1$ and $S_2$ as well as the input sequence variable,
$\forall x \in \, (\text{Use}(S_1) \cap \text{Use}(S_2)) \cup \{{id}_I, , {id}_{IO}\}\,:\,
\vals_{1_i}(x) = \vals_{2_i}(x)$.

\item Values of the I/O sequence variable in value stores $\vals_{1_i}$ and $\vals_{2_i}$  are equivalent,
    $\vals_{1_i}({id}_{IO}) \equiv \vals_{2_i}({id}_{IO})$;
\end{itemize}
\end{enumerate}

When there are  reachable configurations
$(S_1, m_1(\text{loop}_c^{1_i}, \vals_{1_i}))$ from $(S_1, m_1(\vals_1))$,
$(S_2, m_2(\text{loop}_c^{2_i}, \vals_{2_i}))$ from $(S_2, \,\allowbreak m_2(\vals_2))$ where all of the following hold:
\begin{itemize}
\item The loop counter of $S_1$ and $S_2$ are of value $i$,
$\text{loop}_c^{1_i}(n_1)\allowbreak = \text{loop}_c^{2_i}(n_2) = i$.

\item Value stores $\vals_{1_i}$ and $\vals_{2_i}$ agree on values of  used variables in both $S_1$ and $S_2$ as well as the input sequence variable,
$\forall x \in \, (\text{Use}(S_1) \cap \text{Use}(S_2)) \cup \{{id}_I\}\,:\,
\vals_{1_i}(x) = \vals_{2_i}(x)$.

\item Values of the I/O sequence variable in value stores $\vals_{1_i}$ and $\vals_{2_i}$ are equivalent,
    $\vals_{1_i}({id}_{IO}) \equiv \vals_{2_i}({id}_{IO})$;
\end{itemize}

By similar argument in base case, we have one of the following holds:
\begin{enumerate}
\item Loop counters for $S_1$ and $S_2$ are always less than $i+1$ if any is present,
$\forall m_1'(\text{loop}_c^{1'})\, m_2'(\text{loop}_c^{2'})\,:\,
(S_1, m_1(\text{loop}_c^1, \vals_1)) ->* (S_1'', m_1'(\text{loop}_c^{1'})),
\text{loop}_c^{1'}(n_1) < i+1,
(S_2, m_2(\text{loop}_c^2,\allowbreak \vals_2)) ->* (S_2'', m_2'(\text{loop}_c^{2'})),
\text{loop}_c^{2'}(n_2) < i$,
$S_1$ and $S_2$ terminate in the same way, produce the equivalent I/O sequence, and have equivalent computation of  defined variables in both $S_1$ and $S_2$ and the input sequence variable,
$(S_1, m_1)\allowbreak \equiv_{H} (S_2, m_2)$
and
$\forall x \in (\text{Def}(S)) \cup
               \{{id}_I\}\,:\,
               (S_1, m_1) \equiv_{x} (S_2, m_2)$;

\item The loop counter of $S_1$ and $S_2$ are of value less than or equal to $i+1$,
and there are no reachable configurations
$(S_1, m_1(\text{loop}_c^{1_i}, \vals_{1_i}))$ from $(S_1, m_1(\vals_1))$,
$(S_2, m_2(\text{loop}_c^{2_i}, \vals_{2_i}))$ from $(S_2, \,\allowbreak m_2(\vals_2))$ where all of the following hold:
\begin{itemize}
\item The loop counters of $S_1$ and $S_2$ are of value $i$,
$\text{loop}_c^{1_{i+1}}(n_1)\allowbreak = \text{loop}_c^{2_{i+1}}(n_2) = i+1$.

\item Value stores $\vals_{1_{i+1}}$ and $\vals_{2_{i+1}}$ agree on values of  used variables in both $S_1$ and $S_2$ as well as the input sequence variable,
$\forall x \in \, (\text{Use}(S_1) \cap \text{Use}(S_2)) \cup \{{id}_I\}\,:\,
\vals_{1_{i+1}}(x) = \vals_{2_{i+1}}(x)$.

\item Values of the I/O sequence variable in value stores $\vals_{1_{i+1}}$ and $\vals_{2_{i+1}}$ are equivalent,
    $\vals_{1_{i+1}}({id}_{IO}) \equiv \vals_{2_{i+1}}({id}_{IO})$;
\end{itemize}

\item There are reachable configurations
$(S_1, m_1(\text{loop}_c^{1_{i+1}}, \vals_{1_{i+1}}))$ from $(S_1, m_1(\vals_1))$,
$(S_2, m_2(\text{loop}_c^{2_{i+1}}, \vals_{2_{i+1}}))$ from $(S_2, \,\allowbreak m_2(\vals_2))$ where all of the following hold:
\begin{itemize}
\item The loop counter of $S_1$ and $S_2$ are of value $i$,
$\text{loop}_c^{1_{i+1}}( n_1)\allowbreak = \text{loop}_c^{2_{i+1}}( n_2) = i$.

\item The loop counter of $S_1$ and $S_2$ are of value $i$,
$\text{loop}_c^{1_{i+1}}(, n_1)\allowbreak = \text{loop}_c^{2_{i+1}}(, n_2) = i$.

\item Value stores $\vals_{1_{i+1}}$ and $\vals_{2_{i+1}}$ agree on values of  used variables in both $S_1$ and $S_2$ as well as the input sequence variable,
$\forall x \in \, (\text{Use}(S_1) \cap \text{Use}(S_2)) \cup \{{id}_{I}\}\,:\,
\vals_{1_{i+1}}(x) = \vals_{2_{i+1}}(x)$;

\item Values of the I/O sequence variable in value stores $\vals_{1_{i+1}}$ and $\vals_{2_{i+1}}$  are equivalent,
    $\vals_{1_{i+1}}({id}_{IO}) \equiv \vals_{2_{i+1}}({id}_{IO})$;
\end{itemize}
\end{enumerate}
\end{proof} 

\subsection{Proof rule for missing variable initializations}

A kind of bugfix we call {\em missing-initialization} includes variable initialization for those in the imported variables relative to the I/O sequence variable in the old program.
\begin{figure}
\begin{small}
\begin{tabbing}
xxxxxx\=xxx\=xxx\=xxx\=xxxxxxxxxxxxxxxx\=xxxx\=xxx\=xxxxx\= \kill
\>1: \>\>\>                                    \>1': \> $b: = 2$ \> \\
\>2: \>{\bf If} $(a > 0)$ {\bf then}\>\>       \>2': \> {\bf If} $(a > 0)$ {\bf then} \> \\
\>3: \> \> $b := c + 1$\>                      \>3': \> \> $b := c + 1$               \> \\
\>4: \> {\bf output} $b + c$ \>\>              \>4': \> {\bf output} $b + c$          \> \\
\>\\
\> \> old\>\>                                         \>  \> new \>
\end{tabbing}
\end{small}
\caption{Missing initialization}\label{fig:missInitExample}
\end{figure}
Figure~\ref{fig:missInitExample} shows an example of missing-initializations. The initialization $b: = 2$ ensures the value used in ``{output} $b + c$" is not to be undefined.
In general, new variable initializations only affect rare buggy executions of the old program,
 where  there are uses of undefined imported variables relative to the I/O sequence variable in the program.
Because DSU is not starting in error state, we assume that, in the proof of backward compatibility,
there are no uses of variables with undefined variables in executions of the old program.

The following is the definition of the update class ``missing initializations".
\begin{definition}\label{def:missingVarInitFuncWide}
{\bf (Missing initializations)}
A statement sequence $S_2$ includes updates of missing initializations compared with a statement sequence $S_1$, written $S_2 {\approx}_{\text{Init}}^S S_1$, iff
$S_2 = S_{\text{Init}};S_1$ where $S_{\text{Init}}$ is a sequence of assignment statements of form
$``lval := v"$ and $\text{Def}(S_{\text{Init}}) \subseteq \text{Imp}(S_1, \{{id}_{IO}\})$;
\end{definition}
Though the bugfix in the update of missing initializations are not in rare execution in the first case in Definition~\ref{def:missingVarInitFuncWide}, the definition shows the basic form of bugfix clearly.

We show that two statement sequences terminate in the same way, produce the same output sequence, and have equivalent computation of defined variables in both programs in valid executions if there are updates of missing initializations between them.
\begin{lemma}\label{lmm:missingVarInitFuncWide}
Let $S_1$ and $S_2$ be two statement sequences respectively where there are updates of ``missing initializations" in $S_2$ compared with $S_1$, $S_2 {\approx}_{\text{Init}}^S S_1$.
If $S_1$ and $S_2$ start in states $m_1(\vals_1)$ and $m_2(\vals_2)$ respectively such that both of the following hold:
\begin{itemize}
\item Value stores $\vals_1$ and $\vals_2$ agree on values of variables  used in both $S_1$ and $S_2$ as well as the input sequence variable and  the I/O sequence variable,
$\forall id \in (\text{Use}(S_1)\cap \text{Use}(S_2)) \cup \{{id}_I,  {id}_{IO}\}\,:\,
 \vals_1(id) = \vals_2(id)$;

\item defined variables in $S_{\text{Init}}$ are of undefined value in value stores
$\vals_1, \vals_2$,
$\forall id \in \text{Def}(S_{\text{Init}})\,:\,
\vals_1(id) = \vals_2(id) = \text{Udf}\llbracket \tau\rrbracket$ where $\tau$ is the type of the variable $id$;

\item There are no use of variables with undefined values in the execution of $S_1$;

\item There are no crash in execution of $S_{\text{Init}}$;
\end{itemize}
then $S_1$ and $S_2$ terminate in the same way, produce the same output sequence, and when $S_1$ and $S_2$ both terminate, they have equivalent computation of  used variables and  defined variables in both $S_1$ and $S_2$ as well as the input sequence variable and the I/O sequence variable,
\begin{itemize}
\item $(S_1,   m_1) \equiv_{H} (S_2,   m_2)$;

\item $(S_1,   m_1) \equiv_{O} (S_2,   m_2)$;

\item $\forall x \in (\text{Def}(S_1)\cup \text{Def}(S_2)) \cup \{{id}_I,  {id}_{IO}\}\,:\,
 \allowbreak (S_1,   m_1) \equiv_{x} (S_2,   m_2)$;
\end{itemize}
\end{lemma}
\begin{proof}
By induction on the sum of the program size of $S_1$ and $S_2$, $\text{size}(S_1) + \text{size}(S_2)$.

\noindent{Base case}.
$S_1 = s$ and $S_2 = S_{\text{Init}};s$ where
$S_{\text{Init}} = ``lval := v"$ and $\text{Def}(lval) \in \text{Use}(s)$;

There are cases regarding $lval$ in $S_{\text{Init}}$.
\begin{enumerate}
\item $lval = id$.

Then the execution of $S_2$ proceeds as follows.
\begin{tabbing}
xx\=xx\=\kill
\>\>                   $(id := v;s, m_2(\vals_2))$\\
\>$->$\>               $(\text{skip};s, m_2(\vals_2[v/(id)]))$\\
\>\>                   by the rule As-Scl\\
\>$->$\>               $(s, m_2(\vals_2[v/(id)]))$  by the rule Seq.
\end{tabbing}

By assumption, $id \in \text{Use}(e)$. By assumption, the value of $id$ is undefined in value store $\vals_1$. Then there is no valid execution of $S_1$. Then it holds that, in valid executions of $S_1$, $S_1$ and $S_2$ terminate in the same way, produce the same output sequence, and have equivalent computation of  defined variables in both $S_1$ and $S_2$. Then this lemma holds.

\item $lval = id[n]$.

Then the execution of $S_2$ proceeds as follows.
\begin{tabbing}
xx\=xx\=\kill
\>\>                   $(id[n] := v;s, m_2(\vals_2))$\\
\>$->$\>               $(\text{skip};s, m_2(\vals_2[v/(id, n)]))$\\
\>\>                   by the rule As-Err\\
\>$->$\>               $(s, m_2(\vals_2[v/(id, n)]))$  by the rule Seq.
\end{tabbing}

By similar argument above, this lemma holds.

\item $lval = {id}_1[{id}_2]$.

Then the execution of $S_2$ proceeds as follows.
\begin{tabbing}
xx\=xx\=\kill
\>\>                   $({id}_1[{id}_2] := v;s, m_2(\vals_2))$\\
\>$->$\>               $({id}_1[v_1] := v;s, m_2(\vals_2[v/(id, n)]))$\\
\>\>                   by the rule Var\\
\>$->$\>               $(\text{skip};s, m_2(\vals_2[v/(id, v_1)]))$  by the rule As-Arr.\\
\>$->$\>               $(s, m_2(\vals_2[v/(id, v_1)]))$  by the rule Seq.\\
\end{tabbing}

By similar argument above, this lemma holds.
\end{enumerate}


\noindent{Induction step}.

The hypothesis is that this lemma holds when the sum $k$ of the program size of $S_1$ and $S_2$ are great than or equal to 3, $k\geq 3$.

We then show that this lemma holds when the sum of the program size of $S_1$ and $S_2$ is $k+1$.

$S_2 = S_{\text{Init}};S_1$ where $S_{\text{Init}}$ is a sequence of assignment statements and $\text{Def}(S_{\text{Init}}) \in \text{Imp}(S_1, \{{id}_{IO}\})$;

The proof is similar to that in the base case.
By assumption, the execution of $S_{\text{Init}}$ does not crash,
$(S_{\text{Init}}, m_2(\vals_2)) ->* (\text{skip}, m_2(\vals_2'))$ where
$\vals_2' = \vals_2[v_1/x_1]...[v_k/x_k]$ and $\forall 1\leq i\leq k\,:\, x_k \in \text{Def}(S_{\text{Init}})$.

By assumption, there are no use of variables with undefined values in the execution of $S_1$ by Theorem~\ref{thm:sameIOtheoremFuncWide} and Theorem~\ref{thm:mainTermSameWayLocal}, this lemma holds.
\end{proof}

\section{Related Work}\label{sec:relatedwork}
We discuss related work on DSU safety and program equivalence in order.

Existing studies on DSU safety could be roughly divided into high level studies and low level ones.
There are a few studies on  high level DSU safety.
In~\cite{KramerConsistentDSU}, Kramer and Magee defined the DSU correctness that the updated system shall ``operate as normal instead of progressing to an error state". This is covered by our requirement that hybrid executions conform to the old program's specification and our accommodation for bug fixes. Moreover, our backward compatibility includes I/O behavior, which is more concrete than the behavior in~\cite{KramerConsistentDSU}.
In~\cite{BloomCorrectDSU}, Bloom and Day proposed a DSU correctness which allows functionality extension that could not produce past behavior. This is probably because Bloom and Day considered updated environment.
On the contrary, we assume that the environment is not updated. In addition, we explicitly present the error state, which is not mentioned in~\cite{BloomCorrectDSU}.
Panzica La Manna~\cite{PanzicaLaManna_criteria} presented a high level correctness only considering scenario-based specifications for controller systems instead of general programs.

There are also studies on low level DSU safety.
Hayden et al.~\cite{hayden12dsucorrect} discussed DSU correctness and concluded that there is only client-oriented correctness.
Zhang et al.~\cite{ZhangCorrectness} asked the developers to ensure DSU correctness.
Magill et al.~\cite{Magill_automap} did ad-hoc program correlation without definitions of any correctness.
We consider that there is general principle of DSU safety.
The difference lies at the abstraction of the program behavior. We model program behavior by concrete I/O while
others~\cite{hayden12dsucorrect,Magill_automap,ZhangCorrectness} consider a general program behavior.

We next discuss existing work on program equivalence.
There is a rich literature on program equivalence and we compare our work only with most related work.
Our study of program equivalence is inspired by original work of Horwitz et al.~\cite{Horwitz88} on program
dependence graphs, but we take a much more formal approach and we consider terminating as well as non-terminating
programs with recurring I/O.
In~\cite{Godlin10}, Godlin and Strichman have a structured study of program equivalence similar to that of ours.
Godlin and Strichman~\cite{Godlin10} restricted the equivalence to corresponding functions and therefore weakens the applicability to general transformations affecting loops such as loop fission, loop fusion and loop invariant code motion.
However, our program equivalence allows loop optimizations such as loop fusion and loop fission.
Furthermore, our syntactic conditions imply more program point mapping because we allow corresponding program point in arbitrary nested statements and in the middle of program that does not include function call. 

\section{Conclusion}\label{sec:conclusion}

In this paper, we propose a formal and practical general definition of DSU correction based on I/O sequences, backward compatibility.
We devised a formal language and adapt the general definition of DSU correctness for executable programs based on our language.
Based on the adapted backward compatibility, we proposed syntactic conditions that help guarantee correct DSUs for both terminating and nonterminating executions.
In addition, we formalize typical program updates that are provably backward compatible, covering both new feature and bugfix.

In the future, we plan to identify more backward compatible update patterns by studying more open source programs.
Though it is dubious if open source programs' evolution history includes typical update patterns, open source programs are the most important source of widely-used programs for our study of DSU.
In addition, we plan to develop an algorithm for automatic state mapping based on our syntactic condition of program equivalence and definition of update classes.

\bibliographystyle{abbrvnat}

\bibliography{reference}

\appendix

\section{Type system}\label{appendix:typesys}
Figure~\ref{fig:typsys} shows an almost standard unsound and incomplete type system.
The type system is unsound because of three reasons,
(a) the possible value mismatch due to the subtype rule from hte type Int to Long,
(b) the implicit subtype between enumeration types and the type Long allowed by our semantics and
(c) the possible array index out of bound.
The type system is incomplete due to the parameterized ``other" expressions.
The notation Dom($\Gamma$) borrowed from Cardelli~\cite{cardelli1996type} in rules Tvar1, Tvar2, Tlabels an Tfundecl refers to the domain of the typing environment $\Gamma$, which are identifiers bound to a type in $\Gamma$.
\begin{figure}[t!]

\normalsize{
\begin{tabular}{p{1cm}p{1cm}ll}
$\boxed{\Gamma \vdash \diamond}$ \\
\end{tabular}
}

\scriptsize{
\begin{tabular}{l}
\inference[{\bf TInit}]
{}
{\Gamma \vdash \diamond}
\\
\\
\end{tabular}
}

\scriptsize{
\begin{tabular}{p{2.5cm}l}
\multicolumn{2}{l}{$\boxed{\Gamma -> \Gamma'}$}  \\
\end{tabular}
}

\scriptsize{
\begin{tabular}{ll}
{\bf (Tvar1)} & {\bf (Tlabels)}\\

\inference[]
{  \Gamma \vdash \diamond\\
    V = V', \tau \, id   &    id \notin \text{Dom}(\Gamma) \\
}
{
    \Gamma, id : \tau \vdash \diamond
} &

\inference[]
{   \Gamma \vdash \diamond & k\geq 1 & id \notin \text{Dom}(\Gamma)\\
    EN = EN', \text{enum }id \{l_1,...,l_k\}
}
{
    \Gamma, id : \{l_1,...,l_k\} \vdash \diamond
}
\end{tabular}

\begin{tabular}{l}
{\bf (Tprompt)} \\

\inference[]
{  \Gamma \vdash \diamond\\
    Pmpt = \{l_1 : n_1, \ldots, l_k : n_k\}   &    \text{pmpt} \notin \text{Dom}(\Gamma) \\
}
{
    \Gamma, \text{pmpt} : \{l_1 : n_1, \ldots, l_k : n_k\} \vdash \diamond
}
\end{tabular}

\begin{tabular}{ll}
{\bf (Tvar2)} \\

\inference[]
{  \Gamma \vdash \diamond & id \notin \text{Dom}(\Gamma)\\
   V = V', \tau \, id[n]     & n > 0 \\
}
{
   \Gamma, id : \text{array}(\tau, n) \vdash \diamond
}
\\
\\
\end{tabular}

\normalsize{
\begin{tabular}{p{1cm}p{1cm}ll}
{$\boxed{\Gamma \vdash \tau}$}  \\
\end{tabular}
}

\scriptsize{
\begin{tabular}{p{1cm}ll}
{\bf (Tint)} & {\bf (Tlong)} & {\bf (Tenum)}\\
\inference[]
{\Gamma \vdash \diamond}
{\Gamma \vdash \text{Int}}

&

\inference[]
{\Gamma \vdash \diamond}
{\Gamma \vdash \text{Long}}

&

\inference[]
{\Gamma \vdash id : \{l_1,...,l_k\}}
{\Gamma \vdash \text{enum} \; id}
\\
\\
\end{tabular}
}
}

\normalsize{
\begin{tabular}{lll}
\multicolumn{3}{l}{$\boxed{\Gamma \vdash e : \tau}$} \\
\end{tabular}
}

\scriptsize{
\begin{tabular}{lll}
{\bf (Topnd)} & {\bf (Tequiv)} & {\bf (TSub)}     \\
\inference
{
}
{
    \Gamma, id:\tau \vdash id:\tau
}
&
\inference
{
    \Gamma \vdash id' : \{l_1,..., l, ..., l_k\} \\
    \Gamma \vdash id : \text{enum} \, id' \\
}
{
    \Gamma \vdash (id == l):\text{Long}
}
&
\inference
{   \Gamma \vdash e:\text{Int}   }
{   \Gamma \vdash e:\text{Long}  }
\\
\\
\end{tabular}

\begin{tabular}{p{1.5cm}p{1.5cm}l}
{\bf (Tarray1)} & {\bf (Tarray2)}  \\ 
\inference
{
    \Gamma \vdash id : \text{array}(\tau, n) \\
    \Gamma \vdash id': \text{Long}
}
{
    \Gamma \vdash id[id'] : \tau
}
&
\inference
{
 \Gamma \vdash  id : \text{array}(\tau, n) \\
 1\leq k \leq n
}
{
 \Gamma \vdash id[k] : \tau
}
\\
\end{tabular}
}

\scriptsize{
\begin{tabular}{lll}
\multicolumn{3}{l}{$\boxed{\Gamma \vdash S}$} \\
\end{tabular}
}

\scriptsize{
\begin{tabular}{p{0.8cm}p{.8cm}p{.8cm}l}%
{\bf (Tassign)} & {\bf (Tinput)} & {\bf (Toutput)} & {\bf (Tseq)}\\
\inference[]
{    \Gamma \vdash lval : \tau \\
     \Gamma \vdash e : \tau }
{    \Gamma \vdash lval := e }
     &
\inference[]
{   \Gamma \vdash id : \tau  \\
    \tau \neq \text{pmpt}}
{   \Gamma \vdash \text{input} \; id}
     &
\inference[]
{   \Gamma \vdash e : \tau    }
{   \Gamma \vdash \text{output} \, e }
     &
\inference[]
{   \Gamma \vdash S_1 \\
    \Gamma \vdash S_2}
{   \Gamma \vdash S_1; S_2 }
\\
\\
\end{tabular}

\begin{tabular}{ll}
{\bf (Tif)}   & {\bf (Twhile)}\\
\inference[]
{   \Gamma \vdash e : \text{Long}\\
 \Gamma \vdash S_1 & \Gamma \vdash S_2}
{   \Gamma \vdash \text{If}(e) \, \text{then} \, \{S_1\} \, \text{else} \, \{S_2\}}
         &
\inference[]
{   \Gamma \vdash e : \text{Long}, \; \Gamma \vdash S}
{   \Gamma \vdash \, \text{while}(e) \{S\}}
    \\
    \\
\end{tabular}
}

\scriptsize{
\begin{tabular}{l}
{$\boxed{\mbox{$\Gamma\vdash P$}}$} \\
\end{tabular}
}

\scriptsize{
\begin{tabular}{l}
{\bf (Tprog)}\\
\inference[]
{   
    {Pmpt} =  \{l_1 : n_1,...,l_k : n_k\} \\
    {EN} = \text{enum} \, {id}_1 \{l_1,...,l_r\},..., \text{enum} \, {id}_k \{l_1',...,l_r'\} \\
    \Gamma \vdash \text{enum} \, {id}_i, 1 \leq i \leq k  &  V = {\tau}_1' \, {id}_1', ..., {\tau}_k' {id}_k'[n] \\
    \Gamma \vdash {id}_j' : {\tau}_j', 1\leq j \leq k'-1 &\Gamma \vdash {id}_k' : \text{array}({\tau}_k', n) \\
    \Gamma \vdash S_{entry}\\
}
{
    \Gamma \vdash Pmpt; EN; V; S_{entry}
}\\
\\
\hline
\end{tabular}
}

\caption{Typing rules}\label{fig:typsys}

\end{figure}

\section{Syntactic definitions}\label{appendix:definitions}
The syntax-directed definitions listed below make our argument independent of existing program analysis partially.
\begin{definition}\label{def:indexUseInLval}
{\em (Idx$(lval)$)} The used variables in index of a left value $lval$, written Idx$(lval)$, are listed as follows:
\begin{enumerate}
\item $\text{Idx}({id})$ = $\emptyset$;

\item $\text{Idx}({id}[n])$ = $\emptyset$;

\item $\text{Idx}({id}_1[{id}_2])$ = \{${id}_2$\};
\end{enumerate}
\end{definition}

\begin{definition}\label{def:baseLval}
{\em (Base$(lval)$)} The base of a left value $lval$, written Base$(lval)$, is listed as follows:
\begin{enumerate}
\item $\text{Base}({id})$ = $\{id\}$;

\item $\text{Base}({id}[n])$ = $\{id\}$;

\item $\text{Base}({id}_1[{id}_2])$ = \{${id}_1$\};
\end{enumerate}
\end{definition}
%
\begin{definition}\label{def:dusevarsExpr}
{\em (Use$(e)$)} The set of used variables in an expression $e$, written Use$(e)$, are listed as follows:
\begin{enumerate}
\item $\text{Use}({lval})$ = $\text{Base}(lval) \, \cup \, \text{Idx}(lval)$;

\item $\text{Use}(id == l) = \{id\}$;

\item $\text{Use}(\text{other}) = \text{Use}(other)$ where function $\text{Use}: \text{other} -> \{id\}$ is parameterized;
\end{enumerate}
\end{definition}

\begin{definition}\label{usevars}
{\em(Use$(S)$)} The used variables in a sequence of statements $S$, written Use$(S)$, are listed as follows:
\begin{enumerate}
\item Use$(skip) = \emptyset$;

\item Use$({lval} := e)$ = $\text{Use}(e) \cup \text{Idx}(lval)$;

\item Use$(\text{output }{e})$ = Use$(e) \cup \{{id}_{IO}\}$;

\item Use$(\text{input }id)$ = $\{{id}_I, {id}_{IO}\}$;

\item Use$(\text{If }(e) \text{ then }\{S_t\} \text{ else }\{S_f\})$ = Use$(e) \cup \text{Use}(S_t) \cup \text{Use}(S_f)$;

\item Use$(\text{while}_{\langle n\rangle}(e) \{S'\})$ = Use$(e) \cup \text{Use}(S')$;

\item For $k>0$, Use$(s_1;...;s_{k+1})$ = $\text{Use}(s_1;...;s_k) \cup \text{Use}(s_{k+1})$;
\end{enumerate}
\end{definition}

\begin{definition}\label{defvars}
{\em (Def$(S)$)} The defined variables in a sequence of statements $S$, written Def$(S)$, are listed as follows:
\begin{enumerate}
\item Def$(skip)$ = $\emptyset$;

\item Def$({id} := e)$ = \{$id$\};

\item Def$(\text{input }id)$ = \{${id}_I, {id}_{IO}, id$\};

\item Def$(\text{output }{e})$ = \{${id}_{IO}$\};

\item Def$(\text{If }(e) \text{ then }\{S_t\} \text{ else }\{S_f\})$ = Def$(S_t)$ $\cup$ Def$(S_f)$;

\item Def$(\text{while}_{\langle n\rangle}(e) \{S\})$ = Def$(S)$;

\item For $k>0$, Def$(s_1;...;s_{k+1})$ = $\text{Def}(s_1;...;s_k) \cup \text{Def}(s_{k+1})$;
\end{enumerate}
\end{definition}

%
\begin{definition}\label{def:stmtInclusion}
($s \in S$) We say a statement $s$ is in a sequence of statements $S$ of a program $P$, written $s \in S$, if one of the following holds:
\begin{enumerate}
\item $S = s$;


\item If $S$ = ``If(e) then \{$S_t$\} else \{$S_f$\}", $(s \in S_t) \vee (s \in S_f)$;

\item If $S$ = ``while(e) \{$S'$\}", $s \in S'$;

\item For $k>0$, if $S = s_1;...;s_{k};s_{k+1}$, $(s \in s_{k+1}) \vee (s \in s_1;...;s_{k})$;
\end{enumerate}
\end{definition}
We write $s \notin S$ if $s \in S$ does not hold.

We show the definition of program size, which is based of our induction proof.
\begin{definition}\label{def:programsize}
$(\text{size}(S))$ The program size of a statement sequence $S$, written $\text{size}(S)$, is listed as follows:
\begin{enumerate}
\item $\text{size}(``\text{skip}") = \text{size}(``id := e") = \text{size}(``{id}_1 := \,\text{call}\, {id}_2(e^{*})") \allowbreak= \text{size}(``\text{input }id") = \text{size}(``\text{output }e") =   1$;

\item $\text{size}(``\text{If}(e) \text{ then }\{S_t\} \text{ else }\{S_f\}") =  1 + \text{size}(S_t) + \text{size}(S_f)$;

\item $\text{size}(``\text{while}(e)\, \{S'\}") = 1 +  \text{size}(S')$;

\item For $k>0$, $\text{size}(s_1;...;s_{k}$  = $\sum\limits_{i=1}^k\text{size}(s_i)$;
\end{enumerate}
\end{definition} 
\section{Properties of imported variables}
\begin{lemma}\label{lmm:ImpPrefixLemma}
$\text{Imp}(S_1;S_2, X) = \text{Imp}(S_1, \text{Imp}(S_2, X))$.
\end{lemma}
\begin{proof}
Let statement sequence $S_2$ = $s_1;s_2;...;s_k$ for some $k > 0$. The proof is by induction on $k$.
\end{proof}
%

\begin{corollary}\label{lmm:dupStmtPrefixLemma}
$\forall i \in \mathbb{Z}_{+}, \text{Imp}(S^{i+1}, X)$ = $\text{Imp}(S, \text{Imp}(S^{\mathit{i}}, X))$.
\end{corollary}
This is by lemma~\ref{lmm:ImpPrefixLemma}.

\begin{lemma}\label{lmm:impVarUnionLemma}
$\text{Imp}(S, A \cup B)$ = $\text{Imp}(S, A) \cup \text{Imp}(S, B)$.
\end{lemma}
\begin{proof}
By structural induction on abstract syntax of statement sequence $S$.
\end{proof}

\begin{lemma}\label{lmm:boundOfLoopDupForImpVar}
For statement $s$ = ``while(e)\{$S$\}" and a set of finite number of variables $X$ such that $X \cap \text{Def }($s$) \; \neq \; \emptyset$,
there is $\beta > 0$ such that $\bigcup_{0\leq i \leq (\beta+1)}$ Imp $(S^i,X)$ $\subseteq$ $\bigcup_{1\leq j\leq \beta}$ Imp$(S^j, X)$.
\end{lemma}
\begin{proof}
By contradiction against the fact that is finite number of variables redefined in statement $s$.
\end{proof}

%
%
%
%
%
%

\section{Properties of expression evaluation}
We wrap the two properties of expression evaluation, which is based on the two properties of ``other" expression evaluation. In the following, we use the notation $\mathcal{E}'$ to expand the domain of the expression meaning function
$\mathcal{E}'\,:\, e -> \vals -> (v_\text{error}, \{0, 1\})$.
\begin{lemma}\label{lmm:expEvalSameVal}
If every variable in $\text{Use}(e)$ of an expression $e$ has the same value w.r.t two value stores, the expression $e$ evaluates to same value against the two value stores,
$(\forall x \in \, \text{Use}(e)\,:\, \vals_1(x) = \vals_2(x)) => (\mathcal{E}'\llbracket e\rrbracket\vals_1 = \mathcal{E}'\llbracket e\rrbracket\vals_2)$.
\end{lemma}
\begin{proof}
The proof is a case analysis of the expression $e$.
\begin{enumerate}
\item $e = lval$;

There are further cases regarding $lval$.
\begin{enumerate}
\item $lval = id$;

By definition, $\text{Use}(e) = \{id\}$.
Besides, there is no integer overflow in both evaluations.
The lemma holds trivially.

\item $lval = id[n]$;

By definition, $\text{Use}(e) = \{id\}$.
Because the array has fixed size, by assumption,
$\vals_1(id, n) = \vals_2(id, n)$ or
$(id, n, *) \notin \vals_1, (id, n, *) \notin \vals_2$.
Besides, there is no integer overflow in both evaluations.
The lemma holds.

\item $lval = {id}_1[{id}_2]$;

By definition, $\text{Use}(e) = \{{id}_1, {id}_2\}$.
By assumption, $\vals_1({id}_2) = \vals_2({id}_2) = n$
By similar argument to the case $lval = id[n]$, the lemma holds.
\end{enumerate}

\item $e = ``id == l"$;

By definition, $\text{Use}(e) = \{id\}$.
W.l.o.g, $id$ is a global variable.
By assumption, $\vals_1(id) = \vals_2(id) = l'$.
If $l' = l$, by rule Eq-T, $(l'==l, m(\vals)) -> (1, m)$.
If $l' \neq l$, by rule Eq-F, $(l'==l, m(\vals)) -> (0, m)$.
Besides, there is no integer overflow in both evaluations.
 The lemma holds.

\item $e = \text{other}$;

By definition, $\text{Use}(e) = \text{Use}(e)$.
By assumption, $\forall x \in \text{Use}(e)\,:\, \vals_1(x) = \vals_2(x) \vee \vals_1(x) = \vals_2(x)$.
The lemma holds by parameterized expression meaning function for ``other" expression.
\end{enumerate}
\end{proof}

\begin{lemma}\label{lmm:expEvalSameTerm}
If every variable in $\text{Err}(e)$ of an expression $e$ has same value w.r.t two pairs of (block, value store),
$\forall x \in  \, \text{Err}(e)\,:\,
  \vals_1(x) = \vals_2(x)$
then one of the following holds:
\begin{enumerate}
\item the expression evaluates to crash against the two value stores,
$(\mathcal{E}'\llbracket e\rrbracket\vals_1 = (\text{error}, v_{\mathfrak{of}})) \wedge
 (\mathcal{E}'\llbracket e\rrbracket\vals_2 \allowbreak= (\text{error}, v_{\mathfrak{of}}))$;

\item the expression evaluates to no crash against the two pairs of (block, value store)
$(\mathcal{E}'\llbracket e\rrbracket\vals_1 \neq (\text{error}, v_{\mathfrak{of}}^1)) \wedge
 (\mathcal{E}'\llbracket e\rrbracket\vals_2 \neq (\text{error}, v_{\mathfrak{of}}^2))$.
\end{enumerate}
\end{lemma}
\begin{proof}
The proof is a case analysis of the expression $e$.
\begin{enumerate}
\item $e = lval$;

There are further cases regarding $lval$.
\begin{enumerate}
\item $lval = id$;

By definition, $\text{Err}(e) = \text{Idx}(e) = \emptyset$. By our semantic, the evaluation of $id$ never crash.
Besides, there is no integer overflow in both evaluations.
 The lemma holds.

\item $lval = id[n]$;

By definition, $\text{Err}(e) = \text{Idx}(e) = \emptyset$.
Because the array ${id}_1$ has a fixed array size, by assumption, either
$(({id}_1, n, v_1) \in \vals_1) \wedge (({id}_1, n, v_2) \in \vals_2)$ or
$(({id}_1, n, v_1) \notin \vals_1) \wedge (({id}_1, n, v_2) \notin \vals_2)$.
Besides, there is no integer overflow in both evaluations.
The lemma holds.

\item $lval = {id}_1[{id}_2]$

By definition, $\text{Err}(e) = \text{Idx}(e) = \{{id}_2\}$.
By assumption, $\vals_1({id}_2) = \vals_2({id}_2) = n$ or
$\vals_1({id}_2) = \vals_2({id}_2) = n$.
By similar argument to the case $lval = id[n]$, the lemma holds.
\end{enumerate}

\item $e = ``id == l"$;

By definition, $\text{Err}(e) = \emptyset$.
W.l.o.g, $id$ is a global variable.
Let $\vals_1(id) = l_1, \vals_2(id) = l_2$.
W.l.o.g., $l_1 = l$ and $l_2 \neq l$, by rule Eq-T,
$(l_1==l, m(\vals)) -> (1, m)$ and, by rule Eq-F,
$(l_2==l, m(\vals)) -> (0, m)$.
Besides, there is no integer overflow in both evaluations.
The lemma holds.

\item $e = \text{other}$;

By definition, $\text{Err}(e) = \text{Err}(e)$.
By assumption, $\forall x \in \text{Err}(e)\,:\, \vals_1(x) = \vals_2(x)$.
The lemma holds by the property of parameterized expression meaning function for ``other" expression.
\end{enumerate}
\end{proof}

With respect to Lemma~\ref{lmm:expEvalSameVal} and Lemma~\ref{lmm:expEvalSameTerm},
we extend semantic rule for expression evaluation as follows.
\begin{figure}[t!]

\begin{tabular} {l}
{\boxed{\mbox{$(r, {\state}) -> (r', {\state}')$}}}
\end{tabular}

\begin{center}
\scriptsize{
\begin{tabular}{l}
$\mathcal{E}': e -> \vals -> (v_\text{error} \times \{0, 1\})$ 
\\
\\

\inference[EEval']
{\mathfrak{f} = 0}
{(e, \state(\mathfrak{f}, \vals)) -> (\mathcal{E}'\llbracket e\rrbracket\vals, \state)}
\\
\\
\hline
\end{tabular}
}
\end{center}
\caption{An extended SOS rule for expressions}\label{fig:sosrulesExprExt}
\end{figure}

\section{Properties of remaining execution}
We assume that crash flag $\mathfrak{f}$ = 0 in given execution state $m(\mathfrak{f})$.
\begin{lemma}\label{lmm:oneStepSeqExec}
$({S_1}, m) -> ({S_1'}, m') => ({S_1;S_2}, m) -> ({S_1';S_2}, m')$.
\end{lemma}
\begin{proof}
The proof is by structural induction on abstract syntax of $S_1$.

{\bf Case 1}. $S_1$ = ``skip".

By rule Seq,
$({skip;S_2}, m) -> ({S_2}, m)$ where $m = m'$.

{\bf Case 2}.
$S_1 = \; ``id := e"$.

There are two subcases.

\hspace{0.5ex} {\bf Case 2.1}.
$(e, m) ->* (v, m)$ for some value $v$.

By rule Assign,

$(id := v, m) -> (\text{skip}, m(\vals[v/x]))$.

Then, by contextual (semantic) rule,

$(id := v;S_2, m) -> (\text{skip};S_2, m(\vals[v/x]))$.

\hspace{0.5ex} {\bf Case 2.2}.
$(e, m) ->* (e', m(1/\mathfrak{f}))$ for some expression $e'$.

Then, by rule crash,

$(id := e', m(1/\mathfrak{f})) -> (id := e', m(1/\mathfrak{f}))$.

Then, by contextual rule,

$(id := e';S_2, m(1/\mathfrak{f})) -> (id := e';S_2, m(1/\mathfrak{f}))$.

{\bf Case 3}.
$S_1$= ``output $e$"

{\bf Case 4}.
$S_1$= ``input $id$"

By similar argument in Case 2, the lemma holds for case 3 and 4.

{\bf Case 5}.
$S_1$= ``If $(e)$ then \{$S_t$\} else \{$S_f$\}".

\hspace{0.5ex} {\bf Case 5.1}.
W.l.o.g., expression $e$ in predicate of $S_1$ evaluates to nonzero in state $m$, written $(e, m) ->* (0, m)$.

By rule If-T,
$({\text{If }(0) \text{ then } \{S_t\} \text{ else }\{S_f\}}, m) -> ({S_t}, m)$.

By contextual (semantic) rule,

$({\text{If }(0) \text{ then } \{S_t\} \text{ else }\{S_f\};S_2}, m) -> ({S_t;S_2}, m)$.

\hspace{0.5ex} {\bf Case 5.2}.
Evaluation of expression $e$ in predicate of $S_1$ crashes, written $(e, m) ->* (e', m(1/\mathfrak{f}))$.

Then, by rule crash,

$({\text{If }(e') \text{ then } \{S_t\} \text{ else }\{S_f\}}, m(1/\mathfrak{f})) ->$

$({\text{If }(e') \text{ then } \{S_t\} \text{ else }\{S_f\}}, m(1/\mathfrak{f}))$.

Then, by contextual rule,

$({\text{If }(e') \text{ then } \{S_t\} \text{ else }\{S_f\}};S_2, m(1/\mathfrak{f})) -> $

$({\text{If }(e') \text{ then } \{S_t\} \text{ else }\{S_f\}};S_2, m(1/\mathfrak{f}))$.

{\bf Case 6}.
$S_1$ = ``$\text{while}_{\langle n\rangle} \; (e) \; \{S\}$".

\hspace{0.5ex}{\bf Case 6.1}
When expression $e$ in predicate of $S_1$ evaluates to nonzero value, written $(e, m) ->* (v, m)$ for some $v \neq 0$,
then, by rule Wh-T,

$(\text{while}_{\langle n\rangle} \; (e) \{S\}, m) -> (S; \text{while}_{\langle n\rangle} \; (e) \{S\}, m(m_c[(k+1)/n]))$ for some nonnegative integer $k$.

Then, by contextual rule,

$({\text{while}_{\langle n\rangle} \; (e) \; \{S\};S_2}, m) ->$

$({S; \text{while }_{\langle n\rangle} \; (e) \; \{S\};S_2}, m(m_c[(k+1)/n]))$.

\hspace{0.5ex}{\bf Case 6.2}
When expression $e$ in predicate of $S_1$ evaluates to zero, written $(e, m) ->* (0, m)$,
then, by rule Wh-F,

$(\text{while}_{\langle n\rangle} \;(e) \{S\}, m) -> (\text{skip}, m(m_c[0/n]))$.

By contextual rule,

$({\text{while}_{\langle n\rangle} \; (e) \; \{S\};S_2}, m) -> (\text{skip};{S_2}, m(m_c[0/n]))$.

\hspace{0.5ex}{\bf Case 6.3}
Evaluation of expression $e$ in predicate of $S_1$ crashes, written $(e, m) ->* (e', m)$.

By rule crash,

$({\text{while}_{\langle n\rangle} \;(e') \{S\}}, m(1/\mathfrak{f})) ->
 ({\text{while}_{\langle n\rangle} \;(e') \{S\}}, m(1/\mathfrak{f}))$.

Then, by contextual rule,

$({\text{while}_{\langle n\rangle} \;(e') \{S\}};S_2, m(1/\mathfrak{f})) -> $

$ ({\text{while}_{\langle n\rangle} \;(e') \{S\}};S_2, m(1/\mathfrak{f}))$.
\end{proof}

\begin{lemma}\label{lmm:multiStepSeqExec}
$(S_1, m) ->* (S_1', m') => ({S_1;S_2}, m) ->* (S_1';S_2, m')$.
\end{lemma}
\begin{proof}
By induction on number of steps $k$ in execution
$({S_1}, m)$ {\kStepArrow [k] } $({S_1'}, m')$.

{\bf Base case}. $k$ = 0 and 1.

By definition,

$({S_1}, m)$ {\kStepArrow [0] } $({S_1}, m)$, and $({S_1;S_2}, m)$ {\kStepArrow [0] } $({S_1;S_2}, m)$.

By lemma~\ref{lmm:oneStepSeqExec},

$({S_1}, m) -> ({S_1'}, m') => ({S_1;S_2}, m) -> ({S_1';S_2}, m')$.

{\bf Induction step}.

The induction hypothesis IH is that, for $k\geq1$,

$({S_1}, m)$ {\kStepArrow [k] } $({S_1'}, m') => ({S_1;S_2}, m)$ {\kStepArrow [k] } $({S_1';S_2}, m')$.

Then we show that,

$({S_1}, m)$ {\kStepArrow [k+1] } $({S_1'}, m') => ({S_1;S_2}, m)$ {\kStepArrow [k+1] } $({S_1';S_2}, m')$.

We decompose the k+1 step execution into

$({S_1}, m) -> ({S_1''}, m'')$ {\kStepArrow [k] } $({S_1'}, m')$.

By lemma~\ref{lmm:oneStepSeqExec},

$({S_1;S_2}, m) -> ({S_1'';S_2}, m'')$.

Next, by IH,

$({S_1'';S_2}, m'')$ {\kStepArrow [k] } $({S_1';S_2}, m')$.
\end{proof}

\begin{corollary}\label{coro:termSeq}
$({S_1}, m) ->* (\text{skip}, m') => ({S_1;S_2}, m) ->* ({S_2}, m')$.
\end{corollary}
\begin{proof}
By lemma~\ref{lmm:multiStepSeqExec},

$({S_1}, m) ->* (\text{skip}, m') => ({S_1;S_2}, m) ->* ({\text{skip};S_2}, m')$.

Then, by rule Seq,

$({\text{skip};S_2}, m') -> ({S_2}, m')$.
\end{proof}

\begin{lemma}\label{lmm:oneStepStmtExclusion}
If one statement $s$ is not in $S$, then, after one step of execution $({S}, m) -> ({S'}, m')$, $s$ is not in the $S'$,
 $(s\notin S) \wedge (({S}, m) -> ({S'}, m')) => (s \notin S')$.
\end{lemma}
\begin{proof}
By induction on abstract syntax of $S$.
\end{proof}

\begin{lemma}\label{lmm:multiStepStmtExclusion}
If one statement $s$ is not in $S$, then, after the execution $({S}, m) ->* ({S'}, m')$, $s$ is not in the $S'$,
 $(s\notin S) \wedge (({S}, m) ->* ({S'}, m')) => (s \notin S')$.
\end{lemma}
\begin{proof}
By induction on the number $k$ of the steps in the execution $({S}, m)$ {\kStepArrow [k] } $({S'}, m')$.
\end{proof}

\begin{lemma}\label{lmm:defExclusion}
If a variable $x$ is not defined in a statement sequence $S$, then, after one step execution of $S$, the value of $x$ is not redefined,
$(x \notin \text{Def}(S)) \land ((S, m(\vals)) -> ({S'}, m'(\vals'))) =>
(x \notin \text{Def}({S'})) \land (\vals'(x) = \vals(x))$
\end{lemma}
\begin{proof}
By structural induction on abstract syntax of statement sequence $S$, we show the lemma holds.

{\bf Case 1}.
$S$ = ``$id := e$". 

By definition, Def($S$) = \{$id$\}.
Then $id \neq x$ by condition that $x \notin \text{ Def }(S)$.

Then there are two subcases.

\hspace{0.5ex} {\bf Case 1.1}
Expression $e$ evaluates to some value $v$, written $(e, m) ->* (v, m)$.

Then, by rule Assign,
$(S, m(\vals)) -> (\text{skip}, m(\vals[v/id]))$ where $m' = m(\vals[v/id])$.

Hence, $\vals'(x) = \vals(x)$. Besides, $x \notin \text{Def (skip)}$ by definition.

\hspace{0.5ex} {\bf Case 1.2}
Evaluation of expression $e$ crashes, written $(e, m) ->* (e', m(1/\mathfrak{f}))$.

Then, by rule crash,

$(id := e', m(1/\mathfrak{f},\vals)) -> (id := e', m(1/\mathfrak{f}, \vals))$ where $m' = m(1/\mathfrak{f}, \vals)$.

Hence, $\vals'(x) = \vals(x)$. Besides, $x \notin \text{Def}(id := e')$ by definition.

{\bf Case 2}.
$S$ = ``$\text{ output }e$".

{\bf Case 3}.
$S$ = ``$\text{ input }id$".

By similar argument in case 1.

{\bf Case 4}.
$S$= ``If $(e)$ then \{$S_t$\} else \{$S_f$\}".

Def $(S) = \text{ Def }(S_f) \cup \text{ Def }(S_t)$ by definition.
Then $x \notin \text{ Def }(S_f) \cup \text{ Def }(S_t)$.

There are two subcases.

\hspace{0.5ex} {\bf Case 4.1}
W.l.o.g., expression $e$ in predicate of $S$ evaluates to nonzero value, written $(e, m) ->* (v, m)$ where $v\neq0$.

Then by rule If-T,
$(\text{If }(v) \text{ then }\{S_t\} \text{ else }\{S_f\}, m(\vals)) -> (S_t, m(\vals))$ where $m' = m$.

Therefore, $\vals'(x) = \vals(x)$.
By argument above, $x \notin \text{ Def }(S_t)$.

\hspace{0.5ex} {\bf Case 4.2}
Evaluation of expression $e$ in predicate of $S$ crashes, written $(e, m) ->* (e', m(1/\mathfrak{f}))$.

Then, by rule crash,

$(\text{If }(e') \text{ then }\{S_t\} \text{ else }\{S_f\}, m(1/\mathfrak{f},\vals)) -> $

$(\text{If }(e') \text{ then }\{S_t\} \text{ else }\{S_f\}, m(1/\mathfrak{f},\vals))$ where $m' = m(1/\mathfrak{f},\vals)$.

Therefore, $\vals'(x) = \vals(x)$.

Besides, $x \notin \text{ Def }(\text{If }(e') \text{ then }\{S_t\} \text{ else }\{S_f\})$ = $\text{ Def }(S_f) \cup \text{ Def }(S_t)$.

{\bf Case 5}.
$S$ = ``$\text{while}_{\langle n\rangle} \; (e) \; \{S'\}$".

Def $(S)$ = Def $(S')$ by definition.
Then $x \notin \text{ Def }(S')$ by condition $x \notin \text{ Def }(S)$.

There are subcases.

\hspace{0.5ex} {\bf Case 5.1}
Expression $e$ evaluates to nonzero value, written $(e, m) ->* (v, m)$ where $v\neq0$.

By rule Wh-T,
$(\text{while}_{\langle n\rangle} \; (v) \; \{S'\}, m(\vals)) -> $

$(S';\text{while}_{\langle n\rangle} \; (e) \; \{S'\}, m(m_c[(k+1)/n]),\vals)$ for some nonnegative integer $k$.

Let $m' = m(m_c[(k+1)/n], \vals)$.
Then $\vals'(x)$ = $\vals(x)$.

Besides, $x \notin \text{ Def }(S';\text{while}_{\langle n\rangle} \; (e) \; \{S'\})$ = $\text{Def }(S') \cup \text{Def }(S)$,
because $x \notin \text{Def }(S')$.

\hspace{0.5ex} {\bf Case 5.2}
Expression $e$ evaluates to zero in state $m$, written $(e, m) ->* (0, m)$.

By rule Wh-F,

$(\text{while}_{\langle n\rangle} \; (0) \; \{S'\}, m(\vals)) -> (\text{skip}, m(m_c[0/n],\vals))$ where $m' = m(m_c[0/n],\vals)$.

Therefore, $\vals'(x) = \vals(x)$.
Besides, $x \notin$ Def (skip).

\hspace{0.5ex} {\bf Case 5.3}
Evaluation of expression $e$ crashes, written $(e, m) ->* (e', m(1/\mathfrak{f}))$.

By rule crash

$(\text{while}_{\langle n\rangle} \; (e') \; \{S'\}, m(1/\mathfrak{f}, \vals)) -> $

$(\text{while}_{\langle n\rangle} \; (e') \; \{S'\}, m(1/\mathfrak{f}, \vals))$ where $m' = m(1/\mathfrak{f}, \vals)$.

Therefore, $\vals'(x) = \vals(x)$.
Besides, $x \notin \text{Def} (\text{while}_{\langle n\rangle} \; (e') \; \{S'\})$  = Def ($S'$) by definition.

{\bf Case 6}.
$S$ = $S_1;S_2$.

By argument in Case 1 to 5, after one step execution
$((S_1, m(\vals)) -> ({S'}, m'(\vals)))$,
$\vals'(x) = \vals(x)$.

By contextual rule, the lemma holds.

%
%
%
%
%
%
\end{proof}

\begin{corollary}\label{coro:defExclusion}
If a variable $x$ is not defined in a statement sequence $S$, then, after an execution of $S$, the value of $x$ is not redefined,
$(x \notin \text{Def}(S) \land (S, m(\vals)) ->* (S', m'(\vals')) =>
(x \notin \text{Def}(S')) \land \vals'(x) = \vals(x))$.
\end{corollary}
\begin{proof}
Let $(S, m)$ {\kStepArrow [k] } $(S', m')$.
The proof is by induction on $k$ using lemma~\ref{lmm:defExclusion}.
\end{proof}

Based on Corollary~\ref{coro:defExclusion}, we extend the result to array variable elements.
\begin{corollary}\label{coro:arrayVarDefExclusionMultiStep}
If an element in an array variable $x[i]$ is not defined in a statement sequence $S$ in a program $P = EN; V; S_{entry}$, then, after an execution of $S$, the value of $x[i]$ is not redefined,
$(x \notin \text{Def}(S)) \wedge
 ((x, i, *) \in \vals) \wedge
 (S, m(\vals)) ->* (S', m'(\vals')) =>
 (x \notin \text{Def}(S')) \wedge \vals'(x, i) = \vals(x, i))$.
\end{corollary}

\begin{lemma}\label{lmm:loopCntRemainsSame}
If all of the following hold:
\begin{enumerate}
\item There is no loop of label $n$ in statements $S$, $``\text{while}_{\langle n\rangle}(e) \{S'\}" \notin S$;

\item The crash flag is not set, $\mathfrak{f} = 0$;

\item There is an entry $n$ in loop counter, $(n, *) \in \text{loop}_c$;

\item There is one step execution, $(S, m(\mathfrak{f}, \text{loop}_c)) -> (S', m'(\text{loop}_c'))$;
\end{enumerate}
then, $\text{loop}_c'(n) = \text{loop}_c(n)$.
\end{lemma}
\begin{proof}
The proof is by induction on abstract syntax of $S$, similar to that for lemma~\ref{lmm:defExclusion}.
\end{proof}

\begin{corollary}\label{coro:loopCntRemainsSame}
If all of the following hold:
\begin{enumerate}
\item There is no loop of label $n$ in statements $S$, $``\text{while}_{\langle n\rangle}(e) \{S'\}" \notin S$;

\item The crash flag is not set, $\mathfrak{f} = 0$;

\item There is an entry $n$ in loop counter, $(n, *) \in \text{loop}_c$;

\item There is multiple steps execution of stack depth $d = 0$, $(S, m(\mathfrak{f}, \text{loop}_c)) ->* (S', m'(\text{loop}_c'))$;
\end{enumerate}
then, $\text{loop}_c'(n) = \text{loop}_c(n)$.
\end{corollary}

\begin{lemma} \label{lmm:stmtExclExecExt}
If all of the following hold:
\begin{enumerate}
\item A non-skip statement $s$ is not in $S$, $(s \neq \text{skip}) \wedge (s \notin S)$;

\item There is one step execution of stack depth $d = 0$, $(S, m) -> (S', m')$,
\end{enumerate}
then, $s \notin S'$.
\end{lemma}

\begin{proof}
By structural induction on abstract syntax of statement sequence $S$, we show the lemma holds.

{\bf Case 1}.
$S$ = ``$id := e$". 

Then there are two subcases.

\hspace{0.5ex} {\bf Case 1.1}
Expression $e$ evaluates to some value $v$, written $(e, m) ->* (v, m)$.

Then, by rule Assign,
$(S, m) -> (\text{skip}, m(\vals[v/id]))$.

Hence, $s \notin \text{skip}$ by definition.

\hspace{0.5ex} {\bf Case 1.2}
Evaluation of expression $e$ crashes, written $(e, m) ->* (e', m(1/\mathfrak{f}))$.

By parameterized type rule TExpr, $\Gamma \not\vdash e'$.
Then, by type rule TAssign, $\Gamma \not\vdash id := e'$.

Then, by rule crash,

$(id := e', m(1/\mathfrak{f})) -> (id := e', m(1/\mathfrak{f}))$.

Because $\Gamma \vdash s$, then $s \neq id := e'$.
Hence, $s \notin id := e'$ by definition.

{\bf Case 2}.
$S$ = ``$\text{ output }e$".

{\bf Case 3}.
$S$ = ``$\text{ input }id$".

By similar argument in case 1.

{\bf Case 4}.
$S$= ``If $(e)$ then \{$S_t$\} else \{$S_f$\}".

$s \notin S_f, s \notin S_t$ by definition.
There are two subcases.

\hspace{0.5ex} {\bf Case 4.1}
W.l.o.g., expression $e$ in predicate of $S$ evaluates to nonzero value, written $(e, m) ->* (v, m)$ where $v\neq0$.

Then by rule If-T,
$(\text{If }(v) \text{ then }\{S_t\} \text{ else }\{S_f\}, m) -> (S_t, m)$.

Therefore, $s \notin S_t$.

\hspace{0.5ex} {\bf Case 4.2}
Evaluation of expression $e$ in predicate of $S$ crashes, written $(e, m) ->* (e', m(1/\mathfrak{f}))$.

Then, by rule crash,

$(\text{If }(e') \text{ then }\{S_t\} \text{ else }\{S_f\}, m(1/\mathfrak{f})) -> $

$(\text{If }(e') \text{ then }\{S_t\} \text{ else }\{S_f\}, m(1/\mathfrak{f}))$.

By parameterized type rule TExpr, $\Gamma \not\vdash e'$.
By type rule Tif, $\Gamma \not\vdash \text{If }(e') \text{ then }\{S_t\} \text{ else }\{S_f\}$.

Because $\Gamma \vdash s$, then $s \neq \text{If }(e') \text{ then }\{S_t\} \text{ else }\{S_f\}$.

Besides, $s\notin S_t, s\notin S_f$ by condition. Therefore, $s \notin \text{If }(e') \text{ then }\{S_t\} \text{ else }\{S_f\}$.

{\bf Case 5}.
$S$ = ``$\text{while}_{\langle n\rangle} \; (e) \; \{S'\}$".

$s \notin S'$ by definition.
There are subcases.

\hspace{0.5ex} {\bf Case 5.1}
Expression $e$ evaluates to nonzero value, written $(e, m) ->* (v, m)$ where $v\neq0$.

By rule Wh-T,
$(\text{while}_{\langle n\rangle} \; (v) \; \{S'\}, m) -> $

$(S';\text{while}_{\langle n\rangle} \; (e) \; \{S'\}, m(m_c[(k+1)/n]))$ for some nonnegative integer $k$.

Then $s \notin S';\text{while}_{\langle n\rangle} \; (e) \; \{S'\}$ by definition.

\hspace{0.5ex} {\bf Case 5.2}
Expression $e$ evaluates to zero in state $m$, written $(e, m) ->* (0, m)$.

By rule Wh-F,

$(\text{while}_{\langle n\rangle} \; (0) \; \{S'\}, m) -> ({skip}, m(m_c[0/n]))$.

Therefore, $s \notin \text{skip}$.

\hspace{0.5ex} {\bf Case 5.3}
Evaluation of expression $e$ crashes, written $(e, m) ->* (e', m(1/\mathfrak{f}))$.

By rule crash,

$(\text{while}_{\langle n\rangle} \; (e') \; \{S'\}, m(1/\mathfrak{f})) -> $

$(\text{while}_{\langle n\rangle} \; (e') \; \{S'\}, m(1/\mathfrak{f}))$.

Then, by type rule Twhile, $\Gamma \not\vdash \text{while}_{\langle n\rangle} \; (e') \; \{S'\}$.
Because $\Gamma \vdash s$, then $s \neq \text{while}_{\langle n\rangle} \; (e') \; \{S'\}$.

Besides $s \notin S'$, then $s \notin \text{while}_{\langle n\rangle} \; (e') \; \{S'\}$ by definition.

{\bf Case 6}.
$S = S_1;S_2$.

By argument in Case 1 to 5, after one step execution
$(S_1, m) -> (S', m')$, $s\notin S'$.

By contextual rule, $(S_1;S_2, m) -> (S';S_2, m')$.

By definition, $s \notin S_2$.

Then, by definition, $s\notin S';S_2$

%
%
%
%
%
%
%
%
\end{proof}
%
%
%
%
%
%

%
\begin{lemma} \label{lmm:loopCntStepwiseInc}
Let $s$ = ``$\text{while}_{\langle n\rangle} \, (e) \, \{S''\}$".
If both of the following hold:
\begin{itemize}
\item $s \in S$;

\item $(S, m(\text{loop}_c)) -> (S', m'(\text{loop}_c'))$;
\end{itemize}
then one of the following holds:
\begin{enumerate}
\item The loop counter of label $n$ is incremented by one, $\text{loop}_c'(n) - \text{loop}_c(n)$ = 1;

\item There is no entry for label $n$ in loop counter, $(n, v) \notin \text{loop}_c'$;

\item The loop counter of label $n$ is not changed, $\text{loop}_c'(n) - \text{loop}_c(n)$ = 0;
\end{enumerate}
\end{lemma}
\begin{proof}
Let $S = s';S''$.
The proof is by induction on abstract syntax of $s'$.
\end{proof}

\end{document}